\documentclass[11pt,a4paper]{article}
\pdfoutput=1

\usepackage{jheppub}

\usepackage{amssymb}
\usepackage{amsfonts}
\usepackage{graphicx}
\usepackage[dvipsnames,table]{xcolor}
\usepackage{epstopdf}
\usepackage{dcolumn}
\usepackage{amsmath}
\usepackage{latexsym,bm}
\usepackage{amsthm}
\usepackage{slashed}
\usepackage{float}
\usepackage{color}
\usepackage{url}
\usepackage{longtable}
\usepackage{subfig}
\usepackage{tikz}
\usepackage[all,cmtip]{xy}

\usepackage{multirow}
\usepackage{longtable}

\usepackage[titletoc]{appendix}
\makeatletter
\newtheorem*{rep@theorem}{\rep@title}
\newcommand{\newreptheorem}[2]{%
\newenvironment{rep#1}[1]{%
 \def\rep@title{#2 \ref{##1}}%
 \begin{rep@theorem}}%
 {\end{rep@theorem}}}
\makeatother
\newtheorem{lemma}{Lemma}[subsection]

\newtheorem{prop}[lemma]{Proposition}

\newreptheorem{conj}{Conjecture}

\theoremstyle{definition}

\newcommand \xoverline[2][0.75]{
    \sbox{\myboxA}{$\m@th#2$}
    \setbox\myboxB\null
    \ht\myboxB=\ht\myboxA
    \dp\myboxB=\dp\myboxA
    \wd\myboxB=#1\wd\myboxA
    \sbox\myboxB{$\m@th\overline{\copy\myboxB}$}
    \setlength\mylenA{\the\wd\myboxA}
    \addtolength\mylenA{-\the\wd\myboxB}
    \ifdim\wd\myboxB<\wd\myboxA
       \rlap{\hskip 0.5\mylenA\usebox\myboxB}{\usebox\myboxA}%
    \else
        \hskip -0.5\mylenA\rlap{\usebox\myboxA}{\hskip 0.5\mylenA\usebox\myboxB}%
    \fi}
\makeatother

\newcommand{\ba}{\begin{aligned}}
\newcommand{\ea}{\end{aligned}}
\def\be{\begin{equation}}
\def\ee{\end{equation}}
\def\bsp{\begin{split}}
\def\esp{\end{split}}
\def\bea{\begin{eqnarray}}
\def\eea{\end{eqnarray}}

\def\mc{\mathcal}

\def\mb{\mathbb}
\def \bp{\begin{pmatrix}}
\def\ep{\end{pmatrix}}

\def\R{\mathbb{R}}
\def\N{\mathcal{N}}
\def\P{\mathbb{P}}
\def\C{\mathbb{C}}
\def\Z{\mathbb{Z}}

\def\Q{\mathbb{Q}}

\def\dim{\mathrm{dim}}

\def\br{\breve}

\usepackage{tikz}
\usepackage{diagbox}
\usetikzlibrary{positioning}
\usetikzlibrary{calc}
\usetikzlibrary{decorations.pathreplacing,calligraphy}

\usepackage{xstring}
\usetikzlibrary{decorations.pathmorphing} 
\usetikzlibrary{decorations.markings} 
\usetikzlibrary{snakes}
\usetikzlibrary{arrows} 
\usetikzlibrary{shapes} 
\usetikzlibrary{matrix} 
\usetikzlibrary{positioning} 
\usepackage[english]{babel} 
\usepackage[autostyle]{csquotes}

\usepackage{tikz}
\usetikzlibrary{arrows}
\usetikzlibrary{arrows.meta}
\usetikzlibrary{shapes.geometric,calc,arrows, positioning,shapes.misc,decorations.markings}
\tikzset{
  big arrow/.style={
    decoration={markings,mark=at position 1 with {\arrow[scale=2,#1]{>}}},
    postaction={decorate},
    shorten >=0.4pt},
  big arrow/.default=black}
  
\pgfdeclarelayer{edgelayer}
\pgfdeclarelayer{nodelayer}
\pgfsetlayers{edgelayer,nodelayer,main} 
\tikzstyle{none}=[inner sep=0pt] 

\tikzstyle{NodeCross}=[draw, shape=circle, cross out, inner sep=0pt, minimum size=6pt,line width=0.25mm]
\tikzstyle{Circle}=[draw, shape=circle, black,  fill=black, inner sep=0pt, minimum size=6pt]
\tikzstyle{Star}=[draw, shape=star, fill=black, star points=8, inner sep=0pt, minimum size=8pt]

\tikzstyle{DashedLine}=[-, densely dashed, line width=0.25mm]
\tikzstyle{DottedLine}=[-, dotted, line width=0.25mm]
\tikzstyle{ThickLine}=[-, line width=0.25mm]
\tikzstyle{ArrowLineRight}=[-, -{Stealth[scale=1.75]}, line width=0.1mm, scale=5]
\tikzstyle{RedLine}=[-, draw={rgb,255: red,191; green,0; blue,0}, fill=none, line width=0.25mm]
\tikzstyle{DottedRed}=[-, dotted, draw={rgb,255: red,191; green,0; blue,0}, fill=none, line width=0.25mm]
\tikzstyle{DashedLineThin}=[-, densely dashed, line width=0.125mm, fill=none, draw=black]
\tikzstyle{ArrowLineRed}=[-, -{Stealth[scale=1.75]}, draw={rgb,255: red,191; green,0; blue,0}, line width=0.1mm, scale=5]
\tikzstyle{brane}=[draw]
\tikzset{D7/.style={circle, draw=black, inner sep=0pt, fill=white, minimum size=3mm}}
\tikzset{hasse/.style={circle, fill,inner sep=2pt}}
\tikzset{flavor/.style={regular polygon,fill=white,regular polygon sides=4,inner sep=2.5pt, draw}}
\tikzset{gauge/.style={circle, draw,inner sep=2.5pt}}
\tikzset{gaugeb/.style={circle, draw,fill=black,inner sep=2.5pt}}
\tikzset{gauger/.style={circle, draw,fill=cyan,inner sep=2.5pt}}
\tikzset{gaugeg/.style={circle, draw,fill=red,inner sep=2.5pt}}
\tikzset{SUd/.style={circle, draw=black, inner sep=0pt, fill=yellow, minimum size=2mm}}
\tikzset{bd/.style={circle, draw=black, inner sep=0pt, fill=black, minimum size=2mm}}
\tikzset{wd/.style={circle, draw=black, inner sep=0pt, fill=white, minimum size=2mm}}
\tikzset{Dynkin/.style={circle, draw=black, inner sep=0pt, fill=white, minimum size=2mm}}
\tikzstyle{ligne}=[draw, thick] 
\tikzset{doublearrow/.style={ draw=black!75, color=black!75, thick, double distance=3pt, }} 

\makeatletter
\def\widebreve#1{\mathop{\vbox{\m@th\ialign{##\crcr\noalign{\kern\p@}%
  \brevefill\crcr\noalign{\kern0.1\p@\nointerlineskip}%
  $\hfil\displaystyle{#1}\hfil$\crcr}}}\limits}

\def\brevefill{$\m@th \setbox\z@\hbox{}%
 \hfill\scalebox{0.7}{\rotatebox[origin=c]{90}{(}} \kern4pt $}
\makeatletter

\usepackage{amsmath}

\usepackage{bbm}

\usepackage{enumitem} 

\usepackage{physics}

\usepackage{chemfig}
\usepackage{import}

\makeatletter
\newcommand\xleftrightarrow[2][]{%
  \ext@arrow 9999{\longleftrightarrowfill@}{#1}{#2}}
\newcommand\longleftrightarrowfill@{%
  \arrowfill@\leftarrow\relbar\rightarrow}
\makeatother

\usepackage{booktabs}

\newcommand{\PROVE}[1]{{\textcolor{red}{PROVE}}}

\usepackage{braket}

\usepackage{mathrsfs}

\newcommand{\mE}{\mathcal{E}}

\newcommand{\ii}{\mathtt{i}}
\newcommand{\vol}{\mathrm{vol}}

\newcommand{\FTfive}{\mathcal{T}^{(\mathrm{5d})}_{X_6}}

\usepackage{amsthm}

	\theoremstyle{definition}
	\newtheorem{defin}[lemma]{Definition}	
	
\graphicspath{ {figs/} }

\title{$(-1)$-form symmetries from M-theory and SymTFTs}

\preprint{USTC-ICTS/PCFT-24-53 \hspace*{0.1in} }

\author[\clubsuit]{Marwan Najjar}
\emailAdd{marwan.najjar@pku.edu.cn}
\author[\diamondsuit]{Leonardo Santilli}
\emailAdd{santilli@tsinghua.edu.cn}
\author[\spadesuit,\clubsuit,\heartsuit]{Yi-Nan Wang}
\emailAdd{ynwang@pku.edu.cn}

\affiliation[\clubsuit]{Center for High Energy Physics, Peking University, \protect\\
Beijing 100871, China}

\affiliation[\diamondsuit]{Yau Mathematical Sciences Center, Tsinghua University, \protect\\ Beijing 100084, China}

\affiliation[\spadesuit]{School of Physics, Peking University, \protect\\
Beijing 100871, China}

\affiliation[\heartsuit]{Peng Huanwu Center for Fundamental Theory, \protect\\
Hefei, Anhui 230026, China}

\abstract{
We explore $(-1)$-form symmetries within the framework of geometric engineering in M-theory. By constructing the Symmetry Topological Field Theory (SymTFT) for selected 5d $\N=1$, 4d $\mc{N}=2$ and 4d $\N=1$ theories, we formalize the geometric origin of these symmetries and compute the mixed anomaly polynomials involving $(-1)$-form and higher-form symmetries. Our findings consistently reveal both discrete and continuous $(-1)$-form symmetries, aligning with established field theory results, while also uncovering new $(-1)$-form symmetry factors and structural insights. In particular, we study the SymTFT of 4d $\mc{N}=1$ theories from M-theory on a class of spaces with $G_2$ holonomy, and obtain properties such as modified instanton sums and 4-group structures observed in other 4d gauge theories. Additionally, we systematically construct symmetry operators for continuous abelian symmetries, refining existing proposals, and providing an M-theory origin for them. }

\begin{document}

\maketitle


\section{Introduction and summary}

Global symmetries play a critical role in analyzing physical systems, particularly in quantum field theory (QFT). Alongside their 't Hooft anomalies, they provide properties of a theory that are intrinsic and invariant under renormalization group (RG) flow \cite{tHooft:1979rat,Callan:1984sa}. Symmetries organize the spectrum into representations, and impose constraints on correlation functions, which are especially useful in strongly coupled theories. Importantly, as a theory evolves along the RG flow, all operators allowed by the global symmetries may emerge through quantum effects. As a result, a comprehensive classification of a model's global symmetries is instrumental in managing the RG flow, allowing the definition of order parameters that describe the infrared (IR) or long-distance behaviour of specific ultraviolet (UV) theories.\par
However, in a wealth of models, conventional symmetry concepts fall short of fully capturing their IR characteristics. A notable example arises in Yang--Mills theories, where confining and deconfining phases cannot be satisfactorily explained by traditional patterns of symmetry preservation or breaking.\par

A crucial advancement in addressing this limitation was introduced in \cite{Aharony:2013hda, Gaiotto:2014kfa}, where the concept of global symmetries was expanded to include topological operators. Within this extended framework, the conventional notion of global symmetry corresponds to topological operators of codimension-one in spacetime, referred to as 0-form symmetries. Importantly, this reformulation enables the inclusion of topological operators supported on codimension other than one. This concept, termed higher-form symmetry, generalizes global symmetries to $p$-form symmetries, where the associated topological operators are supported on codimension-$(p+1)$, i.e., on $(d-p-1)$-dimensional submanifolds. 

With this expanded notion, confining and deconfining phases of non-abelian gauge theories can be characterized by electric and magnetic 1-form symmetries, where the relevant order parameters are one-dimensional defect operators, commonly known as Wilson and ’t Hooft lines \cite{Wilson:1974sk,tHooft:1977nqb}.

The discussion above is naturally embedded within the broader context of superstring theory and M-theory, particularly through approaches such as the AdS/CFT correspondence and the geometric engineering program. In this work, we focus on the latter approach, which we describe below.

\paragraph{M-theory, geometric engineering, and SymTFT.}

M-theory establishes a fundamental framework for exploring the vast landscape of supersymmetric quantum field theories (SQFTs). A highly effective approach to construct SQFTs is the geometric engineering program \cite{Katz:1996fh}, by placing M-theory on a product manifold $M_{d} \times X_{11-d}$. This program establishes a dictionary between the topology and geometry of compactification spaces and effective lower-dimensional QFTs.\par

By requiring that the internal space has $SU(n)$ or $G_2$ holonomy ensures that the resulting QFT in $d$-dimensions is supersymmetric. This relationship can be schematically represented as,
\begin{equation}
    X_{11-d} \ \xrightarrow{ \,\,\, \text{geometric engineering}\,\,\,} \ \mathcal{T}^{\,(d)}_{X_{11-d}} \, \in \, \text{SQFT}_d\,,
\end{equation}
where $X_{11-d}$ admits a singular limit, in which it is described as the metric cone over the so-called link $L_{10-d}$ of the singularity, 
\begin{equation}\label{eq:X=ConeL}
    X_{11-d}= \text{Cone} (L_{10-d})\,.
\end{equation}\par
To gain a broader view of the extensive work in the field of geometric engineering, we refer to the following partial lists of works discussing M-theory on $G_2$-manifolds \cite{Atiyah:2000zz,Acharya:2000gb,Acharya:2001dz, Atiyah:2001qf,Acharya:2001hq,Beasley:2002db,Berglund:2002hw,Acharya:2004qe,Anderson:2006mv,Halverson:2014tya,Halverson:2015vta, Braun:2018fdp,Kennon:2018eqg,Braun:2018vhk,Acharya:2020vmg,DelZotto:2021ydd,Braun:2023fqa} and Calabi--Yau threefolds \cite{Intriligator:1997pq,Esole:2014bka,Esole:2014hya,DelZotto:2017pti,Xie:2017pfl,Esole:2017hlw,Closset:2018bjz,Jefferson:2018irk,Apruzzi:2019opn,Apruzzi:2019enx,Saxena:2020ltf,Apruzzi:2019kgb,Collinucci:2020jqd,Closset:2020scj,Eckhard_2020,Acharya:2021jsp,Tian:2021cif,Closset:2021lwy,DeMarco:2022dgh,Mu:2023uws,DeMarco:2023irn,Alexeev:2024bko}.

The generalized global symmetries and the associated anomalies of a given QFT $\mathcal{T}^{\,(d)}$ are captured by a bulk $(d+1)$-dimensional topological theory $\widehat{\mathcal{T}}^{(d+1)}$, supported on $Y_{d+1}\cong [0, \infty)\times M_{d}$,\footnote{In the literature $Y_{d+1}$ is usually taken as $[0, 1]\times M_{d}$, here we set the length of the interval to be infinite to achieve a better match with the geometry. This does not affect the topological data.} commonly referred to as the Symmetry Topological Field Theory (SymTFT). The physical theory is located at the boundary $\{0\}\times M_{d}$.\par

In geometric engineering, the SymTFT $\widehat{\mathcal{T}}_{L_{10-d}}^{\,(d+1)}$ is naturally derived from the link $L_{10-d}$, with the radial direction inherited from \eqref{eq:X=ConeL},
\begin{equation}
    L_{10-d} \ \xrightarrow{ \,\,\, \text{ geometric engineering }\,\,\, }\  \widehat{\mathcal{T}}^{\,(d+1)}_{L_{10-d}} \, \in \, \text{SymTFT}_{d+1}\,.
\end{equation}\par
Equivalently, we may describe $X_{11-d}$ as an $L_{10-d}$-fibration,
\begin{equation}\label{eq:LfibreX}
	L_{10-d} \ \hookrightarrow \ X_{11-d} \ \twoheadrightarrow \ [0,\infty) \,.
\end{equation}
Integrating over the fibre, one is left with the physical theory $\mathcal{T}^{(d)}_{X_{11-d}}$ localized at the tip of the cone $M_d \times \{0 \}$, coupled to a Symmetry Theory living on $Y_{d+1}$. The boundary conditions at infinity for the M-theory fields induce topological boundary conditions for the field content of the SymTFT, and different choices of boundary conditions at infinity select distinct global forms of $\mathcal{T}^{(d)}_{X_{11-d}}$ \cite{Witten:1998wy,GarciaEtxebarria:2019caf}. The setup is illustrated in Figure \ref{fig:sandwich}.\par

The top-down approach to deriving the SymTFT from geometric engineering has been successfully deployed in recent years \cite{Apruzzi:2021nmk,Hubner:2022kxr,DelZotto:2022joo,Apruzzi:2022rei,vanBeest:2022fss,Heckman:2022xgu,Apruzzi:2023uma,Baume:2023kkf,DelZotto:2024tae,GarciaEtxebarria:2024fuk,Franco:2024mxa,Cvetic:2024dzu,Tian:2024dgl,Cvetic:2024mtt,Gagliano:2024off}.\footnote{Other types of geometric compactifications have been used to deduce the SymTFT in \cite{Gukov:2020btk,Bashmakov:2022jtl,Antinucci:2022cdi,Chen:2023qnv,Bashmakov:2023kwo,Cui:2024cav,Chen:2024fno}.} Our endeavour continues this line of research, with the goal of systematically classifying and studying $(-1)$-form symmetries and their anomalies, as explained hereafter.

\begin{figure}[H]
\centering
\includegraphics[width=0.75\textwidth]{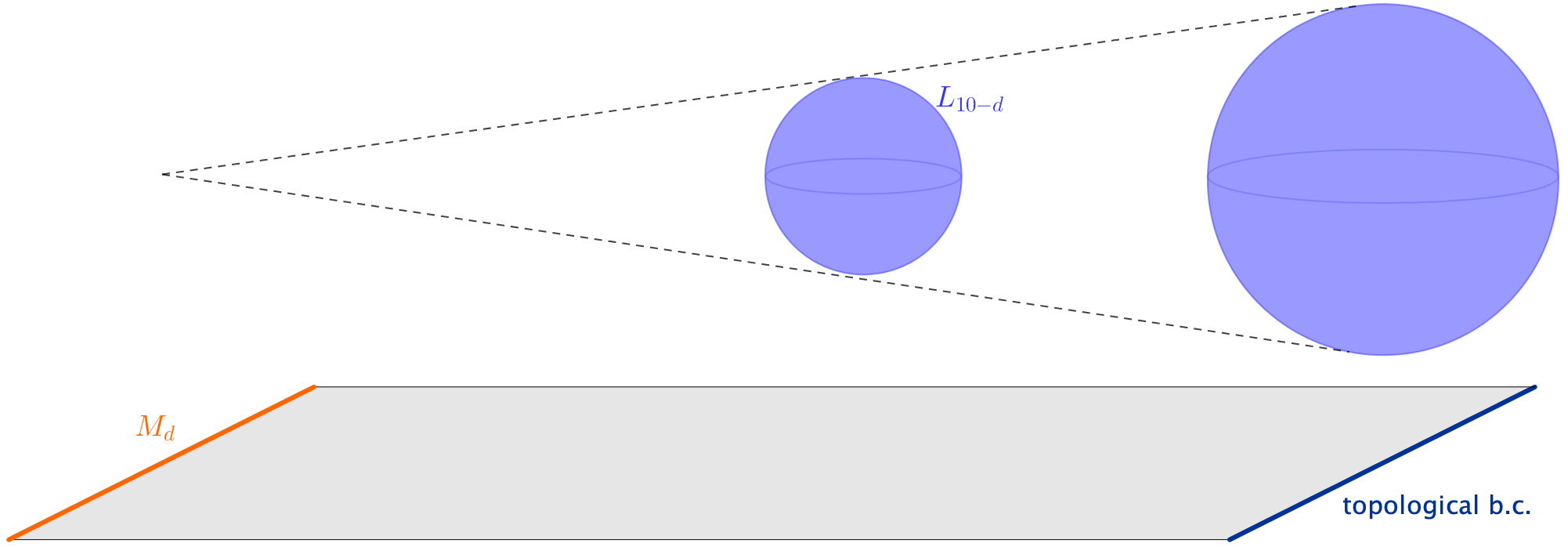}
\caption{SymTFT from geometric engineering.}
\label{fig:sandwich}
\end{figure}\par

\paragraph{\texorpdfstring{$(-1)$}{(-1)}-form symmetries.}

The idea that certain transformations in the space of couplings are better described in the language of symmetries was pioneered in \cite{Cordova:2019jnf}. In a shift of perspective, the topological terms in the action are thought of as the insertion of spacetime-filling topological operators in the path integral. In this way, the corresponding couplings are reinterpreted as scalar background gauge fields, for the symmetry generated by these topological operators. This will be henceforth referred to as `$(-1)$-form symmetry'.\par
While not symmetries in the canonical sense, one gets a lot of mileage out of treating $(-1)$-form symmetries as such. Their anomalies and dynamical applications have been explored in \cite{Cordova:2019jnf,Cordova:2019uob,Brennan:2020ehu,Damia:2022seq,Aloni:2024jpb,Garcia-Valdecasas:2024cqn,Brennan:2024tlw}, while implications for holography and the swampland program have been considered in \cite{McNamara:2020uza,Heckman:2024oot}.\par
The formalism to treat $(-1)$-form symmetries on equal footing as higher-form symmetries was put on firm grounds in \cite{Yu:2020twi,Santilli:2024dyz}. One lesson from \cite{Santilli:2024dyz} is that one should distinguish between two classes of $(-1)$-form symmetries.
\begin{itemize}
    \item \textit{Chern--Weil} symmetries are continuous abelian $(-1)$-form symmetries, whose background fields serve as couplings for characteristic classes in the theory. The prototypical example is the Chern--Weil symmetry in 4d $SU(N)$ Yang--Mills theories: the $\theta_{\text{YM}}$-angle accompanying the second Chern class $\frac{1}{8 \pi^2} F \wedge F$ of the gauge bundle is reinterpreted as a gauge field for this symmetry, and the invariance under $\theta_{\text{YM}} \mapsto \theta_{\text{YM}} + 2 \pi \Z$ is reinterpreted as invariance under large background gauge transformations.
    \item \textit{Finite} $(-1)$-form symmetries are electric/magnetic dual to finite $(d-1)$-form symmetries. Consider a parent theory with topological local operators, generating a $(d-1)$-form symmetry acting on domain walls. The various ways of gauging the latter are classified by discrete $\theta$-angles \cite{Vafa:1986wx}. The dual symmetry emerging in the orbifold theory is a finite $(-1)$-form symmetry, and the discrete $\theta$-angle is identified with the discrete background gauge field for this $(-1)$-form symmetry. Background gauge transformations permute the different orbifolds. This latter effect is what makes $(-1)$-form symmetries pertinent in the context of decomposition \cite{Hellerman:2006zs,Sharpe:2014tca,Sharpe:2019ddn,Robbins:2020msp,Sharpe:2021srf,Pantev:2022kpl,Sharpe:2022ene}.
\end{itemize}\par
It is well understood that, to fully specify a theory, one needs to choose a polarization for the mutually non-local defects \cite{Aharony:2013hda}, for each pair of $p$-form/$(d-p-2)$-form symmetries. The M-theory uplift of this statement was formulated in \cite{DelZotto:2015isa,Albertini:2020mdx}. However, the choice of polarization for $p=-1$ has been mostly overlooked in the literature.\footnote{During the finishing of the draft, the authors noticed the appearance of \cite{Yu:2024jtk}, nonetheless the overlap is minimal. In particular, that work does not apply to M-theory.}

\subsection{Summary of results}

\paragraph{BF terms from differential cohomology.}
We uncover a prescription to obtain the BF terms from M-theory, for both discrete and continuous symmetries. Our result systematically builds on \cite{Apruzzi:2023uma} for finite symmetries, and extends the construction to continuous symmetries.\par
The derivation consists of refining the kinetic term $G_4 \wedge G_7$ in the M-theory action to a differential character. The precise way to do so is spelled out in Subsection \ref{sec:MtheoryupliftH12}. Then, from the reduction of this character along the various homology classes in the link, we find all the BF terms:
\begin{itemize}
    \item Reducing both pieces along $n$-torsional cycles in the link, we obtain the standard BF terms of the form $B_{p+1} \smile \delta A_{d-p-1}$, where $B_{p+1}$ is the background gauge field for a finite abelian $p$-form symmetry $\Z_n^{\scriptscriptstyle [p]}$, and $A_{d-p-1}$ is the background gauge field for the magnetic dual symmetry. 
    \item Reducing both pieces along free cycles, we obtain  BF-like terms of the form $F_{p+2} \wedge h_{d-p-1}$, where $F_{p+2}$ is the curvature of a $U(1)^{\scriptscriptstyle [p]}$ $p$-form symmetry and $h_{d-p-1}$ is a quantized but not closed field. Additionally, from the Chern--Simons term in M-theory, we derive a correction which combines with the BF-like term to give $F_{p+2} \wedge \widetilde{h}_{d-p-1}$, where now $ \widetilde{h}_{d-p-1}$ is closed and with periods in $2 \pi \Z$.
\end{itemize}
In other words, despite the different-looking form of the BF couplings for discrete and continuous symmetries, we are able to deduce both of them at once, by proposing a refinement of the kinetic action of M-theory to differential cohomology, and combining with the geometric engineering dictionary.\par
The BF-like term we derive is consistent with the existing proposals for $U(1)$ symmetries \cite{Brennan:2024fgj,Antinucci:2024zjp,Apruzzi:2024htg}. In particular, our derivation provides a top-down perspective on the BF-like terms of \cite{Apruzzi:2024htg}. Along the way, in Subsection \ref{sec:braneSymTFT} we deal with subtle issues to extend the prescriptions to M-theory compactifications.

\paragraph{Geometric engineering and $(-1)$-form symmetry.}

M-theory and superstring theory branes provide natural mechanisms for constructing charged defects and topological symmetry operators acting on them. Symmetry topological operators are obtained through one of the following scenarios: 
\begin{itemize}
    \item Finite symmetries are generated by BPS branes that wrap torsional cycles in $L_{10-d}$ \cite{Heckman:2022muc};
    \item Continuous symmetries are generated by non-BPS fluxbranes that wrap free cycles in $L_{10-d}$ \cite{Cvetic:2023plv,Bergman:2024aly}.
\end{itemize}
These constructions are reviewed in detail in Subsection \ref{sec:braneSymTFT}, where we further clarify and build upon the existing literature to deal with continuous abelian symmetries in M-theory. One crucial development in Subsection \ref{sec:braneSymTFT} (further elaborated upon in Appendix \ref{app:MthwithM5}) is our discussion on the role of $P_7$-fluxbranes. These are certain non-BPS branes with seven-dimensional worldvolume capturing the holonomy of $P_7 $, which can be defined to be the projection of $G_7 + \frac{1}{2} C_3 \wedge G_4$ on the seven-dimensional manifold. Accounting for such a combination of $G_7$ and $G_4$ is essential to ensure that the resulting defect is topological.\par
These considerations apply to $(-1)$-form symmetries as well. In M-theory, 
\begin{itemize}
    \item Finite $(-1)$-form symmetries are generated by M5-branes that fill the spacetime $M_d$ and wrap torsional $(6-d)$-cycles in $L_{10-d}$;
    \item Continuous $(-1)$-form symmetries are generated by $P_7$-fluxbranes that fill the spacetime $M_d$ and wrap free $(7-d)$-cycles in $L_{10-d}$.
\end{itemize}
This leads us to conjecture that, for QFTs that can be realized via geometric engineering in M-theory, finite $(-1)$-form symmetries only exist in $d \le 5$, while continuous ones exist in $d \le 7$.\par
Using Poincar\'e duality, we are led to predict the presence of $(-1)$-form symmetries in the SymTFT based on the cohomology of the link:
\begin{itemize}
    \item Finite $(-1)$-form symmetries if $\mathrm{Tor} H^4 (L_{10-d}, \Z)$ supports non-trivial classes; the discrete 0-form gauge field is obtained by expanding the differential cohomology refinement of $G_4$ along such torsion 4-cocycles;
    \item Continuous $(-1)$-form symmetries if $H^3 (L_{10-d},\Z)_{\text{free}}$ supports non-trivial classes; the 1-form curvature is obtained by expanding the differential cohomology refinement of $G_4$ along such free 3-cocycles.
\end{itemize}\par
The derivation is spelled out in full generality in Section \ref{sec:Mgeneral}. Our study of finite $(-1)$-form symmetries from torsion cohomology yields a vast generalization and systematization of the early work \cite{Sharpe:2000qt}.\par
For 5d theories, these conditions are satisfied, for instance, by the link $Y^{N,0}$ of Calabi--Yau threefolds obtained as orbifolds of the conifold. A more detailed discussion of this example can be found in Section \ref{sec:5d}. For 4d theories, one suitable example is the nearly-K\"ahler manifold $\mathbb{S}^{3}\times \mathbb{S}^{3}$, which serves as the link of the spin bundle over the 3-sphere \cite{10.1215/S0012-7094-89-05839-0}. Another relevant case is the direct product space $\mathbb{S}^{1}\times Y^{N,0}$. Both of these examples are studied in detail in Section \ref{sec:4d-SymTFT}.

\paragraph{$(-1)$-form symmetry of 5d SCFTs from threefold singularities.} 
M-theory on a canonical threefold singularity engineers a five-dimensional SCFT, whose global symmetries have been studied using combinations of geometric and field theoretic methods \cite{Albertini:2020mdx,Morrison:2020ool,Bhardwaj:2020phs,Bhardwaj:2020ruf,Bhardwaj:2020avz,Apruzzi:2021vcu,Genolini:2022mpi}. In Subsection \ref{sec:SymTFT5DCY3}, building on \cite{Apruzzi:2021nmk}, we systematically address the SymTFT of these theories. We complement the existing literature in two ways:
\begin{itemize}
    \item The SymTFT action we derive contains the BF terms for both discrete and continuous symmetries. 
    \item We pinpoint the role of $(-1)$-form symmetries. 
\end{itemize}
Objects charged under the electric 1-form symmetry are engineered from M2-branes wrapping non-compact primitive curves in the threefold. We observe that it is always possible to wrap an M5-brane on the same non-compact curve, giving rise to a domain wall charged under a 4-form symmetry.\par
We argue that, gauging a subgroup of the 4-form symmetry, the theory acquires new discrete parameters, acted on by a finite $(-1)$-from symmetry. The geometric origin of these defects and symmetries is provided in Subsections \ref{sec:5dDefectM2M5}-\ref{sec:5dminus1geometry}. \par
Furthermore, we exemplify the computations in two explicit instances: orbifolds of the conifold in Subsection \ref{sec:SymTFT5dSUN0}, and orbifolds of $\C^3$ in Subsection \ref{sec:C3orbi}.

\paragraph{$G_2$-manifolds, $(-1)$-form symmetry and 4-groups.} 

Modelling M-theory on non-compact seven-dimensional manifolds with $G_2$-holonomy results in effective four-dimensional field theory with 4 real supercharges. Interactions by means of super-Yang--Mills (SYM) theories are engineered through codimension-4 singularities \cite{Acharya:2000gb,Acharya:2004qe}. The resulting SYM theories feature ADE-type gauge algebras that align with the ADE classification of the surface singularities.\par

In Subsection \ref{sec:MonG2SymTFT}, we investigate the SymTFT that is associated with specific 4d $\N=1$ $\mathfrak{su}(N)$ SYM theories constructed in \cite{Acharya:2020vmg}. The gauge theory arises through a quotient of the asymptotically locally conical B7-space introduced in \cite{Brandhuber:2001yi}. A detailed review of this construction is provided in Subsection \ref{sec:GE4dN=1}.\par
Specifically, the $\mathfrak{su}(N)$ SYM theories occupy a semi-classical branch of a broad moduli space, with all theories on this moduli branch sharing the same SymTFT. Consequently, the physical interpretation of the SymTFT symmetries and the non-trivially acting symmetries vary with the branch of the moduli space.\par

In this work, we  focus on the branch characterized by the seven-dimensional orbifold 
\begin{equation}
    X_7 = (\R^{4}/\Z_{N}\times \mathbb{S}^{3})/\Z_{p}\, .
\end{equation}
Through a detailed SymTFT analysis, in Subsection \ref{sec:MonG2SymTFT} we highlight the following key features: 
\begin{itemize} 
    \item We identify the subgroup of the electric 1-form symmetry read off from geometry that acts effectively on the QFT.\par
    The maximal 1-form symmetry from geometric engineering is $\Z_{pN}$. However, the $\mathfrak{su}(N)$ gauge algebra arises from the action of $\Z_{N}$, with $\Z_{p}$ acting freely. Since the electric 1-form symmetry depends on the actual subgroup producing the codimension-4 singularity, the electric 1-form symmetry in an electric polarization is indeed $\Z_{N}$, as anticipated by gauge theoretic arguments \cite{Aharony:2013hda}. Similarly, the argument applies to the dual magnetic 1-form symmetry. 
    \item We observe the existence of a universal, continuous $2$-form symmetry. 
    \item We uncover the existence of two continuous $(-1)$-form symmetries. In M-theory, these symmetries are due to the link space being an orbifold of $\mathbb{S}^{3}\times \mathbb{S}^{3}$. On the branch of interest, one of these symmetries is the expected Chern--Weil shifting the $\theta_{\text{YM}}$-angle. We thus discover a new Chern--Weil symmetry, stemming from the higher-dimensional origin of the 4d theory.
    \item Furthermore, we show the existence of discrete $(-1)$-form symmetries. Our model exhibits a discrete $\Z_{p}$ $(-1)$-form symmetry, alongside the magnetic dual $3$-form symmetry. They provide additional discrete $\theta$-angles, whose physical implication is further discussed in Subsection \ref{sec:application}.
\end{itemize}

Simultaneous gauging of multiple symmetries may lead to higher-group structures; see e.g. \cite{Cordova:2018cvg, Benini:2018reh,Hidaka:2020izy,Hidaka:2021kkf,Bhardwaj:2022scy,Copetti:2023mcq,Kang:2023uvm,Liu:2024znj}. Our analysis reveals that this model exhibits a 4-group structure, which emerges when the electric $\Z_{N}$ 1-form symmetry, the discrete $\Z_{p}$ 3-form symmetry, and a discrete subgroup of the continuous 2-form symmetry are gauged concurrently. Notably, we identify the same 4-group structure found in \cite{Tanizaki:2019rbk}.\par
Additionally, we show that gauging only a discrete subgroup of the 2-form symmetry imposes a constraint on the instanton number of the gauge theory. This reproduces the modified instanton sum results presented in \cite{Tanizaki:2019rbk}, building on earlier discussions in \cite{Seiberg:2010qd}.\par
In this way, we establish a geometric engineering picture for the mechanism of \cite{Tanizaki:2019rbk}.

\section{SymTFT from M-theory}
\label{sec:Mgeneral}

In our conventions, the integrand of M-theory is normalized to be $\exp{ 2 \pi \ii  \left( S^{\text{M}}_{\text{CS}} + S^{\text{M}}_{\text{kin}}\right)}$.

\subsection{Differential cohomology and M-theory}

\subsubsection{Basics of differential cohomology}
\label{sec:diffcoho}
It has long been understood that the proper language to formalize the M-theory action is (generalized) differential cohomology \cite{Hopkins:2002rd}. In this subsection we briefly provide the necessary tools in the theory of differential characters of Cheeger--Simons \cite{Cheeger:1985}. We collect further details in Appendix \ref{app:diffchar}, and, for a more comprehensive exposition, we refer to \cite{Bar:2014}, and to the reviews aimed at physicists in \cite[Sec.2]{Freed:2006yc} and \cite{Szabo:2012hc}.\par
A useful slogan is that $\br{H}^{p+2} (M)$ classifies isomorphism classes of $p$-gerbes with $(p+1)$-connection on $M$, for $p \in \Z_{\geq -2}$.

\paragraph{Differential characters.}
A differential character of degree $p$ on $M$ is a certain homomorphism $Z_{p-1} (M) \longrightarrow U(1) $ from the abelian group of $(p-1)$-chain to $U(1)$. Differential characters of degree $p$ form a group, called the Cheeger--Simons group $\br{H}^{p} (M)$. The product operation
\begin{equation}
    \star \ : \ \br{H}^p (M) \otimes \br{H}^q (M) \longrightarrow \br{H}^{p+q} (M) .
\end{equation}
endows $\br{H}^{p} (M)$ with the structure of a graded ring. More details are in Appendix \ref{app:diffchar}.\par
The differential characters come equipped with two associated maps:
\begin{align}
    \mathscr{F} \ : \ & \br{H}^{p} (M) \longrightarrow \Omega^{p}_{\Z} (M) , \label{diffcharFmap} \\
    c \ : \ & \br{H}^{p} (M) \longrightarrow H^{p} (M,\Z) , \label{diffcharCmap}
\end{align}
called, respectively, field strength map and characteristic class map.\par
We adopt the notation of \cite{Apruzzi:2021nmk} and write $\br{F}_p$ for the differential character with field strength 
\begin{equation}\label{eq:Fmapdiffcoho}
    \mathscr{F} (\br{F}_p) = \frac{1}{2\pi} F_p \in \Omega_{\mathrm{closed}}^p (M). 
\end{equation}
Besides, in our conventions, $\mathscr{F} (\br{F}_p), c (\br{F}_p)$ have integral periods, and the discrete gauge fields have $\Z$-periodicity.\par

\paragraph{Integration.}
There exists an integration map
\begin{equation}\label{eq:diffcohoint}
    \int^{\br{H}}_{M} \ : \ \br{H}^{\dim(M)+1} (M) \longrightarrow \R/\Z .
\end{equation}
The holonomy of $\br{F}_p \in \br{H}^{p} (M)$ along a $(p-1)$-chain $\Sigma_{p-1}$ is (the exponential of)
\begin{equation}
    2 \pi \ii \int^{\br{H}}_{\Sigma_{p-1}}\br{F}_p  \in 2 \pi \ii \R/\Z .
\end{equation}

\subsubsection{M-theory action in differential cohomology}
\label{sec:MtheoryupliftH12}
We ought to consider the uplift to differential cohomology of the M-theory action. 

\paragraph{Topological term in M-theory action.}
We define the topological Chern--Simons integrand of M-theory to be the holonomy of the differential character 
\begin{equation}
    - \frac{1}{6} \br{G}_4 \star \br{G}_4 \star \br{G}_4 -  \br{G}_4 \star \br{\mathsf{X}}_8 \ \in \br{H}^{12} (M_{11}) 
\end{equation}
along the fundamental cycle $[M_{11}]$. Explicitly\footnote{Throughout the paper the normalization of the action is chosen so that the integrand of the partition function is $\sim e^{2\pi \ii S}$.}:
\begin{equation}\label{eq:MtheorySpre}
    2 \pi \ii S^{\text{M}}_{\text{CS}} :=  2 \pi \ii \int_{M_{11}}^{\br{H}} \left[- \frac{1}{6} \br{G}_4 \star \br{G}_4 \star \br{G}_4 -  \br{G}_4 \star \br{\mathsf{X}}_8 \right] .
\end{equation}
Here $G_4$ is the curvature of the M-theory 3-form gauge field $C_3$ and 
\begin{equation}
    \mathsf{X}_8 = \frac{1}{192} \left[ \mathsf{p}_1 (T M_{11} )^{\wedge 2} -4 \mathsf{p}_2 (T M_{11} )\right]
\end{equation}
with $\mathsf{p}_i  (T M_{11} )$ the $i^{\text{th}}$ Pontryagin class of the tangent bundle to $M_{11}$.\par
Throughout we will mainly focus on 
\begin{equation}\label{eq:MtheoryS}
    S^{\text{M}}_{\text{CS}} =  - \frac{1}{6} \int_{M_{11}}^{\br{H}} \br{G}_4 \star \br{G}_4 \star \br{G}_4 .
\end{equation}

\paragraph{Kinetic term in M-theory action.}
In the democratic formalism (see Appendix \ref{app:MthwithM5} for a review), the bosonic part of the kinetic term for the 3-form gauge field in M-theory is written as 
\begin{equation}\label{eq:SkinM0}
   S^{\text{M}}_{\text{kin}} = \frac{1}{(2\pi)^2} \int_{M_{11}} G_4 \wedge G_7 ,
\end{equation}
where the globally defined $G_7$ is subject to the modified Bianchi identity 
\begin{equation}\label{modifiedBIG7}
	\dd G_7 = - \frac{1}{2 (2\pi)} G_4 \wedge G_4 .
\end{equation}
The factors of $2\pi$ in \eqref{eq:SkinM0} and \eqref{modifiedBIG7} are often reabsorbed in a field redefinition, e.g. in \cite[App.A2]{Apruzzi:2023uma}. Throughout, we stick to the standard conventions that all field strengths have periods in $2 \pi \Z$.\par

To refine \eqref{eq:SkinM0} in differential cohomology, we exploit that $\dd G_7 \ne 0$ defines a non-trivial, closed 8-form (in particular, it is exact) to introduce a differential character $\br{\left( \dd G_7 \right)} \in \br{H}^{8} (M_{11})$.\par
We define the uplift to differential cohomology of the M-theory kinetic term as the holonomy of $\br{G}_4 \star \br{\left( \dd G_7 \right)} \in \br{H}^{12} (M_{11})$: 
\begin{equation}\label{eq:MtheoryKin}
    S^{\text{M}}_{\text{kin}} = \int_{M_{11}}^{\br{H}}\br{G}_4 \star \br{\dd G}_7  ,
\end{equation}\par
Promoting $G_7$ to a differential character in $\br{H}^{8} (M_{11})$ is necessary to derive the BF terms for discrete symmetries, which descend from torsional cycles in the geometry, and in addition, it allows us to deal with both continuous and discrete symmetries at once.\par

\paragraph{SymTFT from geometric engineering.}
As explained in the introduction, we take $M_{11} = X_{11-d} \times M_d$, with $X_{11-d}$ a conical singularity whose link is $L_{10-d}$.\par 
Inspired by \cite{Apruzzi:2021nmk,vanBeest:2022fss,GarciaEtxebarria:2024fuk}, we consider the differential cohomology refinement of the (bosonic part of the) M-theory action, and reduce it along $L_{10-d}$. The reduction is achieved via fibre-wise integration, cf. \eqref{eq:LfibreX}.\par
In this way, we derive the complete action of the SymTFT:
\begin{itemize}
    \item Reducing \eqref{eq:MtheoryS} we derive the anomaly terms in the SymTFT action;
    \item Reducing \eqref{eq:MtheoryKin} we derive the BF terms in the SymTFT action.
\end{itemize}\par
Our main focus is on $X_6$ a canonical threefold singularity, and $X_7$ a $G_2$-manifold.

\subsection{Branes, charged defects, and symmetry operators}
\label{sec:braneSymTFT}

In this section, we discuss the brane construction of the spacetime defects and symmetry operators for both discrete and continuous symmetries, following \cite{Heckman:2022muc,Cvetic:2023plv}.\par

\paragraph{Charged defects.}
The defects of interest here are non-dynamical extended objects, and therefore, the brane configurations that generate them must wrap non-compact cycles in the resolved or deformed cone space $\widetilde{X}_{11-d}$.\par
The set of spacetime supersymmetric $m$-dimensional defects, denoted by $\mathbb{D}^{m}$, consists of BPS branes wrapping cycles of the link space $L_{10-d}$ and extending along the radial direction of $ \widetilde{X}_{11-d}$. Specifically \cite{DelZotto:2015isa,Albertini:2020mdx}, 
\begin{equation}\label{def:defect1}
    \mathbb{D}^{m} := \bigcup_{p=2,5} \{\text{M$p$-branes on } H_{p-m}(L_{10-d},\Z) \times [0, \infty) \} .
\end{equation}\par

\paragraph{Symmetry topological operators.}
The elements of \eqref{def:defect1} are natural objects to be charged under higher-form symmetries, both discrete and continuous. These symmetries are generated by $(d-m-1)$-dimensional topological operators \cite{Gaiotto:2014kfa}.\par
In geometric engineering, we construct these topological operators through (not necessarily BPS) branes wrapping cycles in $L_{10-d}$, and extending transversely to the radial direction of the cone \cite{Heckman:2022muc,Cvetic:2023plv}. Specifically, the set of all such $(m'+1)$-dimensional operators is 
\begin{equation}\label{def:symmetryoperator1}
    \mathbb{U}^{m'+1} =\bigcup_{p}  \{\text{$p$-branes wrapping $H_{p-m'}(L_{10-d},\Z)$ and transverse to $[0,\infty)$}\}\,.
\end{equation}
A given symmetry operator $\mathcal{U}(\Sigma_{m'+1})\in\mathbb{U}^{m'+1}$ is determined by
\begin{itemize}
    \item Brane type: whether the brane is BPS or not. In M-theory, BPS branes are the usual M2 and M5-branes, while the non-BPS branes are fluxbranes in the sense of \cite{Cvetic:2023plv}.
    \item Cycle type: whether the brane wraps a torsional cycle $\gamma_{p-m'} \in \mathrm{Tor}H_{p-m'}(L_{10-d},\Z)$, or a free cycle $\gamma_{p-m'} \in H_{p-m'}(L_{10-d},\Z)_{\text{free}}$.
\end{itemize}\par
For an $m$-dimensional defect in \eqref{def:defect1} to carry charge under $\mathcal{U}(\Sigma_{m'+1})$, their linking pairing must be non-trivial both in $L_{10-d}$ and in spacetime. For the spacetime part, it is non-trivial only if $m'+m=d-2$ \cite{Gaiotto:2014kfa}.\par
We now proceed to discuss the conditions in the link space, and construct the symmetry operators using the framework of differential cohomology. For further discussion on the concept of linking pairing and how it relates to our approach, we refer to Appendix \ref{sec:DelZottoLink}.

\subsubsection{Discrete symmetries from M-branes}

\begin{table}[th]
\centering{\small 
\begin{tabular}{|l|c|c|}
\hline 
& M2 & M5 \\
\hline
Tor$H_{k}(L_{10-d},\Z)\times [0,\infty)$  & \cellcolor{yellow!20}Electric $(2-k)$d defect  & \cellcolor{orange!20}Magnetic $(5-k)$d defect \\
\hline
Tor$H_{9-d-k}(L_{10-d},\Z)$ & \cellcolor{orange!20}$(5-k)$-form sym. generator & \cellcolor{yellow!20}$(2-k)$-form sym. generator \\
\hline
\end{tabular}}
\caption{Branes wrapping torsional cycles give rise to finite symmetries.}
\label{tab:M2M5defect}
\end{table}

We summarize the discussion of defects and symmetry operators for finite symmetries in Table \ref{tab:M2M5defect}.

\paragraph{Branes wrapping torsional cycles.} 

Finite symmetries are generated by the sub-collection of \eqref{def:symmetryoperator1} with M$p$-branes wrapping $\gamma_k \in \mathrm{Tor}H_{k}(L_{10-d},\Z)$.\par
The differential cohomology integral \eqref{eq:diffcohoint} induces a perfect pairing
\begin{equation}\label{eq:pairingTorTor}
        \mathrm{Tor}H^{k}(L_{10-d},\Z) \times \mathrm{Tor}H^{11-k-d}(L_{10-d},\Z)  \ \longrightarrow \ \Q/\Z\,,
\end{equation}
which, together with Poincar\'e duality, determines the linking pairing $\ell_{L}(\gamma_{k},\gamma_{k'})$ in $L_{10-d}$, as explained in Appendix \ref{sec:DelZottoLink}. An immediate necessary conditions for $\ell_{L}(\gamma_{k},\gamma_{k'})\neq 0$ is $k' = 9-d-k$.\par

Combining the requirements on codimension in spacetime and in the link, we have that the $p$-brane that realizes the $(p-k)$-dimensional defect and the $p'$-brane which realizes the $(p'+1-k')$-dimensional symmetry operator, must satisfy
\begin{equation}\label{p+p'+4=d+D}
    (p+1) + (p'+1) = \left[ (p-k) + (p'+1-k') \right]+ \left[k + k' \right] +1 = [d-1] +[9-d] +1 = 9\,.
\end{equation}
The total dimension of the worldvolume of the two branes is of codimension 2, thus they are electric/magnetic dual. Exchanging the role of M2-branes and M5-branes, thus exchanging electric and magnetic gauge fields in 11d, corresponds to switch the roles of charged defects and symmetry operators, as shown in Table \ref{tab:M2M5defect}.\par

\paragraph{Finite symmetry generators from M-theory.}
To define the worldvolume effective action of the symmetry operator, we promote the field strength sourced by the M2 or M5-brane to a differential character, and compute its holonomy. Explicitly:
\begin{equation}\label{eq:symmetryoperator-discrete(1)}
\begin{aligned}
    \mathcal{U}^{\text{M2 on }\gamma_{k}} (\Sigma_{3-k}) &=  \exp{ 2\pi \ii \int^{\br{H}}_{\Sigma_{3-k} \times \gamma_{k}} \br{G}_4 } =  \exp{ 2\pi \ii \int^{\br{H}}_{\Sigma_{3-k} \times L_{10-d}} \br{\mathsf{PD} } (\gamma_{k}) \star \br{G}_4 } , \\
    \mathcal{U}^{\text{M5 on }\gamma_{k}} (\Sigma_{6-k}) &= \exp{ 2\pi \ii \int^{\br{H}}_{\hat{\Sigma}_{7-k} \times \gamma_{k}} \br{\dd G}_7 } =  \exp{ 2\pi \ii \int^{\br{H}}_{\hat{\Sigma}_{7-k} \times L_{10-d}} \br{\mathsf{PD} } (\gamma_{k}) \star \br{\dd G}_7 } ,
\end{aligned}
\end{equation}
where $\mathsf{PD}(\cdot)$ is the Poincar\'e duality isomorphism. In the second line, $\hat{\Sigma}_{7-k}$ is any chain with $\partial \hat{\Sigma}_{7-k} = \Sigma_{6-k}$ and, by the exactness of $\dd G_7$, the result only depends on $\Sigma_{6-k}$ and not on its extension $\hat{\Sigma}_{7-k}$.\par
Expanding the differential character in the basis of $\br{H}^{\bullet} (L_{10-d})$, $\br{\mathsf{PD} } (\gamma_{k}) \star$ acts as a projector. We exemplify the relation in the case $\mathrm{Tor} H_{k} (L_{10-d},\Z) \cong \Z_n$, and refer to Appendix \ref{app:IntTorCycles} for more details and for the general case.\par
Writing schematically $\br{G}_4= \br{B}_{3-k} \star \br{t}_{k+1} + \cdots $ and $\br{\dd G}_7= \br{\delta A}_{6-k} \star \br{t}_{k+1} + \cdots $, we get 
\begin{equation}\label{eq:symmetryoperator-discrete(2)}
\begin{aligned}
    \mathcal{U}^{\text{M2 on }\gamma_{k}} (\Sigma_{3-k}) &=  \exp{ 2\pi \ii \left( \int^{\br{H}}_{L_{10-d}} \br{\mathsf{PD} } (\gamma_{k}) \star \br{t}_{k+1} \right) \int_{\Sigma_{3-k}} B_{3-k} } , \\
    \mathcal{U}^{\text{M5 on }\gamma_{k}} (\Sigma_{6-k}) &= \exp{ 2\pi \ii \left( \int^{\br{H}}_{L_{10-d}} \br{\mathsf{PD} } (\gamma_{k}) \star \br{t}_{k+1} \right) \int_{\Sigma_{6-k}} A_{6-k} } .
\end{aligned}
\end{equation}
Note that $B_{3-k}$ and $A_{6-k}$ are discrete $\Z_n$ gauge field with $\Z$-periodicity. Formally, the second line of \eqref{eq:symmetryoperator-discrete(2)} looks analogous to expanding $G_7$ along the torsional cycle $\gamma_k$; however, $G_7$ does not admit a refinement in differential cohomology, necessary to integrate over torsional classes, thus we had to pass through the exact 8-form $\dd G_7$ to arrive at \eqref{eq:symmetryoperator-discrete(2)}.\par
The integrals appearing in \eqref{eq:symmetryoperator-discrete(2)} are discussed extensively in Appendix \ref{app:IntTorCycles}. The outcome of the pairing \eqref{eq:pairingTorTor} between $\mathsf{PD} (\gamma_{k})$ and $t_{k+1}$ is (cf. Appendix \ref{app:IntTorCycles})
\begin{equation}
    \int^{\br{H}}_{L_{10-d}} \br{\mathsf{PD} } (\gamma_{k}) \star \br{t}_{k+1} = \begin{cases} - \frac{1}{n} & \text{ $\mathsf{PD} (\gamma_{k})$ and $t_{k +1}$ are Pontryagin dual}, \\ 0 & \text{ otherwise}. \end{cases}
\end{equation}\par
In summary, we get the topological operators that generate the $\Z_n$ symmetries:
\begin{equation}
    \mathcal{U} (\Sigma_{3-k}) = \exp{ - \frac{2 \pi \ii}{n} \int_{\Sigma_{3-k}} B_{3-k}} , \qquad \mathcal{U} (\Sigma_{6-k}) = \exp{ - \frac{2 \pi \ii}{n} \int_{\Sigma_{6-k}} A_{6-k}} ,
\end{equation}
where $B_{3-k}, A_{6-k}$ are, respectively, electric and magnetic $\Z_n$ gauge fields.

\subsubsection{Continuous abelian symmetries from fluxbranes}
\label{sec:BraneSymTFTFree}

\paragraph{Branes wrapping free cycles.}

We now aim at constructing the topological operators associated with $U(1)$ symmetries. To achieve this goal, we wrap $p$-branes on free cycles in the link space $L_{10-d}$. In this case, the integral \eqref{eq:diffcohoint} in 11d induces a perfect pairing,
\begin{equation}\label{eq:pairing-free-free}
    H^{k}(L_{10-d},\Z)_{\mathrm{free}} \times H^{10-k-d}(L_{10-d},\Z)_{\mathrm{free}}  \ \longrightarrow \ \Z\,,
\end{equation} 
leaving behind a differential character in $\br{H}^{d+1} (M_d)$. Consequently, this leads to the relation
\begin{equation}\label{p+p'+3=d+D-1}
    (p+1) + (p'+1) = \left[ (p-k) + (p'+1-k') \right]+ \left[k + k' \right] +1 = [d-1] +[10-d] +1 = 10\, ,
\end{equation}
as opposed to \eqref{p+p'+4=d+D}. This dimension counting implies that the branes behind the defect and the symmetry operator do not form an electric/magnetic dual pair.\par
The implications of \eqref{p+p'+3=d+D-1} are summarized as:
\begin{equation*}
    \begin{tabular}{r|l}
      charged defect   & worldvolume of symmetry operator \\
      \hline
        M2 & 7d \\
        M5 & 4d
    \end{tabular}
\end{equation*}

\paragraph{Fluxbranes.}
Following \cite{Cvetic:2023plv}, we discuss the generation of continuous symmetries using fluxbranes. Importantly, we point out an issue in using the naive attempt to use $G_7$-fluxbranes, and show how to correctly identify the fluxbrane that produces a topological operator.\par
Let us focus on magnetic defects first, realized by M5-branes, which source $G_7$. From \eqref{p+p'+3=d+D-1} and the Hodge duality relation $G_4\propto \ast G_7$, \cite{Gutperle:2001mb} argues that the brane we are after is interpreted as a $G_4$-fluxbrane. For our purposes, we only need to consider their topological aspects. See \cite{Cvetic:2023plv} for more details on the fluxbrane realization of continuous abelian symmetries.\par
The symmetry topological operator derived from a $G_4$-fluxbrane wrapping $\gamma_k \in H_k (L_{10-d},\Z)_{\mathrm{free}}$ is: 
\begin{equation}\label{symmetry-operator-fluxbrane-G4-0}
    \mathcal{U}^{G_4\text{-flux along }\gamma_{k}} (\Sigma_{4-k}) = \exp{ \ii \varphi \int_{\Sigma_{4-k} \times \gamma_{k}} \frac{G_4}{2\pi} } =  \exp{ \ii \frac{\varphi}{2\pi} \int_{\Sigma_{4-k} \times L_{10-d}} \mathsf{PD} (\gamma_{k}) \wedge G_4 } .
\end{equation}
Here, the quantization of flux implies the identification $\varphi \sim\varphi +2\pi$, in agreement with the expectation that the symmetry is $U(1)$.\par
Taking $v_{k} \in H^{k}(L_{10-d},\Z)_{\text{free}}$ and expanding $G_4= F_{4-k} \wedge v_{k} + \cdots $, we obtain 
\begin{equation}\label{symmetry-operator-fluxbrane-G4-1}
\begin{aligned}
    \mathcal{U}^{G_4\text{-flux along }\gamma_{k}} (\Sigma_{4-k})  &= \exp{ \ii \frac{\varphi}{2\pi} \left( \int_{L_{10-d}} \mathsf{PD}  (\gamma_{k}) \wedge v_{k}\right) \int_{\Sigma_{4-k}} F_{4-k} } \\
    &= \exp{ \ii \frac{\varphi}{2\pi}  \int_{\Sigma_{4-k}} F_{4-k} } .
\end{aligned}
\end{equation}
In the second line we have assumed that the integral over $L_{10-d}$ is non-vanishing, and normalized it to 1 without loss of generality.\par 

One may attempt to propose a symmetry topological operator for the electric symmetry as a $G_7$-fluxbrane, linking with the M2-brane providing the charged defect. Tentatively, 
\begin{equation}
    \mathcal{U}^{G_7\text{-flux along }\gamma_{k}} (\Sigma_{7-k}) \stackrel{\text{naive}}{\sim} \exp{ \ii \frac{\varphi}{2\pi} \int_{\Sigma_{7-k} \times \gamma_{k}} G_7 }  .
\end{equation}
However, $G_7$ is not closed and, working with this naive attempt, the would-be operator is not topological.\par
Therefore, we propose an enhanced formulation of this symmetry operator, by more carefully tracking the effective action on the fluxbrane. In this way we address the issues.

\paragraph{Hopf--Wess--Zumino action.}

In M-theory, the relation \eqref{p+p'+3=d+D-1} can be interpreted as the intersection pairing between the M$p$-branes and the corresponding Page charges \cite{Page:1983mke}. For the case where the defects are given by the M5-brane, its Page charge is given by the flux of $G_{4}$, which satisfies, $\dd G_{4}=0$. Thus, the $G_4$-fluxbrane is the brane that supports the Page charge, in agreement with \eqref{symmetry-operator-fluxbrane-G4-1}.\par
In the following, we construct the closed 7-form $P_{7}$ whose fluxbranes generate $U(1)$ symmetries. To do that, we start from the topological action on M5-branes. Additional discussion on the effective M-theory action in presence of M5-branes is deferred to Appendix \ref{app:MthwithM5}.\par
The worldvolume theory on an M5-brane admits a self-dual 3-form, $H_{3}$, which satisfies \cite{Howe:1997vn}
\begin{equation}\label{eq:dH3isG4}
    \dd H_{3} \,=\, \iota^{\ast}_{\text{M5}}G_{4}\,.
\end{equation}
Here, $\iota_{\text{M5}} :\Sigma_6^{\text{M5}} \hookrightarrow M_d \times X_{11-d} $ is the embedding of the M5-brane worldvolume $\Sigma_6^{\text{M5}}$ into the 11d space, whereby $\iota^{\ast}_{\text{M5}}G_{4}$ is the pullback of the M-theory $G_{4}$ to the M5-brane worldvolume.\par

To write the full topological action on the fluxbrane, we ought to consider the coupling between $H_{3}$ and $G_{4}$. This is implemented by the Hopf--Wess--Zumino (HWZ) action on a seven-dimensional worldvolume.\par 
To properly define this action, we embed the M5-brane worldvolume into $\Sigma_7$, which supports the $G_7$-flux. $\Sigma_7$ is obtained by extending the M5-brane worldvolume inside $L_{10-d}$, filling the transverse direction between the M5 and the defect M2-brane. In other words, $\Sigma_7$ is a 7-manifolds that is transverse to the radial direction $[0,\infty)$ of the SymTFT, links with the worldvolume of the M2-brane in spacetime $M_d$, and intersects transversely with the worldvolume of the M2-brane in $L_{10-d}$.\par
For comparison, in Figure \ref{fig:M2M5P7link} we show an M5-brane linking with an M2-brane inside $L_{10-d}$, and a fluxbrane intersecting the M2-brane.\par

\begin{figure}[th]
    \centering
    \begin{tikzpicture}
        \node[violet,ellipse,draw,very thick,align=center] at (-4,0) { \hspace{56pt} \\  \hspace{56pt} };
        \node[violet,ellipse,draw,very thick,fill,fill opacity=0.2,align=center] at (4,0) { \hspace{56pt} \\  \hspace{56pt} };
        \draw[red,very thick] (-4,2) -- (-4,0);
        \draw[red,very thick] (-4,-0.75) -- (-4,-2);
        \draw[red,very thick] (4,2) -- (4,0);
        \draw[red,very thick] (4,-0.75) -- (4,-2);

        \node[red,anchor=west] at (-4,1.5) {\small defect M2};
        \node[red,anchor=east] at (4,1.5) {\small defect M2};
        \node[violet,anchor=north west,align=left] at (-3.5,-0.5) {\small topological\\ \small  M5-brane};
        \node[violet,anchor=north east,align=right] at (3.5,-0.5) {\small topological\\ \small $P_7$-fluxbrane};
    \end{tikzpicture}
    \caption{Left: an M5-brane linking with an M2-brane in the link $L_{10-d}$. This configuration is relevant for finite symmetries. Right: a fluxbrane intersecting the same M2-brane in the link $L_{10-d}$. This configuration is relevant for $U(1)$ symmetries.}
    \label{fig:M2M5P7link}
\end{figure}
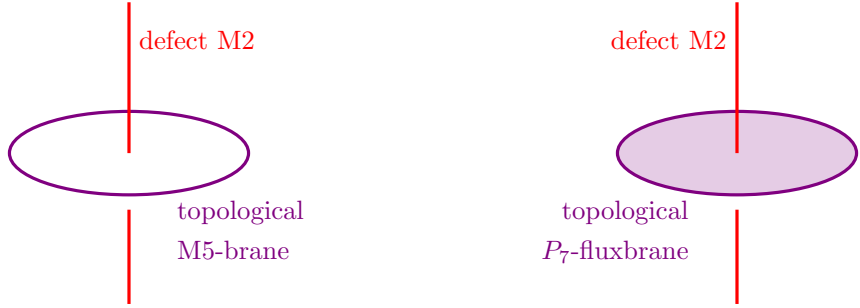

If we denote $\iota_7 : \Sigma_6^{\text{M5}} \hookrightarrow \Sigma_7$ the embedding of the M5-brane worldvolume into $\Sigma_7$, we can define a map
\begin{equation}
    \phi^{\ast} := \iota_{7,\ast} \circ \iota^{\ast}_{\text{M5}} \ : \ \Omega^{\bullet} (M_d \times X_{11-d}) \longrightarrow \Omega^{\bullet} (\Sigma_7) .
\end{equation}
This allows to formulate the HWZ action as \cite{Bandos:1997ui,Intriligator:2000eq},
\begin{equation}\label{eq:HWZ-action}
 S_{\text{HWZ}}\,=\, \frac{1}{2\pi} \int_{\Sigma_{7}}  \phi^{\ast}G_{7} +\frac{1}{4\pi} \iota_{7,\ast} H_{3}\wedge\phi^{\ast}G_{4}\,.  
\end{equation}
We can identify this action with the M2-brane Page charge \cite{Page:1983mke}.\footnote{This identification has been considered previously in different contexts \cite{Pilch:2015vha,Anabalon:2022fti,Fiorenza:2019ain}.} The integrand,
\begin{equation}\label{Page-P7-charge}
    P_{7}\, := \phi^{\ast}G_{7} +\frac{1}{4\pi} \iota_{7,\ast} H_{3}\wedge\phi^{\ast}G_{4}\, ,
\end{equation}
is both closed and quantized, i.e. 
\begin{itemize}
    \item $\dd P_7 =0$ by virtue of \eqref{modifiedBIG7} and \eqref{eq:dH3isG4};
    \item the periods of $P_7$ are in $2 \pi \Z$. 
\end{itemize}

\paragraph{More on Page charges.}
We have derived the Page charge $P_7$ on the worldvolume of a fluxbrane. Alternatively, one may introduce a differential form $\mathsf{P}_7$ directly on the eleven-dimensional space, which is locally defined to be (see \cite[Sec.5]{Moore:2004jv} for a more exhaustive treatment)
\begin{equation}\label{eq:P7in11d}
    \mathsf{P}_7 = G_7 - \frac{1}{12\pi} C_3 \wedge G_4 .
\end{equation}
This allows to write 
\begin{equation}
    S_{\text{kin}}^{\text{M}} + S_{\text{CS}}^{\text{M}} = \frac{1}{(2\pi)^2} \int_{M_{11}} G_4 \wedge \mathsf{P}_7 ,
\end{equation}
and the equations of motion guarantee the Bianchi identities $\dd G_4 = 0 = \dd \mathsf{P}_7$.\par 
In practice, we may work with $\mathsf{P}_7$ throughout, and the procedure works equally well for the analysis of continuous symmetries. This Page charge has the advantage of being defined directly in the 11d action; however, the latter form of the action only admits a trivial refinement to differential cohomology, and is not suited for finite symmetries.\par 
We instead follow the fluxbrane approach and work with $P_7$, given in \eqref{Page-P7-charge} in terms of globally defined quantities supported on an extension of the fluxbrane worldvolume.
It is also worth mentioning that our approach allows for fields supported on the brane, as discussed around \eqref{eq:def-of-H3}.

\paragraph{Symmetry topological operator form $P_7$-fluxbrane.}

We propose that the symmetry topological operator for a $U(1)$ symmetry is a fluxbrane for the class $P_7$ in \eqref{eq:HWZ-action}.\par 
Thus, building on the structure of the HWZ action, we introduce the symmetry topological operator
\begin{equation}\label{eq:Page-symmetry-operator}
   \mathcal{U}^{P_7\text{-flux along }\gamma_{k}} (\Sigma_{7-k}) = \exp{ \ii \frac{\varphi}{2\pi} \int_{\Sigma_{7-k} \times \gamma_{k}} P_7 } =  \exp{ \ii \frac{\varphi}{2\pi} \int_{\Sigma_{7-k} \times L_{10-d}} \mathsf{PD} (\gamma_{k}) \wedge P_7 } .
\end{equation}
Since $P_{7}$ is closed and quantized, this ensures that the operator is topological; furthermore, $\varphi$ is a $2\pi$-periodic parameter. Therefore, \eqref{eq:Page-symmetry-operator} generates a $U(1)$ symmetry.\par
Taking $v_{k} \in H^{k}(L_{10-d},\Z)_{\text{free}}$ and expanding $P_7= \widetilde{h}_{7-k}\wedge v_{k} + \cdots $, we arrive at
\begin{equation}\label{symmetry-operator-fluxbrane2}
\begin{aligned}
    \mathcal{U}^{P_7\text{-flux along }\gamma_{k}} (\Sigma_{7-k}) &= \exp{ \ii\frac{\varphi}{2\pi} \left( \int_{L_{10-d}} \mathsf{PD} (\gamma_{k}) \wedge v_{k}\right) \int_{\Sigma_{7-k}} \widetilde{h}_{7-k} } \\
    &= \exp{ \ii \frac{\varphi}{2\pi}  \int_{\Sigma_{7-k}} \widetilde{h}_{7-k} } ,
\end{aligned}
\end{equation}
where we have assumed that the integral over $L_{10-d}$ is non-vanishing, and normalized it to 1 without loss of generality.\par
\eqref{symmetry-operator-fluxbrane2} gives an operator defined in the SymTFT, occupying a worldvolume $\Sigma_{7-k}$ transverse to the radial direction. The field strength $\widetilde{h}_{7-k}$ is closed and quantized. These properties are inherited from the Page charge $P_{7}$ as defined in \eqref{Page-P7-charge}.\par
Clearly, the expansion of $P_7$ cannot be taken independently from that of $G_7$ and $G_4$, by definition. Thus, writing $G_7 = h_{7-k} \wedge v_k + \cdots$, we have 
\begin{equation}\label{eq:fromP7tpG7}
    \widetilde{h}_{7-k} = \phi^{\ast} (h_{7-k} + g_{7-k}) ,
\end{equation}
where the second piece $g_{7-k}$ is the projection of $(\iota_{\text{M5},\ast }H_3 )\wedge G_4$ along the directions transverse to $v_k$, which would become a wedge product of two differential forms after dimensional reduction as well.

\paragraph{Exchanging M2 and M5.}

The main proposal of this subsection, summarized in Table \ref{tab:M2P7M5P4}, is that a defect engineered by an M$p$-brane is acted upon by a topological operator defined by the fluxbrane for the corresponding Page charge.\par
\begin{table}[th]
    \centering
    \begin{tabular}{|c|c|}
    \hline
    Defect & Symmetry operator \\
    \hline
    M2 & $P_7$-fluxbrane \\
    M5 & $G_4$-fluxbrane \\
    \hline
    \end{tabular}
    \caption{Defects charged under continuous symmetries are realized by BPS M$p$-branes. The symmetry topological operators are realized by the fluxbranes for the corresponding Page charge.}
    \label{tab:M2P7M5P4}
\end{table}\par

In an electric frame we have: M2-branes provide the charged defects; M5-branes source the Page charges of the defect M2-branes; and the symmetry operator is the $P_7$-fluxbrane.\par
Now, consider an electromagnetic transformation in M-theory, which corresponds to exchanging M2 and M5-branes. This exchanges the roles of $p$ and $p'$ in \eqref{p+p'+3=d+D-1}. We get a dual frame with: M5-branes providing the charged defects; M2-brane sourcing the Page charges of the defect M5-branes; and the symmetry operator being the $P_4$-fluxbrane, with $P_4=G_4$.\par
As a result, this duality transformation in 11d relates $m$-form and $(d-m-3)$-form $U(1)$ symmetries. This should not be confused with the electric/magnetic duality in spacetime, schematically,
\begin{equation}
\begin{aligned}
    U(1)^{\scriptscriptstyle [m]} \quad & \xrightarrow{ \ \text{ M2 $\leftrightarrow$ M5 } \ } \quad & U(1)^{\scriptscriptstyle [d-m-3]} \\
    U(1)^{\scriptscriptstyle [m]} \quad & \xrightarrow{ \ \text{flat gauging} \ } \quad & \Z^{\scriptscriptstyle [d-m-2]} 
\end{aligned}
\end{equation}
where $\Z^{\scriptscriptstyle [d-m-2]} $ in the second line is the character group of $U(1)^{\scriptscriptstyle [m]}$.\par

\subsubsection{Symmetry Theory for continuous abelian symmetries}

From a field theory viewpoint, there are proposals for a SymTFT action of continuous symmetries \cite{Brennan:2024fgj,Antinucci:2024zjp,Apruzzi:2024htg,Bonetti:2024cjk}. Nonetheless, its M-theory origin has not yet been worked out. We now proceed to fill this gap (see however the very recent \cite{Gagliano:2024off} for related discussion in Type II string theory).\par

Let us mention from the onset that our derivation is in line with the `Symmetry Theory' of \cite{Apruzzi:2024htg}. One important difference is that \cite{Apruzzi:2024htg} considers Maxwell-type theories. M-theory does not belong to this class, because of the Chern--Simons term. As a consequence, the relation between the SymTFT action and the topological operator is less direct, albeit the action we derive agrees with the prescription of \cite{Apruzzi:2024htg}. We come back to this point at the end of the current subsection.\par
Our proposal is as follows.
\begin{itemize}
    \item We begin with $-G_4 \wedge \ast G_4$ and take a topological limit, in which we treat $G_4$ and $G_7= -\ast G_4/(2\pi)$ as the electric and magnetic field strengths. Upon direct dimensional reduction, this matches the BF terms of \cite{Apruzzi:2024htg}.
    \item We refine the action to differential cohomology, as prescribed in Subsection \ref{sec:MtheoryupliftH12}.\par
    We remark that the advantage of using differential characters is to unify the treatment of finite and continuous symmetries.\footnote{These comment applies to Type II superstring theory as well \cite{Gagliano:2024off}, with the due changes.} Had we only been interested in $U(1)$ symmetries, we would have obtained the identical BF couplings using $G_4 \wedge G_7$. This is a by-product of our refinement in Subsection \ref{sec:MtheoryupliftH12} and the basic properties of differential characters (see Appendix \ref{app:diffchar}).
    \item Once we insert an M5-brane or the associated fluxbrane, we cannot neglect the presence of $H_3$, supported on the worldvolume of the brane. Thus, when constructing defects charged under $U(1)$ symmetries, or the associated symmetry operators, we utilize $P_7$ rather than $G_7$.\par
    This latter step is the main novelty and difference with respect to the Maxwell-type theories discussed in \cite{Apruzzi:2024htg}, and also with respect to other proposals in \cite{Brennan:2024fgj,Antinucci:2024zjp}.
\end{itemize}
We thus obtain the following description of the Symmetry Theory of $U(1)$ symmetries:
\begin{itemize}
    \item The BF-terms are computed reducing \eqref{eq:MtheoryKin} along free cycles in the link;
    \item The defect operator realized by an M$p$-brane is the holonomy of the differential character $\br{P}_{p+2}$;
    \item The symmetry operator acting on an M$p$-brane is given by the flux of the corresponding page charge, i.e. $P_7$-fluxbrane for an M2-brane defect and $P_4$-fluxbrane for an M5-brane defect (with $P_4=G_4$).
\end{itemize}
We summarize our findings in Table \ref{tab:fluxbraneMbrane}.\par

\begin{table}[th]
    \centering
    \begin{tabular}{|c|c|c|c|}
    \hline
    Defect brane & Defect action & Symmetry brane & Symmetry operator action \\
    \hline
    M2 & $\int^{\br{H}}_{\Sigma_3^{\text{M2}}} \br{G}_4 $ & $P_7$-fluxbrane & $\frac{\varphi}{2\pi} \int_{\Sigma_7} P_7 $ \\
    M5 & $\int^{\br{H}}_{\Sigma_6^{\text{M5}}} \br{P}_7 $ & $G_4$-fluxbrane & $\frac{\varphi}{2\pi} \int_{\Sigma_4} G_4 $ \\
    \hline
    \end{tabular}
    \caption{Defects and symmetry operators for $U(1)$ symmetries from M-theory.}
    \label{tab:fluxbraneMbrane}
\end{table}\par

\paragraph{Bulk operators.}
Let us fix $v_k \in H^k (L_{10-d},\Z)_{\text{free}}$. We expand $G_4= F_{4-k} \wedge v_k + \cdots $ and $P_7 = \widetilde{h}_{7-k} \wedge v_k + \cdots $, omitting terms that yield a vanishing contribution. We also use the relation between the expansion of $P_7$ and that of $G_7$, as in \eqref{eq:fromP7tpG7}.\par 
Wrapping an M$p$-brane on $\mathsf{PD}(v_k)$, we thus obtain the defect operators
\begin{equation}
\begin{aligned}
    \text{M2:} \quad & \exp{2\pi\ii \int^{\br{H}}_{\Sigma_{3-k} \times \mathsf{PD}(v_k)} \br{G}_4 } =  \exp{2\pi\ii \int^{\br{H}}_{\Sigma_{3-k}} \br{F}_{4-k} } , \\
    \text{M5:} \quad & \exp{ 2\pi\ii\int^{\br{H}}_{\Sigma_{6-k} \times \mathsf{PD}(v_k)} \br{P}_7 } =  \exp{2\pi\ii \int^{\br{H}}_{\Sigma_{6-k}} \phi^{\ast} \left( \br{h}_{7-k} + \br{g}_{7-k} \right) } .
\end{aligned}
\end{equation}
The left-hand side is the refinement in differential cohomology of the non-globally defined holonomies of gauge fields (and recall the $2\pi$-normalization of \eqref{diffcharFmap}, between differential characters and field strengths).
On the other hand, the symmetry operators are given by 
\begin{equation}
\begin{aligned}
    \text{acting on M2:} \quad & \exp{\ii \int_{\Sigma_{7-k}}\frac{\varphi}{2\pi} ~\phi^{\ast} \left( h_{7-k} + g_{7-k} \right) } , \\
    \text{acting on M5:} \quad & \exp{\ii \int_{\Sigma_{4-k}} \frac{\varphi}{2\pi} ~ F_{4-k} } .
\end{aligned}
\end{equation}
These expressions are consistent with each other.

\paragraph{Completion of BF terms and comparison with topological operators.}
Note that, from the above discussion and the relation between BF terms and topological operators in known examples, one may expect a BF term of the form 
\begin{equation}
    \frac{1}{2\pi} \int_{Y_{d+1}} F_{4-k} \wedge \left( h_{7-k'} + g_{7-k'} \right) 
\end{equation}
in the SymTFT, with $k+k'=10-d$. However, reducing \eqref{eq:MtheoryKin}, we obtain the BF-like term 
\begin{equation}
    2\pi \int^{\br{H}}_{M_d \times X_{11-d}} \br{G}_4 \star \br{\dd G}_7 = \frac{1}{2\pi} \int_{Y_{d+1}} F_{4-k} \wedge h_{d-3+k} .
\end{equation}
This is because the BF couplings are part of the action, and are computed in absence of any brane insertion. Once we insert the fluxbrane, we cannot neglect $H_3$ sourced on its worldvolume. This gives the correction $\phi^{\ast} g_{d-3+k}$ in the action of the symmetry operator, not seen from the BF terms alone.\par
Nonetheless, reducing \eqref{eq:MtheoryS} along free cycles, we always obtain a term of the form
\begin{equation}\label{eq:reduce-G43-to-BF-like}
    \frac{2\pi}{2}\int^{\br{H}}_{Y_{d+1}} \br{F}_{4-k} \star \br{G}^{\prime}_{d-2+k} = \frac{1}{2\pi} \int_{Y_{d+1}} F_{4-k} \wedge \widetilde{g}_{d-3+k} ,
\end{equation}
where $G^{\prime}_{d-2+k}$ is the reduction of $G_4 \wedge G_4$ along the $(10-d-k)$ directions transverse to $v_k$. On the right-hand side, we have introduced $\widetilde{g}_{d-3+k} $ subject to $\dd \widetilde{g}_{d-3+k} = \frac{1}{2} G^{\prime}_{d-2+k}$, where the factor $1/2$ will typically cancel out in the examples due to the symmetry of the term $G_4 \wedge G_4$. The pullback $\phi^{\ast} \widetilde{g}_{d-3+k}$ matches precisely with $\phi^{\ast} g_{d-3+k}$.\par
A simple computation shows that the generic term takes the form 
\begin{equation}
    \br{G}^{\prime}_{d-2+k} = \sum_{k_1,k_2=0}^{4} \delta_{k_1+k_2, 10-d-k} c_{k_1,k_2} \br{F}^{\prime, 1}_{4-k_1} \star \br{F}^{\prime, 2}_{4-k_2} ,
\end{equation}
with coefficient $c_{k_1,k_2}=\pm 1$ if any two among $\br{F}_{4-k},\br{F}^{\prime, 1}_{4-k_1} $ and $\br{F}^{\prime, 2}_{4-k_2}$ are equal and $c_{k_1,k_2}=\pm 2$ otherwise, and sign depending on the values of $k,k_1,k_2$ and $d$. In all examples, we observe that at most two of the three symmetries whose backgrounds are $\br{F}_{4-k},\br{F}^{\prime, 1}_{4-k_1} $ and $\br{F}^{\prime, 2}_{4-k_2}$ act non-trivially on the physical theory. Whenever this is the case, the coefficient of \eqref{eq:reduce-G43-to-BF-like} is just enough to combine it with the BF terms of the non-trivially acting symmetries. While we lack a more comprehensive explanation, this procedure is suited to match the combined BF terms and the symmetry topological operators in all cases considered in the present work.\par
To sum up, if we only look at the BF-like term derived from $G_4 \wedge G_7$ in the action \eqref{eq:MtheoryKin}, we would get $F_{4-k} \wedge h_{d-3+k}$. However, if we include a specific contribution from \eqref{eq:MtheoryS}, we get the expected BF term 
\begin{equation}\label{eq:BF-likewithG4}
    \frac{1}{2\pi} \int_{Y_{d+1}} F_{4-k} \wedge \left( h_{d-3+k} + \widetilde{g}_{d-3+k} \right) ,
\end{equation}
with the pullback $\phi^{\ast} \left( h_{d-3+k} + \widetilde{g}_{d-3+k} \right) = \widetilde{h}_{d-3+k}$ to the worldvolume of the $P_7$-fluxbrane giving the correct topological operator. Observe that the combination in \eqref{eq:BF-likewithG4} corresponds to the expansion of \eqref{eq:P7in11d}.\par 
To be slightly more precise, we actually get a BF-like term 
\begin{equation}
    2\pi \int_{Y_{d+1}}^{\br{H}} \br{F}_{4-k} \wedge \br{\widetilde{H}}_{d-2+k} ,
\end{equation}
with $\dd \left( h_{d-3+k} + \widetilde{g}_{d-3+k} \right) = \widetilde{H}_{d-2+k}$, which reduces to \eqref{eq:BF-likewithG4} using that it is the curvature of the sum of two globally defined contributions.

\section{Geometric engineering of \texorpdfstring{$(-1)$}{(-1)}-form symmetry in 5d SCFTs}
\label{sec:5d}
Generalized symmetries of 5d SCFTs have been extensively studied in the literature \cite{Albertini:2020mdx,Morrison:2020ool,Bhardwaj:2020phs,Apruzzi:2021vcu,Apruzzi:2021nmk}. In this section, we mostly retrieve known results in the formalism of differential cohomology. The novelties consist in a unified derivation of all BF couplings, as well as a concentrated focus on $(-1)$-form symmetries, their geometric classification and the choices of polarization, that eluded prior scrutiny.

\subsection{SymTFT of 5d SCFTs from M-theory}
\label{sec:SymTFT5DCY3}

We start with a review of the derivation of generalized symmetries and SymTFTs in the 5d SCFT $\FTfive$ from M-theory on a canonical threefold singularity $X_6$~\cite{Apruzzi:2021nmk}.\par
We start with a sketch of the computation of \cite{Apruzzi:2021nmk}, then go on and discuss two new features:
\begin{itemize}
	\item The BF couplings, including both discrete and continuous symmetries;
	\item The $(-1)$-form symmetries.
\end{itemize}
We explain how defects and topological operators are engineered by  M2- and M5-branes in Subsection \ref{sec:5dDefectM2M5}.

\subsubsection{SymTFT from M-theory}
\label{sec:SymTFT5dM}

\paragraph{Geometric data.} The eleven-dimensional space is $M_{11} = X_6 \times M_5$. The link $L_5$ has cohomology 
\begin{equation}
\label{eq:L5coho}
	H^{\bullet} (L_5,\Z)= \  \Z \ \oplus \ 0 \ \oplus \ \begin{matrix} \Z^{b^2} \\ \oplus \\ \mathrm{Tor} H^2(L_5,\Z) \end{matrix} \ \oplus \ \begin{matrix} \Z^{b^2} \\ \oplus \\ \mathrm{Tor} H^3(L_5,\Z) \end{matrix} \ \oplus \ \mathrm{Tor} H^2(L_5,\Z) \ \oplus \ \Z  
\end{equation}
where $b^2$ is the second Betti number. Here and in the following, we make repeated use of the pair of isomorphisms (related by Poincar\'e duality)
\begin{equation}
\begin{aligned}
    \mathrm{Tor} H^p(L_d,\Z) &\cong\mathrm{Tor} H^{d-p+1}(L_d,\Z) ,\\
    \mathrm{Tor} H_p(L_d,\Z) &\cong\mathrm{Tor} H_{d-p-1}(L_d,\Z) ,
\end{aligned}
\end{equation}
to reduce to the independent torsion (co)homology groups. These relations stem from the universal coefficient theorem, combined with Poincar\'e duality, and the fact that the torsion groups we deal with are finite abelian groups, which can be identified with their Pontryagin dual, via $\mathrm{Hom} (\Z_n, U(1)) \cong \Z_n$.\par
From the above, the link $L_5$ has homology groups
\begin{equation}
	H_{\bullet} (L_5,\Z)= \  \Z \ \oplus \ \mathrm{Tor} H^2(L_5,\Z) \ \oplus \ \begin{matrix} \Z^{b^2} \\ \oplus \\ \mathrm{Tor} H^3(L_5,\Z) \end{matrix} \ \oplus \ \begin{matrix} \Z^{b^2} \\ \oplus \\ \mathrm{Tor} H^2(L_5,\Z) \end{matrix} \ \oplus \  0 \ \oplus \ \Z  .
\end{equation}\par
Furthermore, we write without loss of generality
\begin{equation}
\begin{aligned}
    \mathrm{Tor} H^{2} (L_5, \Z) &= \Z_{m_1} \oplus \cdots \oplus \Z_{m_{\mu}} \\
    \mathrm{Tor} H^{3} (L_5, \Z) &= \Z_{n_1} \oplus \cdots \oplus \Z_{n_{\widetilde{\mu}}} . 
\end{aligned}
\end{equation}
We denote the generators of the free and torsion parts of $H^{k} (L_5, \Z)$ by
\begin{equation}
    \begin{aligned}
        &\begin{tabular}{|c|c|c|c|c|c|c|}
        \hline
           $k$ & 0 & 1 & 2 & 3 & 4 & 5 \\
           \hline
            $H^{k} (L_5, \Z)_{\text{free}}$ & 1 &  & $v_2^{i}$ & $v_3^{i}$ &  & $\vol_5$ \\
            \hline
            $\mathrm{Tor} H^{k} (L_5, \Z)$ & & & $t_2^{\alpha}$ & $t_3^{\beta}$ & $t_4^{\alpha}$  & \\
            \hline
        \end{tabular}\\
        & i=1, \dots, b^2 , \ \alpha=1, \dots, \mu , \  \beta=1, \dots, \widetilde{\mu} .
    \end{aligned}
\end{equation}
The generators of $ H^{\bullet} (L_5, \Z)_{\mathrm{free}}$ satisfy 
\begin{equation}\label{eq:v2v3inL5}
	\int_{L_5} v_2^{i} \wedge v_3^{j} = \delta^{ij} .
\end{equation}\par
In this basis, we expand $\br{G}_4 \in \br{H}^{4} (M_{11})$ and $\br{\dd G}_7 \in \br{H}^{8} (M_{11})$ as
\begin{equation}
    \begin{aligned}
        \br{G}_4 &= \br{F}_4 \star \br{1} + \sum_{\alpha=1}^{\mu} \br{B}_2^{\alpha} \star \br{t}_2^{\alpha} + \sum_{i=1}^{b^2}\br{F}_2^{i} \star \br{v}_2^{i} + \sum_{\beta=1}^{\widetilde{\mu}} \br{B}_1^{\beta} \star \br{t}_3^{\beta} + \sum_{i=1}^{b^2}\br{F}_1^{i} \star \br{v}_3^{i} + \sum_{\alpha=1}^{\mu} \br{B}_0^{\alpha} \star \br{t}_4^{\alpha}, \\
        \br{\dd G}_7 &= \br{f}_8 \star \br{1} + \sum_{\alpha=1}^{\mu} \br{\mathcal{B}}_6^{\alpha} \star \br{t}_2^{\alpha} + \sum_{i=1}^{b^2}\br{f}_6^{i} \star \br{v}_2^{i} + \sum_{\beta=1}^{\widetilde{\mu}} \br{\mathcal{B}}_5^{\beta} \star \br{t}_3^{\beta} + \sum_{i=1}^{b^2}\br{f}_5^{i} \star \br{v}_3^{i} + \sum_{\alpha=1}^{\mu} \br{\mathcal{B}}_4^{\alpha} \star \br{t}_4^{\alpha} + \br{f}_3 \star \br{\vol}_5 . 
    \end{aligned}
\label{eq:expansionG5D}
\end{equation}
From the exactness of $\dd G_7$ we have 
\begin{equation}
\label{eq:G8trivial5d}
	\mathscr{F} (\br{f}_{p+1}^{\bullet}) = \frac{1}{2\pi} \dd h_{p}^{\bullet}, \qquad c(\br{f}_{p+1}^{\bullet}) =0, \qquad \mathscr{F} (\br{\mathcal{B}}_{p+1}^{\bullet})= \delta A_{p}^{\bullet}, 
\end{equation}
where in the latter expression, $A_{p}^{\bullet}$ are identified with cochains and $\delta$ is the coboundary operation.\par
To lighten the subsequent expressions, we define the matrices $\mathrm{CS}^{(5)}, \widetilde{\mathrm{CS}}^{(5)}$ and $\kappa^{(5)}$ with entries 
\begin{equation}
\begin{aligned}
     (\mathrm{CS}^{(5)})_{\alpha \alpha^{\prime} \alpha^{\prime\prime}} & = \frac{1}{6}\int_{L_5}^{\br{H}} \br{t}_2^{\alpha} \star \br{t}_2^{\alpha^{\prime}} \star \br{t}_2^{\alpha^{\prime\prime}} , \\
     (\widetilde{\mathrm{CS}}^{(5)})_{\beta \beta^{\prime}} & = \int_{L_5}^{\br{H}} \br{t}_3^{\beta} \star \br{t}_3^{\beta^{\prime}} , \\
     (\kappa^{(5)})_{\alpha \alpha^{\prime}} & = \int_{L_5}^{\br{H}} \br{t}_4^{\alpha} \star \br{t}_2^{\alpha^{\prime}} .
\end{aligned}
\label{eq:L5CSmatrices}
\end{equation}
$\widetilde{\mathrm{CS}}^{(5)}$ is an anti-symmetric matrix, whereas $\mathrm{CS}^{(5)}$ and $\kappa^{(5)}$ are symmetric. We also define 
\begin{equation}
\begin{aligned}
    (\eta^{(5)})_{\alpha \alpha^{\prime} i} & = \frac{1}{2}\int_{L_5}^{\br{H}} \br{t}_2^{\alpha} \star \br{t}_2^{\alpha^{\prime}} \star \br{v}_2^{i} \\
     (\widetilde{\eta}^{(5)})_{\beta i} & = \int_{L_5}^{\br{H}} \br{t}_3^{\beta} \star \br{v}_3^{i} , \\
     (\widetilde{\kappa}^{(5)})_{\alpha i} & = \int_{L_5}^{\br{H}} \br{t}_4^{\alpha} \star \br{v}_2^{i} .
\end{aligned}
\end{equation}
The integrals are studied in Appendix \ref{sec:torintegralsL5}. Moreover, we introduce the shorthand notation
\begin{equation}
\begin{aligned}
     \mathrm{CS}^{(5)} \left(  x \smile y \smile w \right)& := \sum_{\alpha, \alpha^{\prime}, \alpha^{\prime\prime}=1}^{\mu} (\mathrm{CS}^{(5)})_{\alpha \alpha^{\prime} \alpha^{\prime\prime}} x^{\alpha} \smile y^{\alpha^{\prime}} \smile w^{\alpha^{\prime\prime}} , \\
     \widetilde{\mathrm{CS}}^{(5)} \left(  x \smile y \right)& := \sum_{ 1 \le \beta <\beta^{\prime}\le\widetilde{\mu}}(\widetilde{\mathrm{CS}}^{(5)})_{\beta \beta^{\prime}} x^{\beta} \smile y^{\beta^{\prime}} , \\
     \kappa^{(5)} \left(  x \smile y \right)& := \sum_{\alpha, \alpha^{\prime}=1}^{\mu} (\kappa^{(5)})_{\alpha \alpha^{\prime}}  x^{\alpha}\smile y^{\alpha^{\prime}} ,
\end{aligned}
\end{equation}
and likewise for $\widetilde{\kappa}^{(5)} (x \smile v)$ and so on. Note the range of summation in right-hand side of the second expression, by anti-symmetry.

\paragraph{Mixed anomalies from M-theory.}
Fibre-wise integrating the M-theory topological action \eqref{eq:MtheoryS} over $L_5$, we get 
\begin{equation}
\begin{aligned}
    S_{\text{twist}}^{L_5} &= \int_{Y_6}^{\br{H}} \sum_{i=1}^{b^2} \br{F}_4 \star \br{F}_2^{i} \star \br{F}_1^{i} \\
    &-\int_{Y_6} \left\{  \mathscr{F}(\br{F}_4) \smile \left[ \kappa^{(5)} \left(  B_0\smile B_2 \right) + \widetilde{\kappa}^{(5)} \left( \mathscr{F} (\br{F}_2) \smile B_0 \right) \right] \right. \\
    & + \mathrm{CS}^{(5)} \left(  B_2 \smile B_2 \smile B_2\right) + \eta^{(5)} \left( \mathscr{F} (\br{F}_2) \smile B_2 \smile B_2 \right)   \\
    & \left. + \mathscr{F} (\br{F}_4) \smile \left[  - \widetilde{\mathrm{CS}}^{(5)} \left(  B_1\smile B_1 \right) + \widetilde{\eta}^{(5)}  \left( \mathscr{F} (\br{F}_1) \smile B_1 \right) \right] \right\} .
\end{aligned}
\label{eq:SymTFT5dN=1full}
\end{equation}
When integrating over $\vol_5$ and $v_2^{i} \wedge v_3^{j}$, the sign convention \eqref{eq:fibreproddiffcoho} of the differential cohomology integral must be used. For instance, the sign in front of the first term on the right-hand side turns out to be $-(-1)^{\dim (L_5)} =+1$.\par
The term $\br{G}_4 \star \br{\mathsf{X}}_8$ was analyzed in \cite[Sec.4]{Apruzzi:2021nmk} (leveraging a theorem of Wall \cite{Wall:1966rcd}) and we refer to that discussion. The outcome is a shift \cite[Eq.(4.20)]{Apruzzi:2021nmk}
\begin{equation}\label{eq:G4G8shiftB2cubic}
    \mathrm{CS}^{(5)}_{\alpha\alpha^{\prime}\alpha^{\prime\prime}} \ \mapsto \ \mathrm{CS}^{(5)}_{\alpha\alpha^{\prime}\alpha^{\prime\prime}} - \frac{1}{24} \delta_{\alpha\alpha^{\prime}}\delta_{\alpha\alpha^{\prime\prime}} \int^{\br{H}}_{L_5} \br{t}_2^{\alpha} \star \br{\mathsf{p}}_1 (TL_5)
\end{equation}
in the coefficient of the cubic $B_2^{\alpha} \smile B_2^{\alpha^{\prime}} \smile B_2^{\alpha^{\prime\prime}}$ term in \eqref{eq:SymTFT5dN=1full}.

\paragraph{BF terms from M-theory.}
We now plug the expansions \eqref{eq:expansionG5D} into the M-theory kinetic action \eqref{eq:MtheoryKin} and fibre-wise integrate over $L_5$. Repeatedly using \eqref{eq:G8trivial5d} and the properties of topologically trivial differential characters from Appendix \ref{app:diffcohoprods}, we obtain:
\begin{equation}
\begin{aligned}
    S_{\text{BF}}^{L_5}= \int_{Y_6} & \left\{ \frac{F_4}{2\pi} \wedge\frac{h_2}{2\pi} + \sum_{i=1}^{b^2} \left[ \frac{F_2^{i}}{2\pi} \wedge\frac{h^{i}_4}{2\pi} + \frac{F_1^{i}}{2\pi} \wedge \frac{h^{i}_5}{2\pi}\right] \right. \\
    +& \left. \kappa^{(5)} \left(  B_0\smile \delta A_5 \right) + \widetilde{\mathrm{CS}}^{(5)} \left(  B_1\smile \delta A_4 \right)  + \kappa^{(5)} \left(  B_2\smile \delta A_3 \right) \right\} \\
    +\int_{Y_6} & \left\{ \widetilde{\kappa}^{(5)} \left( \frac{F_2}{2\pi} \smile \delta A_3 \right) + \widetilde{\eta}^{(5)} \left( \frac{F_1}{2\pi} \smile \delta A_4 \right) \right. \\
    +  & \left. \widetilde{\kappa}^{(5)} \left( B_0 \smile \frac{ \dd h_5}{2\pi} \right) + \widetilde{\eta}^{(5)} \left(  B_1\smile \frac{ \dd h_4}{2\pi} \right) \right\}\,.
\end{aligned}
\label{eq:BF5dN=1full}
\end{equation}
To get the first line, the property \eqref{eq:v2v3inL5} has also been used, and moreover we have used that $\br{\dd G}_7$ is topologically trivial, to write 
\begin{equation}
	\int_{Y_6}^{\br{H}} \br{F}_p^{\bullet} \star \br{f}_{7-p}^{\bullet} = \int_{Y_6} \mathscr{F}(\br{F}_p^{\bullet}) \wedge\frac{h_{6-p}^{\bullet}}{2\pi} =  \int_{Y_6} \frac{F_p^{\bullet}}{2\pi} \wedge\frac{h_{6-p}^{\bullet}}{2\pi} .
\end{equation}\par
The first line contains BF couplings for electric $U(1)$ symmetries and their magnetic dual; the second line contains BF couplings for finite electric/magnetic dual pairs; the last two lines mix continuous and discrete symmetries and contain higher-derivative terms. Integrating by parts, these latter terms vanish in the bulk of $Y_6$, and only leave a boundary contribution, that is not relevant for the SymTFT thus we drop it.\par
We are left with the BF couplings 
\begin{equation}
\begin{aligned}
  S_{\text{BF}}^{L_5}=\int_{Y_6} & \left\{ \frac{F_4}{2\pi} \wedge\frac{h_2}{2\pi} \ + \sum_{i=1}^{b^2} \left[ \frac{F_2^{i}}{2\pi} \wedge\frac{h^{i}_4}{2\pi} \ + \frac{F_1^{i}}{2\pi} \wedge \frac{h^{i}_5}{2\pi}\right] \right. \\
    & \left. \sum_{\alpha=1}^{\mu}  \left[ B_0^{\alpha}\smile \delta A_5^{\alpha} \ + B_2^{\alpha}\smile \delta A_3^{\alpha} \right] \ - \sum_{ 1 \le \beta <\beta^{\prime}\le\widetilde{\mu}} \frac{\ell_{\beta, \beta^{\prime}}}{\mathrm{gcd}(n_{\beta}, n_{\beta^{\prime}})}  B_1^{\beta}\smile \delta A_4^{\beta^{\prime}}   \right\} .
\end{aligned}
\end{equation}
We have used the explicit coefficients computed in Appendix \ref{sec:torintegralsL5}, and $\ell_{\beta ,\beta^{\prime}} \in \Z$.\par
As elucidated at the end of Subsection \ref{sec:BraneSymTFTFree} (see around \eqref{eq:BF-likewithG4}), the BF-like terms in the first line combine nicely with the contribution from the first line of \eqref{eq:SymTFT5dN=1full}.

\subsubsection{Defects and symmetry operators from M-theory}
\label{sec:5dDefectM2M5}
We summarize the brane realization of defects and symmetry topological operators in 5d SCFTs in Table \ref{tab:M2M5-5d}. Here we focus on finite symmetries, and recall that
\begin{equation}
\begin{aligned}
	\mathrm{Tor} H_{\bullet} (L_5, \Z) &= 0 \oplus \mathrm{Tor} H_{1}(L_{5},\Z) \oplus \mathrm{Tor} H_{2}(L_{5},\Z)\oplus \mathrm{Tor} H_{3}(L_{5},\Z) \oplus 0 \oplus 0 \\
	&\cong  0 \oplus \mathrm{Tor} H^{2}(L_{5},\Z) \oplus \mathrm{Tor} H^{3}(L_{5},\Z)\oplus \mathrm{Tor} H^{2}(L_{5},\Z) \oplus 0 \oplus 0 .
\end{aligned}
\end{equation}
The discussion in the current subsection is similar to \cite{Albertini:2020mdx}, and extends it to include 4-form and $(-1)$-form symmetries. The electric/magnetic dual pairs as realized in M-theory are shown in Figure \ref{fig:EMSym5d}; to fully specify a theory, we ought to choose a polarization for each one of the three pairs.

\begin{table}[ht]
\centering
\begin{tabular}{|l|c|c|c|c|}
\hline 
& M2 & & M5 & \\
\hline
Tor$H_{1}(L_{5},\Z)\times [0,\infty)$ \hspace{4pt} & Wilson line & \begin{color}{Red}$\diamondsuit$\end{color} & Domain wall & ${\spadesuit}$  \\
Tor$H_{2}(L_{5},\Z)\times [0,\infty)$ \hspace{4pt} & Local operator & \begin{color}{blue}$\circ$\end{color} & 3d defect & \begin{color}{blue}$\triangle$\end{color} \\
Tor$H_{3}(L_{5},\Z)\times [0,\infty)$ \hspace{4pt} & --- & $\clubsuit$ & Magnetic string & \begin{color}{Red}$\heartsuit$\end{color}\\
\hline
Tor$H_{1}(L_{5},\Z)$ \hspace{4pt} & 2-form sym. generator & \begin{color}{Red}$\heartsuit$\end{color} &  $(-1)$-form sym. generator & $\clubsuit$  \\
Tor$H_{2}(L_{5},\Z)$ \hspace{4pt} & 3-form sym. generator & \begin{color}{blue}$\triangle$\end{color} &  0-form sym. generator & \begin{color}{blue}$\circ$\end{color} \\
Tor$H_{3}(L_{5},\Z)$ \hspace{4pt} & 4-form sym. generator & $\spadesuit$ & 1-form sym. generator & \begin{color}{Red}$\diamondsuit$\end{color} \\
\hline
\end{tabular}
\caption{Branes wrapping torsional cycles in $L_5$ give rise to finite $p$-form symmetries. We mark with equal symbol the charged defect and the corresponding symmetry generators.}
\label{tab:M2M5-5d}
\end{table}

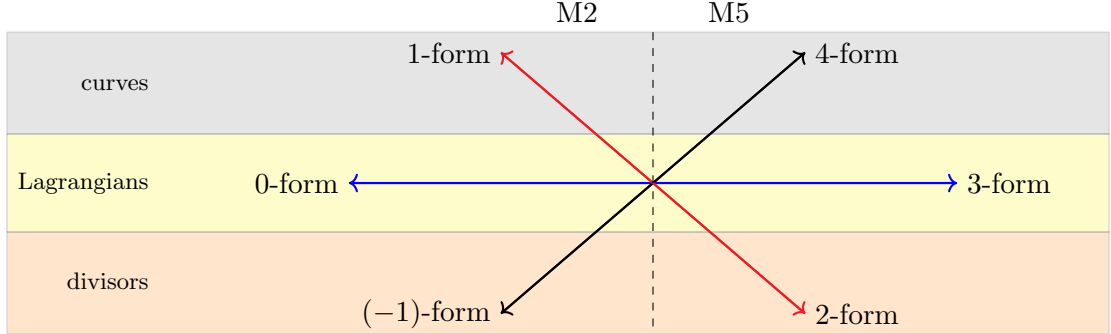
\begin{figure}
\centering
\begin{tikzpicture}
    \draw[fill=gray,opacity=0.2] (-8.5,2) -- (6,2) -- (6,0.65) -- (-8.5,0.65) -- (-8.5,2);
    \draw[fill=yellow,opacity=0.2] (-8.5,-0.65) -- (6,-0.65) -- (6,0.65) -- (-8.5,0.65) -- (-8.5,-0.65);
    \draw[fill=orange,opacity=0.2] (-8.5,-2) -- (6,-2) -- (6,-0.65) -- (-8.5,-0.65) -- (-8.5,-2);
	\node[anchor=east] (1f) at (-2,1.73) {1-form};
	\node[anchor=west] (2f) at (2,-1.73) {2-form};
	\node[anchor=east] (0f) at (-4,0) {0-form};
	\node[anchor=west] (3f) at (4,0) {3-form};
	\node[anchor=east] (mf) at (-2,-1.73) {$(-1)$-form};
	\node[anchor=west] (4f) at (2,1.73) {4-form};
	
	\draw[->,blue,thick] (-4,0) -- (4,0);
	\draw[->,blue,thick] (4,0) -- (-4,0);
	\draw[->,black,thick] (2,1.73) -- (-2,-1.73);
	\draw[->,black,thick] (-2,-1.73) -- (2,1.73);
	\draw[->,Red,thick] (-2,1.73) -- (2,-1.73);
	\draw[->,Red,thick] (2,-1.73) -- (-2,1.73);

    \draw[dashed,thin] (0,2) -- (0,-2);
    \node[anchor=south] at (-1,2) {M2}; 
    \node[anchor=south] at (1,2) {M5};

    \node[anchor=east,align=right] at (-6.5,1.3) {\footnotesize curves};
    \node[anchor=east,align=right] at (-6.5,0) {\footnotesize Lagrangians};
    \node[anchor=east,align=right] at (-6.5,-1.3) {\footnotesize divisors};
 
\end{tikzpicture}
\caption{Electric/magnetic dual symmetries in 5d SCFTs. The charged defects are realized by M2/M5-branes wrapping relative cycles in the Calabi--Yau.}
\label{fig:EMSym5d}
\end{figure}

\paragraph{Discrete 1-/2-form symmetries.}

The 1-form and 2-form symmetries of 5d SCFTs have been intensively studied \cite{Albertini:2020mdx,Morrison:2020ool,Bhardwaj:2020phs,Apruzzi:2021vcu}.
\begin{itemize}
	\item Electric 1-form symmetry.
	\begin{itemize}
	\item[---] The charged defects are Wilson lines, which are given by M2-branes wrapping non-compact (relative) torsional 2-cycles in $X_6$ which correspond to torsional 1-cycles in the link $L_5$. If we consider M-theory on the resolved $\widetilde{X}_6$, which is a gauge theory deformation of the SCFT, these line operators would correspond to the gauge Wilson lines.
	\item[---] The symmetry generators originate from M5-branes wrapping a dual torsional 3-cycle in the link.
	\end{itemize}
	\item Magnetic 2-form symmetry.
	\begin{itemize}
	\item[---] The charged defects are magnetic strings, which are given by M5-branes wrapping non-compact (relative) torsional 4-cycles in $X_6$ which correspond to torsional 3-cycles in the link $L_5$. 
	\item[---] The symmetry generators originate from M2-branes wrapping a dual torsional 1-cycle in the link.
	\end{itemize}
\end{itemize}
The finite 1-form symmetry of $\FTfive$ is understood, in a gauge theory deformation, as the symmetry that acts non-trivially on Wilson loops modulo screening. Here the screening comes from the worldline of electrically charged particles, originating from M2-branes wrapping a compact 2-cycle in $\widetilde{X}_6$ and extending in the time direction in spacetime.

\paragraph{Discrete $(-1)$-/4-form symmetries.}

\begin{itemize}
	\item 4-form symmetry.
	\begin{itemize}
	\item[---] The charged defects are (3+1)d domain walls, which are given by M5-branes wrapping non-compact (relative) torsional 2-cycles in $X_6$ which correspond to torsional 1-cycles in the link $L_5$. 
	\item[---] The symmetry generators originate from M2-branes wrapping a dual torsional 3-cycle in the link.
	\end{itemize}
	\item $(-1)$-form symmetry.
	\begin{itemize}
	\item[---] There are no charged defects under lower form symmetries. The would-be defect, filling the empty slot in Table \ref{tab:M2M5-5d}, morally originates from M2-branes wrapping relative holomorphic surfaces, and extending in a `$(-1)$d subspace' of the spacetime. 
	\item[---] The topological operators originate from M5-branes wrapping a torsional 1-cycle in the link, and filling the spacetime.
	\end{itemize}
\end{itemize}

We emphasize that discrete 4-form symmetries are classified in M-theory from the same geometric objects that classify Wilson lines. The domain walls are screened by dynamical $(3+1)$d objects from M5-branes wrapping compact curves, paralleling the 1-form symmetry story but trading M2-branes for M5-branes.\par
Starting with an electric polarization and gauging such a 4-form symmetry, leads to a theory with less domain walls, and a finite $(-1)$-form symmetry. The choice of family of domain walls that survive the gauging is selected by the discrete torsion, and the $(-1)$-form symmetry acts by shifting this choice. See \cite{Sharpe:2019ddn,Santilli:2024dyz} for related discussion in QFT.

\paragraph{Discrete 0-/3-form symmetries.} 

0-form and 3-form symmetries are magnetic dual to each other in 5d, and they were discussed in \cite{Closset:2020scj}. 
\begin{itemize}
	\item Electric 0-form symmetry.
	\begin{itemize}
	\item[---] The charged defects are local operators, and are realized from M2-branes wrapping torsional 2-cycles in the link and extending in the radial direction, intersecting the 5d spacetime at a point.
	\item[---] The symmetry generators originate from M5-branes wrapping a torsional 2-cycle in the link, and filling a codimension-1 sub-manifold in spacetime.
	\end{itemize}
	\item Magnetic 3-form symmetry.
	\begin{itemize}
	\item[---] The charged 3d defects originate from M5-branes on non-compact Lagrangian 3-cycles in any resolution $\widetilde{X}_6$. Equivalently, the M5-branes wrap a torsional 2-cycle in the link, extend in the radial direction and intersect the spacetime in codimension-2.
	\item[---] The topological operators originate from M2-branes wrapping a torsional 2-cycle in the link, and extending along a 1d curve in spacetime.
	\end{itemize}
\end{itemize}

\subsubsection{\texorpdfstring{$(-1)$}{(-1)}-form symmetry, polarization, and geometry}
\label{sec:5dminus1geometry}

\paragraph{Discrete $(-1)$-form symmetry and polarization.}
Considering again the SymTFT action \eqref{eq:SymTFT5dN=1full}, we focus on the twist terms 
\begin{equation}
	-\int_{Y_6} \left\{  \mathscr{F}(F_4) \smile \left[ \kappa^{(5)} \left(  B_0\smile B_2 \right) + \widetilde{\kappa}^{(5)} \left( \mathscr{F} (F_2) \smile B_0 \right) \right] + \kappa^{(5)} \left(  B_0\smile \delta A_5 \right)  \right\} ,
\end{equation}
which provide the SymTFT for discrete $(-1)$-form symmetries. The full symmetry is $\Z_{m_1}^{\scriptscriptstyle [-1]} \times \cdots \times \Z_{m_{\mu}}^{\scriptscriptstyle [-1]} $, and the discrete gauge fields are $B_0^{\alpha}$. In particular, the first summand couples them to the background fields $B_2^{\alpha^{\prime}}$ for the electric 1-form symmetry.\par
From \eqref{eq:BF5dN=1full}, the BF term $B_0\smile \delta A_5$ provides the electric/magnetic coupling with the discrete gauge fields $A_5^{\alpha^{\prime}}$ for the 4-form symmetries acting on domain walls, of which the $(-1)$-form symmetries is the Pontryagin dual. We observe that
\begin{itemize}
	\item The isomorphism $\mathrm{Tor}H^2 (L_5, \Z) \cong \mathrm{Tor}H^4 (L_5, \Z)$ implies that, upon expanding $\br{G}_4$, it sources the equal discrete 1-form and $(-1)$-form symmetries;
	\item There is an exact parallel between the electric/magnetic BF couplings 
        \begin{equation}
            \kappa^{(5)} \left(  B_0\smile \delta A_5 \right) \quad \text{ and } \quad \kappa^{(5)} \left(  B_2\smile \delta A_3 \right) .
        \end{equation}
\end{itemize}\par

\paragraph{Continuous $(-1)$-form symmetry.}
In the SymTFT action \eqref{eq:SymTFT5dN=1full}-\eqref{eq:BF5dN=1full}, we also note the presence of $F_1^{i}$, providing the curvatures for $U(1)^{b^2}$ $(-1)$-form symmetries. On the spacetime $M_5$, these terms source the 0-form background fields, coupled to the 5-form classes $h_5^{i} \in H^{5} (M_5, \Z)$.\par
The number of continuous $(-1)$-form symmetries is precisely equal to the flavor rank of the 5d SCFT.

\paragraph{$(-1)$-form symmetry and geometry.}

In this subsection we combine the above observations with the brane analysis in Subsection \ref{sec:5dDefectM2M5}. The fact that 1-form and $(-1)$-form symmetries are classified by the same geometric data is explained as follows.\par

\begin{table}[th]
    \centering
    \begin{tabular}{|l|c|c|}
        \hline 
Allowed? & M2 & M5 \\
\hline
Tor$H_{1}(L_{5},\Z)\times [0,\infty)$ \hspace{4pt} & $\checkmark $ & $\checkmark $  \\
Tor$H_{3}(L_{5},\Z)\times [0,\infty)$ \hspace{4pt} & & $\times$ \\
\hline
Tor$H_{1}(L_{5},\Z)$ \hspace{4pt} & $\times$ & $\times$ \\
Tor$H_{3}(L_{5},\Z)$ \hspace{4pt} & $\checkmark $ & $\checkmark $ \\
\hline
    \end{tabular}
    \caption{Setup that allows branes on the maximal collection of primitive curves in the threefold singularity. This setup realizes both the maximal possible 1-form and 4-form symmetries.}
    \label{tab:primitivecurves}
\end{table}

Imagine we start with a global form of the 5d SCFT realizing the maximal possible 1-form symmetry as a global symmetry. This allows all possible ways of wrapping M2-branes on $\mathrm{Tor} H_1(L_5,\Z) \times [0,\infty)$, which correspond to relative primitive curves in $X_6$. Wrapping M5-branes on the same maximal set of primitive curves in $X_6$ engineers domain walls in the 5d theory. The choice is tabulated in Table \ref{tab:primitivecurves}. Then, we expect that all the $(-1)$-form symmetries act trivially, while the maximal 4-form symmetry is realized.\par
Next, we gauge the 4-form symmetry. The gauging operation requires to choose a discrete theta-angle\footnote{This is not to be confused with the discrete theta-angle in for instance 5d $\mathfrak{su}(2)_\theta$ theories with $\theta=0,\pi$. The latter corresponds to different geometries, whereas here we are discussing different polarizations of 5d theories from M-theory on the same topology.} which, in our setup, is identified with a discrete scalar gauge field $B_0 $.
In the modern language, to insert a discrete theta-angle means to stack the 5d SCFT with an SPT phase in the gauging process \cite{Kapustin:2014gua,Bhardwaj:2022kot}. In the present context, the SPT phase is provided by a topological M5-brane filling the spacetime, and wrapping the same torsional 1-cycles in the link that classify the 1-form symmetry.

\subsubsection{Continuous abelian symmetries from fluxbranes}
\label{sec:5dN1fluxbrane}
The SymTFT derived from M-theory on the canonical threefold singularity $X_6$ uncovered the presence of the following continuous symmetries: 
\begin{itemize}
    \item A total of $b^2$ $U(1)^{\scriptscriptstyle [-1,i]}$ $(-1)$-form symmetries, with curvatures $F_1^{i}$;
    \item A total of $b^2$ $U(1)^{\scriptscriptstyle [0,i]}$ $0$-form symmetries, with curvatures $F_2^{i}$; these match with the maximal torus of the flavor symmetry;
    \item A universal $U(1)^{\scriptscriptstyle [2]}$ 2-form symmetry. 
\end{itemize}
We now proceed to describe the realization of the charged defects and symmetry operators from M-branes.\par

\paragraph{Defects charged under continuous symmetries.}
The defects charged under $U(1)$ symmetries are constructed from M-branes wrapping homology classes in $H_{\bullet} (L_5, \Z)_{\text{free}}$ and extending in the radial direction $[0,\infty)$. The electric defects are listed in Table \ref{tab:M2U(1)3fold}.

\begin{table}[ht]
\centering
\begin{tabular}{|l|c|c|}
\hline 
& M2 & charged under \\
\hline
$H_{0}(L_{5},\Z)_{\text{free}}\times [0,\infty) $ \hspace{2pt} & Surface defect & $U(1)^{\scriptscriptstyle [2]}$  \\
$H_{2}(L_{5},\Z)_{\text{free}}\times [0,\infty)$ \hspace{2pt} & Local operator & $U(1)^{\scriptscriptstyle [0, i]}$ \\
$H_{3}(L_{5},\Z)_{\text{free}}\times [0,\infty)$ \hspace{2pt} & --- & $U(1)^{\scriptscriptstyle [-1, i]} $ \\
\hline
\end{tabular}
\caption{Electric defects charged under continuous symmetries in 5d $\N=1$ SCFTs, and their M-theory realization.}
\label{tab:M2U(1)3fold}
\end{table}\par
For completeness, we note that we can further construct magnetic defects using M5-branes wrapping $H_{\bullet}(L_{5},\Z)_{\text{free}}\times [0,\infty)$. We list these defects in Table \ref{tab:M5U(1)3fold} but do not investigate them further.\footnote{3d defects from M5-branes wrapping free cycles have been studied in selected 5d SCFTs in \cite{Banerjee:2018syt,Banerjee:2019apt,Banerjee:2020moh,Santilli:2023fuh}.}\par

\begin{table}[ht]
\centering
\begin{tabular}{|l|c|}
\hline 
& M5 \\
\hline
$H_{2}(L_{5},\Z)_{\text{free}}\times [0,\infty) $ \hspace{2pt} & 3d defect \\
$H_{3}(L_{5},\Z)_{\text{free}}\times [0,\infty)$ \hspace{2pt} & Surface defect  \\
$H_{5}(L_{5},\Z)_{\text{free}}\times [0,\infty)$ \hspace{2pt} & Local operator \\
\hline
\end{tabular}
\caption{Magnetic defects charged under continuous symmetries in 5d $\N=1$ SCFTs, and their M-theory realization.}
\label{tab:M5U(1)3fold}
\end{table}\par

\paragraph{Symmetry operators for continuous symmetries.}
We now derive, from M-theory reduction, the topological operators that generate the symmetries listed in Table \ref{tab:M2U(1)3fold}. We find perfect agreement with the expectations from the QFT analysis.\par
\begin{itemize}
    \item Continuous electric $(-1)$-form symmetry.\par
        The $U(1)^{\scriptscriptstyle [-1,i]}$ symmetries are generated by spacetime-filling topological operators. These are realized from $P_7$-fluxbranes wrapping free 2-cycles, which admit a basis $\left\{ \mathsf{PD} (v_3^{i}) \right\}_{i=1, \dots, b^2}$. The operator representing $e^{\ii \varphi} \in U(1)$ is computed as:
        \begin{equation}
        \begin{aligned}
            \exp \left\{  \ii \varphi S^{P_7\text{-flux along }\mathsf{PD} (v_3^{i})} \right\} &= \exp \left\{ \ii \frac{\varphi}{2\pi} \int_{\mathsf{PD} (v_3^{i})\times M_5} P_7 \right\} \\
            &= \exp \left\{ \ii \frac{\varphi}{2\pi} \int_{L_5 \times M_5} v_3^{i} \wedge \left[ \sum_{j=1}^{b^2} \widetilde{h}_5^{j} \wedge v_2^{j} + \cdots \right] \right\}  \\
            &= \exp \left\{ \ii \frac{\varphi}{2\pi} \int_{M_5} \phi^{\ast} \left( h_5^{i} + g_5^{i} \right) \right\} .
        \end{aligned}
        \end{equation}
        In the second line, the ellipses omit the terms that eventually vanish, and in the third line we have used \eqref{eq:v2v3inL5}. The first piece $h_5^{i}$ comes from $G_7$, and the correction $g_{5}^{i}$ can be written in terms of the fields $\mathsf{h}_{\bullet}$ in the expansion of $H_3$ as 
        \begin{equation}
            g_5^{i} = \frac{1}{4\pi} \left(F_4 \wedge \mathsf{h}_1^{i} + \mathsf{h}_3 \wedge F_2^{i} \right) ,
        \end{equation}
        subject to $\dd \mathsf{h}_1^{i} = F_2^{i}$ and $\dd \mathsf{h}_3 = F_4$. This correction term is consistent with the pullback of $\int^{\br{H}}_{Y_6} \br{F}_4 \star \br{F}_2^{i}$ from \eqref{eq:BF-likewithG4} to the worldvolume of the fluxbrane.
    \item Continuous electric 0-form symmetry.\par
        The $U(1)^{\scriptscriptstyle [0,i]}$ symmetries are generated by $P_7$-fluxbranes wrapping free 3-cycles, which admit a basis $\left\{ \mathsf{PD} (v_2^{i}) \right\}_{i=1, \dots, b^2}$, and extending along a codimension-1 submanifold $\Sigma_4$ in spacetime. The symmetry operator is:
        \begin{equation}
        \begin{aligned}
            \exp \left\{  \ii \varphi S^{P_7\text{-flux along }\mathsf{PD} (v_2^{i})} (\Sigma_4)\right\} &= \exp \left\{ \ii \frac{\varphi}{2\pi}\int_{\mathsf{PD} (v_2^{i})\times \Sigma_4} P_7 \right\} \\
            &= \exp \left\{ \ii \frac{\varphi}{2\pi} \int_{L_5 \times \Sigma_4} v_2^{i} \wedge \left[ \sum_{j=1}^{b^2} \widetilde{h}_4^{j} \wedge v_3^{j} + \cdots \right] \right\}  \\
            &= \exp \left\{ \ii \frac{\varphi}{2\pi} \int_{\Sigma_4} \phi^{\ast} \left( h_4^{i} + g_4^{i} \right) \right\} .
        \end{aligned}
        \end{equation}
        Again we omit terms that eventually vanish, and use \eqref{eq:v2v3inL5}. The correction term in this case is written as 
        \begin{equation}
            g_4^{i} = \frac{1}{4\pi} \left( F_4 \wedge \mathsf{h}_0^{i} + \mathsf{h}_3 \wedge F_1^{i} \right) .
        \end{equation}
    \item Continuous electric 2-form symmetry.\par
        The symmetry topological operator of the universal $U(1)^{\scriptscriptstyle [2]}$ electric 2-form symmetry is engineered by a $P_7$-fluxbrane wrapping the whole $L_{5}$, and filling $\Sigma_2 \subset M_5$. The symmetry operator is:
        \begin{equation}
        \begin{aligned}
            \exp \left\{  \ii \varphi S^{P_7\text{-flux along }L_5} \right\} &=\exp \left\{ \ii \frac{\varphi}{2\pi} \int_{L_5 \times \Sigma_2} P_7 \right\} \\
            &= \exp \left\{ \ii \frac{\varphi}{2\pi} \int_{\Sigma_2 } \phi^{\ast} \left( h_2 + g_2 \right)\right\} .
        \end{aligned}
        \end{equation}
        This time neither $H_3$ nor $G_4$ admit an expansion in $\vol_5$, but their product yields a correction term
        \begin{equation}
        \begin{aligned}
            g_2 & = \frac{1}{4\pi} \int_{L_5} H_3 \wedge G_4 \\
            & = \frac{1}{4\pi} \sum_{i,j=1}^{b^2} \int_{L_5} \left( \mathsf{h}_1^{i} \wedge v_2^{i} + \mathsf{h}_0^{i} \wedge v_3^{i} \right) \wedge \left( F_2^{j} \wedge v_2^{j} + F_1^{j} \wedge v_3^{j} \right)  \\
            &=  \frac{1}{4\pi} \sum_{i=1}^{b^2}\left( \mathsf{h}_1^{i} \wedge F_1^{i} + \mathsf{h}_0^{i} \wedge F_2^{i} \right) .
        \end{aligned}
        \end{equation}
\end{itemize}

\subsection{SymTFT of 5d \texorpdfstring{$\mathfrak{su}(N)$}{SU(N)} SYM from M-theory}
\label{sec:SymTFT5dSUN0}
We study the 5d $\mathcal{N}=1$ SCFT that flows to $\mathfrak{su} (N)$ Yang--Mills theory on its Coulomb branch, geometrically engineered from M-theory on a toric Calabi--Yau threefold.

\subsubsection{Geometry and torsion cohomology}

\paragraph{Toric geometric data.} 
\begin{figure}
\centering
\begin{tikzpicture}
	\draw (-1,0) -- (0,0) -- (1,4) -- (0,4) -- (-1,0);
	\node at (-1,0) {$\bullet$};
	\node at (0,0) {$\bullet$};
	\node at (1,4) {$\bullet$};
	\node at (0,4) {$\bullet$};
	\node at (0,2) {$\scriptstyle \vdots$};
	\node at (0,1) {$\bullet$};
	\node at (0,3) {$\bullet$};
	\node[anchor=south] at (0,4) {$D_4$};
	\node[anchor=south] at (1,4) {$D_3$};
	\node[anchor=east] at (-1,0) {$D_1$};
	\node[anchor=west] at (0,0) {$D_2$};
	\node at (0,0) {$\bullet$};
	\node at (1,4) {$\bullet$};
	\node at (0,4) {$\bullet$};
	
	\node[anchor=south west] at (4,2) {$\mE \equiv D_1 $};
	\node[anchor=north west] at (4,2) {$\mE^{\prime} \equiv -N (D_1 +D_4) $};
\end{tikzpicture}
\caption{Toric diagram for the 5d $\mathcal{N}=1$ SCFT with $\mathfrak{su}(N)_0$ gauge theory description.}
\label{fig:YN0toric}
\end{figure}
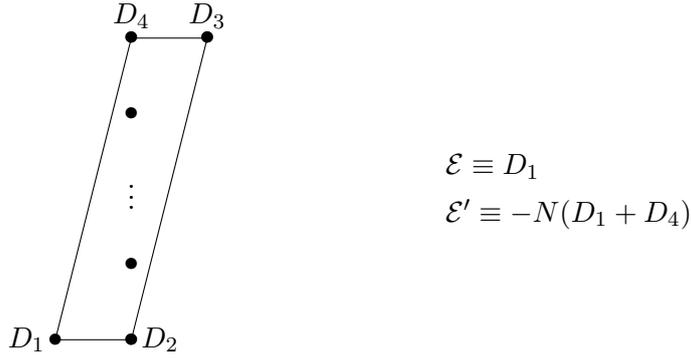
Let $\widetilde{X}_6$ be a toric crepant resolution of the threefold singularity. It contains $N-1$ compact toric divisors $\{ S_j \}_{j=1, \dots, N-1}$, and $4$ non-compact toric 4-cycles $\{ D_1, \dots, D_4 \}$.\par 
As is well known, not all the toric divisors are independent, and they are subject to three relations. We impose the three relations on the non-compact divisors, and take a non-compact $\mE$ together with the $(N-1)$ compact ones to be the independent ones. We choose the free 4-cycle $\mE$ to be the generator of the Cartan subalgebra of the flavor symmetry. Our conventions for this toric geometry are illustrated in Figure \ref{fig:YN0toric}. The generators are chosen to match the identification with the physical operators in \cite{Bhardwaj:2020phs}.\par

\paragraph{Geometric data of the link.}
The link of the singularity is $Y^{N,0}$, whose (co)homology is 
\begin{equation}
    H_{5-\bullet} (Y^{N,0},\Z) = H^{\bullet} (Y^{N,0},\Z) = \ \Z  \ \oplus \ 0 \ \oplus \ \begin{matrix} \Z \\ \oplus \\ \Z_N\end{matrix} \ \oplus \ \Z  \  \oplus \ \Z_N \ \oplus \ \Z .
\end{equation}
The generators of $\br{H}^{\bullet} (Y^{N,0})$ are denoted as $\br{1}, \br{v}_2,\br{t}_2, \br{v}_3, \br{t}_4, \br{\vol}_5$.\par
The 3-chain $N\mathsf{PD} (t_2)$ in $Y^{N,0}$ pulls back to a compact divisor $\mE^{\prime}$ in $\widetilde{X}_6$, which in our conventions is 
\begin{equation}
    \mE^{\prime} \equiv -N (D_1 +D_4)  \cong \sum_{j=1}^{N-1} j S_j ,
\end{equation}
with triple intersection number $(\mE^{\prime})^3 \in \Z$. Furthermore we fix the generator $\br{t}_4$ by duality, as explained in Appendix \ref{sec:torintegralsL5},
\begin{equation}
	\int^{\br{H}}_{L_5} \br{t}_4 \star \br{t}_2 = - \frac{1}{N} .
\end{equation}
Thus, its Poincar\'e dual 1-cycle satisfies the linking pairing 
\begin{equation}
\label{eq:normt4inYN0}
	\ell_{Y^{N,0}} \left( N \mathsf{PD} (t_2), N \mathsf{PD} (t_4)\right) = -N ,
\end{equation}
which agrees with \cite[App.B1]{Albertini:2020mdx}.\par
From differential cohomology we also have
\begin{equation}
\begin{aligned}
	\frac{1}{N^3} \ell_{Y^{N,0}} \left( N \mathsf{PD} (t_2),  N \mathsf{PD} (t_2), N \mathsf{PD} (t_2)\right) &= \int_{Y^{N,0}}^{\br{H}} \br{t}_2\star\br{t}_2 \star \br{t}_2 = - \frac{1}{N} \langle c (t_2) \smile c(t_2) , \Sigma^{\prime}_4 \rangle \\
	\frac{1}{N^2} \ell_{Y^{N,0}} \left( N \mathsf{PD} (t_2),  N \mathsf{PD} (t_2), \mathsf{PD} (v_2)\right) &= \int_{Y^{N,0}}^{\br{H}} \br{t}_2\star\br{t}_2 \star \br{v}_2 = - \frac{1}{N} \langle c (t_2) \smile c(t_2) , \Sigma_4 \rangle
\end{aligned}
\end{equation}
for $\partial \Sigma_4^{\prime}=  N \mathsf{PD} (t_2)$, $\partial \Sigma_4= \mathsf{PD} (v_2)$. We thus get the predictions
\begin{equation}
	\frac{1}{N^2}(\mE^{\prime})^3 \in \Z , \qquad \frac{1}{N} (\mE^{\prime})^2 \cdot \mE \in \Z .
\end{equation}
In fact, it turns out that \cite[Eq.(4.29)]{Apruzzi:2021nmk}
\begin{equation}
\label{eq:CSYN0inZ}
\begin{aligned}
	6 \mathrm{CS}^{Y^{N,0}} &:= \int^{\br{H}}_{Y^{N,0}} \br{t}_2 \star \br{t}_2 \star \br{t}_2 = \frac{(\mE^{\prime})^3}{N^3} =  N-1 \\
	2\eta^{Y^{N,0}} &:= \int_{Y^{N,0}}^{\br{H}} \br{t}_2 \star \br{t}_2 \star \br{v}_2 = \frac{\mE^{\prime} \cdot \mE^{\prime} \cdot \mE}{N^2}  =- \frac{N-1}{N} .
\end{aligned}
\end{equation}
Moreover, since there is only one generator in $\mathrm{Tor} H^4 (Y^{N,0}, \Z)$, we also have $\br{t}_2 \star\br{t}_2\in \Z \br{t}_4$; the coefficient is fixed comparing \eqref{eq:normt4inYN0} with \eqref{eq:CSYN0inZ}. From self-consistency we get  
\begin{equation}
	\widetilde{\kappa}^{Y^{N,0}} :=\int_{Y^{N,0}}^{\br{H}} \br{t}_4 \star \br{v}_2 = \frac{1}{N} .
\end{equation}

\paragraph{Example: Local \texorpdfstring{$\mathbb{P}^1 \times \mathbb{P}^1$}{P1xP1}.}
We cross-check our calculations in $\mathfrak{su}(2)$ Yang--Mills theory, engineered by the local del Pezzo surface $\mathbb{P}^1 \times \mathbb{P}^1$. In Appendix \ref{sec:torintegralsL5}, we compare the results with the triple intersection numbers in \cite[Eq.(B.18)]{Closset:2018bjz}, finding perfect agreement.

\subsubsection{SymTFT of 5d \texorpdfstring{$\mathfrak{su}(N)_0$}{SU(N)} SYM}
\label{sec:SymTFT5dSUNnoCS}
We expand 
\begin{equation}
\begin{aligned}
	\br{G}_4 &= \br{F}_4 \star \br{1} + \br{B}_2^{\scriptscriptstyle [1,\mathrm{e}]} \star \br{t}_2 + \br{F}_2^{\scriptscriptstyle [0,\mathrm{in}]}  \star \br{v}_2 + \br{F}_1^{\scriptscriptstyle [-1]} \star \br{v}_3 + B_0^{\scriptscriptstyle [-1]} \star \br{t}_4 , \\
	\br{\dd G}_7 &= \br{f}_8 \star \br{1} + \br{\mathcal{B}}_6 \star \br{t}_2 + \br{f}_6 \star \br{v}_2 + \br{f}_5 \star \br{v}_3 + \br{\mathcal{B}}_4\star \br{t}_4 + \br{f}_3 \star \br{\vol}_5 .
\end{aligned}
\end{equation}

\paragraph{Mixed anomalies from M-theory.}
Plugging the form of $\br{G}_4$ above in the topological M-theory action \eqref{eq:MtheoryS}, and integrating along the link $Y^{N,0}$, we get 
\begin{equation}
\begin{aligned}
	S_{\text{twist}}^{Y^{N,0}} = \int_{Y_6} &\left\{ - \frac{(\mE^{\prime})^3}{6N^3} B_2^{\scriptscriptstyle [1,\mathrm{e}]} \smile B_2^{\scriptscriptstyle [1,\mathrm{e}]} \smile B_2^{\scriptscriptstyle [1,\mathrm{e}]} + \frac{1}{N} \mathscr{F} (\br{F}_4) \smile B_2^{\scriptscriptstyle [1,\mathrm{e}]} \smile B_0^{\scriptscriptstyle [-1]} \right. \\
	& - \frac{((\mE^{\prime})^2 \cdot \mE) }{2N^2} \mathscr{F} (\br{F}_2^{\scriptscriptstyle [0,\mathrm{in}]} ) \smile B_2^{\scriptscriptstyle [1,\mathrm{e}]} \smile B_2^{\scriptscriptstyle [1,\mathrm{e}]}  \\
	& \left. - \widetilde{\kappa}^{Y^{N,0}}  \mathscr{F} (\br{F}_4) \smile \mathscr{F} (\br{F}_2^{\scriptscriptstyle [0,\mathrm{in}]} ) \smile B_0^{\scriptscriptstyle [-1]} \right\} \ + \int_{Y_6}^{\br{H}} \br{F}_4 \star \br{F}_2^{\scriptscriptstyle [0,\mathrm{in}]}  \star \br{F}_1^{\scriptscriptstyle [-1]}  .
\end{aligned}
\end{equation}
In this example, the additional piece \eqref{eq:G4G8shiftB2cubic} from $\int^{\br{H}}_{L_5} \br{G}_4 \star \br{\mathsf{X}}_8$ vanishes \cite[Eq.(4.30)]{Apruzzi:2021nmk}. From here, we detect the following backgrounds. 
\begin{itemize}
\item[---] Discrete gauge fields: $B_2^{\scriptscriptstyle [1,\mathrm{e}]}$ for the $\Z_N^{\scriptscriptstyle [1,\mathrm{e}]}$ electric 1-form symmetry; $B_0^{\scriptscriptstyle [-1]}$ for a $\Z_N^{\scriptscriptstyle [-1]}$ $(-1)$-form symmetry.
\item[---] Continuous gauge fields: the curvature $F_2^{\scriptscriptstyle [0,\mathrm{in}]}$ of the $U(1)^{\scriptscriptstyle [0,\mathrm{in}]}$ flavor symmetry acting on the instantons; additionally, the curvatures $F_4$ and $F_1^{\scriptscriptstyle [-1]}$ of continuous 2-form and $(-1)$-form symmetries, respectively.
\end{itemize}
Let us analyze the potential 't Hooft anomalies.
\begin{itemize}
\item The first term yields a potential 't Hooft anomaly for the electric 1-form symmetry $\Z_N^{\scriptscriptstyle [1,\mathrm{e}]}$, whose background gauge field is $B_2^{\scriptscriptstyle [1,\mathrm{e}]}$. However, due to \eqref{eq:CSYN0inZ}, the coefficient is such that the 1-form symmetry is anomaly-free.
\item The second term indicates a mixed 't Hooft anomaly between the 1-form symmetry and the finite $(-1)$-form symmetry.
\item The third term couples the background gauge field for $\Z_N^{\scriptscriptstyle [1,\mathrm{e}]}$ with the field strength $F_2^{\scriptscriptstyle [0,\mathrm{in}]}$ of the continuous $U(1)^{\scriptscriptstyle [0,\mathrm{in}]}$ instanton symmetry. From \eqref{eq:CSYN0inZ}, the anomalous coupling has coefficient $-(\mE^{\prime})^2 \cdot \mE /N^2 =  (N-1)/N$. This mixed 't Hooft anomaly was first analyzed from field theoretic methods in \cite{BenettiGenolini:2020doj}.
\item The fourth term couples the instanton field strength with the discrete gauge field for the $(-1)$-form symmetry.
\end{itemize}
The field strength $F_1^{\scriptscriptstyle [-1]}$ only appears in this part of the SymTFT through the last term.

\paragraph{BF terms from M-theory.}
Inserting the expansions above into \eqref{eq:MtheoryKin} we find:
\begin{equation}
\begin{aligned}
	S_{\text{BF}}^{Y^{N,0}} = &\int_{Y_6}^{\br{H}} \left[ \br{F}_2^{\scriptscriptstyle [0,\mathrm{in}]} \star \br{f}_5 + \br{F}_1^{\scriptscriptstyle [-1]} \star \br{f}_6 + \br{F}_4 \star \br{f}_3 \right] \\
	-&  \frac{1}{N} \int_{Y_6} \left[ B_2^{\scriptscriptstyle [1,\mathrm{e}]} \smile a_4  + B_0^{\scriptscriptstyle [-1]} \smile a_6 \right] \\
	+& \widetilde{\kappa}^{Y^{N,0}} \int_{Y_6}\left[ \mathscr{F} (\br{F}_2^{\scriptscriptstyle [0,\mathrm{in}]}) \smile a_4 + \mathscr{F} (\br{f}_6) \smile B_0^{\scriptscriptstyle [-1]} \right]
\end{aligned}
\end{equation}
up to integer shifts. Using that $\br{\dd G}_7$ is topologically trivial by construction, this reduces to 
\begin{equation}
\begin{aligned}
	S_{\text{BF}}^{Y^{N,0}} = &\int_{Y_6}\left[ \frac{F_2^{\scriptscriptstyle [0,\mathrm{in}]} }{2\pi} \wedge \frac{h_4}{2 \pi} + \frac{F_1^{\scriptscriptstyle [-1]}}{2\pi} \wedge \frac{h_5}{2\pi} + \frac{F_4}{2\pi}\wedge \frac{h_2}{2\pi} \right] \\
	-& \frac{1}{N} \int_{Y_6} \left[ B_2^{\scriptscriptstyle [1,\mathrm{e}]} \smile \delta A_3^{\scriptscriptstyle [2,\mathrm{m}]} + B_0^{\scriptscriptstyle [-1]} \smile \delta A_5^{\scriptscriptstyle [4]} \right] \\
	+& \widetilde{\kappa}^{Y^{N,0}} \int_{Y_6} \left[ \delta B_0^{\scriptscriptstyle [-1]} \smile \frac{h_5}{2\pi} + \frac{F_2^{\scriptscriptstyle [0,\mathrm{in}]} }{2\pi}  \smile \delta A_3^{\scriptscriptstyle [2,\mathrm{m}]} \right] .
\end{aligned}
\label{eq:SBFYN0}
\end{equation}
Focusing on the BF terms in the second line, we observe:
\begin{itemize}
\item The usual BF coupling $- \frac{1}{N} B_2^{\scriptscriptstyle [1,\mathrm{e}]} \smile \delta A_3^{\scriptscriptstyle [2,\mathrm{m}]}$ of background gauge fields for the electric 1-form and magnetic 2-form symmetries;
\item The analogous BF coupling $-\frac{1}{N} B_0^{\scriptscriptstyle [-1]} \smile \delta A_5^{\scriptscriptstyle [4]}$ between the electric/magnetic pair of $(-1)$-form and 4-form symmetries.
\end{itemize}
The last line contains higher-derivative terms mixing continuous and finite symmetries, and can be ignored.

\subsubsection{SymTFT of 5d \texorpdfstring{$\mathfrak{su}(N)_k$}{SU(N)} SYM from M-theory}
\label{sec:SymTFT5dSUNk}

For completeness, we repeat the analysis by adding a Chern--Simons term and studying 5d $\mathcal{N}=1$ $\mathfrak{su}(N)_k$ gauge theory. This theory is also geometrically engineered in toric geometry, and the link of the singularity is $Y^{N,k}$. This case has been studied in the literature \cite{Morrison:2020ool,Albertini:2020mdx,Apruzzi:2021nmk,DelZotto:2024tae} as well, thus we will be brief, and simply highlight the changes from the case $k=0$.\par

\paragraph{Geometric data of the link.}
The link of the singularity is $Y^{N,k}$, with (co)homology 
\begin{equation}
    H_{5-\bullet} (Y^{N,k}) = H^{\bullet} (Y^{N,k}) = \ \Z  \ \oplus \ 0 \ \oplus \ \begin{matrix} \Z \\ \oplus \\ \Z_{\mathrm{gcd} (N,k)}\end{matrix} \ \oplus \ \Z  \  \oplus \ \Z_{\mathrm{gcd} (N,k)} \ \oplus \ \Z .
\end{equation}
A computation as above yields 
\begin{equation}
\label{eq:CSYNkinZ}
\begin{aligned}
	6 \mathrm{CS}^{Y^{N,k}} &:= \int^{\br{H}}_{Y^{N,k}} \br{t}_2 \star \br{t}_2 \star \br{t}_2 = \frac{(\mE^{\prime})^3}{(\mathrm{gcd}(N,k))^3} , \\
	2\eta^{Y^{N,k}} &:= \int_{Y^{N,k}}^{\br{H}} \br{t}_2 \star \br{t}_2 \star \br{v}_2 = \frac{\mE^{\prime} \cdot \mE^{\prime} \cdot \mE}{(\mathrm{gcd}(N,k))^2} ,
\end{aligned}
\end{equation}
together with the predictions
\begin{equation}
	\frac{1}{(\mathrm{gcd}(N,k))}(\mE^{\prime})^3 \in \Z , \qquad \frac{1}{\mathrm{gcd}(N,k)} (\mE^{\prime})^2 \cdot \mE \in \Z,
\end{equation}
in agreement with \cite[Eq.(4.29)]{Apruzzi:2021nmk}.

\paragraph{Mixed anomalies from M-theory.}
Proceeding as in Subsection \ref{sec:SymTFT5dSUNnoCS}, we compute $S_{\text{CS}}^{Y^{N,k}} $. The differences with the $k=0$ theory are:
\begin{itemize}
\item The presence of a Chern--Simons level $k$ breaks the finite 1-form and $(-1)$-form symmetries down to $\Z_{\mathrm{gcd}(N,k)}$, in agreement with the field theory result;
\item The coupling $B_2^{\scriptscriptstyle [1,\mathrm{e}]} \smile B_2^{\scriptscriptstyle [1,\mathrm{e}]} \smile B_2^{\scriptscriptstyle [1,\mathrm{e}]}$ comes with coefficient $\mathrm{CS}^{Y^{N,k}} \notin \Z$ if $k \ne 0$, whence the 1-form symmetry is anomalous in this situation.
\end{itemize}

\paragraph{BF terms from M-theory.}
Continuing as in Subsection \ref{sec:SymTFT5dSUNnoCS} we easily find the BF couplings.
They are analogous as in the previous case, where now the electric/magnetic pairs of 1-/2-form symmetries and $(-1)$-/4-form symmetries are broken to $\Z_{\mathrm{gcd}(N,k)}$ by the presence of the Chern--Simons term.

\subsection{SymTFT of 5d orbifold SCFTs}
\label{sec:C3orbi}

We now consider 5d SCFTs geometrically engineered placing M-theory on the orbifold threefolds $\C^3 /\Gamma$, where $\Gamma$ is a finite subgroup of $SU(3)$. The 1-/2-form symmetries of these 5d SCFTs have been thoroughly studied in \cite{Tian:2021cif,Acharya:2021jsp,DelZotto:2022fnw}. We briefly complement the analysis therein by pursuing the $(-1)$-form symmetries.

\paragraph{Orbifold geometric data.} 

We consider $\Gamma \subset SU(3)$ acting on $\C^3$ such that the orbifold $\C^3/\Gamma$ is a canonical singularity; the possibilities have been classified in \cite{yau1993gorenstein}. The link of the singularity is 
\begin{equation}
    L_5= \mathbb{S}^5/\Gamma .
\end{equation}
Differently from the more familiar case of du Val singularities, whose links are smooth, the action of $\Gamma$ on $\mathbb{S}^5$ may have fixed point, thus $L_5$ itself may contain singularities.\par
The study of the free part of the (co)homology of $L_5$ will not add to the existing literature, thus we henceforth restrict our attention to the torsion part, and refer to the general analysis in Subsection \ref{sec:SymTFT5DCY3} for the treatment of the free part.\par

\paragraph{Geometric data of the link.}
If $\Gamma$ acts freely on $\mathbb{S}^5$, the (co)homology of the link has 
\begin{equation}
    \mathrm{Tor}H_{5-\bullet} (\mathbb{S}^5/\Gamma,\Z) = \mathrm{Tor}H^{\bullet} (\mathbb{S}^5/\Gamma,\Z) = \ 0 \ \oplus \ 0 \ \oplus \ \Gamma^{\mathrm{ab}} \ \oplus \ 0  \  \oplus \ \Gamma^{\mathrm{ab}} \ \oplus \ 0 ,
\end{equation}
where $\Gamma^{\mathrm{ab}}$ is the abelianization of $\Gamma$. We have used the universal coefficient theorem to get $\mathrm{Tor} H_3 (\mathbb{S}^5/\Gamma,\Z)$ as the Pontryagin dual of $\mathrm{Tor} H_1 (\mathbb{S}^5/\Gamma,\Z)$.\par
When the action of $\Gamma$ on $\mathbb{S}^5$ has fixed points, let $\Gamma_{\text{fixed}} \subseteq \Gamma $ be the normal subgroup generated by the elements with fixed points, and $\Gamma_{\triangleright} := \Gamma / \Gamma_{\text{fixed}}$. Then, thanks to a theorem of \cite{Armstrong1968}, \cite{DelZotto:2022fnw} argues that 
\begin{equation}
    H_{1} (\mathbb{S}^5/\Gamma,\Z) = \Gamma_{\triangleright}^{\mathrm{ab}} \cong \Z_{m_1} \times \cdots \times \Z_{m_{\mu}} .
\end{equation}
By the universal coefficient theorem and Poincar\'e duality, we also have
\begin{equation}
    \mathrm{Tor}H^{\bullet} (\mathbb{S}^5/\Gamma,\Z) = \ 0  \ \oplus \ 0 \ \oplus \ \Gamma_{\triangleright}^{\mathrm{ab}} \ \oplus \ 0  \  \oplus \ \Gamma_{\triangleright}^{\mathrm{ab}} \ \oplus \ 0.
\end{equation}
In conclusion, for our purposes we only need to consider $\Gamma_{\triangleright} \subset \Gamma$, which acts freely.\par
The generators of $\br{H}^{\bullet} (\mathbb{S}^5/\Gamma)$ are chosen according to Subsection \ref{sec:SymTFT5DCY3}.

\paragraph{$(-1)$-form symmetry and polarization.}
\label{sec:SymTFT5dorbi}
Specializing the analysis of Subsection \ref{sec:SymTFT5DCY3} to the present geometry, we deduce the following.
\begin{itemize}
    \item The maximal electric 1-form symmetry is $\Z_{m_{1}}^{\scriptscriptstyle [1,\mathrm{e}]}\times \cdots \times \Z_{m_{\mu}}^{\scriptscriptstyle [1,\mathrm{e}]} $. For each such symmetry we have a magnetic dual $\Z_{m_{\alpha}}^{\scriptscriptstyle [2,\mathrm{m}]}$, and a BF term of the form $- \frac{1}{m_{\alpha}} B_2^{\alpha,{\scriptscriptstyle [1,\mathrm{e}]}} \smile \delta A_3^{\alpha, {\scriptscriptstyle [2,\mathrm{m}]}}$.
    \item For every electric 1-form symmetry, there is a $(-1)$-form symmetry $\Z_{m_{\alpha}}^{\scriptscriptstyle [-1]}$, accompanied by the dual 4-form symmetry $\Z_{m_{\alpha}}^{\scriptscriptstyle [4]}$, with the analogous BF term.\par 
    The symmetry operators for the $(-1)$-form symmetries are realized by M5-branes wrapping a cycle in $\mathrm{Tor} H_1 (\mathbb{S}^5/\Gamma,\Z)$, while the symmetry operators for the 4-form symmetries are realized by M2-branes wrapping a cycle in $\mathrm{Tor} H_3 (\mathbb{S}^5/\Gamma,\Z)$.
\end{itemize}
Therefore, for every pair of electric 1-form/magnetic 2-form symmetries in \cite{DelZotto:2022fnw}, we should include the $(-1)$/4-form pair. To fully specify the global form of the theory, one has to also specify a polarization for the non-trivially acting subset of these symmetries.

\section{Geometric engineering of \texorpdfstring{$(-1)$}{(-1)}-form in 4d SYM theories}
\label{sec:4d-SymTFT}

\subsection{SymTFT of 4d \texorpdfstring{$\mathcal{N}=2$}{N=2} KK theories from M-theory}
\label{sec:4dN=2fromM}

One way to construct four-dimensional $\mathcal{N}=2$ gauge theories is via geometric engineering from Type IIA string theory on a non-compact Calabi--Yau threefold $X_6$ \cite{Katz:1996fh}. At finite string coupling, Type IIA string theory is described by M-theory on a circle of finite radius. Hence, upon further compactification on $X_6$, one obtains a 5d theory on a circle, descending to a 4d $\mathcal{N}=2$ KK theory.\par
We derive the SymTFT of 4d $\mathcal{N}=2$ KK theories engineered by M-theory on $X_6 \times \mathbb{S}^1$ from the circle compactification of the SymTFT derived in Subsection \ref{sec:SymTFT5DCY3}. The strategy is illustrated in Figure \ref{fig:4Dfrom5D}.

\begin{figure}[htb]
\centering
\begin{tikzpicture}
	\node[align=center] (A) at (0,0) {Type IIA on $X_6$\\ \footnotesize at finite coupling};
	\node[align=center] (4T) at (8,0) {4d $\mathcal{N}=2$\\ KK theory}; 
	\node[draw,rectangle,rounded corners,align=center] (4S) at (8,-2) {SymTFT of 4d $\mathcal{N}=2$\\ KK theory}; 
	\node (M) at (0,3) {M-theory on $X_6$};
	\node (5T) at (8,3) {5d $\mathcal{N}=1$}; 
	\node[draw,rectangle,rounded corners] (5S) at (8,5) {SymTFT of 5d $\mathcal{N}=1$};
	
	\path[->] (M) edge node[align=center,anchor=east] {\footnotesize compactify\\ \footnotesize on $\mathbb{S}^1$} (A);
	\path[->] (A) edge node[align=center,anchor=north east,pos=0.67] {\footnotesize Link \footnotesize reduction} (4S);
	\path[->] (A) edge node[align=center,anchor=south] {\footnotesize Geometric engineering} (4T);

	\path[->,thick,blue] (M) edge node[align=center,anchor=south east,pos=0.67] {\footnotesize Link \footnotesize reduction} (5S);
	\path[->] (M) edge node[align=center,anchor=south] {\footnotesize Geometric engineering} (5T);
	
	\path[->] (5T) edge node[align=center,anchor=west] {\footnotesize compactify\\ \footnotesize on $\mathbb{S}^1$} (4T);
	\path[->,thick,purple] (5S.south east) edge[bend left] node[align=center,anchor=west] {\footnotesize $\mathbb{S}^1$\\ \footnotesize reduction} (4S.north east);
\end{tikzpicture}
\caption{Geometric engineering and SymTFT of 4d $\mathcal{N}=2$ KK theories via M-theory. The step \begin{color}{blue}$\rightarrow$\end{color} has been carried out in Subsection \ref{sec:SymTFT5DCY3}. In this section we perform the step \begin{color}{purple}$\rightarrow$\end{color}.}
\label{fig:4Dfrom5D}
\end{figure}
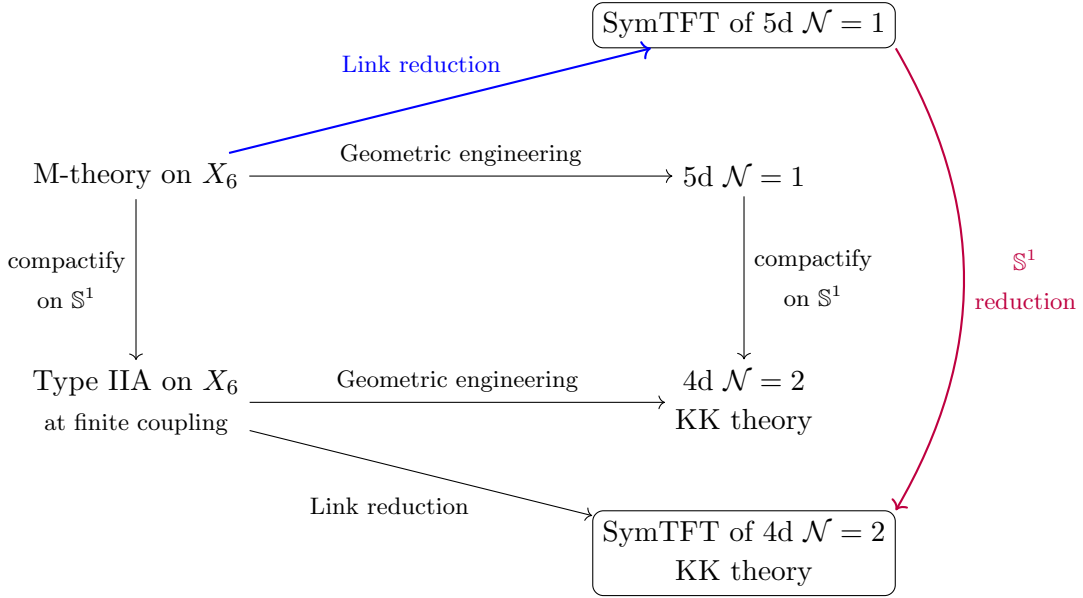

\subsubsection{SymTFT from circle reduction}

The 11d geometry is $X_6 \times \mathbb{S}^1 \times M_4$. The SymTFT living on $Y_6 \cong [0, \infty)\times \mathbb{S}^1 \times M_4$ is a specialization of the result in Subsection \ref{sec:SymTFT5DCY3}. To reduce along the free cycle generated by $[\mathbb{S}^1]$, we denote by $v$ the (pullback to $\mathbb{S}^1 \times M_4 $ of the) fundamental class of $\mathbb{S}^1$, and expand 
\begin{equation}
\begin{aligned}
    \br{F}_{\bullet}^{\cdot} &= \br{F}_{\bullet}^{\cdot} \star \br{1} + \br{\Phi}_{\bullet -1}^{\cdot} \star \br{v} , \\
    \br{B}_{\bullet}^{\cdot} &= \br{B}_{\bullet}^{\cdot} \star \br{1} + \br{\mathsf{b}}_{\bullet -1}^{\cdot} \star \br{v} , \\
    \br{h}_{\bullet}^{\cdot} &= \br{h}_{\bullet}^{\cdot} \star \br{1} + \br{\phi}_{\bullet -1}^{\cdot} \star \br{v} , \\
    \br{A}_{\bullet}^{\cdot} &= \br{A}_{\bullet}^{\cdot} \star \br{1} + \br{\mathsf{a}}_{\bullet -1}^{\cdot} \star \br{v} ,
\end{aligned}
\end{equation}
with a slight abuse of notation.

\paragraph{Mixed anomalies from circle reduction.}
From \eqref{eq:SymTFT5dN=1full}, the SymTFT of the 4d $\mathcal{N}=2$ KK theory engineered by M-theory on $X_6 \times \mathbb{S}^1$ captures the following 't Hooft anomalies:
\begin{equation}
\begin{aligned}
    S_{\text{twist}}^{L_5 \times \mathbb{S}^1} = \int_{Y_5}^{\br{H}} & \sum_{i=1}^{b^2} \left[ \br{\Phi}_3 \star \br{F}_2^{i} \star \br{F}_1^{i} + \br{F}_4 \star \br{\Phi}_1^{i} \star \br{F}_1^{i} - \br{F}_4 \star \br{F}_2^{i} \star \br{\Phi}_0^{i} \right] \\
     -\int_{Y_5} & \left\{  \mathscr{F}(\br{F}_4) \smile \left[ \kappa^{(5)} \left(  B_0\smile \mathsf{b}_1 \right) + \widetilde{\kappa}^{(5)} \left( \mathscr{F} (\br{\Phi}_1) \smile B_0 \right) \right] \right. \\
     & + \mathscr{F}(\br{\Phi}_3) \smile \left[ \kappa^{(5)} \left(  B_0\smile B_2 \right) + \widetilde{\kappa}^{(5)} \left( \mathscr{F} (\br{F}_2) \smile B_0 \right) \right] \\
      & + 3 \mathrm{CS}^{(5)} \left(  B_2 \smile B_2 \smile \mathsf{b}_1\right) + 2 \eta^{(5)} \left( \mathscr{F} (\br{F}_2) \smile B_2 \smile \mathsf{b}_1 \right) \\
      & + \eta^{(5)} \left( \mathscr{F} (\br{\Phi}_1) \smile B_2 \smile B_2 \right)   \\
      &  - \sum_{1 \le \beta \ne \beta^{\prime} \le \widetilde{\mu}} \mathscr{F} (\br{F}_4) \smile \widetilde{\mathrm{CS}}^{(5)} \left(  B_1^{\beta}\smile \mathsf{b}_0^{\beta^{\prime}} \right) \\ & + \mathscr{F} (\br{F}_4) \smile \left[ \widetilde{\eta}^{(5)}  \left( \mathscr{F} (\br{F}_1) \smile \mathsf{b}_0 \right)  - \widetilde{\eta}^{(5)}  \left( \mathscr{F} (\br{\Phi}_0) \smile B_1 \right) \right] \\
      & \left. + \mathscr{F} (\br{\Phi}_3) \smile \left[  - \widetilde{\mathrm{CS}}^{(5)} \left(  B_1\smile B_1 \right) + \widetilde{\eta}^{(5)}  \left( \mathscr{F} (\br{F}_1) \smile B_1 \right) \right] \right\} .
\end{aligned}
\end{equation}

\paragraph{BF terms from circle reduction.}
Next, we obtain the BF terms of the 4d $\mathcal{N}=2$ KK theory from the circle compactification of \eqref{eq:BF5dN=1full}. We find:
\begin{equation}
\begin{aligned}
    S_{\mathrm{BF}}^{L_5 \times \mathbb{S}^1} = \int_{Y_5} & \left\{ \frac{\Phi_3}{2\pi}\wedge\frac{h_2}{2\pi} + \frac{F_4}{2\pi}\wedge\frac{\phi_1}{2\pi} + \sum_{i=1}^{b^2} \left[ \frac{\Phi_1^{i}}{2\pi}\wedge\frac{h^{i}_4}{2\pi} + \frac{F_2^{i}}{2\pi}\wedge\frac{\phi^{i}_3}{2\pi} - \frac{\Phi_0^{i}}{2\pi} \wedge \frac{h^{i}_5}{2\pi} + \frac{F_1^{i}}{2\pi} \wedge \frac{\phi^{i}_4}{2\pi}\right] \right. \\
    +& \kappa^{(5)} \left(  B_0\smile \delta \mathsf{a}_4 \right) + \widetilde{\mathrm{CS}}^{(5)} \left(  B_1\smile \delta \mathsf{a}_3 \right) - \widetilde{\mathrm{CS}}^{(5)} \left(  \mathsf{b}_0\smile \delta A_4 \right)  \\
     + & \kappa^{(5)} \left(  B_2\smile \delta \mathsf{a}_2 \right)  + \kappa^{(5)} \left(  \mathsf{b}_1\smile \delta A_3 \right) \\
    +& \widetilde{\kappa}^{(5)} \left( \frac{F_2}{2\pi} \smile \delta \mathsf{a}_2 \right) + \widetilde{\kappa}^{(5)} \left( \frac{\Phi_1}{2\pi} \smile \delta A_3 \right) +\widetilde{\eta}^{(5)} \left( \frac{F_1}{2\pi } \smile \delta \mathsf{a}_3 \right) +\widetilde{\eta}^{(5)} \left( \frac{\Phi_0}{2\pi} \smile \delta A_4 \right) \\
    + \ & \left.  \widetilde{\kappa}^{(5)} \left( \delta B_0 \smile \frac{\phi_4}{2\pi} \right) + \widetilde{\eta}^{(5)} \left(  \delta B_1\smile \frac{\phi_3}{2\pi} \right) + \widetilde{\eta}^{(5)} \left(  \delta \mathsf{b}_0\smile \frac{h_4}{2\pi} \right) \right\} .
\end{aligned}
\end{equation}

\subsubsection{\texorpdfstring{$(-1)$}{(-1)}-form symmetry and polarization}
\label{sec:minus1form5d}
Our main focus is on the realization of $(-1)$-form symmetries through geometric engineering. There are two sources of $(-1)$-form symmetries in the 4d KK theories:
\begin{itemize}
    \item Those descending directly from $(-1)$-form symmetries in the 5d SCFT prior to circle compactification;
    \item The reduction of 0-form symmetries in the 5d SCFT prior to circle compactification.
\end{itemize}
In the former case, 0-form gauge fields in 5d reduce to 0-form gauge fields in 4d, without producing any new symmetry; in the latter case, the holonomy of the 1-form gauge field around the circle provides a 4d 0-form gauge field.

\paragraph{Discrete $(-1)$-form symmetry.}
The SymTFT of the 5d $\mathcal{N}=1$ SCFT geometrically engineered by $X_6$ contains a $(-1)$-form symmetry $\Z_{m_1}^{\scriptscriptstyle [-1]} \times \cdots \times \Z_{m_{\mu}}^{\scriptscriptstyle [-1]}$, together with a discrete 0-form symmetry $\Z_{n_1}^{\scriptscriptstyle [0]} \times \cdots \times \Z_{n_{\widetilde{\mu}}}^{\scriptscriptstyle [0]}$. The direct reduction of the first $\mu$ factors descends to the finite $(-1)$-form symmetry of the 4d theory. The reduction of the latter $\widetilde{\mu}$ factors gives both 0-form and $(-1)$-form symmetries in 4d.\par
In conclusion, the finite $(-1)$-form symmetry in a 4d $\mathcal{N}=2$ KK theory is 
\begin{equation}
    \underbrace{\Z_{m_1}^{\scriptscriptstyle [-1]} \times \cdots \times \Z_{m_{\mu}}^{\scriptscriptstyle [-1]}}_{\mathrm{Tor}H_1 (L_5,\Z)} \times \underbrace{\Z_{n_1}^{\scriptscriptstyle [-1]} \times \cdots \times \Z_{n_{\widetilde{\mu}}}^{\scriptscriptstyle [-1]}}_{\mathrm{Tor}H_2 (L_5,\Z)} .
\end{equation}
Using $(\kappa^{(5)})_{\alpha \alpha^{\prime}} = -\frac{1}{m_{\alpha}} \delta_{\alpha \alpha^{\prime}}$, we identify the BF couplings:
\begin{equation}
    - \sum_{\alpha=1}^{\mu} \frac{1}{m_{\alpha}} B_0^{\alpha} \smile \delta \mathsf{a}_4^{\alpha} , 
\end{equation}
which are the direct dimensional reduction of the electric/magnetic coupling involving the $(-1)$-form symmetry. Here, $\mathsf{a}_4^{\alpha}$ are background gauge fields for a 3-form symmetry acting on domain walls. As we predicted by general arguments, it carries the exact same BF coupling as the background fields for the electric/magnetic 1-form symmetries $\Z_{m_{\alpha}}^{\scriptscriptstyle [1,\mathrm{e}]}, \Z_{m_{\alpha}}^{\scriptscriptstyle [1,\mathrm{m}]}$.\par
Additionally, we have the BF couplings for the new electric/magnetic pairs of $(-1)$-/3-form symmetries:
\begin{equation}
    - \sum_{1 \le \beta < \beta^{\prime} \le \widetilde{\mu} }\frac{\ell_{\beta , \beta^{\prime}}}{\mathrm{gcd}(n_{\beta},n_{\beta^{\prime}})} \mathsf{b}_0^{\beta} \smile \delta A_4^{\beta^{\prime}} 
\end{equation}
for some integers $\ell_{\beta , \beta^{\prime}} \in \Z / \mathrm{gcd}(n_{\beta},n_{\beta^{\prime}}) \Z $ that depend on the choice of generators of $\mathrm{Tor} H^3 (L_5, \Z)$.\par 
We can make a field redefinition of $A_4^{\beta}$ involving $\ell_{\beta , \beta^{\prime}} A_4^{\beta^{\prime}}$ only for $\beta^{\prime} >\beta$, to put the BF term in a more canonical form. Note that the only other place where $A_4^{\beta}$ appears is through 
\begin{equation}
    \sum_{\beta,i} (\widetilde{\eta}^{(5)} )_{\beta i} \int_{Y_5} \frac{\Phi_0^{i}}{2\pi} \smile \delta A_4^{\beta} = \sum_{\beta,i} (\widetilde{\eta}^{(5)} )_{\beta i} \int_{M_4} \frac{\Phi_0^{i}}{2\pi} \smile A_4^{\beta} \in \Z ,
\end{equation}
which plays no role in the SymTFT. Thus, the field redefinition that adjusts the BF coupling between $\Z_{n_{\beta}}^{\scriptscriptstyle [-1]}$ and $\Z_{n_{\beta^{\prime}}}^{\scriptscriptstyle [3]}$ does not affect the rest of BF couplings in the SymTFT. After this redefinition, the choice of polarization between these genuinely 4d $(-1)$-/3-form symmetries is only non-trivial when the topology of $X_6$ is such that there exist $(n_{\beta},n_{\beta^{\prime}})$ with $\mathrm{gcd}(n_{\beta},n_{\beta^{\prime}}) \ne 1$.

\paragraph{Continuous $(-1)$-form symmetry.} 
We have found a total of $U(1)^{b^2} \times U(1)^{b^2} =U(1)^{2b^2} $ continuous $(-1)$-form symmetries, with background field strengths $F_1^{i}, \Phi_1^{i}$, $i=1, \dots, b^2$.\par
Half of these symmetries descend directly from the $(-1)$-form symmetry in 5d. On the other hand, there are genuinely 4d lower form symmetries with curvature $\Phi_1^{i}$ and coupling 
\begin{equation}
    \sum_{i=1}^{b^2} \int_{Y_5} \frac{\Phi_1^{i}}{2 \pi} \wedge \frac{h_4^{i}}{2 \pi} .
\end{equation}
These descend from the instanton 0-form symmetries in 5d, and we identify them with Chern--Weyl symmetries.\par

\subsubsection{Defects and symmetry operators from M-theory}

\paragraph{Discrete symmetries.}
The M-theory setup allows for configurations of M2 and M5-branes wrapping torsional cycles in $L_5 \times \mathbb{S}^1$. We list all the allowed defects and symmetry generators from M2 and M5-branes in Table \ref{tab:M2M5-4d}.\par

\begin{table}[ht]
\centering
\begin{tabular}{|l|c|c|c|c|}
\hline 
& M2 & & M5 & \\
\hline
Tor$H_{1}(L_{5},\Z)\times [0,\infty)$ \hspace{4pt} & Wilson line & \begin{color}{Red}$\diamondsuit$\end{color} & &  \\
Tor$H_{2}(L_{5},\Z)\times [0,\infty)$ \hspace{4pt} & local operator & \begin{color}{blue}$\circ$\end{color} & domain wall & \begin{color}{blue}$\triangle$\end{color} \\
Tor$H_{3}(L_{5},\Z)\times [0,\infty)$ \hspace{4pt} & --- & $\clubsuit$ & surface defect & $\sharp$\\
\hline
Tor$H_{1}(L_{5},\Z) \times \mathbb{S}^1\times [0,\infty)$ \hspace{4pt} & local operator & \begin{color}{blue}$\Box$\end{color} & domain wall & ${\spadesuit}$  \\
Tor$H_{2}(L_{5},\Z) \times \mathbb{S}^1\times [0,\infty)$ \hspace{4pt} & --- & $\flat$ & surface defect & \begin{color}{blue}$\triangledown$\end{color} \\
Tor$H_{3}(L_{5},\Z) \times \mathbb{S}^1\times [0,\infty)$ \hspace{4pt} & & & 't Hooft line & \begin{color}{Red}$\heartsuit$\end{color}\\
\hline
\hline
Tor$H_{1}(L_{5},\Z)$ \hspace{4pt} & 1-form sym. generator & \begin{color}{Red}$\heartsuit$\end{color} &  &   \\
Tor$H_{2}(L_{5},\Z)$ \hspace{4pt} & 2-form sym. generator & \begin{color}{blue}$\triangledown$\end{color} &  $(-1)$-form sym. generator & $\flat$ \\
Tor$H_{3}(L_{5},\Z)$ \hspace{4pt} & 3-form sym. generator & $\spadesuit$ & 0-form sym. generator & \begin{color}{blue}$\Box$\end{color} \\
\hline
Tor$H_{1}(L_{5},\Z) \times \mathbb{S}^1$ \hspace{4pt} & 2-form sym. generator & $\sharp$ & $(-1)$-form sym. generator & $\clubsuit$ \\
Tor$H_{2}(L_{5},\Z) \times \mathbb{S}^1$ \hspace{4pt} & 3-form sym. generator & \begin{color}{blue}$\triangle$\end{color}& 0-form sym. generator & \begin{color}{blue}$\circ$\end{color} \\
Tor$H_{3}(L_{5},\Z) \times \mathbb{S}^1$ \hspace{4pt} & & & 1-form sym. generator & \begin{color}{Red}$\diamondsuit$\end{color} \\
\hline
\end{tabular}
\caption{Branes wrapping torsional cycles in $L_5 \times \mathbb{S}^1$ give rise to finite $p$-form symmetries. We mark with equal symbol the charged defect and the corresponding symmetry generators.}
\label{tab:M2M5-4d}
\end{table}

\paragraph{Continuous abelian symmetries.}
It is also possible to consider defects charged under $U(1)$ symmetries. We list the electric defects from M2-branes wrapping free cycles in $H_{\bullet} (L_5 \times \mathbb{S}^1, \Z)_{\text{free}}$ in Table \ref{tab:M2U(1)S1L5}.
\begin{itemize}
    \item We distinguish between the Chern--Weil symmetries $U(1)^{\scriptscriptstyle [-1,i,\mathrm{CW}]}$, with curvature $\Phi_1^{i}$ as explained in Subsection \ref{sec:minus1form5d}, and the `new' $(-1)$-form symmetries inherited from 5d, with curvature $F_1^{i}$.
    \item The universal 2-form symmetry in 5d descends to a pair of universal 2-form and 1-form symmetries in the 4d KK theory.
\end{itemize}\par
The symmetry generators are obtained from $P_7$-fluxbranes. They descend from the direct $\mathbb{S}^1$ reduction of the result in Subsection \ref{sec:5dN1fluxbrane}, thus we do not repeat the derivation.\par

\begin{table}[ht]
\centering
\begin{tabular}{|l|c|c|}
\hline 
& M2 & charged under \\
\hline
$H_{0}(L_{5},\Z)_{\text{free}}\times [0,\infty) $ \hspace{2pt} & Surface defect & $U(1)^{\scriptscriptstyle [2]}$  \\
$H_{2}(L_{5},\Z)_{\text{free}}\times [0,\infty)$ \hspace{2pt} & Local operator & $U(1)^{\scriptscriptstyle [0, i]}$ \\
$H_{3}(L_{5},\Z)_{\text{free}}\times [0,\infty)$ \hspace{2pt} & --- & $U(1)^{\scriptscriptstyle [-1, i,\mathrm{new}]} $ \\
\hline
$H_{0}(L_{5},\Z)_{\text{free}}\times\mathbb{S}^1 \times [0,\infty) $ \hspace{2pt} & Line defect & $U(1)^{\scriptscriptstyle [1]}$  \\
$H_{2}(L_{5},\Z)_{\text{free}}\times\mathbb{S}^1 \times [0,\infty)$ \hspace{2pt} & --- & $U(1)^{\scriptscriptstyle [-1,i,\mathrm{CW}]}$ \\
$H_{3}(L_{5},\Z)_{\text{free}}\times\mathbb{S}^1 \times [0,\infty)$ \hspace{2pt} &  &  \\
\hline
\end{tabular}
\caption{Electric defects charged under continuous symmetries in the 4d $\N=2$ KK theory, and their M-theory realization.}
\label{tab:M2U(1)S1L5}
\end{table}

\subsubsection{Example: SymTFT of 4d \texorpdfstring{$\mathcal{N}=2$}{N=2} \texorpdfstring{$\mathfrak{su}(N)$}{su(N)} SYM}

To showcase the general derivation, we write down the SymTFT of the 4d $\mathcal{N}=2$ KK theory from the circle compactification of 5d $\mathcal{N}=1$ $\mathfrak{su}(N)$ Yang--Mills from Subsection \ref{sec:SymTFT5dSUN0}. The link of the threefold singularity is $L_5=Y^{N,0}$, with 
\begin{equation}
    \mathrm{Tor} H_1 (Y^{N,0}, \Z) \cong \Z_N \qquad \mathrm{Tor} H_2 (Y^{N,0}, \Z) = \emptyset , \qquad \mathrm{Tor} H_3 (Y^{N,0}, \Z) \cong \Z_N 
\end{equation}
and non-trivial coefficients 
\begin{equation}
    6 \mathrm{CS}^{Y^{N,0}} = N-1 , \qquad 2 \eta^{Y^{N,0}} = \frac{1}{N} \mod 1 , \qquad \widetilde{\kappa}^{Y^{N,0}} = \frac{1}{N}  \mod 1 .
\end{equation}

\paragraph{Mixed anomalies.}
Focusing on the anomalous terms, we have:
\begin{equation}
\begin{aligned}
    S_{\text{twist}}^{Y^{N,0}\times \mathbb{S}^1} = \frac{1}{N} \int_{Y_5} & \left[ \frac{1}{2}  \mathscr{F} (\Phi_1^{\scriptscriptstyle [-1,\mathrm{CW}]}) \smile B_2^{\scriptscriptstyle [1, \mathrm{e}]} \smile B_2^{\scriptscriptstyle [1, \mathrm{e}]} \ + \mathscr{F} (\Phi_3) \smile B_2^{\scriptscriptstyle [1, \mathrm{e}]} \smile B_0^{\scriptscriptstyle [-1]}  \right. \\
    & + \mathscr{F} (F_4) \smile \mathsf{b}_1 \smile B_0^{\scriptscriptstyle [-1]}  \ + \mathscr{F} (F_2) \smile B_2^{\scriptscriptstyle [1, \mathrm{e}]} \smile \mathsf{b}_1 \\
    & \left. - \mathscr{F} (F_4) \smile \mathscr{F} (\Phi_1^{\scriptscriptstyle [-1,\mathrm{CW}]})  \smile B_0^{\scriptscriptstyle [-1]} \ - \mathscr{F} (\Phi_3) \smile \mathscr{F} (F_2)  \smile B_0^{\scriptscriptstyle [-1]} \right] \\
    + \int_{Y_5}^{\br{H}} & \left[ \br{F}_4 \star \br{\Phi}_1^{\scriptscriptstyle [-1,\mathrm{CW}]} \star \br{F}_1^{\scriptscriptstyle [-1,\mathrm{new}]} \right] .
\end{aligned}
\label{eq:4dN2KKtwist}
\end{equation}
The first term contains the curvature $\Phi_1^{\scriptscriptstyle [-1,\mathrm{CW}]}$ of a continuous $U(1)^{\scriptscriptstyle [-1,\mathrm{CW}]}$ $(-1)$-form symmetry, coupling to the background fields for the electric 1-form symmetry. We recognize in $\Phi_1^{\scriptscriptstyle [-1,\mathrm{CW}]}$ the curvature of the 0-form gauge field, that is identified with the theta angle in the 4d gauge theory. That is to say, $\Phi_1^{\scriptscriptstyle [-1,\mathrm{CW}]}$ gives a background for the Chern--Weil symmetry. This provides a top-down derivation of the mixed anomaly discussed in \cite{Santilli:2024dyz} (in a non-supersymmetric setting).\par
On the contrary, $F_1^{\scriptscriptstyle [-1,\mathrm{new}]}$ provides a background for a `new' $U(1)^{\scriptscriptstyle [-1,\mathrm{new}]}$ $(-1)$-form symmetry, inherited from the continuous $(-1)$-form symmetry of the 5d theory prior to circle compactification.

\paragraph{BF couplings.}
Let us now work out the BF couplings. The circle reduction of \eqref{eq:SBFYN0}, after dropping the terms that play no role in our discussion, yields:
\begin{equation}
\begin{aligned}
    S_{\mathrm{BF}}^{Y^{N,0}\times \mathbb{S}^1} =& \int_{Y_5} \left[ \frac{\Phi_1^{\scriptscriptstyle [-1,\mathrm{CW}]}}{2 \pi} \wedge \frac{h_4}{2\pi} + \frac{F_2}{2\pi} \wedge \frac{\phi_3}{2\pi} + \frac{F_1^{\scriptscriptstyle [-1,\mathrm{new}]}}{2 \pi} \wedge \frac{\phi_4}{2\pi} + \frac{F_4}{2\pi} \wedge \frac{\phi_1}{2\pi} + \frac{\Phi_3}{2\pi} \wedge \frac{h_2}{2\pi}  \right] \\
    &- \frac{1}{N} \int_{Y_5} \left[ B_2^{\scriptscriptstyle [1, \mathrm{e}]} \smile \delta A_2^{\scriptscriptstyle [1, \mathrm{m}]}  \ + B_0^{\scriptscriptstyle [-1]} \smile \delta A_4^{\scriptscriptstyle [3]}  \ + \mathsf{b}_1 \smile \delta A_3  \right] .
\end{aligned}
\end{equation}
\begin{itemize}
    \item From the first line, we read off that 
    \begin{itemize}
        \item[---] $h_4$ combines with half of the last term in \eqref{eq:4dN2KKtwist} to give, after projection onto the physical boundary, the instanton current, which couples of the curvature $\Phi_1^{\scriptscriptstyle [-1,\mathrm{CW}]}$ of the Chern--Weil $(-1)$-form symmetry;
        \item[---] There is an additional 4-form which couples to the background for the `new' continuous $(-1)$-form symmetry.
    \end{itemize}
    \item From the second line, we have three standard BF coupling for the discrete electric/magnetic pairs, respectively the electric/magnetic 1-form symmetry, the $(-1)$-/3-form symmetry and the 0-/2-form symmetry, all of which are $\Z_N$ and all of which have same coefficient $-1/N$.
\end{itemize}

\subsection{Geometric engineering of 4d \texorpdfstring{$\N=1$}{N=1} SYM theories}\label{sec:GE4dN=1}

This subsection is an interlude to describe the geometric engineering from M-theory on $G_2$-manifolds, and we come back to analysis of the SymTFT for the resulting 4d theory in Subsection \ref{sec:MonG2SymTFT}.\par
Here, we closely follow the construction outlined in \cite{Acharya:2020vmg,Najjar:2022eci} to geometrically engineer four-dimensional $\N=1$ $\mathfrak{su}(N)$ gauge theories and describe their associated moduli spaces. We compactify M-theory on an orbifold with $G_2$ holonomy, of the form $X_7= \text{B7}/\Gamma_{p,N,q}$, whose link is 
\begin{equation}
    L_6 = \mathrm{Link} (X_7) = (\mathbb{S}^3 \times \mathbb{S}^3) /\Gamma_{p,N,q} .
\end{equation}
We now proceed to explain the ingredients in this formula and the ensuing geometric engineering.

\subsubsection{Abelian quotients of the B7 family}

\paragraph{The B7 family and its symmetries.} The B7 family is a 1-parameter family of $G_2$-metrics on $\R^{4}\times \mathbb{S}_{\mathrm{b}}^{3}$ with a co-homogeneity one metric. The seven-dimensional space is asymptotically locally conical (ALC), smooth and simply connected. An explicit member of the family was found in the physics literature \cite{Brandhuber:2001yi} and later the full 1-parameter family of ALC metrics was found in \cite{bazaikin2013complete}.\par
Writing $\R^4$ as the cone over\footnote{The notation $\mathbb{S}^3_{\mathrm{b}},\mathbb{S}^3_{\mathrm{f}}$ refers to the base and fibre of the spinor bundle over $\mathbb{S}^3$.} $\mathbb{S}_{\mathrm{f}}^3$, we have $\R^4 \times \mathbb{S}_{\mathrm{b}}^3 \cong [0,\infty) \times \mathbb{S}^3_1 \times \mathbb{S}^3_2 / \sim$, where $\sim$ identifies the cones obtained exchanging the two 3-spheres. Using quaternions $\mathsf{q}_i \in \mathbb{H}$ to write $\mathbb{S}^3_i$ as 
\begin{equation}
    \mathbb{S}^3_i = \left\{ \mathsf{q}_i \in \mathbb{H} \ \,   | \ \, \vert \mathsf{q}_i \rvert =1 \right\} ,
\end{equation}
the map is given by \cite[Sec.2]{Foscolo:2018mfs} 
\begin{equation}
\begin{aligned}
        \left[ 0,\infty \right) \times \mathbb{S}^3_1 \times \mathbb{S}^3_2 / \sim \ & \longrightarrow \ \R^4 \times \mathbb{S}_{\mathrm{b}}^3 \\
        (t, \mathsf{q}_1, \mathsf{q}_2 ) \ & \mapsto \ (t \mathsf{q}_1, \mathsf{q}_1 \bar{\mathsf{q}}_2 ) .
\end{aligned}
\end{equation}\par
For our purposes, we thus have two pairs of 3-spheres: $\mathbb{S}_1^3 \times \mathbb{S}^3_2$ and $\mathbb{S}_{\mathrm{f}}^3 \times \mathbb{S}_{\mathrm{b}}^3$, and these two pairs are related by
\begin{equation}
\label{eq:S3S3toS3S3}
    \psi : \ (\mathsf{q}_1, \mathsf{q}_2 ) \mapsto (\mathsf{q}_1, \mathsf{q}_1 \bar{\mathsf{q}}_2 ) .
\end{equation}
In particular, $\mathbb{S}_{\mathrm{f}}^3$ is identified right away with $\mathbb{S}^3_1$, whereas the map to $\mathbb{S}_{\mathrm{b}}^3$ is non-trivial. Each $\mathbb{S}^3_i$ has an action of $SU(2)_{\ell,i} \times SU(2)_{r,i}$, with the first (respectively second) $SU(2)$ acting on the left (respectively right); however, only the diagonal subgroup $ \triangle SU(2)_{r} \subset  SU(2)_{r,1} \times SU(2)_{r,2}$ preserves \eqref{eq:S3S3toS3S3}. We thus consider the action of $ (\omega_1, \omega_2, \eta) \in SU(2)_{\ell,1} \times SU(2)_{\ell,2} \times SU(2)_{r}$ on $\mathbb{S}_1^3 \times \mathbb{S}^3_2$, and combine it with \eqref{eq:S3S3toS3S3} to get:
\begin{equation}
    \psi \left( (\omega_1, \omega_2, \xi ) \cdot (\mathsf{q}_1,\mathsf{q}_2) \right)= \psi \left( \omega_1 \mathsf{q}_1 \xi, \omega_2 \mathsf{q}_2 \xi \right) = (\omega_1 \mathsf{q}_1 \xi, \omega_1 \mathsf{q}_1 \bar{\mathsf{q}}_2 \omega_2^{-1} ).
\end{equation}

\paragraph{Abelian quotients of B7.} Let us consider the quotient taken in \cite{Acharya:2020vmg,Najjar:2022eci}, which is a generalization of that taken in \cite{Foscolo:2018mfs}. We study B7$/ \Gamma_{p,N,q}$, with the abelian finite group $\Gamma_{p,N,q}$ characterized by 
\begin{itemize}
    \item[(i)] a subgroup $\Z_{N} \subset \triangle U(1)_{r}$ of the maximal torus of $\triangle SU(2)_r$;
    \item[(ii)] a subgroup $\Z_p \subset U(1)_{\ell,2}$ of the maximal torus of $SU(2)_{\ell,2}$;
    \item[(iii)] a character $\chi_q$ of $\Z_p$, which we classify by $q \in \Z / p\Z$.
\end{itemize}
Furthermore, we impose 
\begin{equation}\label{eq:gcdpNq}
    \mathrm{gcd}(p,q)=1, \qquad \mathrm{gcd}(p,N)=1 , \qquad N>p .
\end{equation}
The first condition guarantees that $\chi_q$ is either trivial (if $q=0$) or generates the character group $\widehat{\Z}_p$, while the second and third conditions are for ease of exposition.\par
We thus have a natural action of $(\delta, \xi) \in \Z_p \times \Z_{pN}$ on $\mathbb{S}^3_1 \times \mathbb{S}^3_2$, which, under \eqref{eq:S3S3toS3S3}, becomes 
\begin{equation}\label{eq:Gammaction}
    \psi \left( (\delta, \xi) \cdot (\mathsf{q}_1, \mathsf{q}_2) \right) = \psi \left( (\mathsf{q}_1 \xi, \delta \mathsf{q}_2 \xi) \right) = (\mathsf{q}_1 \xi, \mathsf{q}_1 \bar{\mathsf{q}}_2\delta^{-1} ) = (\xi \mathsf{q}_1, \delta^{-1} \mathsf{q}_1 \bar{\mathsf{q}}_2 ) .
\end{equation}
Explicitly, the finite group $\Gamma_{p,N,q}$ of order $pN$ is \cite[Sec.3.1]{Acharya:2020vmg}
\begin{equation}\label{def:Gamma}
    \Gamma_{p, N, q} = \left\{ \ (\delta,\xi) \in \Z_{p}\times \Z_{pN}  \ \, | \ \,  \xi^{N} = \chi_q (\delta)   \right\} .
\end{equation}\par

We can now define the 7-manifold $X_7:= \text{B7}/\Gamma_{p,N,q}$ as the abelian quotient of the B7 manifold by this group. The link is 
\begin{equation}
    L_6 = L_{6,(p,N,q)} = (\mathbb{S}^3_{\mathrm{f}} \times \mathbb{S}^3_{\mathrm{b}} ) / \Gamma_{p,N,q} ,
\end{equation}
with action of $\Gamma_{p,N,q}$ induced from the action on B7, as given in \eqref{eq:Gammaction}.\par
The quotient structure, and hence the topology of $L_{6,(p,N,q)}$, depends on the choice of $\chi_q$, that is, on the value of the parameter $q$.
\begin{itemize}
    \item $q=0$. In this case, the constraint in \eqref{def:Gamma} implies that $\xi\in \Z_N$ belongs to the subgroup of $N^{\text{th}}$ roots of unity in $\Z_{pN}$, and $\Gamma_{p,N,0} \cong \Z_N \times \Z_p$. The action \eqref{eq:Gammaction} on the two 3-spheres in this case factorizes into
    \begin{equation}\label{M/gammasplit}
        L_{(p,N,0)} \ \cong \ (\mathbb{S}_{\mathrm{f}}^{3}/\Z_{N})\times (\mathbb{S}_{\mathrm{b}}^{3}/\Z_{p})\,.
    \end{equation}
    The requirement $\mathrm{gcd}(N,p)=1$ ensures that the action is free, as discussed in \cite{Friedmann:2002ct,Friedmann:2012uf,cortes2015locally}.
    \item $q=1$. In this case, the constraint in \eqref{def:Gamma} implies that $\xi^N$ generates the $\Z_p$ subgroup of $\Z_{pN} \cong \Z_p \times \Z_N$ (using \eqref{eq:gcdpNq}). Then, $\Gamma_{p,N,1} \cong \Z_{p} \times \Z_N$, which by \eqref{eq:Gammaction} acts as $\Z_{pN}$ on $\mathbb{S}^3_{\mathrm{f}}$ and as $\Z_p$ on $\mathbb{S}^3_{\mathrm{b}}$. We obtain
    \begin{equation}\label{M/gammanosplit}
        L_{6,(p,N,1)} \ \cong \ ((\mathbb{S}_{\mathrm{f}}^{3}/\Z_{N})\times \mathbb{S}_{\mathrm{b}}^{3})/\Z_{p} \ \cong \ (\mathbb{S}_{\mathrm{f}}^{3}/\Z_{pN})\times   (\mathbb{S}_{\mathrm{b}}^{3}/\Z_{p})\, ,
    \end{equation}
    holding for the specific construction of the action \eqref{eq:Gammaction} of $\Gamma_{p,N,1}$ subject to the hypotheses \eqref{eq:gcdpNq}.
    \item $q>1$. This case is similar to $q=1$. Thanks to the assumption \eqref{eq:gcdpNq}, also in this case $\xi^N$ generates a $\Z_p$ subgroup of $\Z_{pN}$, and the derivation parallels the case $q=1$.
\end{itemize}

\subsubsection{M-theory on B7/\texorpdfstring{$\Gamma_{p,N,q}$}{Gamma}}

Recall that, by construction, the action \eqref{eq:Gammaction} lifts to the B7 manifold, leading to the quotient space $X_7=\text{B7}/\Gamma_{p,N,q}$.
The physics of M-theory on B7$/\Gamma_{p,N,q}$ has been analyzed explicitly in \cite{Acharya:2020vmg,Najjar:2022eci}. Here, we briefly review the key results and refer the reader to these original works for further details.

\begin{itemize}
    \item Let us first take M-theory on B7$/\Gamma_{p,N,0}$. In this case, the quotient group factorizes, as shown around \eqref{M/gammasplit}, and 
\begin{equation}
    \text{B7}/\Gamma_{p,N,0}\cong (\C^2/\Z_N) \times (\mathbb{S}^3_{\mathrm{b}} /\Z_p) .
\end{equation}
In the context of four-dimensional physics, the $A_{N-1}$ singularity gives rise to an effective 7d $\mathfrak{su}(N)$ gauge theory \cite{Acharya:2004qe}, further compactified on $\mathbb{S}^3_{\mathrm{b}} /\Z_p$ down to 4d. As a by-product of this twisted compactification, propagating degrees of freedom in 7d split into fields that propagate in 4d, and fields that live on $\mathbb{S}^3_{\mathrm{b}}/\Z_p$ and are massive from a 4d perspective. The $\Z_{p}$ quotient on $\mathbb{S}^3_{\mathrm{b}}$ leads to the presence of holonomies for these $\mathfrak{su}(N)$-valued gauge fields on $\mathbb{S}^3_{\mathrm{b}}/\Z_p$, which (from the 7d perspective) are $\Z_{p}$-invariant Wilson loops \cite{Acharya:2020vmg,Najjar:2022eci}. Their effect is to break the gauge group in a manner that preserves the overall rank of the theory,
\begin{equation}\label{breakingSUN}
    \mathfrak{su}(N) \, \rightarrow\, \mathfrak{su}(n_{0})\oplus \mathfrak{su}(n_{1}) \oplus \cdots \oplus \mathfrak{su}(n_{p-1})\oplus \mathfrak{u}(1)^{\oplus s-1} ,
\end{equation}
where the integers $\{ n_{j} \}_{j=0,1,\cdots, p-1}$ and $s$ should satisfy $N=\sum_{j}n_{j} +s -1$. This Higgsing mechanism has been studied in detail in \cite{Cachazo:2002zk,Friedmann:2002ct, Hosomichi:2005ja,Acharya:2020vmg,Najjar:2022eci}.
\item For the case $q\neq 0 \mod p$ discussed around \eqref{M/gammanosplit}, the $\Z_{p}$ subgroup of $\Z_{pN}$ acts freely on $\C^2/\Z_N$, while $\Z_N$ acts freely on $\mathbb{S}^3_{\mathrm{b}} /\Z_p$. Hence, the effective 7d gauge theory is again $\mathfrak{su}(N)$, transverse to the $A_{N-1}$ singularity $\C^2/\Z_N$ \cite[Sec.4]{Acharya:2020vmg}. The $\Z_{p}$ quotient on $\mathbb{S}^3_{\mathrm{b}}$ again leads to the presence of Wilson loops, which break the $\mathfrak{su}(N)$ gauge group into a product of smaller subgroups while preserving the rank. 
\end{itemize}

\paragraph{The classical moduli space.} In this work, we will not focus on the detailed moduli space of the above gauge theory, as it does not play a crucial role in our analysis. Readers may skip this part, but those interested are encouraged to consult \cite{Acharya:2020vmg,Najjar:2022eci} for a more thorough treatment. However, we provide a bird's-eye view for completeness.

Other regions in the moduli space can be reached through the $G_2$-flop transition \cite{Atiyah:2001qf,Acharya:2004qe}. For example, in the $q=0$ case, the $G_2$-flop results in $\R^{4}/\Z_{p}\times \mathbb{S}^{3}/\Z_{N}$. Dimensionally reducing along the Hopf fibre in $\mathbb{S}^{3}/\Z_{N}$, we arrive at Type IIA string theory on the six-dimensional geometry that takes the form of a $\Z_{p}$ quotient of the conifold Calabi-Yau threefold, as discussed in \cite{Acharya:2020vmg,Najjar:2022eci}. The particular $\Z_{p}$ quotient of the conifold we described was originally studied in \cite{Davies:2011is,Davies:2013pna} and gives a family of hyperconifolds.\par
Since we are reducing from M-theory to Type IIA along a $\mathbb{S}^1/\Z_N$ fibred over the $\mathbb{S}^2$-base of the Hopf fibration, in the resulting Type IIA setup the geometry is accompanied by $N$ quanta of Ramond–Ramond (RR) 2-form flux through the $\mathbb{S}^{2}\subset T^{1,1}/\Z_{p}$ base of the conifold. The presence of the 2-form flux will generate a superpotential for the low-energy 4d theory, which ultimately breaks the supersymmetry down to $\N=1$.  Indeed, this description provides the physical interpretation of the topological condition presented in \cite{Foscolo:2018mfs}, where complete $G_2$-holonomy orbifolds arise as $\mathbb{S}^{1}$-bundles over Calabi-Yau 3-fold cones.\par
In the context of Type IIA superstring theory and its geometric engineering, the effective physics description depends on the distribution of flux quanta within the geometry of the singular hyperconifolds. In fact, for a particular flux distribution, the geometry may be completely or partially resolved. In other words, the 2-form fluxes may prevent the different $\C\P^{1}$s from collapsing, determined by the form of the flux-induced superpotential and its critical points.\par 
For those cases where all the flux goes through the base 2-sphere, the hyperconifold under consideration gets partially resolved. In this case, it can be described by $\Z_{p}$ quotient of the known resolved conifold $\mathcal{O}(-1)\oplus\mathcal{O}(-1)\longrightarrow \C\P^{1}$. For this case, the flux superpotential has a critical point \cite{Acharya:2020vmg,Najjar:2022eci}, hence the 4d effective physics is given by $\mathfrak{su}(p)$ $\N=1$ gauge theory.\par
Other branches of the moduli space arise from different flux distributions, that lead to partial resolutions of the singular geometry. These branches are associated with critical points of the flux superpotential \cite{Acharya:2020vmg,Najjar:2022eci}.

\subsection{SymTFT of 4d \texorpdfstring{$\mathcal{N}=1$}{N=1} theories from M-theory on \texorpdfstring{$G_2$}{G2}-manifolds}
\label{sec:MonG2SymTFT}

\subsubsection{SymTFT of 4d \texorpdfstring{$\mathfrak{su}(N)$}{SU(N)} SYM from M-theory}

We construct the SymTFT that corresponds to the 4d $\N=1$ $\mathfrak{su}(N)$ SYM theories geometrically engineered by M-theory on $X_7=\text{B7}/\Gamma_{p,N,q}$. Throughout we stick to the hypotheses of the previous subsection, in particular we repeatedly use \eqref{eq:gcdpNq}.

\paragraph{Geometric data.}
The link of the $G_2$-manifold $X_7$ is denoted $L_{6,(p,N,q)}$ and given in \eqref{M/gammasplit} if $q=0$, and in \eqref{M/gammanosplit} if $q \ne 0 $. In both cases the link is the product of two lens spaces. For convenience, we provide a brief review of K\"unneth's theorem in Appendix \ref{sec:Kunneth}.\par
To realize discrete $(-1)$-form symmetries and their dual $3$-form symmetries, we look for torsional 2-cycles and 3-cycles in the link. The derivation of the torsion homology of a direct product is reviewed in Appendix \ref{sec:Kunneth}, with final result \eqref{eq:AbelianKunneth}. Applied to the case at hand, we get
\begin{equation}
\begin{aligned}
    \mathrm{Tor}H_2 \left( (\mathbb{S}^3_{\mathrm{f}} /\Z_n) \times (\mathbb{S}^3_{\mathrm{b}} /\Z_p) , \Z\right) &\cong \Z_{\mathrm{gcd}(n,p)} \\ 
    \mathrm{Tor}H_3 \left( (\mathbb{S}^3_{\mathrm{f}} /\Z_n) \times (\mathbb{S}^3_{\mathrm{b}} /\Z_p) , \Z \right) &\cong \Z_{\mathrm{gcd}(n,p)} 
\end{aligned}
\end{equation}
with $n= N$ if $ q=0$, and $n=pN$ if $q \ne 0$. Note that, if $q=0$, $\mathrm{gcd}(n,p)=\mathrm{gcd}(N,p)=1$.\par
\begin{itemize}
    \item $q=0$. $L_{6,(p,N,0)}$ is given in \eqref{M/gammasplit}, and in this case, both $\mathrm{Tor}H_{3}(L_{6,(p,N,0)},\Z)$ and its dual $\mathrm{Tor}H_{2}(L_{6,(p,N,0)},\Z)$ trivialize. The homology groups at $q=0$ are 
    \begin{equation}
        H_{\bullet} (L_{6,(p,N,0)},\Z) = \Z  \ \oplus \ \begin{matrix} \Z_{N} \\ \oplus \\ \Z_p \end{matrix} \ \oplus \ 0 \ \oplus \   \Z^{\oplus 2}  \ \oplus \  \begin{matrix} \Z_N \\ \oplus \\ \Z_{p} \end{matrix} \ \oplus \ 0 \ \oplus \ \Z \,.
    \end{equation}
\item $q\neq 0$. $L_{6,(p,N,q\ne 0)}$ is given in \eqref{M/gammanosplit}, and it acquires additional torsional 2-cycles and 3-cycles. The homology groups are
\begin{equation}\label{homologyMpNq}
    H_{\bullet} (L_{6,(p,N,q\ne 0)}, \Z) = \Z  \ \oplus \ \begin{matrix} \Z_{pN} \\ \oplus \\ \Z_p \end{matrix} \ \oplus \ \Z_{\mathrm{gcd}(pN,p)} \ \oplus \  \begin{matrix} \Z^{\oplus 2} \\ \oplus \\ \Z_{\mathrm{gcd}(pN,p)} \end{matrix} \ \oplus \  \begin{matrix} \Z_{pN} \\ \oplus \\ \Z_{p} \end{matrix} \ \oplus \ 0 \ \oplus \ \Z \,.
\end{equation}
For the remainder of this section, we focus on this case and denote $L_{6,(p,N,1)}$ simply as $L_{6}$. 
\end{itemize}\par
From \eqref{homologyMpNq} and $\mathrm{gcd}(pN,p)=p$ we immediately have 
\begin{equation}\label{cohoL6pq}
    H^{\bullet} (L_{6}, \Z) = \Z  \ \oplus \ 0  \ \oplus \ \begin{matrix} \Z_{pN} \\ \oplus \\ \Z_{p} \end{matrix} \ \oplus \  \begin{matrix} \Z^{\oplus 2} \\ \oplus \\ \Z_{p} \end{matrix} \ \oplus \  \Z_p \ \oplus \ \begin{matrix} \Z_{pN} \\ \oplus \\ \Z_p \end{matrix} \ \oplus \ \Z \,.
\end{equation}
Let us now fix the generators of the differential cohomology as follows. 
\begin{itemize}
    \item[---] We let $\br{\vol}_3^{\mathrm{f}}$ (respectively $\br{\vol}_3^{\mathrm{b}}$) be the pullback to $L_6$ of the generator of $\br{H}^3 (\mathbb{S}^3_{\mathrm{f}} /\Z_{pN})$ (respectively of $\br{H}^3 (\mathbb{S}^3_{\mathrm{b}} /\Z_{p})$). 
    \item[---] Besides, we denote by $t^{\mathrm{f}}_2 \in H^2 (\mathbb{S}^3_{\mathrm{f}} /\Z_{pN} , \Z)$ the Poincar\'e dual to the torsional Hopf fibre in the lens space, and likewise for $t^{\mathrm{b}}_2$. Then, we abuse of notation and denote by $\br{t}_2^{\mathrm{f}}$ (respectively $\br{t}_2^{\mathrm{b}}$) their pullback to $L_6$, which we take as generators of $\br{H}^2 (L_6)$.
    \item[---] We get the generators of $\br{H}^5 (L_6)$ to be $\br{\vol}_3^{\mathrm{f}} \times \br{t}_2^{\mathrm{b}}$ and $\br{\vol}_3^{\mathrm{b}} \times \br{t}_2^{\mathrm{f}}$, where the product $\times$ on a product manifold is reviewed in Appendix \ref{app:diffcohoprods}.
    \item[---] Finally, we have the generators $\br{t}_4 \in \br{H}^4 (L_6)$ and $\br{t}_3 \in \br{H}^3 (L_6)$, where the latter cannot be expressed in terms of $\br{t}_2^{\bullet}$, while the former may be chosen to be $\br{t}_4 = \br{t}_2^{\mathrm{f}} \times \br{t}_2^{\mathrm{b}} $ (see Appendix \ref{app:diffcohoprods}).
\end{itemize}
Thus, the complete set of generators of $\br{H}^{\bullet} (L_6)$ is:
\begin{equation}
    \left\{ \br{1} \right\} , \ 0 , \ \left\{ \br{t}_2^{\mathrm{f}} , \br{t}_2^{\mathrm{b}} \right\} , \ \left\{ \br{\vol}_3^{\mathrm{f}}, \br{\vol}_3^{\mathrm{b}}, \br{t}_{3} \right\} , \  \{\br{t}_{4}\}, \  \left\{ \br{\vol}_3^{\mathrm{f}} \times \br{t}_2^{\mathrm{b}}, \br{\vol}_3^{\mathrm{b}} \times \br{t}_2^{\mathrm{f}} \right\} ,  \ \left\{ \br{\vol}_3^{\mathrm{f}} \times \br{\vol}_3^{\mathrm{b}}\right\} .
\end{equation}
In this basis, we expand $\br{G}_4 \in \br{H}^{4} (M_{11})$ and $\br{\dd G}_7 \in \br{H}^{8} (M_{11})$ as
\begin{equation}
    \begin{aligned}
        \br{G}_4 &= \br{F}_4 \star \br{1} + \sum_{\alpha=\mathrm{f},\mathrm{b}} \br{B}_2^{\alpha} \star \br{t}_2^{\alpha} + \sum_{i=\mathrm{f},\mathrm{b}}\br{F}_1^{i} \star \br{\vol}_3^{\neq i} + \br{B}_1\star \br{t}_3 + \br{B}_0\star \br{t}_4\\
        \br{\dd G}_7 &= \br{f}_8 \star \br{1} + \sum_{\alpha=\mathrm{f},\mathrm{b}} \br{\mathcal{B}}_6^{\alpha} \star \br{t}_2^{\alpha} + \sum_{i=\mathrm{f},\mathrm{b}}\br{f}_5^{i} \star \br{\vol}_3^{ i} + \br{\mathcal{B}}_5 \star \br{t}_3 + \br{\mathcal{B}}_4 \star \br{t}_4 \\
        & \qquad\qquad + \sum_{\alpha=\mathrm{f},\mathrm{b}} \br{\mathcal{B}}_3^{\alpha} \star (\br{t}_2^{\alpha} \times \br{\vol}_3^{\neq \alpha}) + \br{f}_2 \star \br{\vol}_3^{\mathrm{f}} \times \br{\vol}_3^{\mathrm{b}}  
    \end{aligned}
\label{eq:expansionG4DN1}
\end{equation}
and, from the exactness of $\dd G_7$ we have 
\begin{equation}\label{eq:trivialdG74dN1}
	\mathscr{F} (\br{f}_{p+1}^{\bullet}) = \frac{1}{2\pi} \dd h_{p}^{\bullet}, \qquad c(\br{f}_{p+1}^{\bullet}) =0, \qquad \mathscr{F} (\br{\mathcal{B}}_{p+1}^{\bullet})= \delta A_{p}^{\bullet} .
\end{equation}
To lighten the notation, we use $\vol_3^{\ne i}$ to indicate $\vol_3^{\mathrm{b}}$ if $i=\mathrm{f}$ and $\vol_3^{\mathrm{f}}$ if $i=\mathrm{b}$. Besides, we define the shorthand notation 
\begin{equation}
    \mathrm{CS}^{(3)}_{\alpha} := \begin{cases} \mathrm{CS}^{(3)}_{\Z_{pN}} & \text{ if } \alpha=\mathrm{f} \\ \mathrm{CS}^{(3)}_{\Z_{p}} & \text{ if } \alpha=\mathrm{b} \end{cases}
\end{equation}
with the right-hand side defined in \eqref{eq:CSGamma} and studied in Appendix \ref{sec:CSlens3d}.\par

\paragraph{Mixed anomalies from M-theory.} 

To find the potential mixed anomalies of the 4d $\N=1$ theory, we plug the form \eqref{eq:expansionG4DN1} of $\br{G}_4$ in the topological M-theory action \eqref{eq:MtheoryS}.\par
Following the discussion in Appendix \ref{sec:DelZottoLink}, we observe that non-trivial linking pairings can arise between the following cohomology classes: (i) between classes in $\mathrm{Tor} H^{3}(L_{6}, \Z)$ and $\mathrm{Tor} H^{4}(L_{6}, \Z)$; or (ii) between classes in $\mathrm{Tor} H^{2}(L_{6},\Z)$ and $\mathrm{Tor} H^{5}(L_{6},\Z)$. Additionally, there is a non-trivial term coming from the direct reduction along classes in $H^{3}(L_{6},\Z)_{\text{free}}$.\par 
By direct computation we arrive at
\begin{equation}\label{eq:premixAnomaliesCS}
\begin{aligned}
    S_{\text{twist}}^{(\mathbb{S}^3_{\mathrm{f}} \times \mathbb{S}^3_{\mathrm{b}} )/\Gamma_{p,N,q}} = & \frac{1}{p} \int_{Y_{5}}  \left[ B_1 \smile B_2^{\mathrm{f}} \smile B_2^{\mathrm{b}} \  + B_0 \smile B_1 \smile\frac{F_4}{2\pi}  \right] \\
      & - \sum_{\alpha=\mathrm{f,b}} \mathrm{CS}^{(3)}_{\alpha} \int_{Y_5} B_2^{\alpha} \smile B_2^{\alpha} \smile \frac{F_1^{\alpha}}{2\pi} \ + \int_{Y_{5}}^{\br{H}} \br{F}_4 \star \br{F}_1^{\mathrm{f}} \star \br{F}_1^{\mathrm{b}}  .
\end{aligned}
\end{equation}
The pertinent integrals over the link $L_6$ have been computed in Appendix \ref{sec:CSlens6d}, and we used the analysis therein to simplify the expression. Plugging the value of $\mathrm{CS}^{(3)}_{\alpha}$ obtained in Appendix \ref{sec:CSlens3d}, \eqref{eq:premixAnomaliesCS} equals   
\begin{equation}\label{eq:mixAnomaliesCS}
\begin{aligned}
     S_{\text{twist}}^{(\mathbb{S}^3_{\mathrm{f}} \times \mathbb{S}^3_{\mathrm{b}} )/\Gamma_{p,N,q}}=  \frac{1}{p}  \int_{Y_{5}} & \left[B_1 \smile B_2^{\mathrm{f}} \smile B_2^{\mathrm{b}} \ + B_0 \smile B_1 \smile\frac{F_4}{2\pi} \right. \\ 
     & \left. + \frac{(p-1)}{2} B_2^{\mathrm{b}} \smile B_2^{\mathrm{b}} \smile \frac{F_1^{\mathrm{b}}}{2\pi} \  + \frac{(pN-1)}{2N} B_2^{\mathrm{f}} \smile B_2^{\mathrm{f}} \smile \frac{F_1^{\mathrm{f}}}{2\pi}  \right]  \\ 
     + \int_{Y_{5}}^{\br{H}} &\br{F}_4 \star \br{F}_1^{\mathrm{f}} \star \br{F}_1^{\mathrm{b}}  .
\end{aligned}
\end{equation}
Since the precise coefficient of the last term in \eqref{eq:mixAnomaliesCS} will be crucial in Subsection \ref{sec:application}, we spell the derivation out in more detail, referring to Appendix \ref{app:diffchar} for the properties of the integration $\int^{\br{H}}$. From the expansion \eqref{eq:expansionG4DN1} we get:
\begin{equation}
\begin{aligned}
    & - \frac{1}{6} \int^{\br{H}}_{M_{4} \times X_{7}} \left ( \begin{matrix} 3 \\ 1 \end{matrix} \right) \cdot \left ( \begin{matrix} 2 \\ 1 \end{matrix} \right) \cdot  \left[ \left( \br{F}_4 \times \br{1} \right) \star \left( \br{F}_1^{\mathrm{f}} \times \br{\vol}_3^{\mathrm{f}} \right) \star \left( \br{F}_1^{\mathrm{b}} \times \br{\vol}_3^{\mathrm{b}} \right) \right] \\ 
    &= (-1)^2 \int^{\br{H}}_{M_{4} \times X_{7}} \left[  \left( \br{\vol}_3^{\mathrm{f}} \star \br{\vol}_3^{\mathrm{b}} \right) \times \br{F}_4 \star \br{F}_1^{\mathrm{f}}  \star \br{F}_1^{\mathrm{b}} \right] \\
    &= (-1)^{\dim (\mathbb{S}^3 \times \mathbb{S}^3)} \int^{\br{H}}_{Y_5} \br{F}_4 \star \br{F}_1^{\mathrm{f}}  \star \br{F}_1^{\mathrm{b}} ,
\end{aligned}
\end{equation}
where in the second line we get a minus sign from reordering odd-degree forms, and in the last line we have used \eqref{eq:fibreproddiffcoho}.\par
Moreover, we expect no contribution from $\int^{\br{H}} \br{G}_4 \star \br{\mathsf{X}}_8$ to \eqref{eq:mixAnomaliesCS}. This is because 
\begin{equation}
    T (\mathbb{S}^3_{\mathrm{f}} \times \mathbb{S}^3_{\mathrm{b}}) \cong \pi^{\ast}_{\mathrm{f}} (T \mathbb{S}^3_{\mathrm{f}} ) \oplus  \pi^{\ast}_{\mathrm{b}} (T \mathbb{S}^3_{\mathrm{b}})
\end{equation}
implies 
\begin{equation}
    \br{\mathsf{p}} \left(  T (\mathbb{S}^3_{\mathrm{f}} \times \mathbb{S}^3_{\mathrm{b}}) \right) = \pi^{\ast}_{\mathrm{f}} \br{\mathsf{p}} \left(  T \mathbb{S}^3_{\mathrm{f}} \right) \star \pi^{\ast}_{\mathrm{b}} \br{\mathsf{p}} \left(  T \mathbb{S}^3_{\mathrm{b}} \right) ,
\end{equation}
and each $\mathbb{S}^3$ does not support non-trivial Pontryagin classes. Thus, the contribution from $\br{\mathsf{X}}_8$ vanishes in this case.\par
From the summands in \eqref{eq:mixAnomaliesCS}, we identify the following background fields:
\begin{itemize}
    \item $F_{4}$ is the field strength of a universal $U(1)^{\scriptscriptstyle [2]}$ 2-form symmetry. This universal 2-form symmetry exists for all physical theories realized on the semi-classical branches of the moduli space.
    \item $B_{2}^{\mathrm{f}}$ is the discrete gauge field of the finite abelian 1-form symmetry $\Z_{pN}^{\scriptscriptstyle [1]}$, while $B_{2}^{\mathrm{b}}$ is the discrete gauge field of the finite abelian 1-form symmetry $\Z_{p}^{\scriptscriptstyle [1]}$.
    \item There is a finite $\Z^{\scriptscriptstyle [0]}_{p}$ 0-form symmetry with background gauge field $B_1$.
    \item There is a finite $\Z^{\scriptscriptstyle [-1]}_{p}$ $(-1)$-form symmetry with background gauge field $B_0$.
    \item Additionally, there are two continuous $(-1)$-form symmetries $U(1)^{\scriptscriptstyle [-1,i]}$ with background curvature $F_1^{i}$; these will be discussed in detail below in Subsection \ref{sec:lowerform4dN1}.
\end{itemize}

\paragraph{BF terms from M-theory.}
Inserting the expressions for $\br{G}_{4}$ and $\br{\dd G}_{7}$ in \eqref{eq:expansionG4DN1} into the differential cohomology refinement of the M-theory kinetic term \eqref{eq:MtheoryKin}, we find the BF terms:
\begin{equation}\label{eq:preBFdiscretetermsG2example}
\begin{aligned}
    S_{\text{BF}}^{(\mathbb{S}^3_{\mathrm{f}} \times \mathbb{S}^3_{\mathrm{b}} )/\Gamma_{p,N,q}} = \int_{Y_5} & \left[ - \frac{1}{\mathrm{gcd}(pN,p)} B_1 \smile \mathscr{F} (\mathcal{B}_4) \ - \frac{1}{\mathrm{gcd}(pN,p)} B_0 \smile \mathscr{F} (\mathcal{B}_5) \right. \\
    & \left. + \sum_{\alpha=\mathrm{f,b}} 2 \mathrm{CS}_{\alpha}^{(3)} B_2^{\alpha}\smile \mathscr{F} (\mathcal{B}_3^{\alpha}) \right] \\
    +  \int_{Y_5}^{\br{H}} & \left[ \br{F}_4 \star \br{f}_2  \ + \sum_{i=\mathrm{f,b}}  \br{F}_1^{i} \star \br{f}_5^{i} \right] ,
\end{aligned}
\end{equation}
where again we have relied on the integrals computed in Appendix \ref{sec:CSlens6d}. From the exactness of $\dd G_7$, we use \eqref{eq:trivialdG74dN1} and, for the second line, we use the properties of topologically trivial differential characters reviewed in Appendix \ref{app:diffchar}.
After simplifications, we obtain
\begin{equation}\label{BFdiscretetermsG2example}
\begin{aligned}
     S_{\text{BF}}^{(\mathbb{S}^3_{\mathrm{f}} \times \mathbb{S}^3_{\mathrm{b}} )/\Gamma_{p,N,q}} = - \frac{1}{p}\int_{Y_5} & \left[ B_1 \smile \delta A_3 \ + B_0 \smile \delta A_4 \ \right. \\
     & \left. + (p-1) B_2^{\mathrm{b}} \smile \delta A_2^{\mathrm{b}} \ + \frac{(pN-1)}{N}  B_2^{\mathrm{f}} \smile \delta A_2^{\mathrm{f}}  \right] \\
    + \int_{Y_5}& \left[ \frac{F_4}{2\pi} \wedge \frac{h_1}{2\pi} \ + \sum_{i=\mathrm{f,b}} \frac{F_1^{i}}{2\pi} \wedge \frac{h_4^{i}}{2\pi} \right] .
\end{aligned}
\end{equation}
In the first two lines, we have the standard electric/magnetic BF couplings for finite abelian symmetries:
\begin{itemize}
    \item There are two discrete electric 1-form symmetries, $\Z_{pN}^{\scriptscriptstyle [1]}$ and $\Z_{p}^{\scriptscriptstyle [1]}$, with background fields $B_2^{\alpha}$, which are coupled to the background fields $A_2^{\alpha}$ for the magnetic dual 1-form symmetries;
    \item The background field $B_1$ for the 0-form symmetry $\Z^{\scriptscriptstyle [0]}_{p}$ is coupled to the background field $A_3$ for the magnetic dual 2-form symmetry;
    \item The background field $B_0$ for the finite $(-1)$-form symmetry $\Z^{\scriptscriptstyle [-1]}_{p}$ is coupled to the background field $A_4$ for the magnetic dual 3-form symmetry acting on domain walls. 
\end{itemize}
Furthermore, we find the SymTFT BF couplings for continuous symmetries. These are given in the last line of \eqref{BFdiscretetermsG2example}, and read:
\begin{itemize}
    \item There are two Chern--Weil symmetries $U(1)^{\scriptscriptstyle [-1, i]}$, whose curvatures $F_1^{i}$ couple to the magnetic $h_4^{i}$;
    \item We also find a BF coupling for the universal $U(1)^{\scriptscriptstyle [2]}$ 2-form symmetry and its magnetic dual $h_1$.
\end{itemize}
The latter BF terms are in agreement with the form in the literature \cite{Apruzzi:2024htg}. A lightning review is provided in Appendix \ref{app:fromSymTFTtoTQFT} (cf. around \eqref{eq:generalSymTFT-U(1)symmetries}). Further comments regarding the symmetry generators are given in Subsection \ref{sec:4dN1fluxbrane}.

\paragraph{The case $p=1$.}
Let us briefly comment on the SymTFT when $p=1$. In such case, the setup reduces to M-theory on $(\C^2/\Z_N) \times \mathbb{S}^3$ studied in \cite{Atiyah:2000zz}. The SymTFT dramatically simplifies:
\begin{equation}
\begin{aligned}
    S_{\text{SymTFT}}^{(\mathbb{S}^3/\Z_N) \times \mathbb{S}^3} &= \frac{N-1}{2N}\int_{Y_5} B_2^{\mathrm{f}} \smile B_2^{\mathrm{f}} \smile \frac{F_1^{\mathrm{f}}}{2\pi} + \int_{Y_5}^{\br{H}} \br{F}_4 \star \br{F}_1^{\mathrm{f}} \star \br{F}_1^{\mathrm{b}} \\
    & + \frac{1}{N} \int_{Y_5} B_2^{\mathrm{f}} \smile \delta A_2^{\mathrm{f}} + \int_{Y_5} \left[ \frac{F_4}{2\pi} \wedge \frac{h_1}{2\pi} \ + \sum_{i=\mathrm{f,b}} \frac{F_1^{i}}{2\pi} \wedge \frac{h_4^{i}}{2\pi} \right] ,
\end{aligned}
\end{equation}
where the first and second lines are the remnants of, respectively, \eqref{eq:mixAnomaliesCS} and \eqref{BFdiscretetermsG2example}.

\subsubsection{\texorpdfstring{$(-1)$}{(-1)}-form symmetry and polarization}
\label{sec:lowerform4dN1}

\paragraph{Discrete $(-1)$-form symmetry and polarization.}
We have found that the 4d $\N=1$ theory engineered by M-theory on $X_7=\text{B7}/\Gamma_{p,N,q}$ has a finite $(-1)$-form symmetry $\Z^{\scriptscriptstyle [-1]}_{p}$ if $q \ne 0 \mod p$. As expected, the background gauge field for this symmetry is coupled through a BF term \eqref{BFdiscretetermsG2example} to the gauge field for the electric/magnetic dual 3-form symmetry. The latter acts on domain walls in the 4d theory. Gauging it with a discrete theta-angle, we obtain the dual $(-1)$-form symmetry, and the $2\pi$-periodicity of this discrete theta-angle is due to the invariance under large background gauge transformations. This is a 4d analogue of the 2d analysis of \cite{Santilli:2024dyz}.

\paragraph{Continuous $(-1)$-form symmetry.}
We recognize two continuous $(-1)$-form symmetries $U(1)^{\scriptscriptstyle [-1,i]}$. They are identified through the anomalous coupling to the electric 1-form fields $B_2^{i}$ in the second line of \eqref{eq:mixAnomaliesCS}. One of them, $F_1^{\mathrm{f}}$ for this branch of the moduli space, will couple to the instanton current of the $\mathfrak{su} (N)$ theory, providing the expected Chern--Weil $(-1)$-form symmetry.

\paragraph{Periodicity of the $\theta_{\text{YM}}$-angle.}
The above mixed anomalies captured by the action $S_{\text{CS}}^{\text{M}}$ use solely the (co)homology of the link space 
\begin{equation*}
    L_{6,(p,N,q)} = (\mathbb{S}^3_{\mathrm{f}} /\Z_{pN}) \times (\mathbb{S}^3_{\mathrm{b}} /\Z_p)\,,
\end{equation*}
which is common to all theories connected through the moduli space. However, in general the various terms in $S_{\text{CS}}^{\text{M}}$ are only relevant on specific branches, depending on the local effective physics. On the branch of interest, characterized by the B7$/\Gamma_{p,N,q}$ geometry, the effective lower-dimensional theory is given by $\mathfrak{su}(N)$ gauge theory, possibly broken as in \eqref{breakingSUN}. These gauge theories contain a topological theta-term, whose coupling we denote $\theta_{\text{YM}}$.\par
Let us focus on the semi-classical branch of the moduli space, where the low-energy effective theory consists of $SU(N)$ $\N=1$ Yang--Mills theory. The electric 1-form symmetry is $\Z_N^{\scriptscriptstyle [1]}$, with background gauge field $B_2^{\mathrm{e}}$. When a non-trivial $B_{2}^{\mathrm{e}}$ is activated, the partition functions of the theory at $\theta_{\text{YM}}+2\pi $ and $\theta_{\text{YM}}$ differs \cite{Kapustin:2014gua, Gaiotto:2014kfa, Gaiotto:2017yup, Gaiotto:2017tne, Cordova:2019uob}, schematically:
\begin{equation}\label{eq:Zdifferby2pi-Anomaly}
    \frac{Z[\theta_{\text{YM}}+2\pi,B_{2}^{\mathrm{e}}]}{Z[\theta_{\text{YM}},B_{2}^{\mathrm{e}}]}    = \exp{2\pi \ii \,\frac{N-1}{N}\int_{M_4}\mathcal{P}(B_{2}^{\mathrm{e}})}.
\end{equation}
Here, 
\begin{equation}
    \mathcal{P} \ : \ H^2 (M_4, \Z/ N\Z) \longrightarrow H^4 (M_4, \Z/ 2N\Z)    
\end{equation}
is the Pontryagin square operation. In the present case, $\mathcal{P}(B_{2}^{\mathrm{e}}) \mod N$ reduces to $ B_{2}^{\mathrm{e}} \smile B_{2}^{\mathrm{e}} $.\footnote{More precisely, it holds that $ \frac{1}{2} B_{2}^{\mathrm{e}} \smile B_{2}^{\mathrm{e}} \in  H^4 (M_4, \Z/ 2N\Z) $. If $N$ is odd, $\mathcal{P}(B_{2}^{\mathrm{e}}) = B_{2}^{\mathrm{e}} \smile B_{2}^{\mathrm{e}} \in  H^4 (M_4, \Z/ N\Z)$, whereas if $N$ is even, the explicit form uses the cup and cup-1 products. We neglect these subtleties and simply use that $\mathcal{P}(B_{2}^{\mathrm{e}}) \mod N$ equals $ B_{2}^{\mathrm{e}} \smile B_{2}^{\mathrm{e}} $ for both odd and even $N$.} Equation \eqref{eq:Zdifferby2pi-Anomaly} indicates a breakdown of the $2\pi$-periodicity of $\theta_{\text{YM}}$, signalling a mixed anomaly between the Chern--Weil $(-1)$-form symmetry and the electric 1-form symmetry $\Z_N^{\scriptscriptstyle [1]}$.\par
The anomalous behaviour \eqref{eq:Zdifferby2pi-Anomaly} can be corrected by inflow via the term \cite[Eq.(2.5)]{Cordova:2019uob}
\begin{equation}\label{eq:F1B2B2-term}
    \frac{N-1}{2N}\, \frac{F_1^{\mathrm{e}}}{2\pi}\,\smile \mathcal{P}(B_{2}^{\mathrm{e}})\,,
\end{equation}
where the coupling $\theta_{\text{YM}}$ is promoted to a periodic scalar, and $F_1^{\mathrm{e}}$ is the curvature of $\theta_{\text{YM}}$.\par
In fact, we have already found this term among the mixed anomalies captured by our SymTFT action $S_{\text{CS}}^{\text{M}}$, in the last term of the second line of \eqref{eq:mixAnomaliesCS}, with the identification 
\begin{equation}
    B_{2}^{\mathrm{f}}= pB_{2}^{\mathrm{e}} ,
\end{equation}
which reduces the $\Z_{pN}^{\scriptscriptstyle [1]}$ gauge field to a $\Z_N^{\scriptscriptstyle [1]}$ one, and with identification of the Chern--Weil $U(1)^{\scriptscriptstyle [-1,\mathrm{CW}]}$ symmetry with the geometric $U(1)^{\scriptscriptstyle [-1,\mathrm{f}]}$ symmetry, as already discussed.\par
In conclusion, our geometric derivation reproduces the expected features of the 4d $\N=1$ $\theta_{\text{YM}}$-term.

\paragraph{SymTFT, physical theory and branches.}

At this stage, we emphasize that the 1-form symmetry is expected to vary as one transitions between branches of the moduli space.\par
For instance, on the branch of the moduli space where the $SU(N)$ gauge theory is realized, the $\Z_p \subset \Z_{pN}$ does not produce singularities \cite{Acharya:2020vmg}. Thus, we expect the defects charged under it to be uncharged under the $SU(N)$ gauge group, decoupling from the dynamics, and leaving $\Z_{pN}/\Z_p \cong \Z_N^{\scriptscriptstyle [1, \mathrm{e}]}$ as the effectively-acting 1-form symmetry of the low-energy theory.\par
This geometric engineering argument matches the field-theoretic remarks in the previous paragraph.\par
Geometrically, this means that the symmetries uncovered in our SymTFT analysis need not act effectively on the low-energy effective theory, and the symmetries actually realized in the theory depend on additional data of the M-theory compactification.\par
In general, the electric and magnetic 1-form symmetry on any branch can be at most $\Z_{pN}\oplus \Z_{p}$. The precise structure of the 1-form symmetry within any given branch is determined by the specific details of the effective physics on that branch.\par 
For branches admitting a non-abelian gauge theory interpretation, the discrete electric 1-form symmetry typically reduces to a subgroup,
\begin{equation}\label{eq:ZM1e-subgroup}
    \Z_{M}^{\scriptscriptstyle [1, \mathrm{e}]}\,\subset\,\Z_{pN}^{\scriptscriptstyle [1]}\,\oplus\, \Z_{p}^{\scriptscriptstyle [1]}\,,
\end{equation}
which is identified with the centre of the gauge group \cite{Aharony:2013hda}.\par

\subsubsection{Defects and symmetry operators from branes}

Charged defects and symmetry topological operators in the 4d $\N=1$ theories engineered from M-theory on the $G_2$-manifold $X_7=\text{B7}/\Gamma_{p,N,q}$ arise from BPS M-brane wrapped on representative of homology classes in the link $L_{6}$, specified in \eqref{homologyMpNq}. In this subsection we focus on finite symmetries, and discuss continuous abelian symmetries in the next subsection.\par 
In Table \ref{tab:M2M5-G2}, we summarize the brane realization of charged and topological defects wrapping torsional cycles in $L_6$. The former extend in the radial direction $[0,\infty)$, as opposed to the latter. The defects listed in Table \ref{tab:M2M5-G2} exhibit non-trivial linking pairing with the corresponding symmetry operators under which they are charged.\par
The electric/magnetic dual pairs as realized in M-theory are recalled in Figure \ref{fig:EMSym4dN1}; to fully specify a theory, we ought to choose a polarization for each one of the three pairs.

\begin{table}[ht]
\centering
\begin{tabular}{|l|c|c|c|c|}
\hline 
& M2 & & M5 & \\
\hline
Tor$H_{1}(L_{6},\Z)\times [0,\infty)$ \hspace{2pt} & Wilson line & \begin{color}{Red}$\diamondsuit$\end{color} &  &   \\
Tor$H_{2}(L_{6},\Z)\times [0,\infty)$ \hspace{2pt} & Local operator & \begin{color}{blue}$\circ$\end{color} & Domain wall & ${\spadesuit}$ \\
Tor$H_{3}(L_{6},\Z)\times [0,\infty)$ \hspace{2pt} & --- & $\clubsuit$ & Surface defect & \begin{color}{blue}$\triangle$\end{color} \\
Tor$H_{4}(L_{6},\Z)\times [0,\infty)$ \hspace{2pt} &  &  & `t Hooft line & \begin{color}{Red}$\heartsuit$\end{color}\\
\hline
Tor$H_{1}(L_{6},\Z)$ \hspace{4pt} & 1-form sym. generator & \begin{color}{Red}$\heartsuit$\end{color} & &  \\
Tor$H_{2}(L_{6},\Z)$ \hspace{4pt} & 2-form sym. generator & \begin{color}{blue}$\triangle$\end{color} & $(-1)$-form sym. generator & $\clubsuit$ \\
Tor$H_{3}(L_{6},\Z)$ \hspace{4pt} & 3-form sym. generator & $\spadesuit$ & 0-form sym. generator & \begin{color}{blue}$\circ$\end{color}  \\
Tor$H_{4}(L_{6},\Z)$ \hspace{4pt} &  &   & 1-form sym. generator & \begin{color}{Red}$\diamondsuit$\end{color} \\
\hline
\end{tabular}
\caption{Branes wrapping torsional cycles in $L_6 = (\mathbb{S}^3_{\mathrm{f}} /\Z_{pN}) \times (\mathbb{S}^3_{\mathrm{b}} /\Z_p)$ give rise to finite symmetries. We mark with equal symbol the charged defect and the corresponding symmetry generators.}
\label{tab:M2M5-G2}
\end{table}

\begin{figure}[th]
\centering
\begin{tikzpicture}
    \draw[fill=gray,opacity=0.2] (-8.5,2) -- (6,2) -- (6,1) -- (-8.5,1) -- (-8.5,2);
    \draw[fill=yellow,opacity=0.2] (-8.5,0) -- (6,0) -- (6,1) -- (-8.5,1) -- (-8.5,0);
    \draw[fill=orange,opacity=0.2] (-8.5,-1) -- (6,-1) -- (6,0) -- (-8.5,0) -- (-8.5,-1);
    \draw[fill=purple,opacity=0.2] (-8.5,-2) -- (6,-2) -- (6,-1) -- (-8.5,-1) -- (-8.5,-2);
	\node[anchor=east] (1f) at (-2,1.75) {1-form};
	\node[anchor=west] (2f) at (2,-1.75) {1-form};
	\node[anchor=east] (0f) at (-4,0.66) {0-form};
	\node[anchor=west] (3f) at (4,-0.66) {2-form};
	\node[anchor=east] (mf) at (-2.75,-0.66) {$(-1)$-form};
	\node[anchor=west] (4f) at (2.75,0.66) {3-form};
	
	\draw[->,blue,thick] (0f) -- (3f);
	\draw[->,blue,thick] (3f) -- (0f);
	\draw[->,black,thick] (mf) -- (4f);
	\draw[->,black,thick] (4f) -- (mf);
	\draw[->,Red,thick] (1f) -- (2f);
	\draw[->,Red,thick] (2f) -- (1f);

    \draw[dashed,thin] (0,2) -- (0,-2);
    \node[anchor=south] at (-1,2) {M2}; 
    \node[anchor=south] at (1,2) {M5};

    \node[anchor=east,align=right] at (-6.5,1.5) {\footnotesize 2-cycles};
    \node[anchor=east,align=right] at (-6.5,0.5) {\footnotesize 3-cycles};
    \node[anchor=east,align=right] at (-6.5,-0.5) {\footnotesize 4-cycles};
    \node[anchor=east,align=right] at (-6.5,-1.5) {\footnotesize 5-cycles};
\end{tikzpicture}
\caption{Electric/magnetic dual symmetries in 4d theories. The charged defects are realized by M2/M5-branes wrapping cycles in the internal $G_2$-manifold.}
\label{fig:EMSym4dN1}
\end{figure}
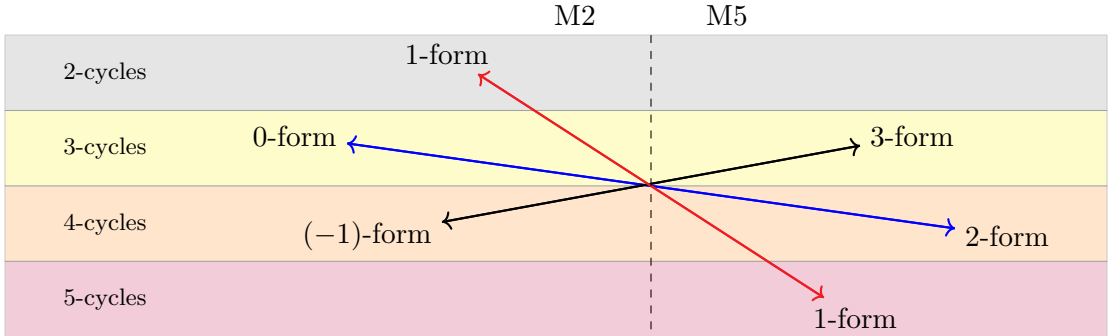

\paragraph{Discrete electric/magnetic 1-form symmetries.}

In 4d, each finite electric 1-form symmetry is accompanied by its magnetic dual 1-form symmetry. In the present situation, with M-theory compactified on $X_7=\text{B7}/\Gamma_{p,N,q}$, we have $\mathrm{Tor}H_1 (L_6,\Z) \cong \Z_n \oplus \Z_p$, thus providing two distinct electric 1-form symmetries $\Z_n^{\scriptscriptstyle [1]}, \Z_p^{\scriptscriptstyle [1]}$, with $n=N$ if $q=0 \mod p $ and $n=pN$ if $q \ne 0 \mod p$. The Pontryagin dual torsion homology groups $\mathrm{Tor}H_4 (L_6,\Z) \cong \widehat{\Z}_n \oplus \widehat{\Z}_p\cong \Z_n \oplus \Z_p $ provide the magnetic dual 1-form symmetries.
\begin{itemize}
	\item Electric 1-form symmetry.
	\begin{itemize}
	\item[---] The charged defects are Wilson lines. They are realized from M2-branes, arranged to wrap a torsional 1-cycles in the link, and extending in the radial direction. There are two generators for torsion classes, namely the Hopf fibres of $\mathbb{S}^3_{\mathrm{f}}/\Z_n$ and $\mathbb{S}^3_{\mathrm{b}}/\Z_p$. Thus, M2-branes wrapping these distinct cycles give rise to Wilson lines with different charges under two different 1-form symmetries.
	\item[---] The symmetry generators originate from M5-branes wrapping a dual torsional 4-cycle in the link, and intersecting the spacetime along a two-dimensional submanifold.
	\end{itemize}
	\item Magnetic 1-form symmetry.
	\begin{itemize}
	\item[---] The charged defects are `t Hooft lines. They are realized from M5-branes wrapping a torsional 4-cycle in the link and extending along the radial direction. There are two generators of torsional 4-cycles, thus the M5-brane can wrap one of the two lens spaces entirely, together with the Hopf fibre of the second lens space. The two choices give rise to `t Hooft loops charged under distinct magnetic 1-form symmetries.
	\item[---] The symmetry generators originate from M2-branes wrapping a dual torsional 1-cycle in the link, and intersecting the spacetime along a two-dimensional submanifold.
	\end{itemize}
\end{itemize}
We derive the expression from the topological symmetry operators, generating these electric/magnetic 1-form symmetry, from the differential cohomology refinement of the action on the corresponding M-brane wrapping a torsional cycle.\par
\begin{itemize}
    \item The symmetry generators of the electric 1-form symmetries are supported on surfaces $\Sigma_2$ in spacetime, and come from an M5-brane wrapping $\mathsf{PD}(t_2^{\alpha}) \in \mathrm{Tor}H_4 (L_6,\Z)$, where $t_2^{\alpha}$ is a generator of the Poincar\'e dual cohomology class. The M5-brane sources $G_7$, which we refine into a differential character $\br{\dd G}_7$ as explained in Subsection \ref{sec:MtheoryupliftH12}. We thus compute 
    \begin{equation}
    \begin{aligned}
        \exp \left\{  2 \pi \ii S_{\text{CS}}^{\text{M5 on }\mathsf{PD}(t_2^{\alpha}) }  (\Sigma_2)\right\} &= \exp \left\{ 2 \pi \ii \int^{\br{H}}_{L_6 \times \Sigma_2 } \br{t}_2^{\alpha} \star \br{\dd G}_7 \right\} \\
        &= \exp \left\{ 2 \pi \ii\int^{\br{H}}_{L_6\times \Sigma_2 } \br{t}_2^{\alpha} \star \left[ \sum_{\alpha^{\prime}=\mathrm{f,b}} \br{\mathcal{B}}_3^{\alpha^{\prime}} \star \br{t}_2^{\alpha^{\prime}} \times \br{\vol}_3^{\neq \alpha^{\prime}} + \cdots \right] \right\}  \\
        &= \exp \left\{ 2 \pi \ii (2\mathrm{CS}^{(3)}_{\alpha})\int_{\Sigma_2 } A_2^{\alpha} \right\}
    \end{aligned}
    \end{equation}
    where, in the second line, we have plugged \eqref{eq:expansionG4DN1} and omitted the terms that eventually drop out. In the last line, the coefficient reads $2\mathrm{CS}^{(3)}_{\alpha}= \frac{n-1}{n}$ if $\alpha=\mathrm{f}$ and $\frac{p-1}{p}$ if $\alpha=\mathrm{b}$, as obtained from the integrals in Appendix \ref{sec:CSlens6d}. We have thus found the expected form, with the symmetry operator given by the holonomy of the gauge field for the gauged dual symmetry. This result is also consistent with the BF term found above.
    \item The symmetry generators of the magnetic 1-form symmetries are supported on surfaces $\check{\Sigma}_2$ in spacetime, and come from an M2-brane wrapping a Hopf link of either lens space, corresponding to a generator of $\mathrm{Tor}H_1 (L_6,\Z)$. The Poincar\'e dual cohomology classes are $\br{t}_2^{\alpha} \times \br{\vol}_3^{\neq \alpha}$. The M2-brane sources $G_4$, which we refine into a differential character $\br{G}_4$ as explained in Subsection \ref{sec:MtheoryupliftH12}, and compute 
    \begin{equation}
    \begin{aligned}
        \exp \left\{  2 \pi \ii S_{\text{CS}}^{\text{M2 on }\mathsf{PD}(t_2^{\alpha} \smile \vol_3^{\neq \alpha}) } (\check{\Sigma}_2) \right\} &= \exp \left\{ 2 \pi \ii \int^{\br{H}}_{L_6 \times \check{\Sigma}_2 } (\br{t}_2^{\alpha} \times \br{\vol}_3^{\neq \alpha}) \star \br{G}_4 \right\} \\
         &  \hspace{-3.3cm} = \exp \left\{ 2 \pi \ii\int^{\br{H}}_{L_6\times \check{\Sigma}_2 } (\br{t}_2^{\alpha} \times \br{\vol}_3^{\neq \alpha}) \star \left[ \sum_{\alpha^{\prime}=\mathrm{f,b}} \br{B}_2^{\alpha^{\prime}} \star \br{t}_2^{\alpha^{\prime}} + \cdots \right] \right\}  \\
         &  \hspace{-3.3cm}= \exp \left\{ 2 \pi \ii (2\mathrm{CS}^{(3)}_{\alpha})\int_{\check{\Sigma}_2 } B_2^{\alpha} \right\}
    \end{aligned}
    \end{equation}
    where again, in the second line we have used \eqref{eq:expansionG4DN1} and omitted the terms that eventually drop out. As expected, we find the symmetry operator to be given by the holonomy of the gauge field for the gauged dual symmetry; in this case, the gauge field $B_2^{\alpha}$ for the gauged electric 1-form symmetry.
\end{itemize}\par

In both cases discussed above, the exact structure of the 1-form symmetry depends on the specific semi-classical branch under consideration, as outlined in the discussion surrounding (\ref{eq:ZM1e-subgroup}).

\paragraph{Discrete $(-1)$/3-form symmetries.}
It is possible to have infinite-tension domain wall defects, that are charged under a 3-form symmetry in 4d. Gauging this symmetry, we obtain a dual finite $(-1)$-form symmetry \cite{Sharpe:2019ddn,Yu:2020twi,Santilli:2024dyz}. In general one must choose a polarization for this electric/magnetic dual pair.\par
In the present case, these symmetries only exists if $q \ne 1$, in which case the symmetry groups are $\Z_{\mathrm{gcd}(pN,p)} =\Z_p$.

\begin{itemize}
    \item 3-form symmetry.
    \begin{itemize}
	\item[---] The charged defects are domain walls. They are realized from M5-branes wrapping a torsional 2-cycle in the link, and extending in the radial direction, intersecting the spacetime along a $(2+1)$d submanifold.
	\item[---] The symmetry generators are point-like in spacetime, and originate from M2-branes wrapping a dual torsional 3-cycle in the link.
	\end{itemize}
    \item $(-1)$-form symmetry.
    \begin{itemize}
	\item[---] There are no $(-1)$-form objects charged under this symmetry. They would morally correspond to M2-branes wrapping a torsional 3-cycles in the link, and extending in the radial direction.
	\item[---] The symmetry generators are spacetime-filling topological operators. They originate from M5-branes wrapping a dual torsional 2-cycle in the link.
	\end{itemize}
\end{itemize}
As for the electric/magnetic 1-form symmetries, we find the symmetry generators from a refinement of the M-brane action to differential cohomology. 
\begin{itemize}
    \item The symmetry generators of the 3-form symmetry are supported on points $\wp \in M_4$ in spacetime, and come from an M2-brane wrapping $\mathsf{PD}(t_3) \in \mathrm{Tor}H_3 (L_6,\Z)$, where $t_3$ is a generator of the Poincar\'e dual cohomology class. We thus compute 
    \begin{equation}
    \begin{aligned}
        \exp \left\{  2 \pi \ii S_{\text{CS}}^{\text{M2 on }\mathsf{PD}(t_3) } (\wp)\right\} &= \exp \left\{ 2 \pi \ii \int^{\br{H}}_{L_6 \times \{ \wp\} } \br{t}_3 \star \br{G}_4 \right\} \\
        &= \exp \left\{ 2 \pi \ii\int^{\br{H}}_{L_6\times \{ \wp \} } \br{t}_3 \star \left[ \sum_{\alpha^{\prime}=\mathrm{f,b}} \br{B}_0 \star \br{t}_4 + \cdots \right] \right\}  \\
        &= \exp \left\{ - \frac{2 \pi \ii}{p} \mathrm{ev}_{\wp} B_0 \right\} .
    \end{aligned}
    \end{equation}
    In the second line, we have omitted the terms that eventually lead to vanishing contributions. In the last line, the coefficient $-\frac{1}{p}$ stems from the integrals in Appendix \ref{sec:CSlens6d}, and $\mathrm{ev}_{\wp} B_0 = B_0 (\wp)$ in the evaluation map at the point $\wp \in M_4$.\par 
    The symmetry operator derived from M-theory takes the expected form, and is given by the exponential insertion of a scalar gauge field for the gauged $(-1)$-form symmetry.
    \item The symmetry generator of the $(-1)$-form symmetries is supported on the whole spacetime $M_4$. It comes from an M5-brane wrapping the generator of $\mathrm{Tor}H_2 (L_6,\Z)$, whose Poincar\'e dual cohomology class is $t_2^{\mathrm{f}} \smile t_2^{\mathrm{b}}$. We compute 
    \begin{equation}
    \begin{aligned}
        \exp \left\{  2 \pi \ii S_{\text{CS}}^{\text{M5 on }\mathsf{PD}(t_2^{\mathrm{f}} \smile t_2^{\mathrm{b}}) }  \right\} &= \exp \left\{ 2 \pi \ii \int^{\br{H}}_{L_6 \times M_4} (\br{t}_2^{\mathrm{f}} \times \br{t}_2^{\mathrm{b}}) \star \br{\dd G}_7 \right\} \\
        &= \exp \left\{ 2 \pi \ii\int^{\br{H}}_{L_6\times M_4 }(\br{t}_2^{\mathrm{f}} \times \br{t}_2^{\mathrm{b}}) \star \left[ \br{\mathcal{B}}_5 \star \br{t}_3 + \cdots \right] \right\}  \\
        &= \exp \left\{ -\frac{2 \pi \ii}{p} \int_{M_4 } A_4 \right\}
    \end{aligned}
    \end{equation}
    where again, in the second line we have omitted the terms that eventually drop out. In the last line, the coefficient comes from the same integral as before.\par 
    As expected, we find the symmetry operator to be given by the integration of the gauge field $A_4$ for the gauged dual symmetry $\Z_p^{\scriptscriptstyle [3]}$ over the full spacetime. This result is also consistent with the BF term found above.
\end{itemize}\par
As with the 1-form symmetries, the $(-1)$-form and 3-form symmetries are common to all branches of the moduli space, but we expect only a subset of them to act effectively, depending on the branch under consideration. Their origin is rooted in K\"unneth's theorem, reviewed in Appendix \ref{sec:Kunneth}.

\paragraph{Discrete 0/2-form symmetries.}
According to Table \ref{tab:M2M5-G2}, it is also possible to engineer defects charged under 0-/2-form symmetries. They stem from M2- and M5-branes wrapping cycles in exchanged fashion compared to the $(-1)$-/3-form symmetry pair. Analogous to the  $(-1)$-/3-form symmetries, the 0-/2-form symmetries only exists if $q \ne 0 $.

\begin{itemize}
    \item Electric 0-form symmetry.
    \begin{itemize}
	\item[---] The charged defects originate from M2-branes wrapping a torsional 2-cycle in the link, and extending in the radial direction, intersecting the spacetime at a point.
	\item[---] The symmetry generators originate from M5-branes wrapping a dual torsional 3-cycle in the link, and filling a codimension-1 submanifold in spacetime.
	\end{itemize}
    \item Magnetic 2-form symmetry.
    \begin{itemize}
	\item[---] The charged defects are surface defects. They are realized from M5-branes wrapping a torsional 3-cycle in the link, and extending in the radial direction, intersecting the spacetime along a 2d submanifold.
	\item[---] The symmetry generators originate from M2-branes wrapping a dual torsional 2-cycle in the link, and running along a loop in spacetime.
	\end{itemize}
\end{itemize}
Once again we compute the symmetry operators from the M-brane action, and find perfect agreement with the expected form of an electric/magnetic dual pair.
\begin{itemize}
    \item The generator of the 0-form symmetry is supported on codimension-1 submanifolds $\Sigma_3 \subset M_4$, and come from an M5-brane wrapping $\mathsf{PD}(t_3) \in \mathrm{Tor}H_3 (L_6,\Z)$. We thus compute 
    \begin{equation}
    \begin{aligned}
        \exp \left\{  2 \pi \ii S_{\text{CS}}^{\text{M5 on }\mathsf{PD}(t_3) } (\Sigma_3) \right\} &= \exp \left\{ 2 \pi \ii \int^{\br{H}}_{L_6 \times \Sigma_3} \br{t}_3 \star \br{\dd G}_7 \right\} \\
        &= \exp \left\{ 2 \pi \ii\int^{\br{H}}_{L_6\times \Sigma_3 }\br{t}_3 \star \left[ \br{\mathcal{B}}_4 \star \br{t}_4 + \cdots \right] \right\}  \\
        &= \exp \left\{ -\frac{2 \pi \ii}{p} \int_{\Sigma_3 } A_3 \right\} .
    \end{aligned}
    \end{equation}
    The symmetry operator derived from M-theory takes the expected form, given by the holonomy of the discrete gauge field $A_3$ for the gauged dual 2-form symmetry.
    \item The generator of the magnetic 2-form symmetry is supported on a loop $\Sigma_1 \subset M_4$. It comes from an M2-brane wrapping the generator of $\mathrm{Tor}H_2 (L_6,\Z)$, whose Poincar\'e dual cohomology class is $t_2^{\mathrm{f}} \smile t_2^{\mathrm{b}}$. We compute 
    \begin{equation}
    \begin{aligned}
        \exp \left\{  2 \pi \ii S_{\text{CS}}^{\text{M2 on }\mathsf{PD}(t_2^{\mathrm{f}} \smile t_2^{\mathrm{b}})} (\Sigma_1)  \right\} &= \exp \left\{ 2 \pi \ii \int^{\br{H}}_{L_6 \times \Sigma_1} (\br{t}_2^{\mathrm{f}} \times \br{t}_2^{\mathrm{b}}) \star \br{G}_4 \right\} \\
        &= \exp \left\{ 2 \pi \ii\int^{\br{H}}_{L_6\times \Sigma_1 }(\br{t}_2^{\mathrm{f}} \times \br{t}_2^{\mathrm{b}}) \star \left[ \br{B}_1 \star \br{t}_3 + \cdots \right] \right\}  \\
        &= \exp \left\{ -\frac{2 \pi \ii}{p} \oint_{\Sigma_1 } B_1 \right\} .
    \end{aligned}
    \end{equation}
    The last line is the holonomy of the gauge field $B_1$ for the gauged dual 0-form symmetry $\Z_p^{\scriptscriptstyle [0]}$.
\end{itemize}\par

\subsubsection{Continuous abelian symmetries from fluxbranes}
\label{sec:4dN1fluxbrane}
The SymTFT derived from M-theory on B7$/\Gamma_{p,N,q}$ uncovered the presence of two $U(1)^{\scriptscriptstyle [-1,i]}$ $(-1)$-form symmetries, and of a $U(1)^{\scriptscriptstyle [2]}$ 2-form symmetry. We now proceed to realize them from M-branes.\par

\paragraph{Defects charged under continuous symmetries.}
The defects charged under $U(1)$ symmetries, in the present setup, are constructed from M-branes wrapping homology classes in $H_{\bullet} (L_6, \Z)_{\text{free}}$ and extending in the radial direction $[0,\infty)$. The electric defects are listed in Table \ref{tab:M2U(1)G2}.

\begin{table}[ht]
\centering
\begin{tabular}{|l|c|c|}
\hline 
& M2 & charged under \\
\hline
$H_{0}(L_{6},\Z)_{\text{free}}\times [0,\infty) $ \hspace{2pt} & Surface defect & $U(1)^{\scriptscriptstyle [2]}$  \\
$H_{3}(L_{6},\Z)_{\text{free}}\times [0,\infty)$ \hspace{2pt} & --- & $U(1)^{\scriptscriptstyle [-1, \mathrm{f}]} \times U(1)^{\scriptscriptstyle [-1, \mathrm{b}]}$ \\
\hline
\end{tabular}
\caption{Electric defects charged under continuous symmetries and their M-theory realization.}
\label{tab:M2U(1)G2}
\end{table}\par
For completeness, we note that we can further construct surface defects from M5-branes wrapping $H_{3}(L_{6},\Z)_{\text{free}}\times [0,\infty)$, and extending in $(1+1)$d in spacetime.

\paragraph{Symmetry operators for continuous symmetries.}
Now we investigate the symmetry operators for continuous $U(1)^{\scriptscriptstyle [p]}$ $p$-form symmetries using fluxbranes wrapping free cycles of the homology groups given in \eqref{homologyMpNq}.\par
We recall from the analysis in Subsection \ref{sec:BraneSymTFTFree} that the topological operators is obtained from $P_7$-fluxbranes. Throughout we use the conventions that $P_7= \sum_k \widetilde{h}_{7-k} \wedge v_k$ with $\widetilde{h}_{7-k}= \phi^{\ast}(h_{7-k} + g_{7-k})$. Thus, it includes the holonomy of $h_{7-k}$, stemming from the $G_7$-flux, together with a correction $g_{7-k}$, stemming from the HWZ action on the brane, which ensures the topological invariance of the resulting operator.\par

\begin{itemize}
    \item Continuous $(-1)$-form symmetry.\par
    The two continuous $U(1)^{\scriptscriptstyle [-1,i]}$ $(-1)$-form symmetries are generated from $P_7$-fluxbranes wrapping either generator in $H_{3}(L_{6},\Z)_{\text{free}}$. With our choices, the latter homology classes are those of the lens spaces themselves. The spacetime-filling topological symmetry operators representing the action of the group element $e^{\ii \varphi } \in U(1)^{\scriptscriptstyle [-1,i]}$ is calculated as
    \begin{equation}\label{Sym-op-(-1)-form-U(1)}
    \begin{aligned}
        \exp \left\{  \ii \varphi S^{P_7\text{-flux along }L_3^{i}} \right\} &= \exp \left\{ \ii \frac{\varphi}{2\pi} \int_{L_6 \times M_4} \vol_3^{\neq i} \wedge P_7 \right\} \\
        &= \exp \left\{ \ii \frac{\varphi}{2\pi} \int_{L_6 \times M_4} \vol_3^{\neq i} \wedge \left[ \sum_{j=\mathrm{f,b}} \widetilde{h}_4^{j} \wedge \vol_3^{j} + \cdots \right] \right\}  \\
        &= \exp \left\{ \ii \frac{\varphi }{2\pi}\int_{M_4 } \phi^{\ast} \left( h_4^{i} + g_4^{i} \right) \right\} .
    \end{aligned}
    \end{equation}
    We have denoted $L_3^{\mathrm{b}}$ any representative in the class $[\mathbb{S}^3_{\mathrm{b}}/\Z_{p}]=\mathsf{PD} (\vol_3^{\mathrm{f}})$; and likewise for $L_3^{\mathrm{f}}$.\par
    The correction $g_4^{i}$ is given by the projection of $H_3 \wedge G_4$ on the subspace transverse to $\vol_3^{i}$:
    \begin{equation}
    \begin{aligned}
        g_4^{i} &= \frac{1}{4\pi} \int_{L_3^{i}} H_3 \wedge G_4 \\
        &= \frac{1}{4\pi} \int_{L_3^{i}} \left( \mathsf{h}_3 + \sum_{j=\mathrm{f},\mathrm{b}}\mathsf{h}_0^{j} \vol_3^{j} \right) \wedge \left( F_4 + \sum_{j=\mathrm{f},\mathrm{b}}F_1^{j} \vol_3^{j}  \right) \\
        &= \frac{1}{4\pi} \left( \mathsf{h}_3 \wedge F_1^{i} + \mathsf{h}_0^{i} F_4 \right) ,
    \end{aligned}
    \end{equation}
    where we recall that $\dd \mathsf{h}_3 = F_4$ and $\dd \mathsf{h}_0^{i} = F_1^{i}$.\par
    Thus, denoting for shortness
    \begin{equation}
        \widetilde{h}_4^{i} := h_4^{i} +  \frac{1}{4\pi} \left( \mathsf{h}_3 \wedge F_1^{i} + \mathsf{h}_0^{i} F_4 \right)  ,
    \end{equation}
    we find that the insertion of $\widetilde{h}_4^{i}$ integrated over the full spacetime generates a $U(1)^{\scriptscriptstyle [-1,i]}$ $(-1)$-form symmetry. In particular, on the semi-classical branch of the moduli space that realizes a $\mathfrak{su}(N)$ gauge theory, we identify $U(1)^{\scriptscriptstyle [-1,\mathrm{f}]}= U(1)^{\scriptscriptstyle [-1,\mathrm{CW}]}$ with the Chern--Weil $(-1)$-form symmetry that shifts the $\theta_{\text{YM}}$-term, and 
    \begin{equation}\label{eq:F4identifiedwithTrFF}
         \frac{\widetilde{h}_4^{\mathrm{f}}}{2\pi}  \quad  \longleftrightarrow\quad \frac{1}{8\pi^{2}}\,\tr{F\wedge F}\,.
    \end{equation}
    We thus establish a correspondence between $\frac{\widetilde{h}_4^{\mathrm{f}}}{2\pi}$ and the second Chern class of the $SU(N)$ gauge bundle, supported by the two crucial properties:
    \begin{itemize}
        \item[---] Flatness: The flatness condition on $\widetilde{h}_4^{i}$ aligns with the Bianchi identity satisfied by the second Chern class.
        \item[---] Quantization: The integral of both sides of \eqref{eq:F4identifiedwithTrFF} over the spacetime is an integer.
    \end{itemize}\par
    \item Continuous 2-form symmetry.\par
    The symmetry topological operator of the universal $U(1)^{\scriptscriptstyle [2]}$ electric 2-form symmetry is engineered by a $P_7$-fluxbrane wrapping the whole six-dimensional link space $L_{6}$. For the generator representing the action of $e^{\ii \varphi } \in U(1)^{\scriptscriptstyle [2]}$, and extending along a loop $\Sigma_1\subset M_4$, we have
    \begin{equation}\label{gen:2-formContinuous}
    \begin{aligned}
        \exp \left\{  \ii \varphi S^{P_7\text{-flux along }L_6} \right\} &=\exp \left\{ \ii \frac{\varphi}{2\pi} \int_{L_6 \times \Sigma_1} P_7 \right\} \\
        &= \exp \left\{ \ii \frac{\varphi}{2\pi} \int_{L_6 \times \Sigma_1} \left[ \widetilde{h}_1 \wedge \vol_3^{\mathrm{f}} \wedge \vol_3^{\mathrm{b}} + \cdots \right] \right\}  \\
        &= \exp \left\{ \ii \frac{\varphi}{2\pi} \oint_{\Sigma_1 } \phi^{\ast}(h_1 + g_1) \right\} .
    \end{aligned}
    \end{equation}
   In this case $\widetilde{h}_1=\phi^{\ast}(h_1 + g_1)$ with correction term 
   \begin{equation}
       g_1= \frac{1}{4\pi} \int_{L_6} H_3 \wedge G_4 = \frac{1}{4\pi} \left( \mathsf{h}_0^{\mathrm{f}} F_1^{\mathrm{b}} + \mathsf{h}_0^{\mathrm{b}} F_1^{\mathrm{f}} \right) 
   \end{equation}
   subject to $\dd \mathsf{h}_0^{i} = F_1^{i}$, hence $2 \pi \dd g_1 = F_1^{\mathrm{f}} \wedge F_1^{\mathrm{b}}$.
\end{itemize}

\paragraph{Continuous $(-1)$-form symmetry and instanton number.}
Plugging the identifications above into the BF couplings for the continuous symmetries, we arrive at the precise identification of the coupling in the QFT action that gives rise to the $(-1)$-form symmetry. Additional discussion on how we can project the SymTFT onto the physical theory is presented in Appendix \ref{app:fromSymTFTtoTQFT}, see also Subsection \ref{sec:application} for an application.\par
    Using the fact that $F_1^{\mathrm{f}}$ is the curvature for the $(-1)$-form symmetry that sources the $\theta_{\text{YM}}$-angle, as we have seen in Subsection \ref{sec:lowerform4dN1}, together with \eqref{eq:F4identifiedwithTrFF}, we obtain the identification:
    \begin{equation}\label{eq:identificationh4f}
        \int_{Y_5} F_1^{\mathrm{f}} \wedge \frac{\widetilde{h}_4^{\mathrm{f}}}{2\pi} \quad  \xrightarrow{ \ \text{ project to physical boundary } \ } \quad \frac{\theta_{\text{YM}}}{8\pi^{2}}\,\int_{M_4} \tr{F\wedge F}\,.
    \end{equation}
    Here we are slightly abusing of notation by letting $\widetilde{h}_4^{\mathrm{f}}$ be defined on the whole $Y_5$.\par 
    On the other hand, combining the contribution $F_1^{\mathrm{f}} \wedge \frac{h_4^{\mathrm{f}}}{2\pi}$ in \eqref{BFdiscretetermsG2example} with the last term in \eqref{eq:mixAnomaliesCS}, we get a closed and quantized $\widetilde{h}_4^{\mathrm{f}}$ whose pullback to the topological operators matches $\widetilde{h}_4^{\mathrm{f}}$ in \eqref{eq:F4identifiedwithTrFF}.\par
    In conclusion, we have the identification \eqref{eq:identificationh4f}, substantiated both from the point of view of the topological operators generating the symmetry, and from the SymTFT action. The symmetry $U(1)^{\scriptscriptstyle [-1,\mathrm{f}]}= U(1)^{\scriptscriptstyle [-1,\mathrm{CW}]}$ discussed here and derived from M-theory, corresponds precisely to the Chern--Weil $(-1)$-form symmetry provided in \cite{Aloni:2024jpb}.\par
    In conclusion: we have successfully provided an M-theory engineering of the $U(1)^{\scriptscriptstyle [-1,\mathrm{f}]}= U(1)^{\scriptscriptstyle [-1,\mathrm{CW}]}$ symmetry, together with its SymTFT and topological operators.

\subsection{Application: 4-group structure}
\label{sec:application}

The goal of this section is to demonstrate a non-trivial application of the SymTFT derived in Subsection \ref{sec:MonG2SymTFT}, which arises from the presence of discrete $(-1)$-form symmetry and its dual 3-form symmetry. Specifically, we will apply the framework discussed in Appendix \ref{app:fromSymTFTtoTQFT} to project the SymTFT to the physical boundary, thereby obtaining a TQFT that couples to the $\mathfrak{su}(N)$ gauge theory on the semi-classical branch of B7$/\Gamma_{p,N,q}$.

We will show that the resulting TQFT aligns with the one found in \cite{Tanizaki:2019rbk}, leading to interesting physical implications. The key point is a mixed anomaly between $(-1)$-form and 3-form symmetries, detected by the SymTFT action via the third term in \eqref{eq:SymTFTincohomology} below. Gauging a discrete subgroup of the $(-1)$-form symmetry leads to a higher-group structure \cite{Tanizaki:2019rbk}. Here we derive this outcome from our results together with the technology of Appendix \ref{app:fromSymTFTtoTQFT}.

\paragraph{SymTFT action.}

We consider the semi-classical branch defined by the B7$/\Gamma_{p,N,q \neq 0}$ geometry, where the effective theory is an $\mathfrak{su}(N)$ gauge theory.\par
The SymTFT, as presented in \eqref{eq:mixAnomaliesCS}-\eqref{BFdiscretetermsG2example}, contains the following terms relevant to our discussion:
\begin{equation}\label{eq:partSymTFTforApplication}
\begin{aligned}
        S_{\text{SymTFT}}^{\text{4d }\mathfrak{su}(N)}\ = \ & \frac{1}{(2\pi)^2}\int_{Y_{5}}  \left[ F_{4} \wedge \widetilde{h}_{1} + F_{1}^{\mathrm{f}}\wedge \widetilde{h}_{4}^{\mathrm{f}} \right] \\
        & + \int_{Y_5} \left[  \frac{1}{p} B_{0}\smile \delta A_{4} - \frac{1}{2N} \frac{F_{1}^{\mathrm{f}}}{2\pi}\smile B_{2}^{\mathrm{e}}\smile B_{2}^{\mathrm{e}} \right]\,.
\end{aligned}
\end{equation}
In this action, the first two terms are BF terms involving the continuous $2$-form and $(-1)$-form symmetries, respectively, while the third term is a BF term coupling a discrete $(-1)$-form symmetry and its dual $3$-form. The appearance of $\widetilde{h}_{\bullet}$ has been extensively discussed in Subsection \ref{sec:4dN1fluxbrane}. In the last term, $B_{2}^{\mathrm{e}}$ is the gauge field for the $\Z_N^{\scriptscriptstyle [1,\mathrm{e}]}$ electric 1-form symmetry, discussed in Subsection \ref{sec:lowerform4dN1}. 
Here the coefficient $- \frac{1}{N}$ is equivalent to the one in \eqref{eq:F1B2B2-term}, modulo integer shifts.

\paragraph{Field redefinition.}

To ease the comparison with the existing literature, we make a field redefinition to pass from the differential cohomology normalization to a normalization of discrete gauge fields more common in the physics literature.\par
$U(1)$-gauge fields require no change; however, for finite symmetries, we trade the $\Z_n$-gauge fields $B_{\bullet}$ with $U(1)$-gauge fields $\tilde{B}_{\bullet}$ as reviewed in Appendix \ref{app:BFtermreview}. The new normalization is obtained replacing $B_{\bullet} \mapsto \frac{n}{2\pi} \tilde{B}_{\bullet}$. Dropping the tilde to reduce clutter, the field redefinition leads us to the new form of \eqref{eq:partSymTFTforApplication}:
\begin{equation}\label{eq:SymTFTincohomology}
\begin{aligned}
        S_{\text{SymTFT}}^{\text{4d }\mathfrak{su}(N)}\ = \ \frac{1}{(2\pi)^2}\int_{Y_{5}}  & \left[ F_{4} \wedge \widetilde{h}_{1} + F_{1}^{\mathrm{f}}\wedge \widetilde{h}_{4}^{\mathrm{f}}  + p B_{0}\smile \delta A_{4} - \frac{N}{4\pi} F_{1}^{\mathrm{f}}\smile  B_{2}^{\mathrm{e}}\smile B_{2}^{\mathrm{e}} \right]\,.
\end{aligned}
\end{equation}
With this form of the SymTFT action, we are ready to analyze its physical properties and implications.

\paragraph{Projecting the SymTFT on the physical theory.} 

We now examine the physical implications of the SymTFT action given in \eqref{eq:SymTFTincohomology}. This is achieved by projecting the SymTFT to the physical boundary. This step involves deriving the effective physical theory by gauging specific symmetries and making precise identifications. To accomplish this, we apply the boundary projection operator $\widetilde{\delta}$, which is introduced in Appendix \ref{app:collapseSymTFT}, see especially around \eqref{eq:widetildedelta}. In particular, this procedure draws upon key transformations and constraints elucidated in \eqref{delta-on-hG} and  \eqref{deltaontopform}.\par

We proceed step-by-step with this process:
\begin{itemize}
    \item Gauging the electric 1-form symmetry.\par
    We gauge the 1-form symmetry $\Z_N^{\scriptscriptstyle [1,\mathrm{e}]}$ by integrating over its background 2-form gauge field $B_{2}^{\mathrm{e}}$, which is promoted to a dynamical one $b_2$ according to 
    \begin{equation}
        \widetilde{\delta}(B_{2}^{\mathrm{e}} \smile B_{2}^{\mathrm{e}}) \,=\, b_{2}\wedge b_{2}\,.
    \end{equation}
    It is known in the literature (see e.g. \cite[Sec.2]{Gaiotto:2017yup}, \cite[Sec.2]{Cordova:2019uob}, or \cite[Sec.3]{Santilli:2024dyz}) that, when gauging the electric 1-form symmetry, one uplifts the $SU(N)$-gauge curvature $F_{2}^{SU(N)}$ to a $U(N)$-gauge curvature $\widetilde{F}_{2}^{U(N)}$. This can be achieved by demanding 
    \begin{equation}
        F_{2}^{SU(N)} \ \mapsto \ \widetilde{F}_{2}^{U(N)} - b_{2}\,,
    \end{equation}
    with the condition
    \begin{equation}
        \tr{\widetilde{F}_{2}^{U(N)}} \,=\, N\,b_{2}\,.
    \end{equation}
    We drop the superscript distinguishing the relevant gauge group in what follows.\par
    Moreover, as already derived in Subsection \ref{sec:lowerform4dN1}, we identify $\widetilde{h}_{4}$ with the generator of the $U(1)$ $(-1)$-form Chern--Weil symmetry. In particular, we impose 
    \begin{equation}
        \widetilde{\delta}\left( \frac{\widetilde{h}_{4}}{2\pi} \right) \,= \,\frac{1}{8\pi^{2}}\tr{\widetilde{F}_{2}\wedge \widetilde{F}_{2}}\,.
    \end{equation}
    Furthermore, by utilizing the discussion in Subsection \ref{sec:lowerform4dN1}, the $\theta_{\text{YM}}$-term for the Yang--Mills theory is introduced, along with a Lagrange multiplier, as
    \begin{equation}
        \widetilde{\delta}(F_{1}) \,=\, \theta_{\text{YM}}\, + \,\chi^{(1)}\, ,
    \end{equation}
    where we have shifted the fixed value of $\theta_{\text{YM}}$ by a background field $\chi^{(1)}$.
    \item Gauging the continuous 2-form symmetry.\par
        The gauge field for the symmetry $U(1)^{\scriptscriptstyle [2]}$ is introduced by setting locally: 
        \begin{equation}
            \widetilde{\delta}(F_{4}) \,=\, \dd c_{3}\,.
        \end{equation}
        Here, $c_{3}$ is the dynamical 3-form gauge field of the 2-form symmetry. Furthermore, we gauge only a $\Z_{K}^{\scriptscriptstyle [2]}$ subgroup of $U(1)^{\scriptscriptstyle [2]}$ by imposing
        \begin{equation}
            K \dd c_{3}=0\,.
        \end{equation}
        This condition ensures the discrete structure of the symmetry.\par
        Additionally, we introduce a discrete $\theta$-term for the topological top form $\dd c_{3}$, as discussed around \eqref{deltaontopform}. To incorporate it, we set 
        \begin{equation}
            \widetilde{\delta}(\widetilde{h}_{1}) \, = \, \theta^{(2)} \,+\, K\,\chi^{(2)}\,,
        \end{equation}
        with $\chi^{(2)}$ acting as a Lagrange multiplier enforcing the condition $K \dd c_{3}=0$.
    \item Gauging the finite 3-form symmetry.\par
        The $\Z_{p}^{\scriptscriptstyle [3]}$ 3-form symmetry is gauged by taking
        \begin{equation}
            \widetilde{\delta}(A_{4}) \,=\, a_{4} .
        \end{equation}
        Additionally, the field $B_{0}$ is identified with the boundary degrees of freedom as,
        \begin{equation}
            \widetilde{\delta}(p B_{0}) \,=\, \theta^{(3)}\,+\,  p\chi^{(3)}\,.
        \end{equation}
        The first summand introduces a $\theta$-term for the $a_{4}$ top form, while the second is a Lagrange multiplier, which enforces the quantization of the periods of $a_4$.
\end{itemize}

At this point, we arrive at the following topological action within the 4d action of the dynamical gauge theory:
\begin{equation}
\begin{aligned}
        2 \pi S_{\text{TQFT}}^{PSU(N)/(\Z_K^{\scriptscriptstyle [2]} \rtimes \Z_p^{\scriptscriptstyle [3]})} & = \int_{M_4}\left[  \frac{\theta_{\text{YM}}}{8\pi^{2}} \tr{\left(\widetilde{F}-b_{2}\right)^{2}} + \frac{\theta^{(2)}}{2\pi}  \dd c_{3} + \frac{\theta^{(3)}}{2\pi}  a_{4} \right] \\
        &+ \int_{M_4}\left[  \frac{1}{8\pi^{2}} \chi^{(1)}\wedge\, \tr{\left(\widetilde{F}-b_{2}\right)^{2}} \right. \\
        &+ \left. \frac{K}{2\pi} \chi^{(2)} \wedge \, \dd c_{3} + \frac{p}{2\pi} \, \chi^{(3)} \wedge \,  a_{4}\,  \right]. 
\end{aligned}
\label{eq:S4dTQFTgauged}
\end{equation}
In the above, the first term is the usual $\theta_{\text{YM}}$-term. The interpretation of the other terms stems from the above discussion. In the following, we will further constrain this action by demanding gauge invariance, uncovering interesting physical implications.

\paragraph{Gauge invariance.} 
Demanding the invariance of the action \eqref{eq:S4dTQFTgauged} under large gauge transformations imposes additional constraints.\par 
The first term is by construction invariant under 1-form gauge transformations. Invariance under large gauge transformations of $c_3$ and $a_4$ typically imposes additional restrictions on the parameters, imposing discrete values for the $\theta$-angles. However, in the present situation, we are gauging two discrete groups simultaneously, and they may combine in higher structures.\par
If we select the gauged subgroup $\Z_K^{\scriptscriptstyle [2]} \subset U(1)^{\scriptscriptstyle [2]}$ to have
\begin{equation}
    K=p,
\end{equation}
we find a non-trivial solution to the requirement of invariance of \eqref{eq:S4dTQFTgauged} under large gauge transformations $c_{3}\  \mapsto \ c_{3} + \Lambda_{3}$. Letting $a_4$ transform non-trivially:
\begin{equation}
\label{eq:a4L3gauge}
    a_4 \ \mapsto \ a_4 + \dd \Lambda_3 ,
\end{equation}
we get the constraints
\begin{equation}
\begin{aligned}
    \theta^{(2)} \, &= \, - \,\theta^{(3)} \\
     \chi^{(2)} \, &=\, -\chi^{(3)}.
\end{aligned}
\end{equation}
To lighten the notation we simply write $\theta$ for $\theta^{(2)}$ from now on.\par

At this point, the TQFT action is given by
\begin{equation}
\begin{aligned}
        2 \pi S_{\text{TQFT}}^{PSU(N)/(\Z_K^{\scriptscriptstyle [2]} \rtimes \Z_p^{\scriptscriptstyle [3]})} & = \int_{M_4}\left[  \frac{\theta_{\text{YM}}}{8\pi^{2}} \tr{\left(\widetilde{F}-b_{2}\right)^{2}} + \frac{\theta}{2\pi}  \left(\dd c_{3} - a_{4}\right) \right] \\
        &+ \int_{M_4}\left[  \frac{1}{8\pi^{2}} \chi^{(1)}\,\wedge\, \tr{\left(\widetilde{F}-b_{2}\right)^{2}} + \frac{p}{2\pi} \, \chi^{(3)} \,\wedge \,  \left( a_{4}\, - \dd c_3 \right)  \right]. 
\end{aligned}
\label{eq:S4dTQFTconstr}
\end{equation}

\paragraph{Higher-group structure.}

Further constraints are found by considering the equations of motion of the Lagrange multipliers. The equation of motion of $\chi^{(1)}$ implies 
\begin{equation}
   \ \frac{1}{8\pi^{2}}\int \tr{\widetilde{F}^{2}}\,  = \,  \frac{N}{8\pi^{2}}\int b_{2}\wedge b_{2} \ \in \ \frac{1}{N}\,\Z\,.
\end{equation}
This is avoidable by imposing $\chi^{(1)}=0$ from the onset, which trivially solves the issue. This is the solution typically considered in the literature.\par
In the present situation, thanks to the simultaneous gauging of discrete 2-form and 3-form symmetries, we now proceed to show that there is another, non-trivial solution, given by
\begin{equation}
    \chi^{(1)} \,=\, \chi^{(3)}\, 
\end{equation}
and allowing for modifications of the $a_4$ gauge field (reminiscent of the Green--Schwartz mechanism). This latter choice leads to interesting physical implications. In the following we fix $\chi^{(1)}=\chi^{(3)}=:\chi$.\par
The relevant part of \eqref{eq:S4dTQFTconstr} subject to this condition is
\begin{equation}\label{4DTQFTaction}
     \int_{M_4} \chi \,\wedge\,\left[ \frac{1}{8\pi^{2}} \tr{\widetilde{F}^{2}}  - \frac{N}{8\pi^{2}} \, b_{2}\wedge b_{2} \,+\,
     \frac{p}{2\pi} \, ( - \dd c_{3} + a_{4} ) \right]\,.
\end{equation}
Recall that usually the discreteness of $a_4$ is imposed with an auxiliary field $a_3$ with $pa_4 = \dd a_3$ (locally). Next, we relax this relation and instead we consider
\begin{equation}
    p a_{4} = \dd a_{3} + \frac{N}{4\pi} b_{2}\wedge b_{2}\,.
\end{equation}
This leads to express \eqref{eq:S4dTQFTconstr} as
\begin{equation}
\begin{aligned}
        2 \pi S_{\text{TQFT}}^{PSU(N)/(\Z_K^{\scriptscriptstyle [2]} \rtimes \Z_p^{\scriptscriptstyle [3]})} & = \int_{M_4}\left[  \frac{\theta_{\text{YM}}}{8\pi^{2}} \tr{\widetilde{F}^{2}} + \frac{\theta}{2\pi}  \left(\dd c_{3} - \frac{1}{p} \dd a_3\right) \right] \\
        &+ \int_{M_4} \frac{1}{8\pi^2} \left( \theta_{\text{YM}} + \frac{\theta}{p} \right) b_2 \wedge b_2  \\
        & + \int_{M_4} \chi \wedge \left[  \frac{1}{8\pi^{2}}\tr{ \widetilde{F}^{2}} - \frac{p}{2\pi} \dd c_3 + \frac{1}{2\pi}\,\dd a_{3} \right]. 
\end{aligned}\label{4DTQFTactionfinal}
\end{equation}
We observe that, in this formulation, the action is invariant under the large $\Lambda_{3}$ gauge transformation by construction. However, now $a_4$ remains invariant under 1-form large gauge transformations $b_2 \mapsto b_2 + \Lambda_2$ only if we impose 
\begin{equation}
    a_{3} \ \mapsto \ a_{3} - \frac{N}{2\pi}b_{2}\wedge \Lambda_{1} - \frac{N}{4\pi}\Lambda_{1}\wedge \dd \Lambda_{1}
\end{equation}
with $\dd \Lambda_1 = \Lambda_2$. This transformation indicates a higher-group structure between the gauged 1-form and 3-form symmetries. This is to be combined with the non-trivial transformation \eqref{eq:a4L3gauge} of $a_3$ under the 2-form large gauge transformation. Altogether, we have the following 4-group structure
\begin{equation}\label{4-group-structure}
   \left( \Z_{N}^{\scriptscriptstyle [1,\mathrm{e}]} \times \Z_{p}^{\scriptscriptstyle [2]} \right)\  \widetilde{\times} \ \Z_{p}^{\scriptscriptstyle [3]} \,.
\end{equation}

This higher group structure was analyzed in \cite{Tanizaki:2019rbk}. The dynamical gauge theory Lagrangian, along with the TQFT action in \eqref{4DTQFTactionfinal}, were referred to as the generalized $SU(N)$ super Yang--Mills (SYM) theory. Here, we review the  main physical implication of the 4-group structure in \eqref{4-group-structure} as discussed in \cite{Tanizaki:2019rbk}.
\begin{itemize}
    \item Existence of $pN$ vacua:\par
    The generalized $SU(N)$ SYM theories have $pN$ vacua rather than the $N$ vacua of the usual $SU(N)$ SYM theory. These extended vacua can be seen due to the existence of $p$ non-dynamical UV domain-walls (DWs) that separate between the $N$ vacua of the usual $SU(N)$ gauge theory. Hence, there are $pN$ vacua. In fact, our top-down analysis shows that such non-dynamical UV DWs exist and can be identified with the $\Z_{p}$ DWs found in  Table \ref{tab:M2M5-G2}. 

    \item Periodicity of $\theta_{\text{YM}}$:\par
    Using the equation of motion of the $\chi$ field in \eqref{4DTQFTactionfinal}, one shows that the combination $(\theta + p\,\theta_{\text{YM}})$ should be $2\pi$-periodic. It then follows that
    \begin{equation}
        \theta \ \sim \ \theta + 2\pi, \qquad \theta_{\text{YM}} \ \sim \ \theta_{\text{YM}} + \frac{2\pi}{p}\,.
    \end{equation}
\end{itemize}

Next, we consider the physical implication of gauging the $\Z_{p}^{\scriptscriptstyle [2]}\subset U(1)^{\scriptscriptstyle [2]}$ 2-form symmetry while turning-off the discrete 1-form and 3-form symmetries.

\paragraph{Modified instanton sum.} Let us now turn off the 4-group structure discussed above, retaining only the $c_{3}$ gauge field associated with the 2-form symmetry. In this case, our TQFT action simplifies  to 
\begin{equation}
\begin{aligned}
    2 \pi S_{\text{TQFT}}^{SU(N)/\Z_p^{\scriptscriptstyle [2]}} & = \int_{M_4} \left[ \frac{\theta_{\text{YM}}}{8\pi^{2}}  \tr{F\wedge F} + \frac{\theta}{2\pi}  \dd c_{3} \right] \\
        &+ \int_{M_4} \chi \,\wedge\,\left[ \frac{1}{8\pi^{2}} \tr{F\wedge F}  -\frac{p}{2\pi}\, \dd c_{3}  \right]\,.
\end{aligned}
\end{equation}

The equation of motion of the $\chi$ field imposes the following constraint: 
\begin{equation}
     \frac{1}{8\pi^{2}} \tr{F\wedge F}  = \frac{p}{2\pi}\, \dd c_{3} \,.
\end{equation}
Integrating both sides over spacetime gives a constraint on the instanton number, forcing it to be a multiple of the integer $p$. This modified instanton sum was first proposed in \cite{Seiberg:2010qd}, building on earlier 2d examples with gauged 1-form symmetry \cite{Pantev:2005rh,Pantev:2005wj,Pantev:2005zs}, and further explored in \cite{Tanizaki:2019rbk}. For more detailed discussions, we refer the reader to these works, as our focus here is on the implications of the above TQFT.\par
Besides, from a calculation analogous to the one above, one shows that the periodicity of the Yang--Mills $\theta_{\text{YM}}$-angle is $2\pi/p$.

\section{Conclusions and outlook}

In this paper, we investigated the appearance of $(-1)$-form symmetries and their mixed 't Hooft anomalies in supersymmetric field theories from the geometric engineering of M-theory on a cone (with singularity at the tip), and computed their SymTFT actions.\par 
For discrete $(-1)$-form symmetries, the 0-form gauge fields arise from the reduction of differential cohomology class $\br{G}_4$ over torsional cycles on the link of the cone. For continuous $(-1)$-form symmetries, the 1-form field strengths arise from the reduction of $\br{G}_4$ over free cycles on the link of the cone. In particular, we propose that the topological operators generating continuous symmetries come from the M-theory action $\sim \int  P_7$, where $P_7$ is the Page charge sourced by a (6+1)d fluxbrane.\par 
This treatment is crucial to obtain well-posed topological operators, and expands on the existing proposals, which are not suited to encompass Chern--Simons terms as the one in the M-theory action.\par
From the dimensional reduction of M-theory topological term $\int C_3\wedge G_4\wedge G_4$, we computed the anomaly (twist) terms of the SymTFT action including the gauge fields for $(-1)$-form symmetries. From the BF-like term $G_4 \wedge G_7$ we derived the BF terms for both continuous and discrete symmetries. In particular, a combination of such the BF term for the continuous symmetries of the form $\int_{Y_{d+1}}F\wedge h$ together with a contribution from the twist term, of the form $\int_{Y_{d+1}}F\wedge \widetilde{g}$, matches the action on the topological operator.\par

We applied the methods to a number of examples. First, for 5d $\mc{N}=1$ SCFTs from M-theory on Calabi-Yau threefold singularities $X_6$, we discussed the presence of discrete and continuous $(-1)$-form symmetries which do not have a conventional gauge theory interpretation. 

We have also studied the SymTFT for 4d theories involving $(-1)$-form symmetries. We discussed 4d $\mc{N}=2$ $\mathfrak{su}(N)$ gauge theories from M-theory on $X_6\times S^1$, as well as 4d $\mc{N}=1$ $\mathfrak{su}(N)$ gauge theories from M-theory on spaces with $G_2$ holonomy, i.e. the orbifold $X_7=\text{B7}/\Gamma_{p,N,q}$ whose link is the quotient space $(\mb{S}^3\times\mb{S}^3)/\Gamma_{p,N,q}$. In these models, we found the mixed anomaly term $F_1\wedge B_2\wedge B_2$ involving the $(-1)$-form symmetry field strength $F_1$ and the background gauge field $B_2$ for 1-form symmetries. For the 4d $\mc{N}=1$ theory engineered from $X_7=\text{B7}/\Gamma_{p,N,q}$, interestingly we found that parts of the 5d SymTFT action can be reduced to topological terms in 4d, and we observed the presence of a 4-group symmetry similar to the algebraic structure in the gauge theory with modified instanton sum~\cite{Tanizaki:2019rbk}.\par

\bigskip
Our extensive and systematic analysis opens up a number of avenues for future research directions:

\begin{itemize}
    \item It would be desirable to extend our results from orbifolds with $G_2$ holonomy, analyzing orbifolds with non-abelian finite quotient groups. The methods we developed herein would shed light on their associated $p$-form symmetries, with special focus on the structure of $(-1)$-form symmetries within these setups.
    \item It is natural to extend our list of examples to incorporate the SymTFT of lower-dimensional theories, emphasizing the role of $(-1)$-form symmetries. This involves studying 3d theories with four or eight supercharges and their dimensional reductions to 2d, integrating perspectives from both field theory and geometric engineering within M-theory and superstring frameworks.
    \item Our developments can be adapted to the construction of $(-1)$-form symmetries within the framework of AdS/CFT (following \cite{vanBeest:2022fss}), with the aim to uncover their holographic origins and dual implications.
    \item Our progress in understanding the BF terms and topological operators for $U(1)$ symmetries to incorporate the effect of the Chern--Simons term in M-theory begs for a more exhaustive treatment. It would be interesting to investigate SymTFTs obtained as the topological limit of a Maxwell--Chern--Simons action, particularly their implications for generalized symmetries and topological phases.
    \item Perhaps the investigation of $P_7$-fluxbranes carried over in this work and their role as generators of $U(1)$ symmetries, may serve to explore the origins of the R-symmetry within the context of geometric engineering in M-theory.
\end{itemize}

\acknowledgments
We thank Babak Haghighat, Tsung-Ju Lee and Yi Zhang for discussions. LS and YNW would like to thank the Peng Huanwu Center for Fundamental Theory for hosting the ``SymTFT workshop'' and the hospitality. LS also thanks SIMIS for hospitality at various stages of this project. The work of LS is supported by the Shuimu Scholars program of Tsinghua University, and by the Natural Science Foundation of China under Grant W2433005 ``String theory, supersymmetry, and applications to quantum geometry''. MN and YNW are supported by National Natural Science Foundation of China under Grant No. 12175004, No. 12422503 and by Young Elite Scientists Sponsorship Program by CAST (2023QNRC001, 2024QNRC001). YNW is also supported by National Natural Science Foundation of China under Grant No. 12247103. 

\appendix 

\section{Torsion classes and differential cohomology}

\subsection{K\"unneth's theorem for principal ideal domains}
\label{sec:Kunneth}
Let $\mathbf{Z}$ be a principal ideal domain, and $H_{\bullet} (L, \mathbf{Z})$ be the homology of the topological space $L$ with coefficients in $\mathbf{Z}$. We will be interested in the case $\mathbf{Z}= \mathbb{Z}$, the ring of integers.\par
K\"unneth's theorem states that, for any two topological spaces $L, L^{\prime}$, and $k \in \Z$, there exists a short exact sequence
\begin{equation}
\begin{aligned}
    0 \longrightarrow \bigoplus_{i+j=k} H_i (L,\mathbf{Z}) \otimes_{\mathbf{Z}} H_j (L^{\prime},\mathbf{Z}) & \longrightarrow H_k (L \times L^{\prime},\mathbf{Z}) \\
    \longrightarrow & \bigoplus_{i+j=k-1} \mathrm{Tor}_1^{\mathbf{Z}} \left( H_i (L,\mathbf{Z}) , H_j (L^{\prime},\mathbf{Z}) \right) \longrightarrow 0 ,
\end{aligned}
\label{eq:Kunneth}
\end{equation}
and furthermore the sequence splits. In this expression, $\mathrm{Tor}_1^{\mathbf{Z}}$ is the first Tor-functor. In particular, it vanishes whenever either of the two arguments is a free $\mathbf{Z}$-module.\par
For our purposes, the theorem means that K\"unneth's formula receives a correction by torsion classes whenever we take the product of two manifolds both supporting torsion (co)homology classes. The (co)homological grading of these torsion terms is shifted by one compared to the `classical' K\"unneth formula with coefficients in a field. For instance, if both $L$ and $L^{\prime}$ have torsional 1-cycles, $L \times L^{\prime}$ will have torsional 2- and 3-cycles.\par

\paragraph{Tor.}
To explain the Tor-functor, let $\mathsf{A}$ be a right $\mathbf{Z}$-module and $\mathsf{B}$ be a left $\mathbf{Z}$-module. $\mathrm{Tor}_k^{\mathbf{Z}} \left(\mathsf{A},\mathsf{B} \right) $ is computed as follows. Pick any projective resolution 
\begin{equation}
    \cdots \longrightarrow \mathsf{P}_2 \longrightarrow \mathsf{P}_1 \longrightarrow \mathsf{P}_0 \longrightarrow \mathsf{A} \longrightarrow 0 ,
\end{equation}
eliminate $\mathsf{A}$ and form a new complex by tensoring with $\mathsf{B}$ over $\mathbf{Z}$:
\begin{equation}
    \cdots \longrightarrow \mathsf{P}_2 \otimes_{\mathbf{Z}} \mathsf{B} \longrightarrow \mathsf{P}_1 \otimes_{\mathbf{Z}} \mathsf{B} \longrightarrow \mathsf{P}_0 \otimes_{\mathbf{Z}} \mathsf{B}  \longrightarrow 0 .
\end{equation}
Then, $\mathrm{Tor}_k^{\mathbf{Z}} \left(\mathsf{A},\mathsf{B} \right) $ is by definition the $k^{\text{th}}$ homology of this complex (note that the grading increases from right to left).\par
The following is well-known, but we sketch the proof for completeness.
\begin{lemma}
In the setup above, set $\mathbf{Z}=\Z$ and $\mathsf{A}=\Z_m,\mathsf{B}=\Z_n$. It holds that 
\begin{equation}\label{eq:AbelianKunneth}
    \mathrm{Tor}_1^{\Z } \left( \Z_m , \Z_n \right)= \Z_{\mathrm{gcd}(n,m)} .
\end{equation}
\end{lemma}
\begin{proof}
To be more careful, we will distinguish the multiplicative group $\Z_n$ of $n^{\mathrm{th}}$ roots of unity, and denote $\Z/ n\Z$ the additive group of $n$ elements defined $\mod n$.\par
We start by writing the short exact sequence 
\begin{equation}
    0 \longrightarrow \Z \xrightarrow{ \ m \cdot \ } \Z \twoheadrightarrow \Z/m\Z \longrightarrow 0 ,
\end{equation}
with the first map given by multiplication by $m$. This provides us with a resolution of $\Z/m\Z$ by the two-term complex 
\begin{equation}
    [\Z \xrightarrow{ \ m \cdot \ } \Z ] .
\end{equation}
Tensoring with $\Z/n\Z$ over the integers and using $\Z \otimes_{\Z} (\Z/n\Z) \cong \Z/n\Z$, we arrive at  
\begin{equation}
    \Z /n\Z \xrightarrow{ \ m \cdot \ } \Z/n\Z .
\end{equation}
We are ultimately interested in multiplicative groups, instead of additive. The desired complex is obtained from the above by the exponential map
\begin{equation}
    \exp \left( \frac{2 \pi \ii }{n} \cdot \right) : \ \Z /n\Z \longrightarrow \Z_n , \qquad \alpha \mapsto e^{2 \pi \ii \alpha/n} .
\end{equation}
Observe that, exponentiating the complex, the multiplication map $m \cdot$ lifts to a map of groups:
\begin{equation}\label{eq:mnroot1}
    \hat{m} : \Z_n \longrightarrow \Z_n , \qquad u \mapsto u^m ,
\end{equation}
for $u$ a $n^{\text{th}}$ root of unity.
By definition, $\mathrm{Tor}^{\Z}_1 \left( \Z_n, \Z_m \right) $ is computed applying the homology functor $H_1$ to the resulting complex
\begin{equation}
    \Z_n \xrightarrow{ \ \hat{m} \ } \Z_m .
\end{equation}
Clearly, inside the left-most term in the complex, the image of the previous boundary map is trivial; the only remaining task is thus to compute $\mathrm{ker} (\hat{m})$. The kernel of \eqref{eq:mnroot1} consists of $n^{\text{th}}$ roots of unity that are also $m^{\text{th}}$ roots of unity. In multiplicative notation:
\begin{equation}
    \mathrm{ker} \left( \text{\eqref{eq:mnroot1}}\right) =  \left\{ u \in \Z_n \ : \ u^m =1 \right\} \cong \Z_{\mathrm{gcd}(n,m)} 
\end{equation}
(seen as a subgroup of $\Z_n$). We conclude that $\mathrm{Tor}_1^{\Z } \left( \Z_m , \Z_n \right) = \Z_{\mathrm{gcd}(n,m)}$ .
\end{proof}

\subsection{Differential cohomology summary}
\label{app:diffchar}
For introductions on differential cohomology, see e.g. \cite[Sec.3]{Hopkins:2002rd}, \cite[Sec.2]{Freed:2006yc}, or the review \cite{Szabo:2012hc}. For a more exhaustive mathematical treatment, we refer to \cite{Bar:2014} or the review \cite{Bunke:2012rsi}.\par

\subsubsection{Cheeger--Simons characters}
\begin{defin}[{\cite{Cheeger:1985}}]
    Let $M$ be a smooth manifold, and $Z_{k} (M)$ the group of smooth $k$-cycles in $M$. A differential character of degree $p$ on $M$ is a homomorphism of abelian groups
    \begin{equation}
        \chi \ : \ Z_{p-1} (M) \longrightarrow U(1) 
    \end{equation}
    such that $\exists !$ closed differential $p$-form $F^{\chi} \in 2 \pi \Omega_{\Z}^{p} (M)$ such that
    \begin{equation}
        \chi (\partial \Sigma ) = \exp \left( \ii \int_{\Sigma} F^{\chi}  \right) 
    \end{equation}
    for every $p$-chain $\Sigma \in C_{p} (M)$. The group of differential characters of degree $p$ is called Cheeger--Simons group $\br{H}^{p} (M)$.\footnote{These conventions on the grading are by now standard (see e.g. \cite{Freed:2006yc,Bunke:2012rsi,Bar:2014}), but they are shifted by 1 with respect to the seminal paper \cite{Cheeger:1985}.}
\end{defin}
The field strength map \eqref{diffcharFmap} acts as $\chi \mapsto \mathscr{F}(\chi) = \frac{1}{2\pi}F^{\chi} $. We adopt the notation of \cite{Apruzzi:2021nmk} and denote the character $\chi$ of degree $p$ and with field strength $\frac{F}{2\pi}$ by $\br{F}_p$.\par
The Cheeger--Simons group $\br{H}^{p} (M)$ is completely characterized by the diagram \cite{Simons:2007}
    \begin{equation}\label{2exactseqbrH}
    \begin{tikzpicture}
    \node (brH) at (0,0) {$\br{H}^{p} (M)$};
    \node (OZ) at (3,1) {$\Omega^{p}_{\Z} (M)$};
    \node (ohr) at (6,2) {$0$};
    \node (ohl) at (-6,-2) {$0$};
    \node (Hhl) at (-3,-1) {$H^{p-1} (M, \R/\Z)$};
    \draw[->] (ohl) -- (Hhl);
    \draw[->] (Hhl) -- (brH);
    \path[->] (brH) edge node[anchor=south] {${\scriptstyle \mathscr{F}}$} (OZ);
     \draw[->] (OZ) -- (ohr);
     
    \node (HZ) at (3,-1) {$H^{p} (M,\Z)$};
    \node (ovb) at (6,-2) {$0$};
    \node (ovt) at (-6,2) {$0$};
    \node (Oc) at (-3,1) {$\Omega^{p-1} (M) / \Omega^{p-1}_{\Z} (M)$};
    
    \path[->] (brH) edge node[anchor=north] {$\scriptstyle c$} (HZ);
    \draw[->] (ovt) -- (Oc);
    \draw[->] (Oc) -- (brH);
    \draw[->] (HZ) -- (ovb);
    \end{tikzpicture}    
    \end{equation}
    where $H^{p-1} (M, \R/\Z)$ classifies gauge-equivalence classes of flat $(p-1)$-form gauge fields, and $\Omega^{p-1} (M) / \Omega^{p-1}_{\Z} (M)$ classifies topologically trivial $(p-1)$-form gauge fields.

\paragraph{Integration.}
Given a fibration $M \hookrightarrow \mathcal{M} \twoheadrightarrow B$ whose fibres are smooth manifolds of $\dim M =d<p$, we define an integration map:
\begin{equation}
    \int^{\br{H}}_{\mathcal{M}/B} \ : \ \br{H}^{p} (\mathcal{M}) \longrightarrow \br{H}^{p-d} (B) .
\end{equation}
In particular, using $\br{H}^1 (\mathrm{pt}) \cong \R/\Z$, we have the integration $\int^{\br{H}}_{M} $ introduced in \eqref{eq:diffcohoint}.

\subsubsection{Differential characters and higher-form symmetries}
Given a $(p-2)$-gerbe on $M$ with $(p-1)$-connection $A_{p-1}$, the periods of the connection define a map
\begin{equation}\label{periodgerbe}
    Z_{p-1} (M) \longrightarrow U(1) , \qquad \Sigma_{p-1} \mapsto \exp \left( 2 \pi \ii \int_{\Sigma_{p-1}} A_{p-1} \right) .
\end{equation}
The field strength and characteristic class are obtained in the usual way. We thus see that differential characters of degree $p$ classify $(p-2)$-gerbes with $(p-1)$-connection. In particular, when $p=2$, we have that $\br{H}^{2} (M)$ is a model for line bundles with connection on $M$, providing a refinement of $H^{2} (M, \Z)$ as a model for line bundles on $M$.\par

\paragraph{Holonomy.}
The holonomy \eqref{periodgerbe} lifts to differential cohomology. The holonomy of a differential character $\br{F}_p \in \br{H}^p (M)$ along $\Sigma_{p-1} \in Z_{p-1} (M)$ is 
\begin{equation}
    \Sigma_{p-1} \mapsto \exp \left( 2 \pi \ii \int^{\br{H}}_{\Sigma_{p-1}} \br{F}_{p} \right) .
\end{equation}\par
Consider a torsional cycle $\Sigma_{p-1} \in \mathrm{Tor} H_{p-1} (M, \Z)$, and let $n\in \mathbb{N}$ be the minimal integer such that $n\Sigma_{p-1} \equiv 0$. One can choose a $p$-chain $Y_p \in C_{p} (M, \Z)$ such that $\partial Y_p = n \Sigma_{p-1}$. Then the holonomy of $\br{F}_p \in \br{H}^p(M)$ along the torsional cycle is \cite[Remark 16]{Bar:2014}
\begin{equation}
\label{eq:TorHol}
	\exp \left[ \frac{2 \pi \ii}{n} \left( \int_{Y_p} \frac{F_p}{2\pi} - \langle c (\br{F}_p) , Y_p \rangle \right)  \right] .
\end{equation}
Here $\langle \cdot,\cdot \rangle$ is the Kronecker pairing between cohomology and homology, and we are abusing of notation by using the same symbol when $Y_p$ is a chain but not a homology class (see \cite[Sec.5.1]{Bar:2014} for the precise definition).

\subsubsection{Products in differential cohomology}
\label{app:diffcohoprods}

\paragraph{Internal product.}
There exists an internal product operation $\star$ that induces a graded ring structure on $\br{H}^{\bullet} (M)$:
\begin{equation}
    \star \ : \ \br{H}^p (M) \otimes \br{H}^q (M) \longrightarrow \br{H}^{p+q} (M) .
\end{equation}
It descends to the wedge product $\wedge$ in $\Omega^{\bullet} (M)$ under the field strength map, and to the cup product $\smile$ in $H^{\bullet} (M, \Z)$ under the characteristic class map.\par 
Useful properties that we use in the main text are:
\begin{itemize}
    \item[(i)] The product operation satisfies 
\begin{equation}
    \br{F}_p \star \br{G}_q = (-1)^{pq} \br{G}_q \star \br{F}_p
\end{equation}
for any two classes $\br{F}_p \in \br{H}^p (M) , \br{G}_q \in \br{H}^q (M)$. In particular, it follows that:
\begin{equation}
\label{eq:symmstarproduct}
    \frac{1}{2} \left( \br{F}_p \star \br{G}_q + \br{G}_q \star \br{F}_p \right) = \begin{cases} 0 & \text{ if } p,q\in 1 +2 \Z \\ \br{F}_p \star \br{G}_q & \text{ if } p\in 2 \Z \text{ or } q\in 2 \Z . \end{cases}
\end{equation}
    \item[(ii)] When $\br{F}_{p}$ is topologically trivial, with $\mathscr{F} (\br{F}_p) = \frac{1}{2\pi} \dd A_{p-1}$, then $\br{F}_p \star \br{G}_q$ is also topologically trivial, with 
    \begin{equation}
    	\mathscr{F} (\br{F}_p \star \br{G}_q) = \dd \left( \frac{A_{p-1}}{2\pi} \wedge \frac{G_q}{2\pi} \right).
    \end{equation}
\item[(iii)] The holonomy of $\br{F}_p \star \br{F}_p$ along a $(2p-1)$-cycle $\Sigma_{2p-1}$ is locally represented by the Chern--Simons functional
\begin{equation}
    \exp \left( 2\pi \ii \int_{\Sigma_{2p-1}} \br{F}_p \star \br{F}_p \right) ~ \stackrel{\text{\tiny locally}}{=} ~ \exp \left( \ii \int_{\Sigma_{2p-1}} A_{p-1} \wedge \frac{F_p}{2\pi} \right) ,
\end{equation}
but the left-hand side is globally defined, as opposed to the right-hand side.
\end{itemize}

\paragraph{External product.}
Consider the product manifold $L \times M$, with projections 
\begin{equation}
    L \xleftarrow{\ \pi \ } L \times M \xrightarrow{ \ \pi^{\prime} \ } M .
\end{equation}
The internal product $\star$ induces an external product 
\begin{equation}
    \times \ : \ \br{H}^p (L) \otimes \br{H}^q (M) \longrightarrow \br{H}^{p+q} ( L \times M ) 
\end{equation}
via 
\begin{equation}
    \br{F}_p \times \br{G}_q := (\pi^{\ast} \br{F}_p ) \star (\pi^{\prime \ast} \br{G}_q ) .
\end{equation}
\begin{prop}[{see e.g. \cite[Lemma 30]{Bar:2014}}]
Let $L,M$ be closed oriented manifolds. The integration defines a pairing 
\begin{equation}
    \int^{\br{H}}_{ L \times M} \ : \ \br{H}^p (L) \otimes \br{H}^q (M) \longrightarrow \R / \Z 
\end{equation}
with 
\begin{equation}\label{eq:fibreproddiffcoho}
    \int^{\br{H}}_{ L \times M} \br{F}_p \times \br{G}_q = \begin{cases} \left( \int_{M} c (\br{G}_q) \right) \int^{\br{H}}_{L} \br{F}_p & \text{ if } p= \dim L +1 \text{ and } q= \dim M \\ (-1)^p \left( \int_{L} c (\br{F}_p) \right) \int^{\br{H}}_{M} \br{G}_q & \text{ if } p= \dim L \text{ and } q= \dim M +1 \\ 0 & \text{ otherwise}. \end{cases}
\end{equation}
\end{prop}
In particular note that, if e.g. $p=\dim L +1$, $\br{F}_p$ is necessarily topologically trivial for dimensional reasons, thus there exists a gauge field $A_{p-1}$ such that 
\begin{equation}
	\int^{\br{H}}_{L} \br{F}_p = \int_L A_{p-1} \mod 1 .
\end{equation}
Additionally, one can show the following.
\begin{prop}
	Let $L_{p-1}$ be closed and oriented, and $Y_q$ be a manifold with boundary. Take $\br{F}_p \in \br{H}^p (L_{p-1}), \br{G}_q \in \br{H}^q (Y_q)$, and moreover assume that $\br{G}_q$ is topologically trivial, $\mathscr{F}(\br{G}_q) = \frac{1}{2\pi} \dd B_{q-1}$. Then 
	\begin{equation}\label{eq:fibreprodtrivialG}
    	\int^{\br{H}}_{L_{p-1} \times Y_q} \br{F}_p \times \br{G}_{q} = \left( \int_{\partial Y_q} B_{q-1} \right) \int^{\br{H}}_{L_{p-1}} \br{F}_p . 
	\end{equation}
\end{prop}
\begin{proof}The statement follows immediately adapting the proof of \cite[Lemma 30]{Bar:2014} to chains with boundaries. Since it is not stated explicitly there, we report the proof for completeness.\par
We start by the fact that necessarily $\mathscr{F}(\br{F}_p) =\frac{1}{2\pi}\dd A_{p-1}$. We have 
\begin{equation}
\begin{aligned}
	\int^{\br{H}}_{L_{p-1} \times Y_q} \br{F}_p \times \br{G}_{q} &= \int^{\br{H}}_{L_{p-1} \times Y_q} \left( \pi^{\ast} \dd \br{A}_{p-1} \right) \star \left( \pi^{\prime \ast} \br{G}_{q} \right) =  \int^{\br{H}}_{L_{p-1} \times Y_q}  \dd \left( \pi^{\ast} \br{A}_{p-1} \star \pi^{\prime \ast} \br{G}_{q} \right) 
\end{aligned}
\end{equation}
where, for the second equality, we use the fact that $\br{G}_q$ is closed and the compatibility \cite[Eq.(54)]{Bar:2014}. The right-hand side is the holonomy of a topologically trivial differential character on ${L_{p-1} \times Y_q}$, hence 
\begin{equation}
\begin{aligned}
	\int^{\br{H}}_{L_{p-1} \times Y_q} \br{F}_p \times \br{G}_{q} &=  \int_{L_{p-1} \times Y_q} \pi^{\ast} \frac{A_{p-1}}{2\pi} \wedge \pi^{\prime \ast} \mathscr{F}(G_{q}) \\
	&= \left( \int_{L_{p-1}} \frac{1}{2\pi} A_{p-1}\right) \left( \int_{Y_q}  \frac{1}{2\pi} G_{q} \right) .
\end{aligned}
\end{equation}
To conclude the proof we plug in the hypothesis $G_q = \dd B_{q-1}$ and get 
\begin{equation}
	\int^{\br{H}}_{L_{p-1} \times Y_q} \br{F}_p \times \br{G}_{q} = \left( \frac{1}{2\pi} \int_{L_{p-1}} A_{p-1} \right) \left( \frac{1}{2\pi} \int_{\partial Y_q} B_{q-1} \right) .
\end{equation}
\end{proof}

\paragraph{Slant product.}
The slant product is defined at the level of (co)chains as the pairing
\begin{equation}
    C^{p+q} (L \times L^{\prime}) \otimes C_q (L^{\prime} ) \longrightarrow C^p (L) .
\end{equation}
It lifts to a pairing in differential cohomology \cite{Hopkins:2002rd}, see also \cite{GarciaEtxebarria:2024fuk} for its applications to geometric engineering.

\subsubsection{Holonomy versus higher linking}
\label{sec:DelZottoLink}
In the main text we ought to evaluate integrals in differential cohomology. In this appendix we compare our approach, based on computing holonomies of Cheeger--Simons characters, with the method of \cite{DelZotto:2024tae}, which uses higher linking pairing.\par
This subsection is meant for the interested reader, to serve as a dictionary between distinct but equivalent approaches, and is not necessary for our derivation.

\paragraph{Holonomy versus higher linking: Free cycles.}
Let $v_p,v_q^{\prime} \in H^{\bullet} (L, \Z)_{\mathrm{free}}$. By Poincar\'e duality, we have 
\begin{equation}
	\int^{\br{H}}_L \br{v}_p \star \br{v}_q^{\prime}  = \ell_L \left(  \mathsf{PD} (v_p) , \mathsf{PD} (v_q^{\prime}) \right) .
\end{equation}
The result vanishes unless $p+q=\dim (L)+1$, which we assume is the case. The right-hand side is computed by taking a $(d-p+1)$-chain (i.e. a $q$-chain) $\Sigma_{q}$ with $\partial\Sigma_q = \mathsf{PD} (v_p)$, and counting the number points at which $\mathsf{PD} (v_q^{\prime})$ intersects $\Sigma_q$, 
\begin{equation}
    \ell_L \left(  \mathsf{PD} (v_p) , \mathsf{PD} (v_q^{\prime}) \right) = \lvert \Sigma_p \cap \mathsf{PD} (v_q^{\prime}) \rvert .
\end{equation}
The intersection number should be counted with sign, consistent with the rule $\lvert \Sigma_p \cap \mathsf{PD} (v_q^{\prime}) \rvert  = (-1)^{pq} \lvert \Sigma_q^{\prime} \cap \mathsf{PD} (v_p) \rvert $.

\paragraph{Holonomy versus higher linking: Torsion cycles.}
Let us now discuss the situation with torsion cocycles $t_p,t_q^{\prime} \in \mathrm{Tor}H^{\bullet} (L, \Z)$.
By Poincar\'e duality, we have 
\begin{equation}\label{eq:diffcohotptqint}
	\int^{\br{H}}_L \br{t}_p \star \br{t}_q^{\prime}  = \ell_L \left(  \mathsf{PD} (t_p) , \mathsf{PD} (t_q^{\prime}) \right) .
\end{equation}
To compute the right-hand side, \cite{DelZotto:2024tae} introduces a $q$-chain $\Sigma_{q}$ and a $p$-chain $\Sigma^{\prime}_{p}$ with 
\begin{equation} 
    \partial \Sigma_{q} = n  \mathsf{PD} (t_p) , \qquad \partial \Sigma_{p}^{\prime} = m \mathsf{PD} (t_q^{\prime}) .
\end{equation}
Here we are using $\dim (L) +1 =p+q$ for the dimensions. Letting $\omega_{p-1} = \mathsf{PD} (\Sigma_q)$ and $\omega_{q-1}^{\prime} = \mathsf{PD} (\Sigma^{\prime}_{p})$ , one has \cite{DelZotto:2024tae}
\begin{equation}\label{eq:higherlinkSigmapq}
	\ell_L \left(  \mathsf{PD} (t_p) , \mathsf{PD} (t_q^{\prime}) \right)  = \frac{1}{nm} \int_{L} \omega_{p-1} \wedge \dd \omega_{q-1}^{\prime} = \frac{1}{nm} \lvert \Sigma_{q} \cap \partial \Sigma_{p}^{\prime} \rvert .
\end{equation}
We emphasize a few aspects of this formula:
\begin{itemize}
\item The intersection is well-posed if $\dim (\Sigma_q) + \dim (\Sigma_p^{\prime}) -1= \dim (L)$, that is to say $p+q=\dim (L)+1$, which agrees with the condition from the differential cohomology integration \eqref{eq:diffcohotptqint}.
\item The final result involves counting intersection points between chains. Therefore, while $\partial \Sigma_p^{\prime} = m \mathsf{PD}(t_q^{\prime})$ corresponds to a trivial homology class, it nevertheless intersects $\Sigma_{q}$ at an integer number of points.
\item There is an implicit choice of sign to be made, to make the counting of points unambiguous. Integrating by parts in \eqref{eq:higherlinkSigmapq}, we have the self-consistency condition (cf. \cite[Eq.(2.34)]{GarciaEtxebarria:2019caf})
\begin{equation}
    \lvert \Sigma_{q} \cap \partial \Sigma_{p}^{\prime} \rvert = (-1)^{pq} \lvert \Sigma_{p}^{\prime} \cap \partial \Sigma_{q} \rvert  .
\end{equation}
\item The linking number actually depends on the choice of representative chain, not just on the homology class. Shifting $\mathsf{PD}(t_p) \mapsto \mathsf{PD}(t_p) + \partial (\lambda_{q})$, the right-hand side of \eqref{eq:higherlinkSigmapq} is modified by 
\begin{equation}
    \frac{1}{nm} \lvert \Sigma_{q} \cap \partial \Sigma_{p}^{\prime} \rvert  \mapsto \frac{1}{nm} \lvert \Sigma_{q} \cap \partial \Sigma_{p}^{\prime} \rvert + \frac{1}{m} \lvert \lambda_{q} \cap \partial \Sigma_{p}^{\prime} \rvert .
\end{equation}
This parallels the differential cohomology statement that different discrete gauge fields in the differential character $\br{t}_p$ have different holonomy, whereby giving a different result.
\item For $p,q$ such that \eqref{eq:diffcohotptqint} is non-vanishing, we have the Pontryagin duality 
\begin{equation}
    \mathrm{Tor} H^q (L, \Z) = \mathrm{Hom} \left( \mathrm{Tor} H^p (L, \Z) , U(1) \right) .
\end{equation}
We write $\mathrm{Tor} H^p (L, \Z) = \bigoplus_{\alpha=1}^{\mu}\Z_{m_{\alpha}}$, which, using the canonical isomorphism $\widehat{\Z}_m \cong \Z_m$, implies 
\begin{equation}
    \mathrm{Tor} H^q (L, \Z)=\bigoplus_{\alpha=1}^{\mu}\Z_{m_{\alpha}} .    
\end{equation}
There are two cases for \eqref{eq:higherlinkSigmapq}:
\begin{itemize}
    \item[(i)] If $n=m_{\alpha}=m$ corresponds to the same direct summand, then  
    \begin{equation}
        \frac{1}{n^2} \lvert \Sigma_{q} \cap \partial \Sigma_{p}^{\prime} \rvert \in \frac{1}{n} \Z ;
    \end{equation}
    \item[(ii)] If $n=m_{\alpha}, m=m_{\ne \alpha}$ correspond to disjoint summands, then we can select the generators in such a way that \eqref{eq:higherlinkSigmapq} vanishes.
\end{itemize}
\end{itemize}

\paragraph{Holonomy versus higher linking: mixing torsion and free cycles.}
Repeating the argument with $t_p \in \mathrm{Tor} H^{p} (L,\Z)$, $v_q\in H^q (L, \Z)_{\mathrm{free}}$, we set $\partial \Sigma_q = n \mathsf{PD}(t_p)$ and compute 
\begin{equation}
	\int^{\br{H}}_L \br{t}_p \star \br{v}_q = \int^{\br{H}}_{ \mathsf{PD}(v_q)} \br{t}_p  = \frac{1}{n} \lvert \Sigma_q \cap \mathsf{PD}(v_q) \rvert .
\end{equation}

\subsection{Integrals of torsion cocycles}\label{app:IntTorCycles}

In this appendix, we compute the holonomy of torsion classes in $\br{H}^{\bullet}(L)$ for various closed manifolds $L$ of interest in the main text.\par

\subsubsection{Torsion classes and tensor products}

Let $\br{t}_p, \br{t}_q^{\prime} $ be the lift to $\br{H}^{\bullet}(L)$ of torsion classes, with $p+q= \dim (L) +1$. Denote for simplicity the torsion cohomology groups generated by $t_p$ and $t_q^{\prime}$ by, respectively, $\Z_n$ and $\Z_m$. The holonomy of $\br{t}_p \star \br{t}_q^{\prime}$ is 
\begin{equation}
	\exp \left( 2 \pi \ii  \int^{\br{H}}_L \br{t}_p \star \br{t}_q^{\prime} \right) .
\end{equation}
\begin{lemma}In the setup above, 
\begin{equation}
	\mathrm{gcd} (n,m) \br{t}_p \star \br{t}_q^{\prime} \equiv 0 \quad \text{in } \br{H}^{p+q} (L) .
\end{equation}
\end{lemma}
\begin{proof}
	First, we observe that, since $n t_p \equiv 0$ and $m t_q \equiv 0$, then $n \br{t}_p \star \br{t}_q^{\prime} \equiv 0 \equiv m \br{t}_p \star \br{t}_q^{\prime} $.\par
	Write $n=g n^{\prime}, m= g m^{\prime}$ with $\mathrm{gcd} (n^{\prime}, m^{\prime})=1$, and assume without loss of generality that $m^{\prime}>n^{\prime}$. We write $m^{\prime}= \mu n^{\prime} + \nu $, $0 < \nu < n^{\prime}$. Then, 
	\begin{equation}
		0 \equiv g m^{\prime} (\br{t}_p \star \br{t}_q^{\prime} ) = g m^{\prime} (\br{t}_p \star \br{t}_q^{\prime} )- \mu g n^{\prime} (\br{t}_p \star \br{t}_q^{\prime} ) = g \nu (\br{t}_p \star \br{t}_q^{\prime} ) .
	\end{equation}
	We can now write $n^{\prime}= \mu^{\prime} \nu + m^{\prime\prime}$ and, repeating the argument with the due changes, we have 
	\begin{equation}
		0 \equiv g n^{\prime} (\br{t}_p \star \br{t}_q ) =  g m^{\prime\prime}(\br{t}_p \star \br{t}_q^{\prime} ) ,
	\end{equation}
	for some integer $0<m^{\prime\prime}<\nu<n^{\prime}<m^{\prime}$. Iterating, we eventually arrive at rest 1 and hence $g \cdot 1 \cdot (\br{t}_p \star \br{t}_q^{\prime} ) \equiv 0$ in cohomology.
\end{proof}
We thus have that $ \br{t}_p \star \br{t}_q^{\prime}$ is $\mathrm{gcd} (n,m) $-torsional, and its holonomy takes values in the multiplicative group $\Z_{\mathrm{gcd} (n,m)}$.\par

\subsubsection{Torsion cohomology of lens spaces}
\label{sec:CSlens3d}
\begin{table}[th]
    \centering
    \begin{tabular}{|c|cccccc|}
    \hline
        $\Gamma$ & $A_{n-1}$ & $D_{\mathrm{even}}$ & $D_{\mathrm{odd}}$ & $E_6$ & $E_7$ & $E_8$\\
        \hline
         $\Gamma^{\mathrm{ab}}$& $\Z_n$ & $\Z_2 \oplus \Z_2$ & $\Z_4$ & $\Z_3$ & $\Z_2$ & $1$ \\
         \hline
    \end{tabular}
    \caption{ADE groups and their abelianization.}
    \label{tab:GammaAb}
\end{table}\par

Let $\Gamma$ a finite ADE group, and consider $\mathbb{S}^3/\Gamma$. The (co)homology groups are:
\begin{equation}
    H_{3-\bullet} (\mathbb{S}^3/\Gamma, \Z ) \cong H^{\bullet} (\mathbb{S}^3/\Gamma, \Z ) = \Z \ \oplus \ 0 \ \oplus \ \Gamma^{\mathrm{ab}} \ \oplus \ \Z ,
\end{equation}
where $\Gamma^{\mathrm{ab}} := \Gamma / [\Gamma, \Gamma ]$ is the abelianization of $\Gamma$, listed in Table \ref{tab:GammaAb}.\par
Let $t_2$ denote the torsional generator of $H^2 (\mathbb{S}^3/\Gamma, \Z )$; its Poincar\'e dual $\mathsf{PD} (t_2)$ is a homology class with representative $\mathbb{S}^1/\Gamma$, the Hopf fibre acted on by $\Gamma$. We introduce the shorthand notation
\begin{equation}
\label{eq:CSGamma}
    \mathrm{CS}_{\Gamma}^{(3)} :=\frac{1}{2}  \int^{\br{H}}_{\mathbb{S}^3/\Gamma} \br{t}_2 \star \br{t}_2 .
\end{equation}
The integral was evaluated in \cite{GarciaEtxebarria:2019caf} and \cite[Sec.3.3]{Apruzzi:2021nmk}, with result
\begin{equation}
    - \mathrm{CS}^{(3)}_{\Gamma} = \begin{cases} \frac{n-1}{2n} & \Gamma = A_{n-1} \\ \frac{2n+1}{8} & \Gamma = D_{2n+1} \\ \frac{5}{3} & \Gamma = E_6 \\ \frac{3}{4} & \Gamma = E_7 \end{cases}
\end{equation}\par

\paragraph{Motivation.}
The only reason to include the present example in this appendix is to compare our approach with other methods used in the literature, showing perfect agreement in this well studied example. Along the way, this example allows us to explain some subtleties concerning the importance of the choice of generator in the computation of the coefficients in the SymTFT action.

\paragraph{Derivation via holonomy.}
We now produce a different derivation of this result, using the holonomy of differential characters. This make manifest the machinery introduced so far. We focus for concreteness on $\Gamma=A_{n-1}$.\par
To compute the integral \eqref{eq:CSGamma}, we evaluate the holonomy of the differential character $\br{t}_2 \star \br{t}_2 \in \br{H}^4 (\mathbb{S}^3/\Z_n)$ on the 3-cycle $[\mathbb{S}^3/\Z_n]$. The character $\br{t}_2 \star \br{t}_2$ is necessarily flat for dimensional reasons, therefore 
\begin{equation}\label{eq:brt2pairing}
    \int_{\mathbb{S}^3/\Z_n}^{\br{H}} \br{t}_2 \star \br{t}_2 = - \frac{1}{n} \langle c(\br{t}_2 \star \br{t}_2) , \Sigma_4 \rangle ,
\end{equation}
where $\langle \cdot,\cdot \rangle $ denotes the perfect pairing between chains and and cochains, and $\Sigma_4$ is a 4-chain with $\partial \Sigma_4 = \mathbb{S}^3/\Z_n$. Explicitly:
\begin{equation}
\begin{aligned}
    \langle c(\br{t}_2 \star \br{t}_2) , \Sigma_4 \rangle &= \int_{\Sigma_4} c(\br{t}_2) \smile c(\br{t}_2) \\
    	&= \int_{\Sigma_2} \mathscr{F}(\br{t}_2) = \oint_{n \mathsf{PD} (t_2)} u_1 ,
\end{aligned}
\end{equation}
where $\partial \Sigma_2 = n \mathsf{PD} (t_2)$. The second equality uses Poincar\'e duality and, in the last equality, $u_1$ is a properly quantized 1-form gauge field associated to $\br{t}_2$. By the quantization condition, we have that the integral takes values in $\Z/n\Z$. The actual value depends on the choice of $u_1$, each choice corresponding to a different generator of $\Z_n$, and hence to a different generator of the electric 1-form symmetry. Inserting this back into \eqref{eq:brt2pairing} and comparing with \eqref{eq:CSGamma}, we obtain 
\begin{equation}
    \mathrm{CS}^{(3)}_{A_{n-1}} = - \frac{1}{2n} \int_{\Sigma_2} \mathscr{F}(\br{t}_2) = - \frac{1}{2n} \oint_{n \mathsf{PD} (t_2)} u_1 .
\end{equation}
The choice of generator of $\br{H}^2 (\mathbb{S}^3/\Z_n)$ with holonomy $\exp ( 2 \pi \ii (n-1)/n)$ reproduces the result of \cite{Apruzzi:2021nmk}.

\paragraph{Derivation via linking pairing.}
For the sake of comparison, we show the agreement of our derivation using differential cohomology, with a derivation along the lines of \cite[App.A]{DelZotto:2024tae}.\par
First, one writes
\begin{equation}
\label{eq:linkPDt2}
	\frac{1}{2}\int_{\mathbb{S}^3/\Z_n}^{\br{H}} \br{t}_2 \star \br{t}_2 = \frac{1}{2} \ell_{\mathbb{S}^3/\Z_n} \left( \mathsf{PD} (t_2) , \mathsf{PD} (t_2) \right) =  \frac{1}{2n^2} \lvert n \mathsf{PD} (t_2)\cap \Sigma_2 \rvert  \mod 1 ,
\end{equation}
for any $\partial \Sigma_2 = n \mathsf{PD} (t_2)$. The second equality has been reviewed in Appendix \ref{sec:DelZottoLink}. Besides, we are making explicit that only the mod 1 part of the result matters (as it ultimately appears in an exponential, multiplied by $2\pi \ii$). To count the number of intersection points, we use that $\mathsf{PD} (t_2)$ is the $\Z_n$-quotient of the Hopf fibre of $\mathbb{S}^3$. We thus evaluate the linking number by considering a Hopf link, fill in a disk with boundary one of the circles, and count the intersection points.\par 
The outcome is $\pm n$, with sign depending on the orientation convention. It turns out that 
\begin{itemize}
    \item[---] The orientation implicitly chosen in \cite[Sec.4.1]{DelZotto:2024tae} gives $+n$; however
    \item[---] To match with the result of \cite[Sec.3.3]{Apruzzi:2021nmk} and \cite[Sec.4]{Cvetic:2021sxm}, we must choose a generator with opposite orientation, thus getting an intersection number of $-n$.
\end{itemize}

\paragraph{Derivation via pullback and intersection.}
Let us quickly review how the result was obtained in \cite{Apruzzi:2021nmk}, referring to the literature for more details. The linking number \eqref{eq:linkPDt2} was computed in \cite[Eq.(3.26)]{Apruzzi:2021nmk}:
\begin{equation}
\label{eq:PDt2LinkApruzzi}
	\frac{1}{2} \ell_{\mathbb{S}^3/\Z_n} \left( \mathsf{PD} (t_2) , \mathsf{PD} (t_2) \right) = \frac{1}{2} \left( \frac{1}{n} \sum_{i=1}^{n-1} i S_i \right) \cdot_{\scriptscriptstyle \widetilde{X}_4} \left( \frac{1}{n} \sum_{j=1}^{n-1} j S_j \right) .
\end{equation}
On the right-hand side, 
\begin{equation}
    \widetilde{X}_4 \longrightarrow \C^2 /\Z_n 
\end{equation}
is the minimal crepant resolution, with exceptional divisors $S_i$, and the pullback of $\mathsf{PD} (t_2)$ to $\widetilde{X}_4$ has been identified with $1/n$ times the central generator $\sum_i iS_i$. Using
\begin{equation}
    S_i \cdot_{\scriptscriptstyle \widetilde{X}_4} S_j = \delta_{j,i+1} + \delta_{j,i-1} -2 \delta_{j,i} , 
\end{equation}
we obtain 
\begin{equation}
\begin{aligned}
    \mathrm{CS}^{(3)}_{A_{n-1}} &= \frac{1}{2n^2} \left[ \sum_{i=1}^{n-2} i(i+1) +  \sum_{i=2}^{n-1} i(i-1) -2  \sum_{i=1}^{n-1} i^2\right] = - \frac{n-1}{2n}
\end{aligned}
\end{equation}
as claimed.\par
It is important to stress that the identification of the pullback of $\mathsf{PD} (t_2)$ to $\widetilde{X}_4$ with $\frac{1}{n} \sum_i S_i$ corresponds to fix a choice of generator of the torsion group $\Z_n$. Physically, this choice is motivated in \cite{Apruzzi:2021nmk} to match with the generator of the electric 1-form symmetry usually considered in the physics literature.

\subsubsection{Torsion cohomology of links of Calabi--Yau threefolds}
\label{sec:torintegralsL5}

\paragraph{Torsion cocycles.}
We compute the integrals \eqref{eq:L5CSmatrices}, which we report here for ease of the reader:
\begin{equation}
\begin{aligned}
     (\mathrm{CS}^{(5)})_{\alpha \alpha^{\prime} \alpha^{\prime\prime}} & = \frac{1}{6}\int_{L_5}^{\br{H}} \br{t}_2^{\alpha} \star \br{t}_2^{\alpha^{\prime}} \star \br{t}_2^{\alpha^{\prime\prime}} \\
     (\widetilde{\mathrm{CS}}^{(5)})_{\beta \beta^{\prime}} & = \int_{L_5}^{\br{H}} \br{t}_3^{\beta} \star \br{t}_3^{\beta^{\prime}} , \\
     (\kappa^{(5)})_{\alpha \alpha^{\prime}} & = \int_{L_5}^{\br{H}} \br{t}_4^{\alpha} \star \br{t}_2^{\alpha^{\prime}} .
\end{aligned}
\label{eq:L5CSapp}
\end{equation}\par
The computation of $\widetilde{\mathrm{CS}}^{(5)}$ is similar to Appendix \ref{sec:CSlens3d}. First, we note that
\begin{equation}
	n_{\beta} t_3^{\beta} \equiv 0 \quad \text{ and } \quad n_{\beta^{\prime}} t_3^{\beta^{\prime}} \equiv 0 \qquad \Longrightarrow \qquad \mathrm{gcd} (n_{\beta}, n_{\beta^{\prime}}) t_3^{\beta} \smile t_3^{\beta^{\prime}} \equiv 0 .
\end{equation}
Then, the holonomy of $\br{t}_3^{\beta} \star \br{t}_3^{\beta^{\prime}}$ along $[L_5]$ is given by 
\begin{equation}
\begin{aligned}
	(\widetilde{\mathrm{CS}}^{(5)})_{\beta \beta^{\prime}} &= \ell_{L_5} \left( \mathsf{PD} (t_3^{\beta}) , \mathsf{PD} (t_3^{\beta^\prime}) \right) \\
	&= -\frac{1}{  \mathrm{gcd} (n_{\beta}, n_{\beta^{\prime}}) } \langle c ( \br{t}_3^{\beta}) \smile c ( \br{t}_3^{\beta^{\prime}}), \Sigma_6 \rangle \ \in \frac{1}{\mathrm{gcd} (n_{\beta}, n_{\beta^{\prime}})} \Z ,
\end{aligned}
\end{equation}
When $\beta^{\prime}=\beta$, the integral vanishes by anti-symmetry. When $\beta^{\prime}\neq\beta$, the computation of this integral boils down to the correct identification of generators of $\br{H}^3 (L_5)$ with the physical generators. Using that the pullback of $\mathsf{PD} (t_3^{\beta})$ to $\widetilde{X}_6$ is a Lagrangian, we can write 
\begin{equation}
    (\widetilde{\mathrm{CS}}^{(5)})_{\beta \beta^{\prime}} = - \frac{\ell_{\beta ,\beta^{\prime}}}{\mathrm{gcd} (n_{\beta}, n_{\beta^{\prime}})} ,
\end{equation}
for some integers $\ell_{\beta ,\beta^{\prime}}\in \Z$ computed by counting intersection numbers of (suitably identified) Lagrangian 3-cycles in $\widetilde{X}_6$.\par
To evaluate the matrix $\kappa^{(5)}$ we use that the holonomy of $\br{t}_4^{\alpha} \star \br{t}_2^{\alpha^{\prime}}$ along $[L_5]$ reduces to 
\begin{equation}\label{eq:kappaalphadelta}
	(\kappa^{(5)})_{\alpha \alpha^{\prime}} = - \frac{1}{m_{\alpha}} \langle c(\br{t}_2^{\alpha^\prime}) , \Sigma_2^{\alpha} \rangle 
\end{equation}
where $\partial \Sigma_2^{\alpha} =  m_{\alpha}\mathsf{PD} (t_4^{\alpha})$. We choose dual bases such that the pairing yields $\delta_{\alpha, \alpha^\prime}$, whence 
\begin{equation}
	(\kappa^{(5)})_{\alpha \alpha^{\prime}} = - \frac{1}{m_{\alpha}}  \delta_{\alpha, \alpha^\prime} .
\end{equation}\par
Finally, we evaluate $\mathrm{CS}^{(5)}$. The computation using differential cohomology shows that the holonomy of $\br{t}_2^{\alpha} \star \br{t}_2^{\alpha^{\prime}} \star \br{t}_2^{\alpha^{\prime\prime}}$ is  
\begin{equation}
\begin{aligned}
	\int_{L_5}^{\br{H}} \br{t}_2^{\alpha} \star \br{t}_2^{\alpha^{\prime}} \star \br{t}_2^{\alpha^{\prime\prime}} &= -  \frac{1}{m_{\alpha}}\langle c (\br{t}_2^{\alpha^{\prime}}) \smile c (\br{t}_2^{\alpha^{\prime\prime}})  , \Sigma_4^{\alpha} \rangle \  \in \frac{1}{ m_{\alpha}} \Z \\
	&= -  \frac{1}{m_{\alpha^{\prime}}}\langle c (\br{t}_2^{\alpha}) \smile c (\br{t}_2^{\alpha^{\prime\prime}})  , \Sigma_4^{\alpha^{\prime}} \rangle \  \in \frac{1}{ m_{\alpha^{\prime}}} \Z \\ 
	&= -  \frac{1}{m_{\alpha^{\prime\prime}}}\langle c (\br{t}_2^{\alpha}) \smile c (\br{t}_2^{\alpha})  , \Sigma_4^{\alpha^{\prime\prime}} \rangle \  \in \frac{1}{ m_{\alpha^{\prime\prime}}} \Z .	
\end{aligned}
\end{equation}
The equality of the three expression implies that, in fact, 
\begin{equation}
	\int_{L_5}^{\br{H}} \br{t}_2^{\alpha} \star \br{t}_2^{\alpha^{\prime}} \star \br{t}_2^{\alpha^{\prime\prime}} \in \frac{1}{\mathrm{gcd} (m_{\alpha},m_{\alpha^\prime}, m_{\alpha^{\prime\prime}} )} \Z .
\end{equation}
We then use that $t_2^{\alpha^{\prime}} \smile t_2^{\alpha^{\prime\prime}}$ can be expanded into the basis of $\mathrm{Tor}H^4 (L_5, \Z)$, as
\begin{equation}
	t_2^{\alpha^{\prime}} \smile t_2^{\alpha^{\prime\prime}} = \sum_{\gamma=1}^{\mu} \widehat{\kappa}_{\gamma \alpha^{\prime}\alpha^{\prime\prime}} t_4^{\gamma} .
\end{equation}
The left-hand side takes values in $\Z / \mathrm{gcd} (m_{\alpha^\prime}, m_{\alpha^{\prime\prime}} ) \Z$, whereas each term in the right-hand side is valued in $\Z/m_{\gamma} \Z$. The coefficient should thus vanish unless $\mathrm{gcd} (m_{\alpha^\prime}, m_{\alpha^{\prime\prime}} ) $ divides $m_{\gamma}$, in which case 
\begin{equation}
	\widehat{\kappa}_{\gamma \alpha^{\prime}\alpha^{\prime\prime}} \in \frac{m_{\gamma}}{\mathrm{gcd} (m_{\gamma},m_{\alpha^\prime}, m_{\alpha^{\prime\prime}} ) } \Z .
\end{equation}
Plugging back into \eqref{eq:L5CSmatrices}, we compute 
\begin{equation}
\begin{aligned}
	6(\mathrm{CS}^{(5)})_{\alpha \alpha^{\prime} \alpha^{\prime\prime}} &= \sum_{\gamma=1}^{\mu} \widehat{\kappa}_{\gamma \alpha^{\prime}\alpha^{\prime\prime}} \int_{L_5}^{\br{H}} \br{t}_2^{\alpha} \star \br{t}_4^{\gamma} \\
	&= \sum_{\gamma=1}^{\mu} \widehat{\kappa}_{\gamma \alpha^{\prime}\alpha^{\prime\prime}} (\kappa^{(5)})_{\gamma \alpha} = - \frac{1}{m_{\alpha}}  \widehat{\kappa}_{\alpha \alpha^{\prime}\alpha^{\prime\prime}} . 
\end{aligned}
\end{equation}
The right-hand side takes values in $[1/ \mathrm{gcd} (m_{\alpha},m_{\alpha^\prime}, m_{\alpha^{\prime\prime}} ) ]\Z$, as claimed.\par
In turn, the coefficients can be computed in a crepant resolution $\widetilde{X}_6 \longrightarrow X_6$, as the triple intersection numbers of the divisors $\mE^{\prime}_{\alpha}$, which are the pullback to $\widetilde{X}_6$ of the 3-chains $m_{\alpha}\mathsf{PD} (t_2^{\alpha})$ in $L_5$.
In other words, 
\begin{equation}
\begin{aligned}
	(\mathrm{CS}^{(5)})_{\alpha \alpha^{\prime} \alpha^{\prime\prime}} = \frac{1}{6 m_{\alpha}m_{\alpha^\prime} m_{\alpha^{\prime\prime}}} & \mE_{\alpha}^{\prime} \cdot_{\scriptscriptstyle \widetilde{X}_6} \mE_{\alpha^{\prime}}^{\prime} \cdot_{\scriptscriptstyle \widetilde{X}_6} \mE_{\alpha^{\prime\prime}}^{\prime} .
\end{aligned}
\end{equation}

\paragraph{Mixed torsion and free cocycles.}
Besides the above coefficients, we need to evaluate  
\begin{equation}
\begin{aligned}
     (\eta^{(5)})_{\alpha \alpha^{\prime} i} & = \frac{1}{2}\int_{L_5}^{\br{H}} \br{t}_2^{\alpha} \star \br{t}_2^{\alpha^{\prime}} \star \br{v}_2^{i} \\
     (\widetilde{\eta}^{(5)})_{\beta i} & = \int_{L_5}^{\br{H}} \br{t}_3^{\beta} \star \br{v}_3^{i} , \\
     (\widetilde{\kappa}^{(5)})_{\alpha i} & = \int_{L_5}^{\br{H}} \br{t}_4^{\alpha} \star \br{v}_2^{i} , \\
     (\widehat{\eta}^{(5)})_{\alpha i j} & = \int_{L_5}^{\br{H}} \br{t}_2^{\alpha} \star \br{v}_2^{i} \star \br{v}_2^{j} .
\end{aligned}
\label{eq:L5ETAmatrices}
\end{equation}
A change of generators of $H^2(L_5,\Z)_{\mathrm{free}}$ shifts 
\begin{equation}
	\widehat{\eta}^{(5)}_{\alpha i j} \mapsto \widehat{\eta}^{(5)}_{\alpha i j} + \sum_{\alpha^{\prime}=1}^{\mu} m_{j}^{\alpha^{\prime}} \eta^{(5)}_{\alpha \alpha^{\prime} i} +  \sum_{\alpha^{\prime},\alpha^{\prime\prime}=1}^{\mu} m_{j}^{\alpha^{\prime}}m_{i}^{\alpha^{\prime\prime}} \mathrm{CS}^{(5)}_{\alpha \alpha^{\prime} \alpha^{\prime\prime}} 
\end{equation}
for an arbitrary matrix of integers, $m_i^{\alpha} \in \Z$. Following \cite[Sec.4]{Apruzzi:2021nmk}, we conveniently fix the generators so that $\widehat{\eta}^{(5)}_{\alpha i j}$ is set to zero.\par
We have 
\begin{equation}
	(\widetilde{\kappa}^{(5)})_{\alpha i}= \int_{L_5}^{\br{H}} \br{t}_4^{\alpha} \star \br{v}_2^{i} = - \frac{1}{m_{\alpha}} \langle c(\br{t}_4^{\alpha}) , \Sigma_4^{i} \rangle = \frac{1}{m_{\alpha}} \lvert m_\alpha \mathsf{PD} (t_4^{\alpha}) \cap \Sigma_4 \rvert 
\end{equation}
for any 4-chain with $\partial \Sigma_4^{i} = \mathsf{PD} (v_2^{i})$. In turn, given a resolution $\widetilde{X}_6 \longrightarrow X_6$, $m_\alpha \mathsf{PD} (t_2^{\alpha})$ pulls back to a non-compact 2-cycle and $\mathsf{PD} (v_2^{i})$ pulls back to a non-compact 4-cycle $\mE_i$, which generates a Cartan factor $U(1)_i$ of the flavor symmetry group. 
Additionally, we have the relation 
\begin{equation}
	(\eta^{(5)})_{\alpha^{\prime} \alpha^{\prime\prime} i} = \frac{1}{2} \sum_{\alpha=1}^{\mu} \underbrace{\frac{ \left( \mE_{\alpha}^{\prime} \cdot_{\scriptscriptstyle \widetilde{X}_6} \mE_{\alpha^{\prime}}^{\prime} \cdot_{\scriptscriptstyle \widetilde{X}_6} \mE_{\alpha^{\prime\prime}}^{\prime} \right)}{m_{\alpha^{\prime}}m_{\alpha^{\prime\prime}}} }_{\in \Z} (\widetilde{\kappa}^{(5)})_{\alpha i} ~.
\end{equation}

\paragraph{Free cocycles trivialize.}
There are additional terms, that we are omitting from the derivation of the action in Subsection \ref{sec:SymTFT5dM}. Reductions only along free cycles in the link would contribute an integer to the action after the integration, thus being trivial. For example, 
\begin{equation}
\begin{aligned}
	\sum_{1 \le i <j \le b^2} & \left( \int_{L_5}^{\br{H}} \br{v}_3^{i} \star \br{v}_3^{j} \right) \int_{Y_6} \mathscr{F}(\br{F}_4) \wedge \mathscr{F}(\br{F}_1^{i}) \wedge \mathscr{F}(\br{F}_1^{j}) \\
	= \sum_{1 \le i <j \le b^2} & \underbrace{\ell_{L_5} \left( \mathsf{PD}(v_3^{i}),\mathsf{PD}(v_3^{j}) \right)}_{\in \Z} \underbrace{\int_{Y_6} \mathscr{F}(\br{F}_4) \wedge \mathscr{F}(\br{F}_1) \wedge \mathscr{F}(\br{F}_1)}_{ \in \Z} ,
\end{aligned}
\end{equation}
where we use that the linking in $L_5$ of the 2-cycles Poincar\'e dual to $v_3^{\bullet}$ is an integer, and the integral of the field strengths is also integer (with our normalization for the map $\mathscr{F}$).

\paragraph{Example: Local \texorpdfstring{$\mathbb{P}^1 \times \mathbb{P}^1$}{P1xP1}.}
Let us exemplify the calculation in the simplest example of 5d $\mathcal{N}=1$ $\mathfrak{su}(2)$ Yang--Mills theory. This will serve as a cross-check of our computations in differential cohomology.\par
$\mathfrak{su}(2)$ Yang--Mills theory is geometrically engineered from M-theory on the local del Pezzo surface $\mathbb{P}^1 \times \mathbb{P}^1$, and has been extensively studied in the literature. We refer to \cite[Sec.4.1]{Closset:2018bjz} for a detailed analysis of the toric geometric data. In the conventions therein, the generator of the $U(1)^{\scriptscriptstyle [0, \mathrm{in}]}$ instanton flavor symmetry is $(D_1-D_2)$, but it does not satisfy our assumption on $v_2$ in Subsection \ref{sec:torintegralsL5}. We therefore choose a different generator, by shifting $\br{v}_2 + \br{t}_2\mapsto \br{v}_2$, corresponding to 
\begin{equation}
	\mE= D_1 , \qquad \mE^{\prime} = 2 D_1 - 2 D_2 .
\end{equation}
Furthermore, demanding \eqref{eq:normt4inYN0} we have that $N \mathsf{PD} (t_4)$ pulls back to the curve $D_1 \cdot D_2$. The triple intersection numbers involving three $D_i$ in this geometry have been computed in \cite[Eq.(B.18)]{Closset:2018bjz}. Applying those formulas, we get:
\begin{equation}
\begin{aligned}
	6 \mathrm{CS}^{Y^{2,0}} &= (D_1 + D_2)^3 = -1 \equiv 2-1 \mod 2  \\
	2\eta^{Y^{2,0}} &= D_1 \cdot (D_1+D_2)^2 = - \frac{1}{2}  \\
	\widetilde{\kappa}^{Y^{2,0}} &= D_1 \cdot D_2 \cdot (D_1+D_2) =  \frac{1}{2}
\end{aligned}
\end{equation}
in perfect agreement with the general formulas \eqref{eq:CSYN0inZ} specialized to $N=2$.

\subsubsection{Torsion cohomology of products of lens spaces}
\label{sec:CSlens6d}
Consider $L_6 = L_{3,\mathrm{f}} \times L_{3,\mathrm{b}}$, together with the natural projections 
\begin{equation}
    L_{3,\mathrm{f}}  \xleftarrow{ \ \pi_{\mathrm{f}} \ } L_{3,\mathrm{f}} \times L_{3,\mathrm{b}} \xrightarrow{ \ \pi_{\mathrm{b}} \ } L_{3,\mathrm{b}} ,
\end{equation}
where 
\begin{equation}
     L_{3,\mathrm{f}} = \mathbb{S}_{\mathrm{f}}^3 / \Z_n , \qquad L_{3,\mathrm{b}} = \mathbb{S}_{\mathrm{b}}^3 / \Z_p .
\end{equation}
The differential cohomology of the lens spaces has been reviewed in Appendix \ref{sec:CSlens3d}. The (co)homology of $L_{3,\mathrm{f}} \times L_{3,\mathrm{b}}$ is computed via K\"unneth's theorem, as reviewed in Appendix \ref{sec:Kunneth}. We obtain:
\begin{equation}
    H^{\bullet} (L_{3,\mathrm{f}} \times L_{3,\mathrm{b}}, \Z) = \Z  \ \oplus \ 0  \ \oplus \ \begin{matrix} \Z_n \\ \oplus \\ \Z_{p} \end{matrix} \ \oplus \  \begin{matrix} \Z^{\oplus 2} \\ \oplus \\ \Z_{\mathrm{gcd}(n,p)} \end{matrix} \ \oplus \  \Z_{\mathrm{gcd}(n,p)} \ \oplus \ \begin{matrix} \Z_{n} \\ \oplus \\ \Z_p \end{matrix} \ \oplus \ \Z \,.
\end{equation}
The generators are chosen according to Subsection \ref{sec:MonG2SymTFT}. We denote $t_2^{\alpha}, \vol_3^{\alpha}$ the generators of $\mathrm{Tor} H^2 (L_{3, \alpha}, \Z)$ and $H^2 (L_{3, \alpha}, \Z)_{\text{free}}$, for $\alpha=\mathrm{f},\mathrm{b}$. This implies that:
\begin{itemize}
    \item The generators of $\br{H}^{2} (L_{3,\mathrm{f}} \times L_{3,\mathrm{b}})$ are $\{ \pi_{\alpha}^{\ast} \br{t}_2^{\alpha} , \ \alpha=\mathrm{f},\mathrm{b} \}$;
    \item The generators of $\br{H}^{3} (L_{3,\mathrm{f}} \times L_{3,\mathrm{b}})$ are $\{ \br{t}_3\} \cup \{ \pi_{\alpha}^{\ast} \br{\vol}_3^{\alpha} , \ \alpha=\mathrm{f},\mathrm{b} \}$;
    \item The generator of $\br{H}^{4} (L_{3,\mathrm{f}} \times L_{3,\mathrm{b}})$ is $\br{t}_2^{\mathrm{f}} \times \br{t}_2^{\mathrm{b}} = (\pi^{\ast}_{\mathrm{f}} \br{t}_2^{\mathrm{f}} ) \star (\pi^{\ast}_{\mathrm{b}} \br{t}_2^{\mathrm{b}} )$;
    \item The generators of $\br{H}^{5} (L_{3,\mathrm{f}} \times L_{3,\mathrm{b}})$ are $\br{t}_2^{\mathrm{f}} \times \br{\vol}_3^{\mathrm{b}} = (\pi^{\ast}_{\mathrm{f}} \br{t}_2^{\mathrm{f}} ) \star (\pi^{\ast}_{\mathrm{b}} \br{\vol}_3^{\mathrm{b}} )$ and $\br{t}_2^{\mathrm{b}} \times \br{\vol}_3^{\mathrm{f}} = (\pi^{\ast}_{\mathrm{b}} \br{t}_2^{\mathrm{b}} ) \star (\pi^{\ast}_{\mathrm{f}} \br{\vol}_3^{\mathrm{f}} )$;
    \item The generator of $\br{H}^{6} (L_{3,\mathrm{f}} \times L_{3,\mathrm{b}})$ is $\br{\vol}_3^{\mathrm{f}} \times \br{\vol}_3^{\mathrm{b}} = (\pi^{\ast}_{\mathrm{f}} \br{\vol}_3^{\mathrm{f}} ) \star (\pi^{\ast}_{\mathrm{b}} \br{\vol}_3^{\mathrm{b}} )$.
\end{itemize}

\paragraph{Torsion cocycles.}
Using only torsion cocycles, we form the classes $(\br{t}^{\alpha}_2 \times \br{t}^{\alpha^{\prime}}_2) \star \br{t}_3 \in \br{H}^7 (L_{3,\mathrm{f}} \times L_{3,\mathrm{b}})$. Integrating them against the fundamental cycle we wish to compute
\begin{equation}
    2^{- \delta_{\alpha, \alpha^{\prime} }}\int^{\br{H}}_{L_{3,\mathrm{f}} \times L_{3,\mathrm{b}}} (\br{t}_2^{\alpha} \times  \br{t}_2^{\alpha^{\prime}})\star \br{t}_3 .
\end{equation}
Using Pontryagin duality between $\mathrm{Tor} H^3 (L_{3,\mathrm{f}} \times L_{3,\mathrm{b}}, \Z)$ and $\mathrm{Tor} H^4 (L_{3,\mathrm{f}} \times L_{3,\mathrm{b}}, \Z)$, we choose the generator $\br{t}_3$ such that 
\begin{equation}
    \int^{\br{H}}_{L_{3,\mathrm{f}} \times L_{3,\mathrm{b}}} (\br{t}_2^{\mathrm{f}} \times  \br{t}_2^{\mathrm{b}})\star \br{t}_3 = - \frac{1}{\mathrm{gcd} (n,p)} 
\end{equation}
when $\alpha \ne \alpha^{\prime}$, analogous to \eqref{eq:kappaalphadelta} for $L_5$. In turn, when $\alpha=\alpha^{\prime}=\mathrm{f}$ we have 
\begin{equation}
    \frac{1}{2}\int^{\br{H}}_{L_{3,\mathrm{f}} \times L_{3,\mathrm{b}}} (\br{t}_2^{\mathrm{f}} \times  \br{t}_2^{\mathrm{f}})\star \br{t}_3 = \mathrm{CS}^{(3)}_{\Z_n} \left( \int_{L_{3,\mathrm{b}}} \pi_{\mathrm{b}, \ast}t_3 \right) =0 ,
\end{equation}
and likewise when $\alpha=\alpha^{\prime}=\mathrm{b}$. Here, $ \mathrm{CS}^{(3)}_{\Gamma}$ was defined in \eqref{eq:CSGamma}, and the vanishing is due to the integral of the pushforward of $t_3$ to either lens space.

\paragraph{Mixed torsion and free cocycles.}
Mixing $\br{t}_2^{\cdot}$ and $\br{\vol}^{\cdot}_3$ we have the integrals
\begin{equation}
    \int^{\br{H}}_{L_{3,\mathrm{f}} \times L_{3,\mathrm{b}}} \br{t}_2^{\alpha} \times  \br{t}_2^{\alpha^{\prime}} \times \br{\vol}^{\alpha^{\prime\prime}}_3 = \int^{\br{H}}_{L_{3,\mathrm{f}} \times L_{3,\mathrm{b}}} (\pi^{\ast}_{\alpha} \br{t}_2^{\alpha} ) \star (\pi^{\ast}_{\alpha^{\prime}} \br{t}_2^{\alpha^{\prime}} ) \star (\pi^{\ast}_{\alpha^{\prime\prime}} \br{\vol}_3^{\alpha^{\prime\prime}} ) ,
\end{equation}
which evaluate to:
\begin{equation}
    \int^{\br{H}}_{L_{3,\mathrm{f}} \times L_{3,\mathrm{b}}} \br{t}_2^{\alpha} \times  \br{t}_2^{\alpha^{\prime}} \times \br{\vol}^{\alpha^{\prime\prime}}_3 = \begin{cases} 2 \mathrm{CS}^{(3)}_{\Z_n} & \alpha = \alpha^{\prime} = \mathrm{f}  \text{ and } \alpha^{\prime\prime} = \mathrm{b} \\ 2 \mathrm{CS}^{(3)}_{\Z_p} & \alpha = \alpha^{\prime} = \mathrm{b}  \text{ and } \alpha^{\prime\prime} = \mathrm{f} \\ 0 &  \alpha = \alpha^{\prime\prime}   \text{ or } \alpha^{\prime}=\alpha^{\prime\prime}  .\end{cases}
\end{equation}
$ \mathrm{CS}^{(3)}_{\Gamma}$ was defined and computed in Appendix \ref{sec:CSlens3d} and, choosing generators consistent with the conventions of \cite{Apruzzi:2021nmk}, we have $2\mathrm{CS}^{(3)}_{\Z_n}= \frac{1-n}{n}$. No other classes in $\br{H}^7 (L_{3,\mathrm{f}} \times L_{3,\mathrm{b}})$.

\section{From SymTFT to QFT}\label{app:fromSymTFTtoTQFT}

In this section, we aim to discuss a general approach for constructing the topological sectors of QFTs from the bulk SymTFT theory. This is achieved by introducing a boundary projection operator that maps SymTFT fields to the physical theory using the so-called sandwich construction. Such topological sectors may combine with specific dynamical degrees of freedom within the physical theory, leading either to non-trivial topological constraints or, in other cases, to enriched symmetry structures, such as higher-group symmetries.\par
For further insights into generalized symmetries, we refer readers to the recent comprehensive lectures on this topic in \cite{Schafer-Nameki:2023jdn, Brennan:2023mmt, Luo:2023ive, Bhardwaj:2023kri, Shao:2023gho,Iqbal:2024pee}.

\subsection{SymTFT review}
\label{app:reviewSymTFT}

The bulk SymTFT theory is a topological theory in $(d+1)$ dimensions that encodes and describes the symmetries, in the generalized sense, of a $d$-dimensional physical theory $\mathcal{T}^{(d)}$. In particular, it captures the discrete and continuous $p$-form group symmetries and their 't Hooft anomalies. Schematically, we can write the bulk topological action as
\begin{equation}
    S_{\text{SymTFT}} \ = \ S_{\text{BF}} + S_{\text{twist}}\,.
\end{equation}
Here, 
\begin{itemize}
    \item The action $S_{\text{BF}}$ contains the BF terms of $p$-form discrete global symmetries as well as the BF-like terms for the continuous global ones. 
    \begin{itemize}
        \item[---] The BF terms for the discrete symmetries are constructed from two different fields, which provides a background for an electric/magnetic dual pair of symmetries, which the boundary physical theory can not admit at the same time. 
        \item[---] For the continuous one, the story is a bit more complicated. One proposal \cite{Apruzzi:2024htg} yields BF-like term that couples two field strengths. Alternative proposals \cite{Brennan:2024fgj,Antinucci:2024zjp} suggest that the field strength of a $U(1)$-bundle is coupled to an $\R$-gauge field.
    \end{itemize}
    \item The twist terms $S_{\text{twist}}$ encode the information about (mixed) 't Hooft anomalies after choosing a fixed topological boundary condition. One important consequence of the latter part of the action is providing obstructions to gauging the higher-form global symmetries.
\end{itemize}

In this work, we find BF-like terms consistent with \cite{Apruzzi:2024htg} from M-theory. The recent work \cite{Gagliano:2024off} sheds light on the top-down derivation of the BF terms of \cite{Brennan:2024fgj,Antinucci:2024zjp} in the setting of Type II string theory.

\paragraph{The sandwich construction.}

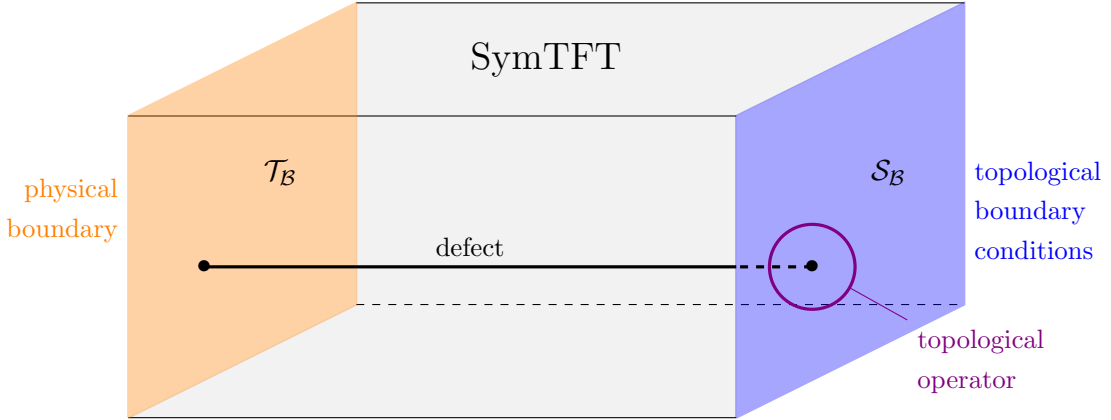
\begin{figure}[th]
\centering
\begin{tikzpicture}
    \draw[gray,thin,fill,opacity=0.1] (0,0) -- (3,1.5) -- (3,5.5) -- (11,5.5) -- (8,4) -- (8,0) -- (0,0);
    \draw[black,thin] (0,0) -- (8,0);
    \draw[black,thin,dashed] (3,1.5) -- (11,1.5);
    \draw[black,thin] (0,4) -- (8,4);
    \draw[black,thin] (3,5.5) -- (11,5.5);
    \draw[orange,thick,fill,opacity=0.35] (0,0) -- (3,1.5) -- (3,5.5) -- (0,4) -- (0,0); 
    \draw[blue,thick,fill,opacity=0.35] (8,0) -- (11,1.5) -- (11,5.5) -- (8,4) -- (8,0);

    \node at (5.5,4.75) {\Large SymTFT};
    \node at (2,3.25) {$\mathcal{T}_{\mathcal{B}}$};
    \node at (10,3.25) {$\mathcal{S}_{\mathcal{B}}$};
    \node[orange,anchor=east,align=right] at (0,2.75) {\small physical\\ \small boundary};
    \node[blue,anchor=west,align=left] at (11,2.75) {\small topological\\ \small boundary\\ \small  conditions};

    \path[black,very thick] (1,2) edge node[anchor=south] {\small defect} (8,2);
    \draw[black,very thick,dashed] (8.05,2) -- (9,2);

    \node[] at (1,2) { $\bullet$ };
     \node[] at (9,2) { $\bullet$ };
    \node[violet,draw,circle,very thick] (to) at (9,2) { \hspace{24pt} };    
    \node[violet,anchor=west,align=left] (tt) at (10.25,.75) {\small topological\\ \small operator};
    \draw[violet,thin] (to) -- (tt);
\end{tikzpicture}
\caption{Visual representation of the SymTFT. The physical theory is located at the orange boundary. The symmetry boundary, is shown in blue. The black line, stretched in the radial direction, corresponds to a defect in the physical theory. A topological symmetry operator acting on the defect is drawn in purple.}
\label{fig:symth}
\end{figure}

The relationship between the $(d+1)$-dimensional SymTFT and the physical $d$-dimensional theory $\mathcal{T}^{(d)}$ can be understood through what is known as the sandwich construction \cite{Freed:2022qnc}. In this framework, we realize the physical theory $\mathcal{T}^{(d)}$ and its associated symmetries by compactifying the SymTFT on an interval. The setup is illustrated in Figure \ref{fig:symth}.\par
This construction involves the following triplet,
\begin{equation}
    (\text{SymTFT},\,\mathcal{T}_{\mathcal{B}},\,\mathcal{S}_{\mathcal{B}})\,.
\end{equation}
Here, we have
\begin{itemize}
    \item SymTFT: A TQFT living on a $(d+1)$ dimensional manifold $Y_{d+1} \cong M_{d}\times [0,1]$. The radial direction of $Y_{d+1}$ can be seen as Euclidean time. Assuming the Euclidean signature on $M_{d}$, then one can naturally construct the Hilbert space of the SymTFT. 
    \item Physical boundary (dynamical boundary) $\mathcal{T}_{\mathcal{B}}$: It is the boundary of the SymTFT where the physical theory $\mathcal{T}^{(d)}$, which is non-topological, is realized.
    \item Symmetry boundary (topological boundary) $\mathcal{S}_{\mathcal{B}}$: It is the boundary with topological boundary conditions for the fields of the SymTFT. This side of the SymTFT depends only on the symmetry structure and it is independent on the dynamics of the physical theory $\mathcal{T}^{(d)}$. 
\end{itemize}

\subsubsection{SymTFT BF terms review}
\label{app:BFtermreview}

\paragraph{$U(1)$ gauge fields for discrete symmetries.}
To ease the comparison with some of the physics literature, we pass to a different normalization of the gauge fields for finite $(p-1)$-form symmetries $\Z_n^{\scriptscriptstyle [p-1]} \subset U(1)^{\scriptscriptstyle [p-1]}$.\par
One natural normalization for a $\Z_n$ gauge field is such that its periods take value in $\Z / n\Z$, and this is the normalization we worked with throughout the main text. Alternatively, one may replace a discrete gauge field $B_{p}$ by a pair of $U(1)$ gauge fields
\begin{equation}\label{eq:2connection}
    \left( \tilde{B}_p, \tilde{B}^{\mathrm{aux}}_{p-1} \right) \qquad \text{ subject to } \qquad n \tilde{B}_p = \tilde{F}^{\mathrm{aux}}_{p} . 
\end{equation}
Here $\tilde{F}^{\mathrm{aux}}_{p} $ is the curvature of $\tilde{B}^{\mathrm{aux}}_{p-1} $. The condition can be imposed using a Lagrange multiplier in the path integral.\par
The data \eqref{eq:2connection} specifies a $p$-connection on a $p$-gerbe over spacetime, see e.g. \cite{Szabo:2012hc} for a review and for the relation with differential cohomology. Integrating the condition over any $p$-cycle $\Sigma_p$, we get the quantization condition 
\begin{equation}\label{eq:checkBquant}
    \int_{\Sigma_p}\tilde{B}_p \in \frac{2 \pi}{n} \Z .
\end{equation}
In practice, this means that we replace everywhere in the SymTFT a $\Z_n$ gauge field $B_p$ with its $U(1)$ avatar $\frac{n}{2\pi}\tilde{B}_p$.

\paragraph{Discrete symmetries.}

Consider a discrete $\Z_n^{\scriptscriptstyle [p-1]}$ symmetry, and its magnetic dual $\Z_n^{\scriptscriptstyle [d-p-1]}$. In the just-described formalism, the backgrounds for these symmetries are given by $U(1)$ gauge fields $\tilde{B}_p$ and $\tilde{A}_{d-p}$, respectively. The SymTFT action contains the BF term
\begin{equation}\label{discreteBFtheory}
  2 \pi \ii S_{\text{BF}}^{\Z_n^{[p-1]}} \  = \ \frac{\ii n}{2\pi}  \int_{Y_{d+1}} \ \tilde{B}_{p}\wedge \tilde{f}_{d-p+1},
\end{equation}
with $\tilde{f}_{d-p+1} = \dd \tilde{A}_{d-p}$ locally. Out of these gauge fields, we can construct the following symmetry operators in $Y_{d+1}$:
\begin{equation}\label{eq:topopReview}
    \mathcal{U}_q(\Sigma_{d-p}) := e^{\ii q \int_{\Sigma_{d-p}}\,\tilde{A}_{d-p}}, \qquad  \mathcal{V}_{\check{q}}(\check{\Sigma}_{p}) := e^{\ii \check{q} \int_{\check{\Sigma}_p}\,\tilde{B}_{p}} .
\end{equation}
From \eqref{eq:checkBquant}, these are gauge-invariant only if $q,\check{q}, \in \Z$, and moreover only their value mod $n$ matters, thus $q,\check{q}, \in \Z/n \Z$.\par
These operators are non-trivially linked, hence their expectation value in the bulk is given as \cite{Birmingham:1991ty}
\begin{equation}
    \expval{\mathcal{U}_q(\Sigma_{d-p})\,\mathcal{V}_{\check{q}}(\Sigma_{p}^{\prime})}\ \sim \ \exp{2\pi \ii \,\frac{q \check{q}}{n}\,\ell_{Y_{d+1}} (\Sigma_{d-p}, \check{\Sigma}_{p})}\, ,
\end{equation}
with $\ell_{Y_{d+1}} (\cdot, \cdot)$ the linking pairing in $Y_{d+1}$.\par
The defects and operators on the physical boundary can be obtained from the bulk symmetry operators by imposing Dirichlet or Neumann boundary conditions on the gauge fields $\tilde{B}_p,\tilde{A}_{d-p}$. In particular, the Dirichlet boundary condition on $\tilde{B}_p$ trivializes the operators $\mathcal{V}_{\check{q}}$, while the Neumann boundary condition on $\tilde{A}_{d-p}$ lets the operators $\mathcal{U}_q$ survive and generate the $\Z_n ^{\scriptscriptstyle [p]}$ symmetry of the QFT on the physical boundary. For visualization, compare with Figure \ref{fig:symth}.

\paragraph{Continuous symmetries.}

The SymTFT of continuous symmetries was analyzed in detail in \cite{Brennan:2024fgj,Antinucci:2024zjp,Apruzzi:2024htg}. Here, we review the basic structure of the SymTFT action for such symmetries.\par
Let us consider a $(p-2)$-form $U(1)^{\scriptscriptstyle [p-2]}$ global symmetry. The main difference to keep in mind is that the magnetic dual is still a $(d-p)$-form symmetry for dimensional reasons, but the Pontryagin dual group is
\begin{equation}
    \widehat{U(1)} = \mathrm{Hom} \left( U(1) , \Z \right) \cong \Z .
\end{equation}
The SymTFT action contains the BF term \cite{Brennan:2024fgj,Antinucci:2024zjp,Apruzzi:2024htg}:
\begin{equation}\label{eq:generalSymTFT-U(1)symmetries}
    2\pi \ii S_{\text{BF}}^{U(1)^{[p-2]}} \  = \ \frac{\ii }{2\pi} \int_{Y_{d+1}} \ F_{p} \wedge h_{d-p+1}\, ,
\end{equation}
Here, $F_{p}$ is the field strength of the $U(1)^{\scriptscriptstyle [p-2]}$ gauge field $A_{p-1}$, i.e. $F_{p}=\dd A_{p-1}$ locally. The interpretation of $h_{d-p+1}$ is more subtle. 
\begin{itemize}
    \item In \cite{Brennan:2024fgj,Antinucci:2024zjp}, it is proposed that $h_{d-p+1}$ is an $\R$-gauge field (connection on a principal $\R$-bundle), with an additional condition imposed on the quantization of periods. 
    \item In \cite{Apruzzi:2024htg}, it is proposed that $h_{d-p+1}$ is a curvature, thus automatically closed and quantized, albeit not immediately related to the gauge field for the dual symmetry $\Z^{\scriptscriptstyle [d-p]}$.
    \item In our work, the top-down approach will naturally lead us along the lines of the proposal in \cite{Apruzzi:2024htg}. The $h_{d-p+1}$ we get from the BF term alone is quantized, however, it is not closed. As explained in Subsection \ref{sec:braneSymTFT}, the SymTFT action we find contains indeed an additional term that, combined with $h_{d-p-1}$, forms a closed and quantized $\widetilde{h}_{d-p+1}$.
\end{itemize}
From the above action, we can construct the following extended operators,
\begin{equation}
    \mathcal{U}^{\text{naive}}_{\varphi}(\Sigma_{d-p}) := e^{\ii \frac{\varphi}{2\pi} \int_{\Sigma_{d-p}}\,h_{d-p}}, \qquad  \mathcal{W}_{\check{q}}(\check{\Sigma}_{p}) := e^{\ii \check{q} \int_{\check{\Sigma}_p}\,A_{p}} 
\end{equation}
which are the direct analogue of \eqref{eq:topopReview}. Here, $\check{q}\in \Z$ is not restricted and is regarded as the symmetry parameter for $\mb{Z}^{[d-p-1]}$, and $\varphi \in [0,2\pi)$.\par
As explained in Subsection \ref{sec:braneSymTFT} and further elaborated upon in Appendix \ref{app:MthwithM5}, the naive definition $\mathcal{U}^{\text{naive}}_{\varphi}$ is not the correct one, and we ought to use 
\begin{equation}
    \mathcal{U}_{\varphi}(\Sigma_{d-p}) := e^{\ii \frac{\varphi}{2\pi} \int_{\Sigma_{d-p}}\,\widetilde{h}_{d-p}}, \qquad  \mathcal{W}_{\check{q}}(\check{\Sigma}_{p}) := e^{\ii \check{q} \int_{\check{\Sigma}_p}\,A_{p}}  .
\end{equation}
These operators have the following expectation value,
\begin{equation}
    \expval{\mathcal{U}_{\varphi}(\Sigma_{d-p}) \mathcal{W}_{\check{q}}(\check{\Sigma}_{p})} = \exp{\ii \, \varphi \check{q}\,\ell_{Y_{d+1}} (\Sigma_{d-p}, \check{\Sigma}_{p})} ,
\end{equation}
from where we learn that the operators $\mathcal{U}_{\varphi}$ acts as the symmetry generator for the continuous $(p-2)$-form symmetry $U(1)^{\scriptscriptstyle [p-2]}$.\par
Note that for $\mathcal{U}_{\varphi}$ and $\mathcal{W}_{\check{q}}$ to be both topological, the flatness conditions $\dd \widetilde{h}_{d-p}=0$ and $\dd A_p=0$ should be imposed. Namely, the SymTFT describes the gauging of $U(1)^{\scriptscriptstyle [p-2]}$ with flat gauge fields, and this is the only way to obtain a dual $\Z^{\scriptscriptstyle [d-p]}$ symmetry afterwards.

\subsection{From SymTFT data to global forms}
\label{app:collapseSymTFT}

Here, we introduce the operator that projects the SymTFT fields and action to the physical theory. In the sandwich construction, this projection corresponds to collapsing the interval direction and pushing the symmetry boundary all the way to the physical boundary $\mathcal{T}_{\mathcal{B}}$.\par
We begin by discussing the gauging of discrete symmetries. Following this, we introduce the projection operator explicitly. Finally, we focus on gauging continuous symmetries, with special attention to top-forms.

\subsubsection{Sandwich of finite symmetries review}
\paragraph{Gauging discrete symmetries.}

On the symmetry boundary $\mathcal{S}_{\mathcal{B}}$, we impose a Dirichlet boundary condition on $\tilde{B}_p$ and a Neumann boundary condition on $\tilde{A}_{d-p}$. From \eqref{eq:topopReview}, this means that $\mathcal{V}_{\check{q}}$ is trivialized at the boundary, while $\mathcal{U}_q$ is not.\par
\begin{table}[th]
    \centering
    \begin{tabular}{|c|c|c|}
    \hline
        Boundary state & $\tilde{B}_p$ & $\tilde{A}_{d-p}$\\
        \hline
         $\ket{D_{b.c.}}$ & Dirichlet & Neumann \\
         $\ket{N_{b.c.}}$ & Neumann & Dirichlet \\
         \hline
    \end{tabular}
    \caption{Boundary conditions for the backgrounds for finite electric/magnetic dual pairs.}
    \label{tab:DirNeubc}
\end{table}

The transformation between the boundary conditions
\begin{equation}
    \left( \tilde{B}_p: \text{ Dirichlet}, \ \tilde{A}_{d-p} :  \text{ Neumann}\right) \quad \longleftrightarrow \quad  \left( \tilde{B}_p: \text{ Neumann}, \ \tilde{A}_{d-p} :  \text{ Dirichlet}\right)
\end{equation}
is achieved by a Fourier transform. In doing so, we use the standard canonical quantization as we view $\tilde{B}_p,\tilde{A}_{d-p}$ as canonically conjugate variables. With the notation as in Table \ref{tab:DirNeubc} for the boundary states, it reads schematically 
\begin{equation}
\begin{aligned}
    \ket{N_{b.c.}} \ &\sim \quad   \sum_{\tilde{B}_{p}\in H^{p}\left( M_{d},\frac{2\pi}{n}(\Z /n \Z) \right)}\,   e^{\frac{\ii n}{2\pi} \int \tilde{A}_{d-p}\,\wedge\, \tilde{B}_{p} } \  \ket{D_{b.c.}}, \\
    \ket{D_{b.c.}} \ &\sim \quad   \sum_{\tilde{A}_{d-p}\in H^{d-p}\left( M_{d},\frac{2\pi}{n}(\Z /n \Z) \right)}\ e^{-\frac{\ii n}{2\pi}\int \tilde{A}_{d-p}\,\wedge\,\tilde{B}_{p}}\ \ket{N_{b.c.}}\,.
\end{aligned}
\end{equation}\par
By the sandwich construction, imposing $\ket{D_{b.c.}}$ on $\mathcal{S}_{\mathcal{B}}$ corresponds to select the global form of the physical theory that realizes $\Z_n^{\scriptscriptstyle [p-1]}$ as a global symmetry. The partition function is:
\begin{equation}
    Z_{\mathcal{T}^{(d)}} [\tilde{B}_p] = \langle \mathcal{T}_{\mathcal{B}} \vert D_{b.c.} \rangle ,
\end{equation}
whose sandwich realization is illustrated in Figure \ref{fig:globalformslab}. Here we write explicitly the dependence on the background for the global symmetry.\par

\begin{figure}[th]
\centering
\begin{tikzpicture}[scale=0.6]
    \draw[gray,thin,fill,opacity=0.1] (0,0) -- (3,1.5) -- (3,5.5) -- (11,5.5) -- (8,4) -- (8,0) -- (0,0);
    \draw[black,thin] (0,0) -- (8,0);
    \draw[black,thin,dashed] (3,1.5) -- (11,1.5);
    \draw[black,thin] (0,4) -- (8,4);
    \draw[black,thin] (3,5.5) -- (11,5.5);
    \draw[orange,thick,fill,opacity=0.35] (0,0) -- (3,1.5) -- (3,5.5) -- (0,4) -- (0,0); 
    \draw[blue,thick,fill,opacity=0.35] (8,0) -- (11,1.5) -- (11,5.5) -- (8,4) -- (8,0);

    \node[orange,align=right] at (1.5,2.75) {$\langle \mathcal{T}_{\mathcal{B}} \vert$};
    \node[blue,align=left] at (9.5,2.75) {$\ket{D_{b.c.}}$};

    \node[anchor=east] at (-1,2.75) {$Z_{\mathcal{T}^{(d)}} = $};
    \node[anchor=west] at (12,2.75) {$= \langle \mathcal{T}_{\mathcal{B}} \vert D_{b.c.} \rangle$};
\end{tikzpicture}
\caption{Partition function of the global form of the physical theory $\mathcal{T}^{(d)}$ that realizes the electric $p$-form symmetry $\Z_n^{\scriptscriptstyle [p-1]}$ as a global symmetry.}
\label{fig:globalformslab}
\end{figure}\par

Assume that there is no obstruction from the twist part to gauge the $(p-1)$-form symmetry. In a sandwich construction, if we exchange the boundary conditions for the fields, we effectively gauge the symmetry $\Z_n^{\scriptscriptstyle [p-1]}$ and realize the magnetic dual $\Z_n^{\scriptscriptstyle [d-p-1]}$ as a global symmetry:
\begin{equation}
\begin{aligned}
    \langle \mathcal{T}_{\mathcal{B}} \vert N_{b.c.} \rangle \ &\sim \quad   \sum_{\tilde{B}_{p}\in H^{p}\left( M_{d},\frac{2\pi}{n}(\Z /n \Z) \right)}\,   e^{\frac{\ii n}{2\pi} \int \tilde{A}_{d-p}\,\wedge\, \tilde{B}_{p} } \  \langle \mathcal{T}_{\mathcal{B}} \vert D_{b.c.} \rangle \\
     &= \quad  \sum_{\tilde{B}_{p}\in H^{p}\left( M_{d},\frac{2\pi}{n}(\Z /n \Z) \right)}\,   e^{\frac{\ii n}{2\pi} \int \tilde{A}_{d-p}\,\wedge\, \tilde{B}_{p} } \ Z_{\mathcal{T}} \left[ \tilde{B}_p \right] \\
     &= \quad  Z_{(\mathcal{T}/\Z_{n}^{[p-1]})} \left[ \tilde{A}_{d-p}\right]\,.
\end{aligned}
\end{equation}

\subsubsection{Boundary projection operator}

We will adopt the (by now standard) notation:
\begin{itemize}
    \item Upper-case Latin letters are background fields for global symmetries. Thus, upper-case letters indicate fields subject to Dirichlet boundary conditions on $\mathcal{S}_{\mathcal{B}}$.
    \item Lower-case Latin letters denote dynamical fields for gauged symmetries. This means that, in the path integral, we sum over backgrounds for the corresponding gauged symmetry. Thus, lower-case letters indicate fields subject to Neumann boundary conditions on $\mathcal{S}_{\mathcal{B}}$.
\end{itemize}

To make explicit the different choices of Dirichlet and Neumann boundary condition, we introduce an operator $\widetilde{\delta}$ that represent the choices that have been made upon the interval compactification. We will refer to this operator as the boundary projection operator. In particular, the operator projects the gauge fields from the SymTFT to the physical boundary, subject to the prescribed boundary conditions at the symmetry boundary. Schematically, $\widetilde{\delta}$ acts on a gauge field $A_{p}$ on $Y_d$ according to\footnote{One may try to formalize this operation by declaring that $\widetilde{\delta}$ is a map from the moduli stack of $p$-connections on a $\mathrm{B}^p U(1)$ $p$-gerbe on $Y_d$ to the analogous stack for $M_d$, whose image is a (possibly reducible) 0-dimensional subscheme. We refrain from a rigorous definition and treat $\widetilde\delta$ as a book-keeping device.} 
\begin{equation}\label{eq:widetildedelta}
    \widetilde{\delta} (A_p) = \begin{cases} \left. A_p \right\rvert_{M_d} & \text{ Dirichlet} , \\ \left. a_p \right\rvert_{M_d} & \text{ Neumann} . \end{cases}
\end{equation}\par
The operator $\widetilde{\delta}$ plays a useful role in mapping the SymTFT gauge fields, or field strengths, to topological terms within the physical QFT, such as Chern classes of the gauge bundle, expressed directly in terms of the dynamical fields of the QFT.\par
Now, let us go back to the example in \eqref{discreteBFtheory}. If the global form of the physical theory realizes the electric $(p-1)$-form global symmetry $\Z_n^{\scriptscriptstyle [p-1]}$, we apply the projection operator $\widetilde{\delta}$ and get
\begin{equation}\label{delta-on-BF}
  \begin{aligned}
      \widetilde{\delta}\left(\frac{\ii n }{2\pi} \int_{Y_{d+1}} \ B_{p}\wedge \dd A_{d-p} \right)  & = \frac{\ii n}{2\pi} \int_{\partial Y_{d+1}} \ \widetilde{\delta}\left(B_{p}\right)\wedge \widetilde{\delta}\left( \dd A_{d-p}\right) \\
      & = \frac{\ii n}{2\pi} \int_{M_{d}} \ B_p\wedge a_{d-p}\,. 
  \end{aligned}
\end{equation}
Using \eqref{eq:2connection}, this action equals
\begin{equation}
    \frac{\ii}{2\pi}\,\int_{M_{d}}\,B_p\,\wedge\, f_{d-p}^{\text{aux}} ,
\end{equation}
with a $U(1)$ gauge field such that, locally, $\dd a_{d-p-1}^{\text{aux}}=f_{d-p}^{\text{aux}}$.

\subsubsection{Sandwich of continuous symmetries}

\paragraph{Gauging continuous symmetries.}
The same construction can be applied to continuous abelian higher-form symmetries. We adopt analogous conventions to those around \eqref{eq:widetildedelta} but, in this case, we work with the field strengths rather than gauge fields, denoted $F_{\bullet},G_{\bullet},H_{\bullet}$.\par
Consider a physical theory that realizes a $U(1)^{\scriptscriptstyle [p-2]}$ $(p-2)$-form symmetry, with background field strength $F_p$. Then,
\begin{equation}
    \begin{aligned}
        \widetilde{\delta}\left(  \frac{\ii}{2\pi}  \int_{Y_{d+1}} \ F_p\wedge H_{d-p+1}  \right) &= \frac{\ii }{2\pi} \int_{M_{d}}\ \widetilde{\delta}\left( F_p\right) \wedge \widetilde{\delta}\left( H_{d-p+1}\right) \\
        &= \frac{\ii}{2\pi} \int_{M_{d}} \  F_p \wedge\, h_{d-p}\,.
    \end{aligned}
\end{equation}
Assuming that there is no obstruction from the twist part of the SymTFT action, we now gauge the $(p-2)$-form symmetry. The projection operator $\widetilde{\delta}$ yields
\begin{equation}\label{delta-on-hG}
    \begin{split}
        \widetilde{\delta}\left(  \frac{i}{2\pi}  \int_{Y_{d+1}} \ F_p\wedge H_{d-p+1} \right) &= \frac{\ii }{2\pi} \int_{M_{d}}\ \widetilde{\delta}\left( F_p\right) \wedge \widetilde{\delta}\left( H_{d-p+1}\right) \\
        &= \frac{\ii}{2\pi} \int_{M_{d}} \  f_{p} \wedge H_{d-p} \, ,
    \end{split}
\end{equation}
where locally $f_p = \dd c_{p-1}$ for a dynamical $U(1)^{\scriptscriptstyle [p-2]}$ gauge field $c_{p-1}$, and the background $H_{d-p}$ comes from ungauging $h_{d-p}$. If we wish to gauge only a subgroup $\Z_n^{\scriptscriptstyle [p-2]} \subset U(1)^{\scriptscriptstyle [p-2]}$, we have from the previous discussion that $n f_p =0$, which can be enforced in practice by modifying $\widetilde{\delta}\left( H_{d-p+1}\right) = nH_{d-p}$.

\paragraph{Top forms in even dimensions.} Top forms in even dimensions can be coupled with a $2\pi$-periodic parameter $\theta$. Such terms are usually included for the top Chern class, but in fact they exist more generally and we will introduce a $\theta$-parameter for every linearly independent cohomology class of top degree.\par
Importantly, such $2\pi$-periodic parameters are not exclusive to topological classes; they may also accompany gauged top-form fields, both discrete and continuous. That is, whenever a theory possesses a non-anomalous $(d-1)$-form symmetry, it can be gauged while introducing a $2\pi$-periodic parameter $\theta_{\mathrm{top}}$ accompanying the top-form gauge field.\par
We incorporate the $\theta$-terms using the operator $\widetilde{\delta}$,
\begin{equation}\label{deltaontopform}
    \begin{aligned}
        \widetilde{\delta}\left(  \frac{\ii}{2\pi}  \int_{Y_{d+1}} \ F_{1}\wedge H_{d}   \right) &= \frac{\ii}{2\pi} \int_{M_{d}}\ \widetilde{\delta}\left( F_{1}\right) \wedge \widetilde{\delta}\left( H_{d}\right) \\
        &= \frac{\ii}{2\pi} \int_{M_{d}} \  \theta  \, \wedge\, \widetilde{\delta}\left( H_{d}\right)\,.
    \end{aligned}
\end{equation}
In this way, $\theta$ is recognized as the background gauge field of a $U(1)^{\scriptscriptstyle [-1]}$ symmetry, with curvature 1-form $F_1$. The realization of $\widetilde{\delta}\left( H_{d}\right)$ on the physical boundary depends on the particular theory being considered. In the example discussed in Subsection \ref{sec:MonG2SymTFT} it was identified with the second Chern class of the gauge bundle.

\section{M-theory action in presence of M-branes}
\label{app:MthwithM5}

The low-energy limit of M-theory is captured by the 11d $\N=1$ supergravity, whose massless field content includes the graviton, the 3-form field $C_{3}$, and the gravitino. In the context of our discussion, we will focus exclusively on $C_{3}$. Thus, the relevant terms of the 11d supergravity action are given by \cite{Cremmer:1978km},
\begin{equation}\label{action-C-field}
 2\kappa^{2}_{\text{11d}}\,   S_{\text{eff}}^{\text{M}} = -\frac{1}{2}\,\int G_{4}\wedge \ast G_{4}\, \, -\frac{1}{6}\,\int C_{3}\wedge G_{4}\wedge G_{4}
\end{equation}
where $G_{4}=\dd C_{3}$ is the field strength associated with the $C_{3}$ gauge field.\par
The Bianchi identity and the equation of motion for the $C_{3}$-field are, respectively, 
\begin{equation}\label{EOM+BI-pre}
    \begin{aligned}
        \dd G_{4} \,&=\,0 \\
        \dd \ast G_{4} \,&=\, \frac{1}{2}\,G_{4}\wedge G_{4}\,.
    \end{aligned}
\end{equation}
To match with the conventions of the main text, one works in units such that 
\begin{equation}
     2\kappa^{2}_{\text{11d}} = (2\pi)^2 
\end{equation}
and defines $G_7 = -\ast G_4 /(2\pi )$. The equations of motion \eqref{EOM+BI-pre} in this formalism read 
\begin{equation}\label{EOM+BI}
    \begin{aligned}
        \dd G_{4} \,&=\,0 \\
        \dd G_{7} \,&=\,-\frac{1}{2(2\pi)}\,G_{4}\wedge G_{4}\,.
    \end{aligned}
\end{equation}
These are equivalently obtained replacing \eqref{action-C-field} with 
\begin{equation}\label{action-C3C6}
    S_{\text{top}}^{\text{M}} = \int_{M_{11}} \left[ \frac{1}{2\pi} G_{4}\wedge G_{7}  -\frac{1}{6 (2\pi)^2}  C_{3}\wedge G_{4}\wedge G_{4} \right].
\end{equation}
The curvatures are quantized to have periods in $2\pi \Z$ and, at this stage, we have recovered the formulation in Subsection \ref{sec:MtheoryupliftH12}, with
\begin{equation}
     S_{\text{top}}^{\text{M}} = 2\pi \left( S_{\text{kin}}^{\text{M}} + S_{\text{CS}}^{\text{M}} \right) .
\end{equation}\par
Locally, a valid solution to \eqref{EOM+BI} is
\begin{equation}\label{local-sol(1)}
    \begin{aligned}
        G_{4} \,&=\,\dd C_{3}\,,  \\
        G_{7} \,&=\, \dd C_{6} - \frac{1}{4\pi}\,C_{3}\wedge G_{4}\,.
    \end{aligned}
\end{equation}
The presence of the $C_{3}$-field, along with its magnetic dual $C_{6}$, allows for solitonic BPS brane solutions in the form of M2-branes and M5-branes. From \eqref{local-sol(1)}, two key observations arise:
\begin{itemize}
    \item The 11d supergravity action is written only in terms of electric $C_{3}$ gauge field degrees of freedom.
    \item The M5-brane is a dyonic object which carries both a magnetic and an electric charges. It couples minimally to the $C_{6}$ gauge field and couples non-minimally to $C_{3}$ \cite{Bandos:1997gd,Sorokin:1998kf}.
\end{itemize}\par

\paragraph{Derivation via regularization.} For completeness, let us briefly mention an alternative approach to derive \eqref{action-C3C6} from \eqref{action-C-field}. A version of this adapted to Type II superstring theory was used in the recent \cite{Yu:2024jtk}, and we highlight similarities and difference.\par
The Hodge $\ast$ in 11d involves the metric of a non-compact manifold $M_{11}$. One may regularize it with a volume cutoff $\Lambda$, and rescale $\ast G_4 \mapsto \Lambda \ast G_4$ in \eqref{action-C-field}. We then introduce $G_7$ as an auxiliary field and make a Hubbard--Stratonovich transformation to write 
\begin{equation}\label{eq:G4G7alternative}
    S_{\text{top}}^{\text{M}} = \frac{1}{2\pi} \int_{M_{11}} \left[ G_{4}\wedge G_{7}  -\frac{1}{6 (2\pi)^2}  C_{3}\wedge G_{4}\wedge G_{4} - \frac{1}{2 \Lambda} G_7 \wedge \ast G_7 \right].
\end{equation}
Integrating out $G_7$ we recover \eqref{action-C-field}. Sending the regulator $\Lambda\to \infty$ at this stage, we formally recover \eqref{action-C3C6}.\par
Despite arriving at the same action, the approach we follow and this alternative derivation are different at quantum level. The main reason is that, in this latter approach, $G_7$ is treated as a field, not a field strength. 
\begin{itemize}
    \item By construction, in this formulation one path integrates \eqref{eq:G4G7alternative} over $G_7$ and $C_3$, which would be an unusual feature of M-theory, as opposed to \eqref{action-C3C6} which is path integrated over $C_3$ and $C_6$.
    \item In our formulation, we recover the Bianchi identity $\dd G_4=0$ from the variation of $C_6$, whereas in this alternative approach, the field is $G_7$, and its variation would impose $G_4=0$ in absence of sources.
    \item A priori, the periods of $G_7$ do not satisfy a quantization condition in this alternative formulation.
\end{itemize}

\paragraph{Coupling of M5-brane in M-theory.}
With the inclusion of an M5-brane, the equations governing the $G$-fluxes in \eqref{EOM+BI} are altered to account for the source terms localized on the brane. If the M5-brane is taken to be supported along $\Sigma_{6}^{\text{M5}}$, it sources a current $j_{6}^{\text{M5}}$ with 
\begin{equation}
    \ast j_{6}^{\text{M5}} \in  2 \pi \Z ~\mathsf{PD}_{11} (\Sigma_{6}^{\text{M5}}) ,
\end{equation}
where $\mathsf{PD}_{11}$ is the Poincar\'e duality map in 11d. That is to say, $\ast j_{6}^{\text{M5}}$ has a $\delta$-function support transverse to the M5-brane worldvolume. The equations \eqref{EOM+BI} are modified into 
\begin{equation}\label{EOM+BI+M5}
 \begin{aligned}
        \dd G_{4} \, &= 2\pi \mathsf{PD}_{11} (\Sigma_{6}^{\text{M5}}) , \\
        \dd G_{7} \, &= -\frac{1}{2(2\pi)} G_{4}\wedge G_{4} \, + \frac{1}{2} H_{3}\wedge\mathsf{PD}_{11} (\Sigma_{6}^{\text{M5}}) \,.
 \end{aligned}
\end{equation}
As already discussed in Subsection \ref{sec:BraneSymTFTFree}, the 3-form $H_{3}$ is constrained to the worldvolume of the M5-brane, and it is related to the pullback of the $G_{4}$ as
\begin{equation}
    \dd H_{3}\,=\, \iota_{\text{M5}}^{\ast}G_{4}\, 
\end{equation}
where $\iota_{\text{M5}} : \Sigma_{6}^{\text{M5}} \hookrightarrow M_{11}$ is the embedding of the M5-brane worldvolume. Locally, $H_3$ takes the form
\begin{equation}\label{eq:def-of-H3}
    H_{3} = \dd B_{2} + \iota_{\text{M5}}^{\ast}C_{3}\,.
\end{equation}
Here, $B_{2}$ is the gauge field that couples to the strings on the worldvolume theory on the M5-branes, arising from M2-brane intersections with the M5-brane.\par
In the presence of an M5-brane, the local solution given in \eqref{local-sol(1)} becomes invalid, necessitating a modification to account for the source term. This is a general feature, first observed by Dirac in \cite{Dirac:1948um}, and subsequently worked out in detail in \cite{Bandos:1997gd, Sorokin:1998kf} for M-theory. Building on these references, the local solutions for $G_4$ and $G_7$ must be adjusted as follows,
\begin{equation}\label{modified-G-fluxes}
    \begin{split}
        G_{4}\  &\mapsto \ \widehat{G}_{4} = G_4 + \ast K_{7}\,,
        \\
        G_{7}\ &\mapsto \ \widehat{G}_{7} = G_7 +\frac{1}{2(2\pi)} H_{3} \wedge \ast K_{7}\,.
    \end{split}
\end{equation}
Here, $K_{7}$ is defined to satisfy
\begin{equation}
        \dd \ast K_{7}  = \ast j_{6}^{\text{M5}} \,= 2\pi \mathsf{PD}_{11} (\Sigma_{6}^{\text{M5}}) .
\end{equation}
After integration by parts and Poincar\'e duality, it is given by 
\begin{equation}\label{eq:K7Sigma7}
       \frac{1}{2\pi} K_{7} = \ast \mathsf{PD}_{11} (\Sigma_{7}) , 
\end{equation}
for any $\Sigma_{7}$ with $\partial \Sigma_{7}=\Sigma_{6}^{\text{M5}}$. This is a higher-dimensional analogue of Dirac’s string \cite{Dirac:1948um} which, in our settings, constitutes a fluxbrane as explained in Subsection \ref{sec:BraneSymTFTFree}.\par

With the new solutions $\widehat{G}_4, \widehat{G}_7$, the kinetic term becomes 
\begin{equation}
        \int_{M_{11}}\, G_{4}\wedge G_{7} \ \, \mapsto \,  \ \int_{M_{11}}\, \widehat{G}_{4}\wedge \widehat{G}_{7} \,.
\end{equation}
Note that, for the Chern--Simons term, we continue to use the original (un-hatted) variables as detailed in \cite{Bandos:1997gd, Sorokin:1998kf}. An immediate computation using \eqref{modified-G-fluxes} with \eqref{eq:K7Sigma7} gives
\begin{equation}
     \int_{M_{11}}\, \widehat{G}_{4}\wedge \widehat{G}_{7} = \int_{M_{11}} G_4 \wedge G_7 + \int_{\Sigma_{7}} \left[ \phi^{\ast} G_7 + \frac{1}{4\pi} \iota_{7,\ast} (H_3) \wedge \phi^{\ast} G_4 \right]
\end{equation}
where $\iota_{7,\ast} (H_3)$ is the push-forward of $H_3$ via $\iota_7 : \Sigma_6^{\text{M5}} \hookrightarrow \Sigma_{7}$, and $\phi^{\ast} G_{\bullet}$ pulls back $G_{\bullet}$ to $\Sigma_7$.\par
Thus, we observe from this perspective that the Hopf--Wess--Zumino topological action is recovered on the $P_{7}$-fluxbrane, as also shown in \eqref{eq:HWZ-action}.\par
Before concluding, let us mention three remarks on the derivation.
\begin{itemize}
    \item[---] From \eqref{eq:def-of-H3}, we observe that $H_{3}$ on the worldvolume of the fluxbranes, used to construct the symmetry topological operator, is essentially the pullback of the $C_{3}$ field from M-theory without activating the local $B_{2}$ gauge field on the worldvolume. In this configuration, the fluxbranes facilitate the construction of invertible topological operators. If a non-trivial $B_{2}$ gauge field is turned on, the symmetry topological operator may instead generate a non-invertible symmetry, as path integration over $B_{2}$ becomes necessary.
    \item[---] Note that this framework allows for the consideration of multiple M5-brane insertions, supported on $\bigcup_{i}\Sigma_6^{\text{M5},i}$, sourcing a current of the form 
        \begin{equation}
            \ast j_{6}^{\text{M5}} =   2 \pi \sum_{i} n_i \mathsf{PD}_{11} (\Sigma_6^{\text{M5},i}) .
        \end{equation}
    \item[---] Finally, the refinement of this derivation to differential cohomology is straightforward. It is given by:
        \begin{equation}
            2\pi \int_{M_{11}} \br{\widehat{G}}_4 \star \br{\dd \widehat{G}}_7 = 2\pi \int_{M_{11}} \br{G}_4 \star \br{\dd G}_7 + \int_{\Sigma_{7}} \left[ \phi^{\ast} G_7 + \frac{1}{4\pi} \iota_{7,\ast} (H_3) \wedge \phi^{\ast} G_4 \right] .
        \end{equation}
\end{itemize}

\bibliographystyle{JHEP}
\bibliography{F-ref.bib}

\providecommand{\href}[2]{#2}\begingroup\raggedright\begin{thebibliography}{100}

\bibitem{tHooft:1979rat}
G.~'t~Hooft, \emph{{Naturalness, chiral symmetry, and spontaneous chiral symmetry breaking}}, \href{http://dx.doi.org/10.1007/978-1-4684-7571-5_9}{\emph{NATO Sci. Ser. B} {\bf 59} (1980) 135--157}.

\bibitem{Callan:1984sa}
C.~G. Callan, Jr. and J.~A. Harvey, \emph{{Anomalies and Fermion Zero Modes on Strings and Domain Walls}}, \href{http://dx.doi.org/10.1016/0550-3213(85)90489-4}{\emph{Nucl. Phys. B} {\bf 250} (1985) 427--436}.

\bibitem{Aharony:2013hda}
O.~Aharony, N.~Seiberg and Y.~Tachikawa, \emph{{Reading between the lines of four-dimensional gauge theories}}, \href{http://dx.doi.org/10.1007/JHEP08(2013)115}{\emph{JHEP} {\bf 08} (2013) 115}, [\href{http://arxiv.org/abs/1305.0318}{{\tt 1305.0318}}].

\bibitem{Gaiotto:2014kfa}
D.~Gaiotto, A.~Kapustin, N.~Seiberg and B.~Willett, \emph{{Generalized Global Symmetries}}, \href{http://dx.doi.org/10.1007/JHEP02(2015)172}{\emph{JHEP} {\bf 02} (2015) 172}, [\href{http://arxiv.org/abs/1412.5148}{{\tt 1412.5148}}].

\bibitem{Wilson:1974sk}
K.~G. Wilson, \emph{{Confinement of Quarks}}, \href{http://dx.doi.org/10.1103/PhysRevD.10.2445}{\emph{Phys. Rev. D} {\bf 10} (1974) 2445--2459}.

\bibitem{tHooft:1977nqb}
G.~'t~Hooft, \emph{{On the Phase Transition Towards Permanent Quark Confinement}}, \href{http://dx.doi.org/10.1016/0550-3213(78)90153-0}{\emph{Nucl. Phys. B} {\bf 138} (1978) 1--25}.

\bibitem{Katz:1996fh}
S.~H. Katz, A.~Klemm and C.~Vafa, \emph{{Geometric engineering of quantum field theories}}, \href{http://dx.doi.org/10.1016/S0550-3213(97)00282-4}{\emph{Nucl. Phys. B} {\bf 497} (1997) 173--195}, [\href{http://arxiv.org/abs/hep-th/9609239}{{\tt hep-th/9609239}}].

\bibitem{Atiyah:2000zz}
M.~Atiyah, J.~M. Maldacena and C.~Vafa, \emph{{An M theory flop as a large N duality}}, \href{http://dx.doi.org/10.1063/1.1376159}{\emph{J. Math. Phys.} {\bf 42} (2001) 3209--3220}, [\href{http://arxiv.org/abs/hep-th/0011256}{{\tt hep-th/0011256}}].

\bibitem{Acharya:2000gb}
B.~S. Acharya, \emph{{On Realizing N=1 superYang-Mills in M theory}},  \href{http://arxiv.org/abs/hep-th/0011089}{{\tt hep-th/0011089}}.

\bibitem{Acharya:2001dz}
B.~S. Acharya and C.~Vafa, \emph{{On domain walls of N=1 supersymmetric Yang-Mills in four-dimensions}},  \href{http://arxiv.org/abs/hep-th/0103011}{{\tt hep-th/0103011}}.

\bibitem{Atiyah:2001qf}
M.~Atiyah and E.~Witten, \emph{{M theory dynamics on a manifold of G(2) holonomy}}, \href{http://dx.doi.org/10.4310/ATMP.2002.v6.n1.a1}{\emph{Adv. Theor. Math. Phys.} {\bf 6} (2003) 1--106}, [\href{http://arxiv.org/abs/hep-th/0107177}{{\tt hep-th/0107177}}].

\bibitem{Acharya:2001hq}
B.~S. Acharya, \emph{{Confining strings from G(2) holonomy space-times}},  \href{http://arxiv.org/abs/hep-th/0101206}{{\tt hep-th/0101206}}.

\bibitem{Beasley:2002db}
C.~Beasley and E.~Witten, \emph{{A Note on fluxes and superpotentials in M theory compactifications on manifolds of G(2) holonomy}}, \href{http://dx.doi.org/10.1088/1126-6708/2002/07/046}{\emph{JHEP} {\bf 07} (2002) 046}, [\href{http://arxiv.org/abs/hep-th/0203061}{{\tt hep-th/0203061}}].

\bibitem{Berglund:2002hw}
P.~Berglund and A.~Brandhuber, \emph{{Matter from G(2) manifolds}}, \href{http://dx.doi.org/10.1016/S0550-3213(02)00612-0}{\emph{Nucl. Phys. B} {\bf 641} (2002) 351--375}, [\href{http://arxiv.org/abs/hep-th/0205184}{{\tt hep-th/0205184}}].

\bibitem{Acharya:2004qe}
B.~S. Acharya and S.~Gukov, \emph{{M theory and singularities of exceptional holonomy manifolds}}, \href{http://dx.doi.org/10.1016/j.physrep.2003.10.017}{\emph{Phys. Rept.} {\bf 392} (2004) 121--189}, [\href{http://arxiv.org/abs/hep-th/0409191}{{\tt hep-th/0409191}}].

\bibitem{Anderson:2006mv}
L.~B. Anderson, A.~B. Barrett, A.~Lukas and M.~Yamaguchi, \emph{{Four-dimensional Effective M-theory on a Singular G(2) Manifold}}, \href{http://dx.doi.org/10.1103/PhysRevD.74.086008}{\emph{Phys. Rev. D} {\bf 74} (2006) 086008}, [\href{http://arxiv.org/abs/hep-th/0606285}{{\tt hep-th/0606285}}].

\bibitem{Halverson:2014tya}
J.~Halverson and D.~R. Morrison, \emph{{The landscape of M-theory compactifications on seven-manifolds with G$_{2}$ holonomy}}, \href{http://dx.doi.org/10.1007/JHEP04(2015)047}{\emph{JHEP} {\bf 04} (2015) 047}, [\href{http://arxiv.org/abs/1412.4123}{{\tt 1412.4123}}].

\bibitem{Halverson:2015vta}
J.~Halverson and D.~R. Morrison, \emph{{On gauge enhancement and singular limits in G$_{2}$ compactifications of M-theory}}, \href{http://dx.doi.org/10.1007/JHEP04(2016)100}{\emph{JHEP} {\bf 04} (2016) 100}, [\href{http://arxiv.org/abs/1507.05965}{{\tt 1507.05965}}].

\bibitem{Braun:2018fdp}
A.~P. Braun, M.~Del~Zotto, J.~Halverson, M.~Larfors, D.~R. Morrison and S.~Sch\"afer-Nameki, \emph{{Infinitely many M2-instanton corrections to M-theory on G$_{2}$-manifolds}}, \href{http://dx.doi.org/10.1007/JHEP09(2018)077}{\emph{JHEP} {\bf 09} (2018) 077}, [\href{http://arxiv.org/abs/1803.02343}{{\tt 1803.02343}}].

\bibitem{Kennon:2018eqg}
A.~Kennon, \emph{{G$_{2}$-Manifolds and M-Theory Compactifications}},  \href{http://arxiv.org/abs/1810.12659}{{\tt 1810.12659}}.

\bibitem{Braun:2018vhk}
A.~P. Braun, S.~Cizel, M.~H\"ubner and S.~Sch\"afer-Nameki, \emph{{Higgs bundles for M-theory on $G_{2}$-manifolds}}, \href{http://dx.doi.org/10.1007/JHEP03(2019)199}{\emph{JHEP} {\bf 03} (2019) 199}, [\href{http://arxiv.org/abs/1812.06072}{{\tt 1812.06072}}].

\bibitem{Acharya:2020vmg}
B.~S. Acharya, L.~Foscolo, M.~Najjar and E.~E. Svanes, \emph{{New G$_{2}$-conifolds in M-theory and their field theory interpretation}}, \href{http://dx.doi.org/10.1007/JHEP05(2021)250}{\emph{JHEP} {\bf 05} (2021) 250}, [\href{http://arxiv.org/abs/2011.06998}{{\tt 2011.06998}}].

\bibitem{DelZotto:2021ydd}
M.~Del~Zotto, J.~Oh and Y.~Zhou, \emph{{Evidence for an algebra of G$_{2}$ instantons}}, \href{http://dx.doi.org/10.1007/JHEP08(2022)214}{\emph{JHEP} {\bf 08} (2022) 214}, [\href{http://arxiv.org/abs/2109.01110}{{\tt 2109.01110}}].

\bibitem{Braun:2023fqa}
A.~P. Braun, E.~Sabag, M.~Sacchi and S.~Schafer-Nameki, \emph{{$G_2$-Manifolds from 4d N=1 Theories, Part I: Domain Walls}}, \href{http://dx.doi.org/10.21468/SciPostPhys.17.4.102}{\emph{SciPost Phys.} {\bf 17} (2024) 102}, [\href{http://arxiv.org/abs/2304.01193}{{\tt 2304.01193}}].

\bibitem{Intriligator:1997pq}
K.~A. Intriligator, D.~R. Morrison and N.~Seiberg, \emph{{Five-dimensional supersymmetric gauge theories and degenerations of Calabi-Yau spaces}}, \href{http://dx.doi.org/10.1016/S0550-3213(97)00279-4}{\emph{Nucl. Phys.} {\bf B497} (1997) 56--100}, [\href{http://arxiv.org/abs/hep-th/9702198}{{\tt hep-th/9702198}}].

\bibitem{Esole:2014bka}
M.~Esole, S.-H. Shao and S.-T. Yau, \emph{{Singularities and Gauge Theory Phases}}, \href{http://dx.doi.org/10.4310/ATMP.2015.v19.n6.a2}{\emph{Adv. Theor. Math. Phys.} {\bf 19} (2015) 1183--1247}, [\href{http://arxiv.org/abs/1402.6331}{{\tt 1402.6331}}].

\bibitem{Esole:2014hya}
M.~Esole, S.-H. Shao and S.-T. Yau, \emph{{Singularities and Gauge Theory Phases II}}, \href{http://dx.doi.org/10.4310/ATMP.2016.v20.n4.a2}{\emph{Adv. Theor. Math. Phys.} {\bf 20} (2016) 683--749}, [\href{http://arxiv.org/abs/1407.1867}{{\tt 1407.1867}}].

\bibitem{DelZotto:2017pti}
M.~Del~Zotto, J.~J. Heckman and D.~R. Morrison, \emph{{6D SCFTs and Phases of 5D Theories}}, \href{http://dx.doi.org/10.1007/JHEP09(2017)147}{\emph{JHEP} {\bf 09} (2017) 147}, [\href{http://arxiv.org/abs/1703.02981}{{\tt 1703.02981}}].

\bibitem{Xie:2017pfl}
D.~Xie and S.-T. Yau, \emph{{Three dimensional canonical singularity and five dimensional $ \mathcal{N} $ = 1 SCFT}}, \href{http://dx.doi.org/10.1007/JHEP06(2017)134}{\emph{JHEP} {\bf 06} (2017) 134}, [\href{http://arxiv.org/abs/1704.00799}{{\tt 1704.00799}}].

\bibitem{Esole:2017hlw}
M.~Esole, M.~J. Kang and S.-T. Yau, \emph{{Mordell-Weil Torsion, Anomalies, and Phase Transitions}},  \href{http://arxiv.org/abs/1712.02337}{{\tt 1712.02337}}.

\bibitem{Closset:2018bjz}
C.~Closset, M.~Del~Zotto and V.~Saxena, \emph{{Five-dimensional SCFTs and gauge theory phases: an M-theory/type IIA perspective}}, \href{http://dx.doi.org/10.21468/SciPostPhys.6.5.052}{\emph{SciPost Phys.} {\bf 6} (2019) 052}, [\href{http://arxiv.org/abs/1812.10451}{{\tt 1812.10451}}].

\bibitem{Jefferson:2018irk}
P.~Jefferson, S.~Katz, H.-C. Kim and C.~Vafa, \emph{{On Geometric Classification of 5d SCFTs}}, \href{http://dx.doi.org/10.1007/JHEP04(2018)103}{\emph{JHEP} {\bf 04} (2018) 103}, [\href{http://arxiv.org/abs/1801.04036}{{\tt 1801.04036}}].

\bibitem{Apruzzi:2019opn}
F.~Apruzzi, C.~Lawrie, L.~Lin, S.~Schafer-Nameki and Y.-N. Wang, \emph{{Fibers add Flavor, Part I: Classification of 5d SCFTs, Flavor Symmetries and BPS States}}, \href{http://dx.doi.org/10.1007/JHEP11(2019)068}{\emph{JHEP} {\bf 11} (2019) 068}, [\href{http://arxiv.org/abs/1907.05404}{{\tt 1907.05404}}].

\bibitem{Apruzzi:2019enx}
F.~Apruzzi, C.~Lawrie, L.~Lin, S.~Sch\"afer-Nameki and Y.-N. Wang, \emph{{Fibers add Flavor, Part II: 5d SCFTs, Gauge Theories, and Dualities}}, \href{http://dx.doi.org/10.1007/JHEP03(2020)052}{\emph{JHEP} {\bf 03} (2020) 052}, [\href{http://arxiv.org/abs/1909.09128}{{\tt 1909.09128}}].

\bibitem{Saxena:2020ltf}
V.~Saxena, \emph{{Rank-two 5d SCFTs from M-theory at isolated toric singularities: a systematic study}}, \href{http://dx.doi.org/10.1007/JHEP04(2020)198}{\emph{JHEP} {\bf 04} (2020) 198}, [\href{http://arxiv.org/abs/1911.09574}{{\tt 1911.09574}}].

\bibitem{Apruzzi:2019kgb}
F.~Apruzzi, S.~Schafer-Nameki and Y.-N. Wang, \emph{{5d SCFTs from Decoupling and Gluing}}, \href{http://dx.doi.org/10.1007/JHEP08(2020)153}{\emph{JHEP} {\bf 08} (2020) 153}, [\href{http://arxiv.org/abs/1912.04264}{{\tt 1912.04264}}].

\bibitem{Collinucci:2020jqd}
A.~Collinucci and R.~Valandro, \emph{{The role of U(1)\textquoteright{}s in 5d theories, Higgs branches, and geometry}}, \href{http://dx.doi.org/10.1007/JHEP10(2020)178}{\emph{JHEP} {\bf 10} (2020) 178}, [\href{http://arxiv.org/abs/2006.15464}{{\tt 2006.15464}}].

\bibitem{Closset:2020scj}
C.~Closset, S.~Schafer-Nameki and Y.-N. Wang, \emph{{Coulomb and Higgs Branches from Canonical Singularities: Part 0}}, \href{http://dx.doi.org/10.1007/JHEP02(2021)003}{\emph{JHEP} {\bf 02} (2021) 003}, [\href{http://arxiv.org/abs/2007.15600}{{\tt 2007.15600}}].

\bibitem{Eckhard_2020}
J.~Eckhard, S.~Schäfer-Nameki and Y.-N. Wang, \emph{Trifectas for tn in 5d}, \href{http://dx.doi.org/10.1007/jhep07(2020)199}{\emph{Journal of High Energy Physics} {\bf 2020} (Jul, 2020) }.

\bibitem{Acharya:2021jsp}
B.~Acharya, N.~Lambert, M.~Najjar, E.~E. Svanes and J.~Tian, \emph{{Gauging discrete symmetries of T$_{N}$-theories in five dimensions}}, \href{http://dx.doi.org/10.1007/JHEP04(2022)114}{\emph{JHEP} {\bf 04} (2022) 114}, [\href{http://arxiv.org/abs/2110.14441}{{\tt 2110.14441}}].

\bibitem{Tian:2021cif}
J.~Tian and Y.-N. Wang, \emph{{5D and 6D SCFTs from $\mathbb{C}^3$ orbifolds}}, \href{http://dx.doi.org/10.21468/SciPostPhys.12.4.127}{\emph{SciPost Phys.} {\bf 12} (2022) 127}, [\href{http://arxiv.org/abs/2110.15129}{{\tt 2110.15129}}].

\bibitem{Closset:2021lwy}
C.~Closset, S.~Sch\"afer-Nameki and Y.-N. Wang, \emph{{Coulomb and Higgs branches from canonical singularities. Part I. Hypersurfaces with smooth Calabi-Yau resolutions}}, \href{http://dx.doi.org/10.1007/JHEP04(2022)061}{\emph{JHEP} {\bf 04} (2022) 061}, [\href{http://arxiv.org/abs/2111.13564}{{\tt 2111.13564}}].

\bibitem{DeMarco:2022dgh}
M.~De~Marco, A.~Sangiovanni and R.~Valandro, \emph{{5d Higgs branches from M-theory on quasi-homogeneous cDV threefold singularities}}, \href{http://dx.doi.org/10.1007/JHEP10(2022)124}{\emph{JHEP} {\bf 10} (2022) 124}, [\href{http://arxiv.org/abs/2205.01125}{{\tt 2205.01125}}].

\bibitem{Mu:2023uws}
J.~Mu, Y.-N. Wang and H.~N. Zhang, \emph{{5d SCFTs from isolated complete intersection singularities}}, \href{http://dx.doi.org/10.1007/JHEP02(2024)155}{\emph{JHEP} {\bf 02} (2024) 155}, [\href{http://arxiv.org/abs/2311.05441}{{\tt 2311.05441}}].

\bibitem{DeMarco:2023irn}
M.~De~Marco, M.~Del~Zotto, M.~Graffeo and A.~Sangiovanni, \emph{{Conformal matter}}, \href{http://dx.doi.org/10.1007/JHEP05(2024)306}{\emph{JHEP} {\bf 05} (2024) 306}, [\href{http://arxiv.org/abs/2311.04984}{{\tt 2311.04984}}].

\bibitem{Alexeev:2024bko}
V.~Alexeev, H.~Arg\"uz and P.~Bousseau, \emph{{Non-toric brane webs, Calabi-Yau 3-folds, and 5d SCFTs}},  \href{http://arxiv.org/abs/2410.04714}{{\tt 2410.04714}}.

\bibitem{Witten:1998wy}
E.~Witten, \emph{{AdS / CFT correspondence and topological field theory}}, \href{http://dx.doi.org/10.1088/1126-6708/1998/12/012}{\emph{JHEP} {\bf 12} (1998) 012}, [\href{http://arxiv.org/abs/hep-th/9812012}{{\tt hep-th/9812012}}].

\bibitem{GarciaEtxebarria:2019caf}
I.~n. Garc\'\i{}a~Etxebarria, B.~Heidenreich and D.~Regalado, \emph{{IIB flux non-commutativity and the global structure of field theories}}, \href{http://dx.doi.org/10.1007/JHEP10(2019)169}{\emph{JHEP} {\bf 10} (2019) 169}, [\href{http://arxiv.org/abs/1908.08027}{{\tt 1908.08027}}].

\bibitem{Apruzzi:2021nmk}
F.~Apruzzi, F.~Bonetti, I.~n. Garc\'\i{}a~Etxebarria, S.~S. Hosseini and S.~Schafer-Nameki, \emph{{Symmetry TFTs from String Theory}}, \href{http://dx.doi.org/10.1007/s00220-023-04737-2}{\emph{Commun. Math. Phys.} {\bf 402} (2023) 895--949}, [\href{http://arxiv.org/abs/2112.02092}{{\tt 2112.02092}}].

\bibitem{Hubner:2022kxr}
M.~Hubner, D.~R. Morrison, S.~Schafer-Nameki and Y.-N. Wang, \emph{{Generalized Symmetries in F-theory and the Topology of Elliptic Fibrations}}, \href{http://dx.doi.org/10.21468/SciPostPhys.13.2.030}{\emph{SciPost Phys.} {\bf 13} (2022) 030}, [\href{http://arxiv.org/abs/2203.10022}{{\tt 2203.10022}}].

\bibitem{DelZotto:2022joo}
M.~Del~Zotto, I.~n. Garc\'\i{}a~Etxebarria and S.~Schafer-Nameki, \emph{{2-Group Symmetries and M-Theory}}, \href{http://dx.doi.org/10.21468/SciPostPhys.13.5.105}{\emph{SciPost Phys.} {\bf 13} (2022) 105}, [\href{http://arxiv.org/abs/2203.10097}{{\tt 2203.10097}}].

\bibitem{Apruzzi:2022rei}
F.~Apruzzi, I.~Bah, F.~Bonetti and S.~Schafer-Nameki, \emph{{Noninvertible Symmetries from Holography and Branes}}, \href{http://dx.doi.org/10.1103/PhysRevLett.130.121601}{\emph{Phys. Rev. Lett.} {\bf 130} (2023) 121601}, [\href{http://arxiv.org/abs/2208.07373}{{\tt 2208.07373}}].

\bibitem{vanBeest:2022fss}
M.~van Beest, D.~S.~W. Gould, S.~Schafer-Nameki and Y.-N. Wang, \emph{{Symmetry TFTs for 3d QFTs from M-theory}}, \href{http://dx.doi.org/10.1007/JHEP02(2023)226}{\emph{JHEP} {\bf 02} (2023) 226}, [\href{http://arxiv.org/abs/2210.03703}{{\tt 2210.03703}}].

\bibitem{Heckman:2022xgu}
J.~J. Heckman, M.~Hubner, E.~Torres, X.~Yu and H.~Y. Zhang, \emph{{Top down approach to topological duality defects}}, \href{http://dx.doi.org/10.1103/PhysRevD.108.046015}{\emph{Phys. Rev. D} {\bf 108} (2023) 046015}, [\href{http://arxiv.org/abs/2212.09743}{{\tt 2212.09743}}].

\bibitem{Apruzzi:2023uma}
F.~Apruzzi, F.~Bonetti, D.~S.~W. Gould and S.~Schafer-Nameki, \emph{{Aspects of categorical symmetries from branes: SymTFTs and generalized charges}}, \href{http://dx.doi.org/10.21468/SciPostPhys.17.1.025}{\emph{SciPost Phys.} {\bf 17} (2024) 025}, [\href{http://arxiv.org/abs/2306.16405}{{\tt 2306.16405}}].

\bibitem{Baume:2023kkf}
F.~Baume, J.~J. Heckman, M.~H\"ubner, E.~Torres, A.~P. Turner and X.~Yu, \emph{{SymTrees and Multi-Sector QFTs}}, \href{http://dx.doi.org/10.1103/PhysRevD.109.106013}{\emph{Phys. Rev. D} {\bf 109} (2024) 106013}, [\href{http://arxiv.org/abs/2310.12980}{{\tt 2310.12980}}].

\bibitem{DelZotto:2024tae}
M.~Del~Zotto, S.~N. Meynet and R.~Moscrop, \emph{{Remarks on geometric engineering, symmetry TFTs and anomalies}}, \href{http://dx.doi.org/10.1007/JHEP07(2024)220}{\emph{JHEP} {\bf 07} (2024) 220}, [\href{http://arxiv.org/abs/2402.18646}{{\tt 2402.18646}}].

\bibitem{GarciaEtxebarria:2024fuk}
I.~n. Garc\'\i{}a~Etxebarria and S.~S. Hosseini, \emph{{Some aspects of symmetry descent}}, \href{http://dx.doi.org/10.1007/JHEP12(2024)223}{\emph{JHEP} {\bf 12} (2025) 223}, [\href{http://arxiv.org/abs/2404.16028}{{\tt 2404.16028}}].

\bibitem{Franco:2024mxa}
S.~Franco and X.~Yu, \emph{{Generalized symmetries in 2D from string theory: SymTFTs, intrinsic relativeness, and anomalies of non-invertible symmetries}}, \href{http://dx.doi.org/10.1007/JHEP11(2024)004}{\emph{JHEP} {\bf 11} (2024) 004}, [\href{http://arxiv.org/abs/2404.19761}{{\tt 2404.19761}}].

\bibitem{Cvetic:2024dzu}
M.~Cveti\v{c}, R.~Donagi, J.~J. Heckman, M.~H\"ubner and E.~Torres, \emph{{Cornering Relative Symmetry Theories}},  \href{http://arxiv.org/abs/2408.12600}{{\tt 2408.12600}}.

\bibitem{Tian:2024dgl}
J.~Tian and Y.-N. Wang, \emph{{A Tale of Bulk and Branes: Symmetry TFT of 6D SCFTs from IIB/F-theory}},  \href{http://arxiv.org/abs/2410.23076}{{\tt 2410.23076}}.

\bibitem{Cvetic:2024mtt}
M.~Cveti\v{c}, M.~Dierigl, L.~Lin, E.~Torres and H.~Y. Zhang, \emph{{Frozen generalized symmetries}}, \href{http://dx.doi.org/10.1103/PhysRevD.111.026018}{\emph{Phys. Rev. D} {\bf 111} (2025) 026018}, [\href{http://arxiv.org/abs/2410.07318}{{\tt 2410.07318}}].

\bibitem{Gagliano:2024off}
F.~Gagliano and I.~n. Garc\'\i{}a~Etxebarria, \emph{{SymTFTs for $U(1)$ symmetries from descent}},  \href{http://arxiv.org/abs/2411.15126}{{\tt 2411.15126}}.

\bibitem{Gukov:2020btk}
S.~Gukov, P.-S. Hsin and D.~Pei, \emph{{Generalized global symmetries of $T[M]$ theories. Part I}}, \href{http://dx.doi.org/10.1007/JHEP04(2021)232}{\emph{JHEP} {\bf 04} (2021) 232}, [\href{http://arxiv.org/abs/2010.15890}{{\tt 2010.15890}}].

\bibitem{Bashmakov:2022jtl}
V.~Bashmakov, M.~Del~Zotto and A.~Hasan, \emph{{On the 6d origin of non-invertible symmetries in 4d}}, \href{http://dx.doi.org/10.1007/JHEP09(2023)161}{\emph{JHEP} {\bf 09} (2023) 161}, [\href{http://arxiv.org/abs/2206.07073}{{\tt 2206.07073}}].

\bibitem{Antinucci:2022cdi}
A.~Antinucci, C.~Copetti, G.~Galati and G.~Rizi, \emph{{\textquotedblleft{}Zoology\textquotedblright{} of non-invertible duality defects: the view from class $ \mathcal{S} $}}, \href{http://dx.doi.org/10.1007/JHEP04(2024)036}{\emph{JHEP} {\bf 04} (2024) 036}, [\href{http://arxiv.org/abs/2212.09549}{{\tt 2212.09549}}].

\bibitem{Chen:2023qnv}
J.~Chen, W.~Cui, B.~Haghighat and Y.-N. Wang, \emph{{SymTFTs and duality defects from 6d SCFTs on 4-manifolds}}, \href{http://dx.doi.org/10.1007/JHEP11(2023)208}{\emph{JHEP} {\bf 11} (2023) 208}, [\href{http://arxiv.org/abs/2305.09734}{{\tt 2305.09734}}].

\bibitem{Bashmakov:2023kwo}
V.~Bashmakov, M.~Del~Zotto and A.~Hasan, \emph{{Four-manifolds and Symmetry Categories of 2d CFTs}},  \href{http://arxiv.org/abs/2305.10422}{{\tt 2305.10422}}.

\bibitem{Cui:2024cav}
W.~Cui, B.~Haghighat and L.~Ruggeri, \emph{{Non-invertible surface defects in 2+1d QFTs from half spacetime gauging}}, \href{http://dx.doi.org/10.1007/JHEP11(2024)159}{\emph{JHEP} {\bf 11} (2024) 159}, [\href{http://arxiv.org/abs/2406.09261}{{\tt 2406.09261}}].

\bibitem{Chen:2024fno}
J.~Chen, W.~Cui, B.~Haghighat and Y.~Sun, \emph{{Modularity of Vafa-Witten Partition Functions from SymTFT}},  \href{http://arxiv.org/abs/2409.19397}{{\tt 2409.19397}}.

\bibitem{Cordova:2019jnf}
C.~C\'ordova, D.~S. Freed, H.~T. Lam and N.~Seiberg, \emph{{Anomalies in the Space of Coupling Constants and Their Dynamical Applications I}}, \href{http://dx.doi.org/10.21468/SciPostPhys.8.1.001}{\emph{SciPost Phys.} {\bf 8} (2020) 001}, [\href{http://arxiv.org/abs/1905.09315}{{\tt 1905.09315}}].

\bibitem{Cordova:2019uob}
C.~C\'ordova, D.~S. Freed, H.~T. Lam and N.~Seiberg, \emph{{Anomalies in the Space of Coupling Constants and Their Dynamical Applications II}}, \href{http://dx.doi.org/10.21468/SciPostPhys.8.1.002}{\emph{SciPost Phys.} {\bf 8} (2020) 002}, [\href{http://arxiv.org/abs/1905.13361}{{\tt 1905.13361}}].

\bibitem{Brennan:2020ehu}
T.~D. Brennan and C.~Cordova, \emph{{Axions, higher-groups, and emergent symmetry}}, \href{http://dx.doi.org/10.1007/JHEP02(2022)145}{\emph{JHEP} {\bf 02} (2022) 145}, [\href{http://arxiv.org/abs/2011.09600}{{\tt 2011.09600}}].

\bibitem{Damia:2022seq}
J.~Aguilera~Damia, R.~Argurio and L.~Tizzano, \emph{{Continuous Generalized Symmetries in Three Dimensions}}, \href{http://dx.doi.org/10.1007/JHEP05(2023)164}{\emph{JHEP} {\bf 23} (2023) 164}, [\href{http://arxiv.org/abs/2206.14093}{{\tt 2206.14093}}].

\bibitem{Aloni:2024jpb}
D.~Aloni, E.~Garc\'\i{}a-Valdecasas, M.~Reece and M.~Suzuki, \emph{{Spontaneously broken (-1)-form U(1) symmetries}}, \href{http://dx.doi.org/10.21468/SciPostPhys.17.2.031}{\emph{SciPost Phys.} {\bf 17} (2024) 031}, [\href{http://arxiv.org/abs/2402.00117}{{\tt 2402.00117}}].

\bibitem{Garcia-Valdecasas:2024cqn}
E.~Garc\'\i{}a-Valdecasas, M.~Reece and M.~Suzuki, \emph{{Monopole Breaking of Chern-Weil Symmetries}},  \href{http://arxiv.org/abs/2408.00067}{{\tt 2408.00067}}.

\bibitem{Brennan:2024tlw}
T.~D. Brennan, \emph{{Constraints on symmetry-preserving gapped phases from coupling constant anomalies}}, \href{http://dx.doi.org/10.1103/PhysRevD.110.L041701}{\emph{Phys. Rev. D} {\bf 110} (2024) L041701}, [\href{http://arxiv.org/abs/2404.11660}{{\tt 2404.11660}}].

\bibitem{McNamara:2020uza}
J.~McNamara and C.~Vafa, \emph{{Baby Universes, Holography, and the Swampland}},  \href{http://arxiv.org/abs/2004.06738}{{\tt 2004.06738}}.

\bibitem{Heckman:2024oot}
J.~J. Heckman, M.~H\"ubner and C.~Murdia, \emph{{On the holographic dual of a topological symmetry operator}}, \href{http://dx.doi.org/10.1103/PhysRevD.110.046007}{\emph{Phys. Rev. D} {\bf 110} (2024) 046007}, [\href{http://arxiv.org/abs/2401.09538}{{\tt 2401.09538}}].

\bibitem{Yu:2020twi}
M.~Yu, \emph{{Symmetries and anomalies of (1+1)d theories: 2-groups and symmetry fractionalization}}, \href{http://dx.doi.org/10.1007/JHEP08(2021)061}{\emph{JHEP} {\bf 08} (2021) 061}, [\href{http://arxiv.org/abs/2010.01136}{{\tt 2010.01136}}].

\bibitem{Santilli:2024dyz}
L.~Santilli and R.~J. Szabo, \emph{{Higher form symmetries and orbifolds of two-dimensional Yang\textendash{}Mills theory}}, \href{http://dx.doi.org/10.1007/s11005-025-01905-4}{\emph{Lett. Math. Phys.} {\bf 115} (2025) 15}, [\href{http://arxiv.org/abs/2403.03119}{{\tt 2403.03119}}].

\bibitem{Vafa:1986wx}
C.~Vafa, \emph{{Modular Invariance and Discrete Torsion on Orbifolds}}, \href{http://dx.doi.org/10.1016/0550-3213(86)90379-2}{\emph{Nucl. Phys. B} {\bf 273} (1986) 592--606}.

\bibitem{Hellerman:2006zs}
S.~Hellerman, A.~Henriques, T.~Pantev, E.~Sharpe and M.~Ando, \emph{{Cluster decomposition, T-duality, and gerby CFT's}}, \href{http://dx.doi.org/10.4310/ATMP.2007.v11.n5.a2}{\emph{Adv. Theor. Math. Phys.} {\bf 11} (2007) 751--818}, [\href{http://arxiv.org/abs/hep-th/0606034}{{\tt hep-th/0606034}}].

\bibitem{Sharpe:2014tca}
E.~Sharpe, \emph{{Decomposition in diverse dimensions}}, \href{http://dx.doi.org/10.1103/PhysRevD.90.025030}{\emph{Phys. Rev. D} {\bf 90} (2014) 025030}, [\href{http://arxiv.org/abs/1404.3986}{{\tt 1404.3986}}].

\bibitem{Sharpe:2019ddn}
E.~Sharpe, \emph{{Undoing decomposition}}, \href{http://dx.doi.org/10.1142/S0217751X19502336}{\emph{Int. J. Mod. Phys. A} {\bf 34} (2020) 1950233}, [\href{http://arxiv.org/abs/1911.05080}{{\tt 1911.05080}}].

\bibitem{Robbins:2020msp}
D.~Robbins, E.~Sharpe and T.~Vandermeulen, \emph{{A generalization of decomposition in orbifolds}}, \href{http://dx.doi.org/10.1007/JHEP10(2021)134}{\emph{JHEP} {\bf 21} (2020) 134}, [\href{http://arxiv.org/abs/2101.11619}{{\tt 2101.11619}}].

\bibitem{Sharpe:2021srf}
E.~Sharpe, \emph{{Topological operators, noninvertible symmetries and decomposition}}, \href{http://dx.doi.org/10.4310/ATMP.2023.v27.n8.a2}{\emph{Adv. Theor. Math. Phys.} {\bf 27} (2023) 2319--2407}, [\href{http://arxiv.org/abs/2108.13423}{{\tt 2108.13423}}].

\bibitem{Pantev:2022kpl}
T.~Pantev, D.~G. Robbins, E.~Sharpe and T.~Vandermeulen, \emph{{Orbifolds by 2-groups and decomposition}}, \href{http://dx.doi.org/10.1007/JHEP09(2022)036}{\emph{JHEP} {\bf 09} (2022) 036}, [\href{http://arxiv.org/abs/2204.13708}{{\tt 2204.13708}}].

\bibitem{Sharpe:2022ene}
E.~Sharpe, \emph{{An introduction to decomposition}},  in \emph{2021-2022 MATRIX Annals}, ch.~8, p.~145–168.
\newblock Springer Nature Switzerland, 2024.
\newblock \href{http://arxiv.org/abs/2204.09117}{{\tt 2204.09117}}.
\newblock \href{http://dx.doi.org/10.1007/978-3-031-47417-0}{DOI}.

\bibitem{DelZotto:2015isa}
M.~Del~Zotto, J.~J. Heckman, D.~S. Park and T.~Rudelius, \emph{{On the Defect Group of a 6D SCFT}}, \href{http://dx.doi.org/10.1007/s11005-016-0839-5}{\emph{Lett. Math. Phys.} {\bf 106} (2016) 765--786}, [\href{http://arxiv.org/abs/1503.04806}{{\tt 1503.04806}}].

\bibitem{Albertini:2020mdx}
F.~Albertini, M.~Del~Zotto, I.~n. Garc\'\i{}a~Etxebarria and S.~S. Hosseini, \emph{{Higher Form Symmetries and M-theory}}, \href{http://dx.doi.org/10.1007/JHEP12(2020)203}{\emph{JHEP} {\bf 12} (2020) 203}, [\href{http://arxiv.org/abs/2005.12831}{{\tt 2005.12831}}].

\bibitem{Yu:2024jtk}
X.~Yu, \emph{{Gauging in Parameter Space: A Top-Down Perspective}},  \href{http://arxiv.org/abs/2411.14997}{{\tt 2411.14997}}.

\bibitem{Brennan:2024fgj}
T.~D. Brennan and Z.~Sun, \emph{{A SymTFT for continuous symmetries}}, \href{http://dx.doi.org/10.1007/JHEP12(2024)100}{\emph{JHEP} {\bf 12} (2024) 100}, [\href{http://arxiv.org/abs/2401.06128}{{\tt 2401.06128}}].

\bibitem{Antinucci:2024zjp}
A.~Antinucci and F.~Benini, \emph{{Anomalies and gauging of U(1) symmetries}}, \href{http://dx.doi.org/10.1103/PhysRevB.111.024110}{\emph{Phys. Rev. B} {\bf 111} (2025) 024110}, [\href{http://arxiv.org/abs/2401.10165}{{\tt 2401.10165}}].

\bibitem{Apruzzi:2024htg}
F.~Apruzzi, F.~Bedogna and N.~Dondi, \emph{{SymTh for non-finite symmetries}},  \href{http://arxiv.org/abs/2402.14813}{{\tt 2402.14813}}.

\bibitem{Heckman:2022muc}
J.~J. Heckman, M.~H\"ubner, E.~Torres and H.~Y. Zhang, \emph{{The Branes Behind Generalized Symmetry Operators}}, \href{http://dx.doi.org/10.1002/prop.202200180}{\emph{Fortsch. Phys.} {\bf 71} (2023) 2200180}, [\href{http://arxiv.org/abs/2209.03343}{{\tt 2209.03343}}].

\bibitem{Cvetic:2023plv}
M.~Cveti\v{c}, J.~J. Heckman, M.~H\"ubner and E.~Torres, \emph{{Fluxbranes, generalized symmetries, and Verlinde\textquoteright{}s metastable monopole}}, \href{http://dx.doi.org/10.1103/PhysRevD.109.046007}{\emph{Phys. Rev. D} {\bf 109} (2024) 046007}, [\href{http://arxiv.org/abs/2305.09665}{{\tt 2305.09665}}].

\bibitem{Bergman:2024aly}
O.~Bergman, E.~Garcia-Valdecasas, F.~Mignosa and D.~Rodriguez-Gomez, \emph{{Non-BPS branes and continuous symmetries}},  \href{http://arxiv.org/abs/2407.00773}{{\tt 2407.00773}}.

\bibitem{Sharpe:2000qt}
E.~R. Sharpe, \emph{{Analogues of discrete torsion for the M theory three form}}, \href{http://dx.doi.org/10.1103/PhysRevD.68.126004}{\emph{Phys. Rev. D} {\bf 68} (2003) 126004}, [\href{http://arxiv.org/abs/hep-th/0008170}{{\tt hep-th/0008170}}].

\bibitem{10.1215/S0012-7094-89-05839-0}
R.~L. Bryant and S.~M. Salamon, \emph{{On the construction of some complete metrics with exceptional holonomy}}, \href{http://dx.doi.org/10.1215/S0012-7094-89-05839-0}{\emph{Duke Mathematical Journal} {\bf 58} (1989) 829 -- 850}.

\bibitem{Morrison:2020ool}
D.~R. Morrison, S.~Schafer-Nameki and B.~Willett, \emph{{Higher-Form Symmetries in 5d}}, \href{http://dx.doi.org/10.1007/JHEP09(2020)024}{\emph{JHEP} {\bf 09} (2020) 024}, [\href{http://arxiv.org/abs/2005.12296}{{\tt 2005.12296}}].

\bibitem{Bhardwaj:2020phs}
L.~Bhardwaj and S.~Sch\"afer-Nameki, \emph{{Higher-form symmetries of 6d and 5d theories}}, \href{http://dx.doi.org/10.1007/JHEP02(2021)159}{\emph{JHEP} {\bf 02} (2021) 159}, [\href{http://arxiv.org/abs/2008.09600}{{\tt 2008.09600}}].

\bibitem{Bhardwaj:2020ruf}
L.~Bhardwaj, \emph{{Flavor symmetry of 5d SCFTs. Part I. General setup}}, \href{http://dx.doi.org/10.1007/JHEP09(2021)186}{\emph{JHEP} {\bf 09} (2021) 186}, [\href{http://arxiv.org/abs/2010.13230}{{\tt 2010.13230}}].

\bibitem{Bhardwaj:2020avz}
L.~Bhardwaj, \emph{{Flavor symmetry of 5$d$ SCFTs. Part II. Applications}}, \href{http://dx.doi.org/10.1007/JHEP04(2021)221}{\emph{JHEP} {\bf 04} (2021) 221}, [\href{http://arxiv.org/abs/2010.13235}{{\tt 2010.13235}}].

\bibitem{Apruzzi:2021vcu}
F.~Apruzzi, S.~Schafer-Nameki, L.~Bhardwaj and J.~Oh, \emph{{The Global Form of Flavor Symmetries and 2-Group Symmetries in 5d SCFTs}}, \href{http://dx.doi.org/10.21468/SciPostPhys.13.2.024}{\emph{SciPost Phys.} {\bf 13} (2022) 024}, [\href{http://arxiv.org/abs/2105.08724}{{\tt 2105.08724}}].

\bibitem{Genolini:2022mpi}
P.~B. Genolini and L.~Tizzano, \emph{{Comments on Global Symmetries and Anomalies of 5d SCFTs}}, \href{http://dx.doi.org/10.1007/s00220-024-05139-8}{\emph{Commun. Math. Phys.} {\bf 405} (2024) 255}, [\href{http://arxiv.org/abs/2201.02190}{{\tt 2201.02190}}].

\bibitem{Brandhuber:2001yi}
A.~Brandhuber, J.~Gomis, S.~S. Gubser and S.~Gukov, \emph{{Gauge theory at large N and new G(2) holonomy metrics}}, \href{http://dx.doi.org/10.1016/S0550-3213(01)00340-6}{\emph{Nucl. Phys. B} {\bf 611} (2001) 179--204}, [\href{http://arxiv.org/abs/hep-th/0106034}{{\tt hep-th/0106034}}].

\bibitem{Cordova:2018cvg}
C.~C\'ordova, T.~T. Dumitrescu and K.~Intriligator, \emph{{Exploring 2-Group Global Symmetries}}, \href{http://dx.doi.org/10.1007/JHEP02(2019)184}{\emph{JHEP} {\bf 02} (2019) 184}, [\href{http://arxiv.org/abs/1802.04790}{{\tt 1802.04790}}].

\bibitem{Benini:2018reh}
F.~Benini, C.~C\'ordova and P.-S. Hsin, \emph{{On 2-Group Global Symmetries and their Anomalies}}, \href{http://dx.doi.org/10.1007/JHEP03(2019)118}{\emph{JHEP} {\bf 03} (2019) 118}, [\href{http://arxiv.org/abs/1803.09336}{{\tt 1803.09336}}].

\bibitem{Hidaka:2020izy}
Y.~Hidaka, M.~Nitta and R.~Yokokura, \emph{{Global 3-group symmetry and 't Hooft anomalies in axion electrodynamics}}, \href{http://dx.doi.org/10.1007/JHEP01(2021)173}{\emph{JHEP} {\bf 01} (2021) 173}, [\href{http://arxiv.org/abs/2009.14368}{{\tt 2009.14368}}].

\bibitem{Hidaka:2021kkf}
Y.~Hidaka, M.~Nitta and R.~Yokokura, \emph{{Global 4-group symmetry and \textquoteright{}t Hooft anomalies in topological axion electrodynamics}}, \href{http://dx.doi.org/10.1093/ptep/ptab150}{\emph{PTEP} {\bf 2022} (2022) 04A109}, [\href{http://arxiv.org/abs/2108.12564}{{\tt 2108.12564}}].

\bibitem{Bhardwaj:2022scy}
L.~Bhardwaj and D.~S.~W. Gould, \emph{{Disconnected 0-form and 2-group symmetries}}, \href{http://dx.doi.org/10.1007/JHEP07(2023)098}{\emph{JHEP} {\bf 07} (2023) 098}, [\href{http://arxiv.org/abs/2206.01287}{{\tt 2206.01287}}].

\bibitem{Copetti:2023mcq}
C.~Copetti, M.~Del~Zotto, K.~Ohmori and Y.~Wang, \emph{{Higher Structure of Chiral Symmetry}},  \href{http://arxiv.org/abs/2305.18282}{{\tt 2305.18282}}.

\bibitem{Kang:2023uvm}
M.~J. Kang and S.~Kang, \emph{{Central extensions of higher groups: Green-Schwarz mechanism and 2-connections}},  \href{http://arxiv.org/abs/2311.14666}{{\tt 2311.14666}}.

\bibitem{Liu:2024znj}
R.~Liu, R.~Luo and Y.-N. Wang, \emph{{Higher-Matter and Landau-Ginzburg Theory of Higher-Group Symmetries}},  \href{http://arxiv.org/abs/2406.03974}{{\tt 2406.03974}}.

\bibitem{Tanizaki:2019rbk}
Y.~Tanizaki and M.~\"Unsal, \emph{{Modified instanton sum in QCD and higher-groups}}, \href{http://dx.doi.org/10.1007/JHEP03(2020)123}{\emph{JHEP} {\bf 03} (2020) 123}, [\href{http://arxiv.org/abs/1912.01033}{{\tt 1912.01033}}].

\bibitem{Seiberg:2010qd}
N.~Seiberg, \emph{{Modifying the Sum Over Topological Sectors and Constraints on Supergravity}}, \href{http://dx.doi.org/10.1007/JHEP07(2010)070}{\emph{JHEP} {\bf 07} (2010) 070}, [\href{http://arxiv.org/abs/1005.0002}{{\tt 1005.0002}}].

\bibitem{Hopkins:2002rd}
M.~J. Hopkins and I.~M. Singer, \emph{{Quadratic functions in geometry, topology, and M theory}}, {\emph{J. Diff. Geom.} {\bf 70} (2005) 329--452}, [\href{http://arxiv.org/abs/math/0211216}{{\tt math/0211216}}].

\bibitem{Cheeger:1985}
J.~Cheeger and J.~Simons, \emph{Differential characters and geometric invariants},  in \emph{{Geometry and topology (College Park, Md., 1983/84)}}, vol.~1167 of \emph{Letc. Notes Math.}, pp.~50--80.
\newblock Springer, 1985.

\bibitem{Bar:2014}
C.~Bar and C.~Becker, \emph{{Differential Characters}}, vol.~2112 of \emph{Letc. Notes Phys.}
\newblock Springer, 2014, \href{http://dx.doi.org/10.1007/978-3-319-07034-6}{10.1007/978-3-319-07034-6}.

\bibitem{Freed:2006yc}
D.~S. Freed, G.~W. Moore and G.~Segal, \emph{{Heisenberg Groups and Noncommutative Fluxes}}, \href{http://dx.doi.org/10.1016/j.aop.2006.07.014}{\emph{Annals Phys.} {\bf 322} (2007) 236--285}, [\href{http://arxiv.org/abs/hep-th/0605200}{{\tt hep-th/0605200}}].

\bibitem{Szabo:2012hc}
R.~J. Szabo, \emph{{Quantization of Higher Abelian Gauge Theory in Generalized Differential Cohomology}}, \href{http://dx.doi.org/10.22323/1.175.0009}{\emph{PoS} {\bf ICMP2012} (2012) 009}, [\href{http://arxiv.org/abs/1209.2530}{{\tt 1209.2530}}].

\bibitem{Gutperle:2001mb}
M.~Gutperle and A.~Strominger, \emph{{Fluxbranes in string theory}}, \href{http://dx.doi.org/10.1088/1126-6708/2001/06/035}{\emph{JHEP} {\bf 06} (2001) 035}, [\href{http://arxiv.org/abs/hep-th/0104136}{{\tt hep-th/0104136}}].

\bibitem{Page:1983mke}
D.~N. Page, \emph{{Classical Stability of Round and Squashed Seven Spheres in Eleven-dimensional Supergravity}}, \href{http://dx.doi.org/10.1103/PhysRevD.28.2976}{\emph{Phys. Rev. D} {\bf 28} (1983) 2976}.

\bibitem{Howe:1997vn}
P.~S. Howe, E.~Sezgin and P.~C. West, \emph{{The Six-dimensional selfdual tensor}}, \href{http://dx.doi.org/10.1016/S0370-2693(97)00365-1}{\emph{Phys. Lett. B} {\bf 400} (1997) 255--259}, [\href{http://arxiv.org/abs/hep-th/9702111}{{\tt hep-th/9702111}}].

\bibitem{Bandos:1997ui}
I.~A. Bandos, K.~Lechner, A.~Nurmagambetov, P.~Pasti, D.~P. Sorokin and M.~Tonin, \emph{{Covariant action for the superfive-brane of M theory}}, \href{http://dx.doi.org/10.1103/PhysRevLett.78.4332}{\emph{Phys. Rev. Lett.} {\bf 78} (1997) 4332--4334}, [\href{http://arxiv.org/abs/hep-th/9701149}{{\tt hep-th/9701149}}].

\bibitem{Intriligator:2000eq}
K.~A. Intriligator, \emph{{Anomaly matching and a Hopf-Wess-Zumino term in 6d, N=(2,0) field theories}}, \href{http://dx.doi.org/10.1016/S0550-3213(00)00148-6}{\emph{Nucl. Phys. B} {\bf 581} (2000) 257--273}, [\href{http://arxiv.org/abs/hep-th/0001205}{{\tt hep-th/0001205}}].

\bibitem{Pilch:2015vha}
K.~Pilch, A.~Tyukov and N.~P. Warner, \emph{{Flowing to Higher Dimensions: A New Strongly-Coupled Phase on M2 Branes}}, \href{http://dx.doi.org/10.1007/JHEP11(2015)170}{\emph{JHEP} {\bf 11} (2015) 170}, [\href{http://arxiv.org/abs/1506.01045}{{\tt 1506.01045}}].

\bibitem{Anabalon:2022fti}
A.~Anabal\'on, M.~Chamorro-Burgos and A.~Guarino, \emph{{Janus and Hades in M-theory}}, \href{http://dx.doi.org/10.1007/JHEP11(2022)150}{\emph{JHEP} {\bf 11} (2022) 150}, [\href{http://arxiv.org/abs/2207.09287}{{\tt 2207.09287}}].

\bibitem{Fiorenza:2019ain}
D.~Fiorenza, H.~Sati and U.~Schreiber, \emph{{Twisted Cohomotopy implies level quantization of the full 6d Wess-Zumino term of the M5-brane}}, \href{http://dx.doi.org/10.1007/s00220-021-03951-0}{\emph{Commun. Math. Phys.} {\bf 384} (2021) 403--432}, [\href{http://arxiv.org/abs/1906.07417}{{\tt 1906.07417}}].

\bibitem{Moore:2004jv}
G.~W. Moore, \emph{{Anomalies, Gauss laws, and Page charges in M-theory}}, \href{http://dx.doi.org/10.1016/j.crhy.2004.12.005}{\emph{Comptes Rendus Physique} {\bf 6} (2005) 251--259}, [\href{http://arxiv.org/abs/hep-th/0409158}{{\tt hep-th/0409158}}].

\bibitem{Bonetti:2024cjk}
F.~Bonetti, M.~Del~Zotto and R.~Minasian, \emph{{SymTFTs for Continuous non-Abelian Symmetries}},  \href{http://arxiv.org/abs/2402.12347}{{\tt 2402.12347}}.

\bibitem{Wall:1966rcd}
C.~T.~C. Wall, \emph{{Classification problems in differential topology. V}}, \href{http://dx.doi.org/10.1007/BF01389738}{\emph{Invent. Math.} {\bf 1} (1966) 355--374}.

\bibitem{Kapustin:2014gua}
A.~Kapustin and N.~Seiberg, \emph{{Coupling a QFT to a TQFT and Duality}}, \href{http://dx.doi.org/10.1007/JHEP04(2014)001}{\emph{JHEP} {\bf 04} (2014) 001}, [\href{http://arxiv.org/abs/1401.0740}{{\tt 1401.0740}}].

\bibitem{Bhardwaj:2022kot}
L.~Bhardwaj, S.~Schafer-Nameki and A.~Tiwari, \emph{{Unifying constructions of non-invertible symmetries}}, \href{http://dx.doi.org/10.21468/SciPostPhys.15.3.122}{\emph{SciPost Phys.} {\bf 15} (2023) 122}, [\href{http://arxiv.org/abs/2212.06159}{{\tt 2212.06159}}].

\bibitem{Banerjee:2018syt}
S.~Banerjee, P.~Longhi and M.~Romo, \emph{{Exploring 5d BPS Spectra with Exponential Networks}}, \href{http://dx.doi.org/10.1007/s00023-019-00851-x}{\emph{Annales Henri Poincare} {\bf 20} (2019) 4055--4162}, [\href{http://arxiv.org/abs/1811.02875}{{\tt 1811.02875}}].

\bibitem{Banerjee:2019apt}
S.~Banerjee, P.~Longhi and M.~Romo, \emph{{Exponential BPS Graphs and D Brane Counting on Toric Calabi-Yau Threefolds: Part I}}, \href{http://dx.doi.org/10.1007/s00220-021-04242-4}{\emph{Commun. Math. Phys.} {\bf 388} (2021) 893--945}, [\href{http://arxiv.org/abs/1910.05296}{{\tt 1910.05296}}].

\bibitem{Banerjee:2020moh}
S.~Banerjee, P.~Longhi and M.~Romo, \emph{{Exponential BPS graphs and D-brane counting on toric Calabi-Yau threefolds: Part II}},  \href{http://arxiv.org/abs/2012.09769}{{\tt 2012.09769}}.

\bibitem{Santilli:2023fuh}
L.~Santilli and C.~F. Uhlemann, \emph{{3d defects in 5d: RG flows and defect F-maximization}}, \href{http://dx.doi.org/10.1007/JHEP06(2023)136}{\emph{JHEP} {\bf 06} (2023) 136}, [\href{http://arxiv.org/abs/2305.01004}{{\tt 2305.01004}}].

\bibitem{BenettiGenolini:2020doj}
P.~Benetti~Genolini and L.~Tizzano, \emph{{Instantons, symmetries and anomalies in five dimensions}}, \href{http://dx.doi.org/10.1007/JHEP04(2021)188}{\emph{JHEP} {\bf 04} (2021) 188}, [\href{http://arxiv.org/abs/2009.07873}{{\tt 2009.07873}}].

\bibitem{DelZotto:2022fnw}
M.~Del~Zotto, J.~J. Heckman, S.~N. Meynet, R.~Moscrop and H.~Y. Zhang, \emph{{Higher symmetries of 5D orbifold SCFTs}}, \href{http://dx.doi.org/10.1103/PhysRevD.106.046010}{\emph{Phys. Rev. D} {\bf 106} (2022) 046010}, [\href{http://arxiv.org/abs/2201.08372}{{\tt 2201.08372}}].

\bibitem{yau1993gorenstein}
S.-T. Yau, S.~S.-T. Yau and Y.~Yu, \emph{Gorenstein quotient singularities in dimension three}, vol.~505.
\newblock American Mathematical Soc., 1993.

\bibitem{Armstrong1968}
M.~A. Armstrong, \emph{The fundamental group of the orbit space of a discontinuous group}, \href{http://dx.doi.org/10.1017/S0305004100042845}{\emph{Mathematical Proceedings of the Cambridge Philosophical Society} {\bf 64} (1968) 299 -- 301}.

\bibitem{Najjar:2022eci}
M.~A.~M. Najjar, \emph{{Field Theory Dynamics from M-theory on Special Holonomy Manifolds}}.
\newblock PhD thesis, King's Coll. London, 2022.

\bibitem{bazaikin2013complete}
Y.~V. Bazaikin and O.~A. Bogoyavlenskaya, \emph{{Complete Riemannian G2 Holonomy Metrics on Deformations of Cones over S3 x S3 }},  \href{http://arxiv.org/abs/1301.6379}{{\tt 1301.6379}}.

\bibitem{Foscolo:2018mfs}
L.~Foscolo, M.~Haskins and J.~Nordstr\"om, \emph{{Infinitely many new families of complete cohomogeneity one G$_2$-manifolds: G$_2$ analogues of the Taub\textendash{}NUT and Eguchi\textendash{}Hanson spaces}}, \href{http://dx.doi.org/10.4171/jems/1051}{\emph{J. Eur. Math. Soc.} {\bf 23} (2021) 2153--2220}, [\href{http://arxiv.org/abs/1805.02612}{{\tt 1805.02612}}].

\bibitem{Friedmann:2002ct}
T.~Friedmann, \emph{{On the quantum moduli space of M theory compactifications}}, \href{http://dx.doi.org/10.1016/S0550-3213(02)00408-X}{\emph{Nucl. Phys. B} {\bf 635} (2002) 384--394}, [\href{http://arxiv.org/abs/hep-th/0203256}{{\tt hep-th/0203256}}].

\bibitem{Friedmann:2012uf}
T.~Friedmann and R.~P. Stanley, \emph{{The String Landscape: On Formulas for Counting Vacua}}, \href{http://dx.doi.org/10.1016/j.nuclphysb.2012.11.019}{\emph{Nucl. Phys. B} {\bf 869} (2013) 74--88}, [\href{http://arxiv.org/abs/1212.0583}{{\tt 1212.0583}}].

\bibitem{cortes2015locally}
V.~Cort{\'e}s and J.~J. V{\'a}squez, \emph{{Locally homogeneous nearly K{\"a}hler manifolds}}, \href{http://dx.doi.org/10.1007/s10455-015-9470-4}{\emph{Annals of Global Analysis and Geometry} {\bf 48} (2015) 269--294}, [\href{http://arxiv.org/abs/1410.6912}{{\tt 1410.6912}}].

\bibitem{Cachazo:2002zk}
F.~Cachazo, N.~Seiberg and E.~Witten, \emph{{Phases of N=1 supersymmetric gauge theories and matrices}}, \href{http://dx.doi.org/10.1088/1126-6708/2003/02/042}{\emph{JHEP} {\bf 02} (2003) 042}, [\href{http://arxiv.org/abs/hep-th/0301006}{{\tt hep-th/0301006}}].

\bibitem{Hosomichi:2005ja}
K.~Hosomichi and D.~C. Page, \emph{{G(2) holonomy, mirror symmetry and phases of N=1 SYM}}, \href{http://dx.doi.org/10.1088/1126-6708/2005/05/041}{\emph{JHEP} {\bf 05} (2005) 041}, [\href{http://arxiv.org/abs/hep-th/0501195}{{\tt hep-th/0501195}}].

\bibitem{Davies:2011is}
R.~Davies, \emph{{Hyperconifold Transitions, Mirror Symmetry, and String Theory}}, \href{http://dx.doi.org/10.1016/j.nuclphysb.2011.04.010}{\emph{Nucl. Phys. B} {\bf 850} (2011) 214--231}, [\href{http://arxiv.org/abs/1102.1428}{{\tt 1102.1428}}].

\bibitem{Davies:2013pna}
R.~Davies, \emph{{Classification and Properties of Hyperconifold Singularities and Transitions}},  \href{http://arxiv.org/abs/1309.6778}{{\tt 1309.6778}}.

\bibitem{Gaiotto:2017yup}
D.~Gaiotto, A.~Kapustin, Z.~Komargodski and N.~Seiberg, \emph{{Theta, Time Reversal, and Temperature}}, \href{http://dx.doi.org/10.1007/JHEP05(2017)091}{\emph{JHEP} {\bf 05} (2017) 091}, [\href{http://arxiv.org/abs/1703.00501}{{\tt 1703.00501}}].

\bibitem{Gaiotto:2017tne}
D.~Gaiotto, Z.~Komargodski and N.~Seiberg, \emph{{Time-reversal breaking in QCD$_{4}$, walls, and dualities in 2 + 1 dimensions}}, \href{http://dx.doi.org/10.1007/JHEP01(2018)110}{\emph{JHEP} {\bf 01} (2018) 110}, [\href{http://arxiv.org/abs/1708.06806}{{\tt 1708.06806}}].

\bibitem{Pantev:2005rh}
T.~Pantev and E.~Sharpe, \emph{{Notes on gauging noneffective group actions}},  \href{http://arxiv.org/abs/hep-th/0502027}{{\tt hep-th/0502027}}.

\bibitem{Pantev:2005wj}
T.~Pantev and E.~Sharpe, \emph{{String compactifications on Calabi-Yau stacks}}, \href{http://dx.doi.org/10.1016/j.nuclphysb.2005.10.035}{\emph{Nucl. Phys. B} {\bf 733} (2006) 233--296}, [\href{http://arxiv.org/abs/hep-th/0502044}{{\tt hep-th/0502044}}].

\bibitem{Pantev:2005zs}
T.~Pantev and E.~Sharpe, \emph{{GLSM's for Gerbes (and other toric stacks)}}, \href{http://dx.doi.org/10.4310/ATMP.2006.v10.n1.a4}{\emph{Adv. Theor. Math. Phys.} {\bf 10} (2006) 77--121}, [\href{http://arxiv.org/abs/hep-th/0502053}{{\tt hep-th/0502053}}].

\bibitem{Bunke:2012rsi}
U.~Bunke, \emph{{Differential cohomology}},  \href{http://arxiv.org/abs/1208.3961}{{\tt 1208.3961}}.

\bibitem{Simons:2007}
J.~Simons and D.~Sullivan, \emph{Axiomatic characterization of ordinary differential cohomology}, \href{http://dx.doi.org/10.1112/jtopol/jtm006}{\emph{Journal of Topology} {\bf 1} (Oct., 2007) 45–56}, [\href{http://arxiv.org/abs/math/0701077}{{\tt math/0701077}}].

\bibitem{Cvetic:2021sxm}
M.~Cvetic, M.~Dierigl, L.~Lin and H.~Y. Zhang, \emph{{Higher-form symmetries and their anomalies in M-/F-theory duality}}, \href{http://dx.doi.org/10.1103/PhysRevD.104.126019}{\emph{Phys. Rev. D} {\bf 104} (2021) 126019}, [\href{http://arxiv.org/abs/2106.07654}{{\tt 2106.07654}}].

\bibitem{Schafer-Nameki:2023jdn}
S.~Schafer-Nameki, \emph{{ICTP lectures on (non-)invertible generalized symmetries}}, \href{http://dx.doi.org/10.1016/j.physrep.2024.01.007}{\emph{Phys. Rept.} {\bf 1063} (2024) 1--55}, [\href{http://arxiv.org/abs/2305.18296}{{\tt 2305.18296}}].

\bibitem{Brennan:2023mmt}
T.~D. Brennan and S.~Hong, \emph{{Introduction to Generalized Global Symmetries in QFT and Particle Physics}},  \href{http://arxiv.org/abs/2306.00912}{{\tt 2306.00912}}.

\bibitem{Luo:2023ive}
R.~Luo, Q.-R. Wang and Y.-N. Wang, \emph{{Lecture notes on generalized symmetries and applications}}, \href{http://dx.doi.org/10.1016/j.physrep.2024.02.002}{\emph{Phys. Rept.} {\bf 1065} (2024) 1--43}, [\href{http://arxiv.org/abs/2307.09215}{{\tt 2307.09215}}].

\bibitem{Bhardwaj:2023kri}
L.~Bhardwaj, L.~E. Bottini, L.~Fraser-Taliente, L.~Gladden, D.~S.~W. Gould, A.~Platschorre et~al., \emph{{Lectures on generalized symmetries}}, \href{http://dx.doi.org/10.1016/j.physrep.2023.11.002}{\emph{Phys. Rept.} {\bf 1051} (2024) 1--87}, [\href{http://arxiv.org/abs/2307.07547}{{\tt 2307.07547}}].

\bibitem{Shao:2023gho}
S.-H. Shao, \emph{{What's Done Cannot Be Undone: TASI Lectures on Non-Invertible Symmetries}},  \href{http://arxiv.org/abs/2308.00747}{{\tt 2308.00747}}.

\bibitem{Iqbal:2024pee}
N.~Iqbal, \emph{{Jena lectures on generalized global symmetries: principles and applications}},  \href{http://arxiv.org/abs/2407.20815}{{\tt 2407.20815}}.

\bibitem{Freed:2022qnc}
D.~S. Freed, G.~W. Moore and C.~Teleman, \emph{{Topological symmetry in quantum field theory}},  \href{http://arxiv.org/abs/2209.07471}{{\tt 2209.07471}}.

\bibitem{Birmingham:1991ty}
D.~Birmingham, M.~Blau, M.~Rakowski and G.~Thompson, \emph{{Topological field theory}}, \href{http://dx.doi.org/10.1016/0370-1573(91)90117-5}{\emph{Phys. Rept.} {\bf 209} (1991) 129--340}.

\bibitem{Cremmer:1978km}
E.~Cremmer, B.~Julia and J.~Scherk, \emph{{Supergravity Theory in 11 Dimensions}}, \href{http://dx.doi.org/10.1016/0370-2693(78)90894-8}{\emph{Phys. Lett. B} {\bf 76} (1978) 409--412}.

\bibitem{Bandos:1997gd}
I.~A. Bandos, N.~Berkovits and D.~P. Sorokin, \emph{{Duality symmetric eleven-dimensional supergravity and its coupling to M-branes}}, \href{http://dx.doi.org/10.1016/S0550-3213(98)00102-3}{\emph{Nucl. Phys. B} {\bf 522} (1998) 214--233}, [\href{http://arxiv.org/abs/hep-th/9711055}{{\tt hep-th/9711055}}].

\bibitem{Sorokin:1998kf}
D.~P. Sorokin, \emph{{Coupling of M-branes in M theory}},  in \emph{{6th International Symposium on Particles, Strings and Cosmology}}, pp.~697--701, 3, 1998.
\newblock \href{http://arxiv.org/abs/hep-th/9806175}{{\tt hep-th/9806175}}.

\bibitem{Dirac:1948um}
P.~A.~M. Dirac, \emph{{The Theory of magnetic poles}}, \href{http://dx.doi.org/10.1103/PhysRev.74.817}{\emph{Phys. Rev.} {\bf 74} (1948) 817--830}.

\end{thebibliography}\endgroup

\end{document}